\documentclass{article}

\makeatletter
\renewcommand\paragraph{\@startsection{paragraph}{4}{\z@}%
                                    {1ex \@plus1ex \@minus.2ex}%
                                    {-1em}%
                                    {\normalfont\normalsize\bfseries}}
\makeatother

\usepackage{amsmath}
\usepackage{amssymb}
\usepackage{latexsym}
\usepackage{verbatim}

\usepackage[final]{graphicx}
\usepackage{psfrag}
\usepackage{mathrsfs}  
\usepackage{amsthm}

\usepackage{fancybox}
\usepackage{enumitem}

\usepackage[export]{adjustbox}
\usepackage{caption,subcaption,multicol,soul}

\newtheorem{theorem}{Theorem}
\newtheorem{proposition}{Proposition}
\newtheorem{lemma}{Lemma}
\newtheorem{corollary}{Corollary}

\theoremstyle{definition}
\newtheorem{example}{Example}
\newtheorem{definition}{Definition}
\newtheorem{assumption}{Assumption}
\newtheorem{sassumption}{Simplifying Assumption}

\theoremstyle{remark}
\newtheorem{remark}{Remark}
\newtheorem{fact}{Fact}

{\begin{list}               
    {$\bullet$ \hfill}{
        \setlength{\leftmargin}{\parindent}
        \setlength{\parsep}{0.04\baselineskip}
        \setlength{\itemsep}{0.5\parsep}
        \setlength{\labelwidth}{\leftmargin}
        \setlength{\labelsep}{0em}}
    }
{\end{list}}

\providecommand{\eref}[1]{\eqref{eq:#1}}  
\providecommand{\cref}[1]{Chapter~\ref{chap:#1}}
\providecommand{\sref}[1]{Section~\ref{sec:#1}}
\providecommand{\appref}[1]{Appendix~\ref{appendix:#1}}
\providecommand{\fref}[1]{Figure~\ref{fig:#1}}

\providecommand{\thref}[1]{Theorem~\ref{thm:#1}}
\providecommand{\defref}[1]{Definition~\ref{def:#1}}
\providecommand{\lemref}[1]{Lemma~\ref{lem:#1}}
\providecommand{\remref}[1]{Remark~\ref{rem:#1}}
\providecommand{\assumpref}[1]{Assumption~\ref{assump:#1}}
\providecommand{\sassumpref}[1]{Simplifying Assumption~\ref{sassump:#1}}
\providecommand{\factref}[1]{Fact~\ref{fact:#1}}
\providecommand{\propref}[1]{Proposition~\ref{prop:#1}}
\providecommand{\corref}[1]{Corollary~\ref{cor:#1}}
\providecommand{\exref}[1]{Example~\ref{ex:#1}}

\providecommand{\R}{\ensuremath{\mathbb{R}}}
\providecommand{\C}{\ensuremath{\mathbb{C}}}
\providecommand{\N}{\ensuremath{\mathbb{N}}}

\providecommand{\W}{\ensuremath{\mathbb{N}_0}}

\providecommand{\abs}[1]{\lvert#1\rvert}
\providecommand{\norm}[1]{\lVert#1\rVert}

\providecommand{\set}[1]{\left\{#1\right\}}

\providecommand{\bydef}{\overset{\text{def}}{=}}
\renewcommand{\dim}{N}
\newcommand{\ssize}{M}
\newcommand{\reg}{\rho}
\newcommand{\uclass}[1]{\mathscr{U}(#1)}
\newcommand{\auxmat}{\mat{A}}
\newcommand{\auxvec}{{a}}
\newcommand{\auxdim}{b}
\newcommand{\order}{k}
\newcommand{\opt}{\mathrm{RLS}}

\renewcommand{\vec}[1]{\ensuremath{\boldsymbol{#1}}}
\providecommand{\mat}[1]{\ensuremath{\boldsymbol{#1}}}

\providecommand{\calS}{\mathcal{S}}

\providecommand{\calU}{\mathcal{U}}
\providecommand{\calV}{\mathcal{V}}

\providecommand{\mA}{\mat{A}} \providecommand{\mB}{\mat{B}}
 
\providecommand{\mD}{\mat{D}}
\providecommand{\mF}{\mat{F}}
\providecommand{\mH}{\mat{H}}
\providecommand{\mI}{\mat{I}} \providecommand{\mJ}{\mat{J}} 
  
\providecommand{\mM}{\mat{M}} \providecommand{\mP}{\mat{P}} 
\providecommand{\mQ}{\mat{Q}} \providecommand{\mR}{\mat{R}}
\providecommand{\mS}{\mat{S}} \providecommand{\mU}{\mat{U}} 
\providecommand{\mV}{\mat{V}}
\providecommand{\mW}{\mat{W}}
\providecommand{\mT}{\mat{T}}
\providecommand{\mZ}{\mat{Z}}
\providecommand{\mSigma}{\mat{\Sigma}}
 \providecommand{\mG}{\mat{G}}
\providecommand{\mPsi}{\mat{\Psi}}
\providecommand{\mPsi}{\mat{\Psi}}
\providecommand{\mX}{\mat{X}}

\providecommand{\mXi}{\mat{\Xi}}
\providecommand{\linmap}{\mathscr{M}}
\providecommand{\mO}{\mat{O}}

 \providecommand{\ve}{\vec{e}}

\providecommand{\vu}{\vec{u}} \providecommand{\vw}{\vec{w}}
\providecommand{\vx}{\vec{x}} \providecommand{\vy}{\vec{y}}
\providecommand{\vz}{\vec{z}} 
 \providecommand{\vzero}{\vec{0}}

\providecommand{\vv}{\vec{v}}

\providecommand{\vDelta}{\vec{\Delta}}
\providecommand{\valpha}{\vec{\alpha}}
\providecommand{\vbeta}{\vec{\beta}}
\providecommand{\vepsilon}{\vec{\epsilon}}

\providecommand{\vone}{\vec{1}}
\providecommand{\vgamma}{\vec{\gamma}}
\providecommand{\vdelta}{\vec{\delta}}
\newcommand{\auxC}{L}

\newcommand{\unif}[1]{\mathsf{Unif}(#1)} 
 
\newcommand{\ortho}{\mathbb{O}} 
\newcommand{\explain}[2]{\overset{\text{\tiny{#1}}}{#2}}
\newcommand{\iter}[2]{{#1}^{(#2)}} 
\newcommand{\E}{\mathbb{E}} 
\renewcommand{\P}{\mathbb{P}} 
\newcommand{\Var}{\mathrm{Var}}
\usepackage[sort&compress,numbers]{natbib} 
\newcommand{\ip}[2]{\left\langle {#1}, {#2} \right\rangle} 
\newcommand{\gauss}[2]{\mathcal{N}\left( #1,#2 \right)} 
\newcommand{\nonlin}{f}

\newcommand{\degree}{D}
\newcommand{\fourier}[3]{\hat{\nonlin}(#1,#2; #3)}
\newcommand{\pweight}[2]{p_{#2}(#1)}
\newcommand{\qweight}[2]{q_{#2}({#1})}

\newcommand{\height}[2][]{h_{#1}{(#2)}}
\newcommand{\valid}{\mathtt{VALID}}
\newcommand{\relev}{\mathtt{RELEV}}
\newcommand{\hermite}[1]{H_{#1}}
\newcommand{\trees}[2]{\mathscr{T}_{#1}(#2)}
\newcommand{\colorings}[2]{\mathscr{C}_{#1}(#2)}
\newcommand{\leaves}[1]{\mathscr{L}(#1)}

\renewcommand{\part}[1]{\mathscr{P}(#1)}
\newcommand{\chrn}[1]{c_{#1}}
\newcommand{\pw}{$\mathrm{PW}_2$}
\newcommand{\serv}[1]{\mathsf{#1}}
\newcommand{\nullleaves}[1]{\mathscr{L}_0(#1)}
\newcommand{\exponent}{\eta}
\newcommand{\noisestd}{\sigma}

\DeclareMathOperator*{\argmin}{\arg\min}

\makeatletter
\newcommand{\myitem}[1]{%
\item{#1} \protected@edef\@currentlabel{#1}%
}
\makeatother



\newenvironment{fminipage}%
  {\begin{Sbox}\begin{minipage}}%
  {\end{minipage}\end{Sbox}\fbox{\TheSbox}}

\newenvironment{algbox}[0]{\vskip 0.2in
\noindent 
\begin{fminipage}{6.3in}
}{
\end{fminipage}
\vskip 0.2in
}

\DeclareFontFamily{U}{mathx}{\hyphenchar\font45}
\DeclareFontShape{U}{mathx}{m}{n}{
      <5> <6> <7> <8> <9> <10>
      <10.95> <12> <14.4> <17.28> <20.74> <24.88>
      mathx10
      }{}
\DeclareSymbolFont{mathx}{U}{mathx}{m}{n}
\DeclareFontSubstitution{U}{mathx}{m}{n}
\DeclareMathAccent{\widecheck}{0}{mathx}{"71}
\DeclareMathAccent{\wideparen}{0}{mathx}{"75}

\renewcommand{\tilde}{\widetilde}
\renewcommand{\hat}{\widehat}
\providecommand{\asymeq}{\explain{\pw}{\simeq}}

\usepackage[margin=1in]{geometry}

\usepackage{cite}

\usepackage{color}

\usepackage{booktabs,multirow}

\usepackage{enumitem}
\usepackage{pmat}
\usepackage{commath}
\usepackage{mathtools,xcolor,authblk}

\usepackage{algorithm,algpseudocode}

\usepackage{hyperref}
 \hypersetup{
     colorlinks=true,
     linkcolor=blue,
     filecolor=blue,
     citecolor = blue,      
     urlcolor=cyan,
     }

\usepackage{esdiff}

\usepackage{pgf,tikz,psfrag}

\usepackage{dsfont}

\newcommand{\Ortho}{\mathbb{O}}

\newcommand{\mLambda}{\mat{\Lambda}}

\DeclareMathOperator{\diag}{diag}

\DeclareMathOperator{\Tr}{Tr}

\DeclareMathOperator{\op}{op}

\newcommand*{\tran}{^{\mkern-1.5mu\mathsf{T}}}
\newcommand*{\herm}{^\ast}

\providecommand{\stepsize}{\zeta}

\providecommand{\vbeta}{\vec{\beta}}

\newcommand{\revised}[1]{{#1}}
\renewcommand*\diff{\mathop{}\!\mathrm{d}}

\title{Spectral Universality of Regularized Linear Regression with\\ Nearly Deterministic Sensing Matrices}
\author[]{Rishabh Dudeja\thanks{rd2714@columbia.edu}}
\author[]{Subhabrata Sen\thanks{subhabratasen@fas.harvard.edu}}
\author[]{Yue M. Lu\thanks{yuelu@seas.harvard.edu}}
\affil[]{Harvard University}

\begin{document}
\maketitle

\begin{abstract}
It has been observed that the performances of many high-dimensional estimation problems are universal with respect to underlying sensing (or design) matrices. Specifically, matrices with markedly different constructions seem to achieve identical performance if they share the same spectral distribution and have ``generic'' singular vectors. We prove this universality phenomenon for the case of convex regularized least squares (RLS) estimators under a linear regression model with additive Gaussian noise. Our main contributions are two-fold: (1) We introduce a notion of  universality classes for sensing matrices, defined through a set of deterministic conditions that fix the spectrum of the sensing matrix and precisely capture the notion of generic singular vectors; (2) We show that for all sensing matrices that lie in the same universality class, the dynamics of the proximal gradient descent algorithm for solving the regression problem, as well as the performance of RLS estimators themselves (under additional strong convexity conditions) are asymptotically identical. In addition to including i.i.d. Gaussian and rotational invariant matrices as special cases, our universality class also contains highly structured, strongly correlated, and even (nearly) deterministic matrices. Examples of the latter include randomly signed versions of incoherent tight frames and randomly subsampled Hadamard transforms. As a consequence of this universality principle, the asymptotic performance of regularized linear regression on many structured matrices constructed with limited randomness can be characterized by using the rotationally invariant ensemble as an equivalent yet mathematically more tractable surrogate.
\end{abstract}

\tableofcontents

\section{Introduction}

A common theme in statistical signal processing and inference is to estimate a signal vector $\vbeta_\star \in \R^{\dim}$ from a set of noisy and potentially highly incomplete measurements $\vy \in \R^\ssize$. A fairly general model is
\begin{align}\label{eq:general-model}
\vy & = g(\mX \vbeta_\star, \vepsilon),
\end{align}
where $\mX \in \R^{\ssize \times \dim}$ is an observed $\ssize \times \dim$ sensing (or design) matrix with $\ssize$ being the sample size, $\vepsilon \in \R^\ssize$ is the unobserved noise, and $g(\cdot, \cdot): \R^2 \mapsto \R$ is some fixed function that acts on each coordinate of its input arguments. This model arises in many (regularized) regression problems, with examples including photon-limited imaging \citep{unser1988maximum,yang2011bits}, phase retrieval \citep{fienup1982phase}, MIMO detection \citep{thrampoulidis2018symbol,hu2020limiting} in wireless communications, signal recovery from quantized measurements \citep{rangan2001recursive}, and robust data fitting \citep{el2013robust,karoui2013asymptotic,donoho2016high,el2018impact}.

There is a long and very rich line of work on studying various estimators for model \eqref{eq:general-model} and its generalizations (see, \emph{e.g.}, \citep{dobson2018introduction} for an overview). In many cases, the cleanest expression for the performance of an estimator is given in the asymptotic regime, where the underlying dimension $N$ and the sample size $M$ are both large and comparable. Indeed, under additional statistical assumptions on the sensing matrix $\mX$, a growing body of work (see, \emph{e.g.},  \citep{donoho2009message,donoho2010counting,bayati2011lasso, bayati2011dynamics, chandrasekaran2012convex,karoui2013asymptotic,amelunxen2014living,thrampoulidis2015regularized, donoho2016high, dobriban2018high,weng2018overcoming,reeves2019replica, barbier2019optimal, sur2019likelihood,sur2019modern,ma2019optimization,candes2020phase,celentano2020lasso,mignacco2020role,lu2020phase,bu2020algorithmic,li2021minimum}) 
analyzes the properties of statistical estimators in the high-dimensional limit, predicting their exact asymptotic performance and often revealing interesting phase transition phenomena. The latter amount to an abrupt change in the performance of an estimator as certain parameters (such as the signal-to-noise ratio or the sampling ratio $\ssize/\dim$) cross critical thresholds. In addition, novel asymptotic null distributions of various classical hypothesis tests have also been characterized under this asymptotic regime, leading to principled, efficient inference in high-dimensions (see, \emph{e.g.}, \citep{bai2009corrections,jiang2012likelihood,jiang2013central,jiang2015likelihood,sur2019likelihood,he2021likelihood}).  Such asymptotic results are highly valuable, as they provide fundamental limits on the degree to which different inference methodology can be successful. Moreover, the precise asymptotic characterizations can also lead to optimal algorithm designs, as demonstrated in recent work \citep{bean2013optimal, hu2019asymptotics,luo2019optimal,celentano2020estimation,wang2020bridge,taheri2021fundamental,mondelli2021optimal,maillard2022construction}. 

Despite considerable recent progress, there remains a significant gap between theory and practice. On the theoretical end, most of the existing research relies upon strong and often unrealistic assumptions on the underlying sensing models. In particular, the sensing matrix $\mX$ in \eqref{eq:general-model} is usually assumed to consist of i.i.d. entries or have rotational-invariant properties. Such idealistic statistical models, while useful and convenient for mathematical proofs, do not resemble the actual systems encountered in practice. Indeed, the sensing matrices $\mX$ encountered in applications are usually structured, and often have strong correlations among the entries.  
In this paper, we seek to narrow this gap between theory and practice by precisely characterizing the statistical properties of estimators in high-dimensional settings where the underlying sensing matrices can be \emph{highly structured}, \emph{strongly correlated} in their components, and even \emph{(nearly) deterministic}. 

To this end, we will establish a \emph{universality principle} for large classes of sensing matrices.\ Broadly speaking, universality is the observation that there exist universal laws that govern the macroscopic behavior of many complex systems, regardless of what the microscopic components of those systems are, or how they interact with each other. In the context of high-dimensional estimation, it refers to the well-documented empirical observations (see, \emph{e.g.}, \citep{donoho2009observed,monajemi2013deterministic,oymak2014case,abbara2020universality}) that, for many estimators and iterative algorithms, certain structured (or even deterministic) sensing matrices seem to exactly match the theoretical performance derived under the i.i.d. Gaussian or rotational-invariant assumptions. In this paper, we first introduce a notion of \emph{universality classes} of sensing matrices, defined through a set of deterministic conditions on the matrices. Under a linear regression model with additive Gaussian noise, we show that the performance of \emph{regularized least squares} (RLS) estimators with general convex regularizers are asymptotically identical for all sensing matrices that lie in a given universality class.

\subsection{Model, RLS Estimators, and Assumptions}
This paper studies a linear version of the general observation model in \eqref{eq:general-model}, with 
\begin{equation}\label{eq:linear-model}
\vy = \mX \vbeta_\star + \vepsilon.
\end{equation}
We also impose the following assumptions on the signal vector $\vbeta_\star$, and the noise $\vepsilon$ in our analysis.

\begin{assumption}[Random Signal and Noise] \label{assump:RSN} The entries of $\vbeta_\star$ are i.i.d. copies of a random variable $\serv{B_\star}$ with finite moments of all orders, whose distribution is uniquely determined by its moments. The entries of the noise vector $\vepsilon$ are i.i.d. $\gauss{0}{\noisestd^2}$ for some $\noisestd \geq 0$. The signal $\vbeta_\star$ and the noise $\vepsilon$ are independent of the sensing matrix $\mX$. 
\end{assumption}

We study general Regularized Least Squares (RLS) estimators, defined as
\begin{subequations}\label{eq:RLS}
\begin{align} 
    \vbeta_\opt(\mX, \vbeta_\star, \vepsilon) \in \argmin_{\vbeta \in \R^{\dim}}  L(\vbeta; \mX, \vy) \label{eq:RLS:1}\\
    L(\vbeta; \mX, \vy) \explain{def}{=} \frac{1}{2\dim}  \|\vy - \mX \vbeta\|^2 + \frac{1}{\dim} \sum_{i=1}^\dim \reg(\beta_i),
\end{align}
\end{subequations}
where $\reg: \R \mapsto \R$ is a regularizer or the penalty function.

\begin{assumption}[Convex Regularizer]\label{assump:convex-reg} The regularizer $\rho: \R \mapsto \R$ is a proper, closed convex function which diverges at $\infty$ (that is, $\rho(x) \rightarrow \infty$ as $|x| \rightarrow \infty$). 
\end{assumption}

The convexity assumption on the regularizer $\rho$  guarantees that the optimization problem in \eqref{eq:RLS} has at least one global minimizer.
A key challenge in showing the universality of the performance of the RLS estimator with respect to the sensing matrix $\mX$ is that the RLS estimator in \eqref{eq:RLS} is an \emph{implicit function} of the matrix $\mX$, defined through a high-dimensional optimization problem. To overcome this challenge, we rely on an \emph{algorithmic approach} to establishing universality. Specifically, we construct a sequence of \emph{explicit approximations} to the implicitly defined estimator $ \vbeta_\opt(\mX, \vbeta_\star, \vepsilon)$ by tracking the iterates of the \emph{proximal gradient algorithm} (see \citep[Chapter 10]{beck2017}) for solving \eqref{eq:RLS}. 

Recall the definition of the proximal operator $\eta: \R \times (0,\infty) \mapsto \R$ associated with the regularizer $\rho$:
\begin{align}\label{eq:prox}
    \eta(x; \gamma) \explain{def}{=} \argmin_{z \in \R}\ \gamma\rho(z) + \frac{(x-z)^2}{2}. 
\end{align}
This operator is well-defined under \assumpref{convex-reg} (see \citep[Theorem 6.3]{beck2017} for a proof). The proximal gradient algorithm executes the following iterations:
\begin{subequations}\label{eq:prox-method}
\begin{align} 
\iter{{\vbeta}}{1} &= \eta(\mX \tran \vy ; \stepsize)\\
    \iter{{\vbeta}}{t+1} &= \eta\big( \iter{{\vbeta}}{t} - \stepsize \mX\tran(\mX \iter{{\vbeta}}{t} - \vy) ; \stepsize\big), \quad \text{for } t \ge 1,
\end{align}
\end{subequations}
where $\stepsize \in (0,\infty)$ is the step-size parameter.

We study the RLS estimator in \eqref{eq:RLS} and the iterates of the proximal gradient algorithm in \eqref{eq:prox} in the high-dimensional asymptotic framework where the signal dimension $\dim \rightarrow \infty$. In this framework, one observes a sequence of regression problems indexed by $\dim$, with sample size $\ssize_\dim$, sensing matrix $\iter{\mX}{\dim} \in \R^{M_\dim \times \dim}$, signal vector $\iter{\vbeta}{\dim}_\star \in \R^{\dim}$, noise $\iter{\vepsilon}{\dim} \in \R^{\ssize_\dim}$, and measurement vector $\iter{\vy}{\dim}\in \R^{\ssize_\dim}$. While our asymptotic analysis does not need to assume a particular scaling of the sample size $\ssize_{\dim}$ with $\dim$, many of our examples will consider the proportional scaling where $\ssize_\dim/\dim \rightarrow \alpha \in (0,\infty)$. For notational simplicity, we will often suppress the dependence of $\ssize_\dim$, $\iter{\mX}{\dim}$, $\iter{\vbeta}{\dim}_\star$, $\iter{\vepsilon}{\dim}$ and $\iter{\vy}{\dim}$ on $\dim$.

\subsection{Demonstrations of Universality}
\label{sec:demonstration}
Before presenting our main results that formally establish the aforementioned universality principle, we first illustrate this phenomenon by considering several concrete examples of sensing matrices. The first ensemble, denoted by $\mathtt{SpikeSine}$, is a real-valued and randomly signed version of the Spikes and Sines matrix considered by \citet{monajemi2013deterministic}. Specifically,
\begin{equation}\label{eq:ensemble_ss}
\mX_\mathtt{SpikeSine} = \frac{1}{\sqrt{2}} \begin{bmatrix} \mI_\ssize & \mQ_\ssize \end{bmatrix} \cdot \mS,
\end{equation}
where $\mI_{\ssize}$ is the $\ssize \times \ssize$ identity matrix and $\mQ_\ssize$ denotes an $\ssize \times \ssize$ orthonormal discrete cosine transform (DCT) matrix, and $\mS = \diag(s_{1:\dim})$ is a diagonal matrix of i.i.d.\ signs $s_{1:\dim} \explain{i.i.d.}{\sim} \unif{\{\pm 1\}}$. The second ensemble, $\texttt{Mask}$, is obtained by concatenating $L$ square matrices:
\begin{equation}\label{eq:ensemble_Hadamard}
\mX_\texttt{Mask} = \begin{bmatrix} \mD_1 \mH_\ssize &  \mD_2 \mH_\ssize& \ldots &  \mD_L \mH_\ssize \end{bmatrix} \cdot \mS.
\end{equation}
Here, $L \in \N$, $\mH_\ssize$ denotes the $\ssize \times \ssize$ orthonormal Hadamard-Walsh matrix, and $\set{\mD_\ell}_{1 \le \ell \le L}$ is a collection of random diagonal matrices. The diagonal entries $(\mD_\ell)_{ii}$ are drawn i.i.d. from some \emph{symmetric} probability distribution with bounded support. We take $(\mD_\ell)_{ii} \explain{i.i.d.}{\sim} \unif{[-1,1]}$ in our experiments. Just as in \eqref{eq:ensemble_ss}, $\mS$ in \eqref{eq:ensemble_Hadamard} is a diagonal matrix of i.i.d. signs. Moreover, $\mD_1, \ldots, \mD_L$ and $\mS$ are mutually independent. The third ensemble, $\texttt{RandDCT}$, and the last ensemble, \texttt{Haar}, have similar constructions in the form of
\begin{equation}\label{eq:ensemble_RandDCT}
     \mX_\texttt{RandDCT}(\mLambda) = \mLambda^{1/2} (\mP \mQ_N \mS)
\end{equation}
and
\begin{equation}\label{eq:ensemble_Haar}
    \mX_\texttt{Haar}(\mLambda) = \mLambda^{1/2} \mV\tran,
\end{equation}
where $\mLambda \in \R^{\dim \times \dim}$ is a deterministic diagonal matrix with non-negative numbers on the diagonal. In \eqref{eq:ensemble_RandDCT}, $\mQ_\dim$ is the DCT matrix of size $\dim \times \dim$, $\mP \in \R^{\dim \times \dim}$ is a uniformly random permutation matrix, and $\mS$ is a diagonal matrix of i.i.d. signs independent of $\mP$. In \eqref{eq:ensemble_Haar}, $\mV \in \unif{\ortho(\dim)}$ is a random orthogonal matrix drawn from the Haar (\emph{i.e.}, uniform) distribution on the group $\ortho(\dim)$ of orthogonal matrices.

By construction, $\mX_\texttt{SpikeSine}$ in \eqref{eq:ensemble_ss} has orthonormal rows, and thus the eigenvalues of $\mX\tran \mX$ for \texttt{SpikeSine} consist of exactly $\ssize$ ones and $\ssize$ zeros. On the other hand, for \texttt{RandDCT} and \texttt{Haar}, the eigenvalues of $\mX\tran\mX$ are given by the diagonal elements of $\mLambda$. Thus, by setting
\begin{equation}\label{eq:Lambda_SS}
    \mLambda = \mLambda_\texttt{SpikeSine} \bydef \diag\{\underbrace{1, \ldots, 1}_{\ssize}, \underbrace{0, \ldots, 0}_{M}\},
\end{equation}
we can make sure that the spectral distribution of $\mX\tran \mX$ is identical for the three ensembles \texttt{SpikeSine}, \texttt{Haar}, and \texttt{RandDCT}. Similarly, it is easy to verify that, with the choice of
\begin{equation}\label{eq:Lambda_Hadamard}
    \mLambda = \mLambda_\texttt{Mask} \bydef \sum_{1 \le \ell \le L} \mD_\ell^2,
\end{equation}
the ensembles \texttt{RandDCT} and \texttt{Haar} can also attain the same spectral distribution of \texttt{Mask}. Notwithstanding their matching spectra, the four ensembles defined above have very different constructions. In fact, with \texttt{SpikeSine} and \texttt{Mask} being rectangular matrices whereas \texttt{RandDCT} and \texttt{Haar} being square matrices, they do not even have the same aspect ratios.
\begin{figure}[t]
         \includegraphics[width=\textwidth,trim={3.4cm 0.1cm 4.1cm 0.4cm},clip]{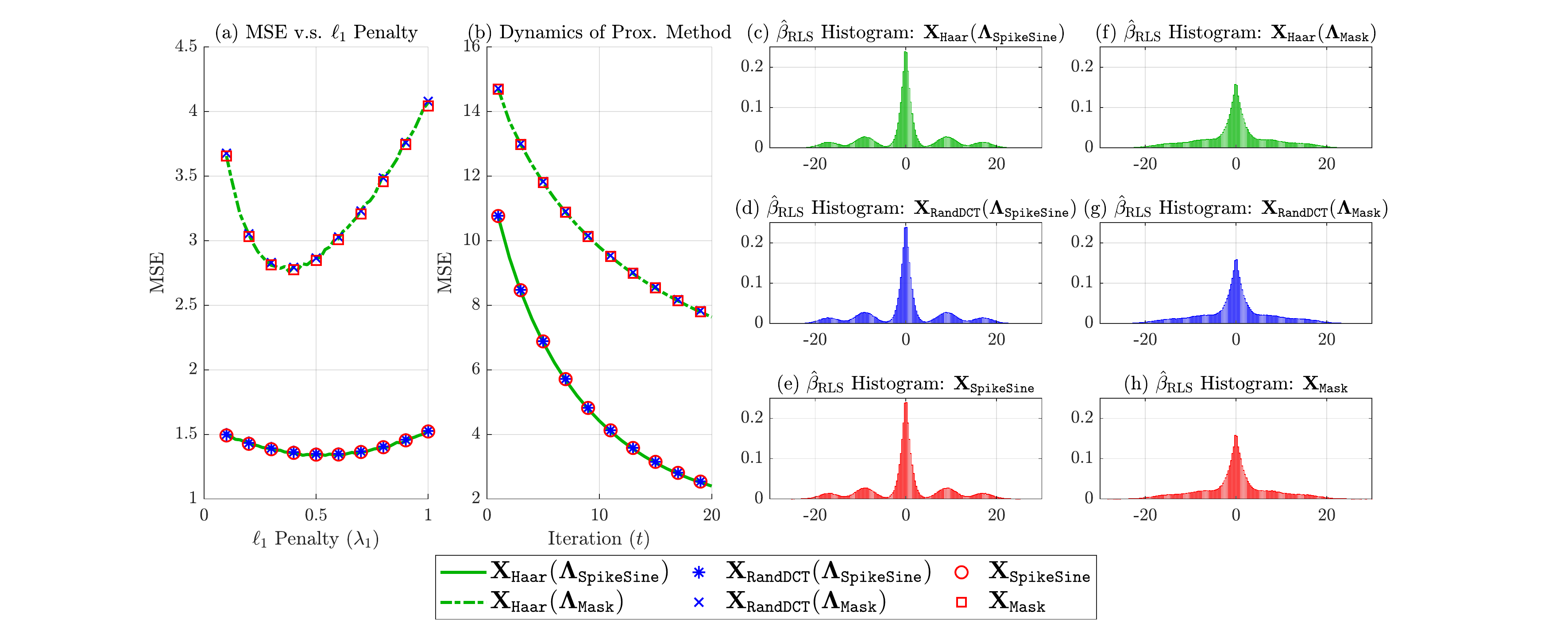}
\caption{Universal behavior of sensing matrices \texttt{SpikeSine}, \texttt{RandDCT}, \texttt{Haar} [with the latter two using $\mLambda= \mLambda_\texttt{SpikeSine}$ in \eqref{eq:Lambda_SS}], and \texttt{Mask}, \texttt{RandDCT}, \texttt{Haar} [with the latter two using $\mLambda= \mLambda_\texttt{Mask}$ in \eqref{eq:Lambda_Hadamard}] for regularized linear regression. Panel (a): Plot of normalized MSE v.s. $\ell_1$ penalty parameter $\lambda_1$. Panel (b): Plot of normalized MSE of the proximal iterate \eqref{eq:prox-method} versus iteration number $t$. Panels (c)-(h): Histogram of the non-zero coordinates of the RLS estimator $\vbeta_\opt$.}
\label{fig:univplot}
\end{figure}

In \fref{univplot}, we compare the performance of the RLS estimator in \eqref{eq:RLS} and the proximal gradient algorithm in \eqref{eq:prox-method} on each of these sensing matrix ensembles. In our experiment, we sampled the coordinates of the $\dim = 2^{21}$ dimensional signal $\vbeta_\star$ i.i.d. from the following 5-point prior:
\begin{align*}
    (\vbeta_\star)_i \explain{i.i.d.}{\sim} \frac{9}{10} \delta_0 + \frac{1}{10} \left( \frac{1}{3} \delta_{12} + \frac{1}{3} \delta_{-12} + \frac{1}{6} \delta_{20} + \frac{1}{6} \delta_{-20} \right).
\end{align*}
The measurements $\vy$ were generated from the linear model \eqref{eq:linear-model} with Gaussian noise of variance $\sigma^2 = 1$. We choose an \emph{elastic net} regularizer, where 
\begin{equation}\label{eq:reg_net}
    \reg(\beta) = \lambda_1 \abs{\beta} + \lambda_2 \beta^2
\end{equation}
for some positive constants $\lambda_1$ and $\lambda_2$. Let $\vbeta_{\opt}(\mX, \vbeta_\star, \vepsilon)$ denote the RLS estimator from \eqref{eq:RLS} and $\iter{{\vbeta}}{t}(\mX, \vepsilon, \vbeta_\star)$ denote the $t$th iterate of the proximal gradient algorithm in \eqref{eq:prox-method}. Panel (a) shows the normalized mean squared error (MSE) of the RLS estimate:
\begin{equation}
    \frac 1 N \norm{{\vbeta}_\opt(\mX, \vepsilon, \vbeta_\star) - \vbeta_\star}^2
\end{equation}
as we vary the $\ell_1$ regularization parameter $\lambda_1$ in \eqref{eq:reg_net} over a grid of values in $(0,1]$, while keeping the ratio $\lambda_2/\lambda_1$ fixed at $10^{-3}$. Thereafter, in panels (b)-(h) we fixed $\lambda_1 = 1$ and $\lambda_2 = 10^{-3}$. Panel (b) plots the normalized MSE of the proximal iterates $\iter{{\vbeta}}{t}(\mX, \vepsilon, \vbeta_\star)$ as a function of iteration $t$. Finally, panels (c-h) show the histograms of the \emph{non-zero} coordinates of the RLS estimator ${\vbeta}_\opt(\mX, \vepsilon, \vbeta_\star)$ for the different sensing matrices. To avoid the high memory and computational cost for explicit generation and manipulation of Haar matrices of dimension $\dim = 2^{21}$, we used the \emph{Householder Dice} algorithm of the third author \citep{lu2021householder} in our experiments involving the \texttt{Haar} ensemble. We observe that, despite their markedly different constructions, the ensembles \texttt{SpikeSine}, \texttt{RandDCT}, \texttt{Haar} [with the latter two using $\mLambda= \mLambda_\texttt{SpikeSine}$ in \eqref{eq:Lambda_SS}], and the ensembles \texttt{Mask}, \texttt{RandDCT}, \texttt{Haar} [with the latter two using $\mLambda= \mLambda_\texttt{Mask}$ in \eqref{eq:Lambda_Hadamard}] seem to achieve identical results in \fref{univplot}.

The universality phenomenon shown in \fref{univplot} is yet another demonstration of the well-known observation in the literature \citep{donoho2009observed,monajemi2013deterministic,oymak2014case,abbara2020universality,ma2021spectral} that the performance of many sensing matrices are universal, as long as they have matching spectral distributions and that their singular vectors are in ``generic positions''. The main goal of this paper is to make this intuition precise and rigorous. Specifically, our contributions are two-fold:

\paragraph{1. Characterization of Universality classes:} We introduce a set of easy-to-verify technical conditions that define a universality class for sensing matrices (see Definition~\ref{def:univ-class}). In addition to having the standard i.i.d. and the rotational invariant ensembles as special cases, the universality classes that we define in this work also include sensing matrices constructed with very limited randomness, such as the \texttt{Haar}, \texttt{Mask}, and \texttt{RandDCT} ensembles considered above. See \sref{universality_examples} for details and examples of other ensembles that lie in the universality classes.

\paragraph{2. Universality of RLS estimators: } In \thref{RLS}, we show that, for all sensing matrices that lie in the same universality class, the dynamics of the proximal gradient algorithm in \eqref{eq:prox-method}---and more generally, the dynamics of a broad class of first-order methods (formally introduced in \sref{GFOM})---for solving the RLS problem in \eqref{eq:RLS} are asymptotically identical (in a sense to be made precise in Definition~\ref{def:PW2-eq}). Moreover, under additional strong convexity conditions (which guarantee the uniqueness of the RLS estimators), we show that the RLS estimators associated with matrices from the same universality class are also asymptotically identical. Finally, while we use regularized linear regression as our primary application, our universality result for first order methods (\thref{GFOM}) is stated more generally and may be applicable to inference problems beyond linear regression.
\begin{remark}
There have been rigorous studies of the universality phenomenon (see, e.g., \citep{donoho2010counting,tulino2010capacity,farrell2011limiting,anderson2014asymptotically}) in contexts related to \fref{univplot}. We provide a detailed discussion of these earlier works in the literature as well as several recent advances \citep{dudeja2020universality,dudeja2022universality,wang2022universality} in \sref{related}.
\end{remark}

\subsection{Main Results}

Our main results are to show that the performance of RLS estimators ${\vbeta}_\opt(\mX, \vbeta_\star, \vepsilon)$ and the dynamics of the proximal gradient method are asymptotically identical for all sensing matrices $\mX$ that lie in a given universality class. We now introduce this universality class in the definition below.

\begin{definition}[Spectral Universality Class]\label{def:univ-class} Given a compactly supported probability measure $\mu$ on $[0,\infty)$, we say that a sensing matrix $\mX$ lies in the universality class $\uclass{\mu}$ if:
\begin{enumerate}
    \item \emph{Random Signs.} $\mX = \mJ \mS$ where $\mJ \in \R^{\ssize \times \dim}$ is a deterministic matrix and $\mS = \diag(s_{1:\dim})$ is a diagonal matrix of i.i.d. signs $s_{1:\dim} \explain{i.i.d.}{\sim} \unif{\{\pm 1\}}$. 
\end{enumerate}
The sequence of deterministic matrices $\mJ \in \R^{\ssize \times \dim}$ satisfies:
\begin{enumerate}
\setcounter{enumi}{1}
    \item \emph{Bounded Operator Norm.} $\|\mJ\|_{\op} \lesssim 1$.
    \item \emph{Convergence of Empirical Spectral Measure.}  For any fixed $k \in \N$, 
    \begin{align*}
    \Tr[(\mJ \tran \mJ)^k] /\dim \rightarrow \int \lambda^k \mu(\diff \lambda) \quad \text{ as $\dim \rightarrow \infty$.}
    \end{align*}
    \item \emph{Generic Right Singular Vectors.}  For any fixed $k \in \N$, $\epsilon > 0$, 
    \begin{align} \label{eq:univ-class-moment-convergence}\Big\| (\mJ \tran \mJ)^k -  \frac{ \Tr[(\mJ \tran \mJ)^k]}{\dim} \mI_{\dim} \Big\|_{\infty} \lesssim \dim^{-1/2+\epsilon}.
    \end{align}
    This means that for any $k \in \N, \epsilon>0$ there are constants $C(k,\epsilon) > 0$, $\dim_0(k,\epsilon) \in \N$ such that:
    \begin{align*}
          \Big\| (\mJ \tran \mJ)^k -  \frac{ \Tr[(\mJ \tran \mJ)^k]}{\dim} \mI_{\dim} \Big\|_{\infty} \leq C(k,\epsilon) \cdot \dim^{-1/2 + \epsilon}  \quad \forall \; \dim \geq \dim_0(k,\epsilon).
    \end{align*}
    In the above display, for a matrix $\mA \in \R^{\dim \times \dim}$, $\|\mA\|_{\infty} \explain{def}{=} \max_{i,j \in [\dim]} |A_{ij}|$ is the entry-wise infinity norm.
\end{enumerate}
\end{definition}

\begin{remark}
We will show in \sref{universality_examples} that all four ensembles considered in \sref{demonstration} (i.e., \texttt{SpikeSine}, \texttt{Mask}, \texttt{RandDCT}, and \texttt{Haar}) belong to $\uclass{\mu}$ for some probability measure $\mu$. Note that in \defref{univ-class}, the randomness of the matrix $\mX = \mJ \mS$ entirely comes from the i.i.d. sign matrix $\mS$; the component $\mJ$ is deterministic. In \texttt{Mask} and \texttt{RandDCT}, however, the corresponding $\mJ$ matrix contains additional sources of randomness ({e.g.}, the random diagonal matrices $\set{\mD_\ell}_\ell$ in \eqref{eq:ensemble_Hadamard} and the random permutation matrix $\mP$ in \eqref{eq:ensemble_RandDCT}). For such cases, the statement that an ensemble lies in $\uclass{\mu}$ should be interpreted as follows: almost surely, any sequence of matrices drawn from the ensemble lies in $\uclass{\mu}$.
\end{remark}

\begin{remark}
The construction of the \texttt{Haar} ensemble in \eqref{eq:ensemble_Haar} does not contain a random sign matrix $\mS$, but we can always append one without changing the distribution of the ensemble. Indeed, due to the rotational invariance of the Haar matrix $\mV$, we have $\mV\tran \explain{d}{=} \mV\tran \mS$.
\end{remark}

To state our result, we will also need the following definition of asymptotic equivalence of random vectors.

\begin{definition}[Asymptotic Equivalence of Random Vectors]\label{def:PW2-eq} Let  $(\iter{\vv}{1}, \dotsc, \iter{\vv}{k})$ and $(\iter{\tilde{\vv}}{1}, \dotsc, \iter{\tilde{\vv}}{k})$  be two collections of $\dim$-dimensional random vectors realized in the same probability space as the signal vector $\vbeta_\star$. We say that  $(\iter{\vv}{1}, \dotsc, \iter{\vv}{k})$ and $(\iter{\tilde{\vv}}{1}, \dotsc, \iter{\tilde{\vv}}{k})$ are asymptotically equivalent in probability with respect to the Wasserstein-$2$ metric  if for any continuous test function $h: \R^{k+1} \rightarrow \R$  (independent of $\dim$) that satisfies:
\begin{subequations}\label{eq:pseudo_Lipschitz}
\begin{align}
    |h(x;\beta) - h(y;\beta)| & \leq L \|x - y\|   (1 + \|x\| + \|y\| + |\beta|^D) \; \forall \; x, y \; \in \; \R^k, \; \beta \; \in \; \R \intertext{and}
    |h(x; \beta)| & \leq L  (1 + \|x\|^D + |\beta|^D)
\end{align}
\end{subequations}
for some finite constants $L \geq 0$ and $D \in \N$, we have,
\begin{align*}
    \frac{1}{\dim} \sum_{i=1}^\dim h(\iter{v}{1}_i, \iter{v}{2}_i, \dotsc, \iter{v}{k}_i; (\beta_\star)_i) - \frac{1}{\dim} \sum_{i=1}^\dim h(\iter{\tilde{v}}{1}_i, \iter{\tilde{v}}{2}_i, \dotsc, \iter{\tilde{v}}{k}_i; (\beta_\star)_i)  \explain{P}{\rightarrow} 0,
\end{align*}
where $\explain{P}{\rightarrow}$ denotes convergence in probability. We denote equivalence in the above sense using the notation $(\iter{\vv}{1}, \iter{\vv}{2}, \dotsc, \iter{\vv}{k}; \vbeta_\star) \asymeq (\iter{\tilde{\vv}}{1}, \dotsc, \iter{\tilde{\vv}}{k}; \vbeta_\star)$. 
\end{definition}

\begin{remark}\label{rem:mse}
Let $x = (x_{1}, \ldots, x_{k}) \in \R^k$. It is easy to verify that for any $j \in [k]$, the function $h(x; \beta) = (x_j - \beta)^2$ satisfies the conditions in \eqref{eq:pseudo_Lipschitz}. Thus, $(\iter{\vv}{1}, \iter{\vv}{2}, \dotsc, \iter{\vv}{k}; \vbeta_\star) \asymeq (\iter{\tilde{\vv}}{1}, \dotsc, \iter{\tilde{\vv}}{k}; \vbeta_\star)$ implies, in particular, that
\begin{equation}
    \frac {\|\iter{{\vv}}{j} - \vbeta_\star \|^2} N  - \frac {\|\iter{\tilde{\vv}}{j} - \vbeta_\star \|^2} N \explain{P}{\rightarrow} 0 \qquad \text{for all } j \in [k].
\end{equation}
\end{remark}

\begin{theorem}\label{thm:RLS} Let $\mu$ be any compactly supported probability measure on $[0,\infty)$. Suppose that  $(\vbeta_\star, \vepsilon)$ satisfy \assumpref{RSN} and the regularizer satisfies \assumpref{convex-reg}. Let $\mX = \mJ \mS, \tilde{\mX} = \tilde{\mJ} \tilde{\mS}$ be two independent sensing matrices in the same universality class $\uclass{\mu}$. Then,
\begin{enumerate}
    \item Universality of Proximal Method Iterates: For any fixed $T \in \N$ (independent of $\dim$):
    \begin{align*}
        (\iter{{\vbeta}}{1}(\mX, \vepsilon, \vbeta_\star), \dotsc, \iter{{\vbeta}}{T}(\mX, \vepsilon, \vbeta_\star); \vbeta_\star ) \asymeq  (\iter{{\vbeta}}{1}(\tilde{\mX}, \vepsilon, \vbeta_\star), \dotsc, \iter{{\vbeta}}{T}(\tilde{\mX}, \vepsilon, \vbeta_\star); \vbeta_\star ).
    \end{align*}
    Here, $\iter{{\vbeta}}{1:T}(\mX, \vepsilon, \vbeta_\star)$ and $\iter{{\vbeta}}{1:T}(\tilde{\mX}, \vepsilon, \vbeta_\star)$ are the iterates generated by the proximal method \eqref{eq:prox-method} on the sensing matrix $\mX$ and $\tilde{\mX}$ respectively, with the signal vector $\vbeta_\star$ and noise vector $\vepsilon$.
    \item Universality of the RLS Estimator: Suppose that in addition, at least one of the following hold:
    \begin{enumerate}
        \item  $\rho$ is $\kappa$-strongly convex for some $\kappa > 0$. That is, for any $x,x^\prime \in \R$:
    \begin{align*}
        \rho(x^\prime) \geq \rho(x) + (x^\prime-x) \partial \rho(x) + \frac{\kappa(x-x^\prime)^2 }{2},
    \end{align*}
    where $\partial \rho(x)$ is any sub-gradient of $\rho$ at $x$.
    \item Or, there are constants $\kappa > 0$ and $\dim_0 \in \N$ such that:
    \begin{align*}
       \lambda_{\min}(\mJ \tran \mJ) \geq \kappa, \quad  \lambda_{\min}(\tilde{\mJ} \tran \tilde{\mJ}) \geq \kappa \quad \forall \; \dim \geq \dim_0.
    \end{align*}
    \end{enumerate}
    Then, the RLS estimator in \eqref{eq:RLS:1} is uniquely specified and
    \begin{align*}
       ({\vbeta}_\opt(\mX, \vepsilon, \vbeta_\star); \vbeta_\star ) \asymeq  ({\vbeta}_\opt(\tilde{\mX}, \vepsilon, \vbeta_\star); \vbeta_\star ).
    \end{align*}
\end{enumerate}
\end{theorem}

As a consequence of \thref{RLS}, if the asymptotic performance of the RLS estimator is known for one sensing matrix in $\uclass{\mu}$, one immediately concludes the same asymptotic characterization for any other sensing matrix in $\uclass{\mu}$. A particularly convenient choice in $\uclass{\mu}$ is the rotationally invariant \texttt{Haar} ensemble $\mX_\texttt{Haar}(\mLambda)$ given in \eqref{eq:ensemble_Haar}, where the diagonal matrix $\mLambda$ is such that the empirical distribution of its diagonal elements $\set{\Lambda_{ii}}$ converges to the probability measure $\mu$. Later, in \sref{sp_inv}, we will verify that $\mX_\texttt{Haar}(\mLambda)$, and more generally, any matrix with matching spectral distribution and whose right singular vectors are sign and permutation invariant, lies in $\uclass{\mu}$. (See \lemref{sign-perm-inv-ensemble} and \exref{rot-inv-ensemble} for details.)

Since the random orthogonal matrix $\mV$ in \texttt{Haar} can be obtained by performing a QR decomposition on an i.i.d. Gaussian random matrix of the same size, the \texttt{Haar} ensemble inherits many nice statistical properties (such as the rotational invariance) of the i.i.d. Gaussian matrix. This makes the \texttt{Haar} ensemble a very convenient model for mathematical analysis. See \emph{e.g.} \citep{ma2017orthogonal, takeuchi2017rigorous, rangan2019vector,dudeja2020analysis,dudeja2020information,maillard2020phase, pmlr-v125-gerbelot20a,gerbelot2020asymptotic,fan2020approximate,ebpca, mondelli2021pca,ma2021analysis, venkataramanan2022estimation} for exact asymptotic characterization of several estimation problems involving the \texttt{Haar} ensemble. In particular, \citet{pmlr-v125-gerbelot20a} have characterized the asymptotic mean squared error of the RLS estimator for $\mX_{\texttt{Haar}}(\mLambda)$, showing that
\begin{align*}
    \frac{\|\vbeta_\star -  \vbeta_\opt(\mX_{\texttt{Haar}}(\mLambda), \vbeta_\star, \vepsilon)\|^2}{\dim} \explain{P}{\rightarrow} \zeta(\mu, \serv{B_\star}),
\end{align*}
for an explicit limiting value $\zeta(\mu, \serv{B_\star})$ determined by $\mu$ and the law of $\serv{B_\star}$ (cf. \assumpref{RSN}). As a consequence of this characterization and \thref{RLS}, we have for any $\mX \in \uclass{\mu}$:
\begin{align} \label{eq:asymp-mse}
     \frac{\|\vbeta_\star -  \vbeta_\opt(\mX, \vbeta_\star, \vepsilon)\|^2}{\dim} \explain{P}{\rightarrow} \zeta(\mu, \serv{B_\star}).
\end{align}
This follows as we can always realize the rotationally invariant sensing matrix $\mX_{\texttt{Haar}}(\mLambda)$ in the same probability space as a given $\mX \in \uclass{\mu}$ by sampling it independently of $\mX$. Hence, by \thref{RLS} and \remref{mse}:
\begin{align*}
    \frac{\|\vbeta_\star -  \vbeta_\opt(\mX_{\texttt{Haar}}(\mLambda), \vbeta_\star, \vepsilon)\|^2}{\dim} -  \frac{\|\vbeta_\star -  \vbeta_\opt(\mX, \vbeta_\star, \vepsilon)\|^2}{\dim} \explain{P}{\rightarrow} 0,
\end{align*}
which immediately yields \eqref{eq:asymp-mse}. 

As suggested by the above discussions, the universality principle established in \thref{RLS} provides a convenient approach to obtaining the asymptotic performance of RLS estimators on structured sensing ensembles such as \texttt{SpikeSine},  \texttt{Mask}, and \texttt{RandDCT}. It is intuitively clear where the challenge lies in directly analyzing these structured ensembles: Compared with \texttt{Haar}, these structured random matrices simply have much less “randomness”. Loosely speaking, the former consists of $\mathcal{O}(\dim^2)$ independent random variables (due to its connection to the i.i.d. Gaussian ensemble), while \texttt{SpikeSine} and \texttt{Mask} only build on $\mathcal{O}(\dim)$ independent random variables. Common tools in high-dimensional performance analysis, such as Gaussian width \citep{chandrasekaran2012convex,amelunxen2014living}, comparison inequalities for Gaussian processes \citep{gordon1985some,thrampoulidis2015regularized}, and state evolution of message passing algorithms \citep{bayati2011dynamics,takeuchi2017rigorous,rangan2019vector,fan2020approximate} are simply not equipped to tightly control the strong correlations that exist in these matrices constructed with such limited randomness. By appealing to the universality principle, we can bypass these technical challenges and study the more mathematically tractable \texttt{Haar} ensemble instead, with the guarantee that results obtained there can be transferred to \texttt{SpikeSine} and \texttt{Mask} in the high-dimensional setting.

\paragraph{Organization:} The rest of the paper is organized as follows. In \sref{universality_examples}, we study several concrete examples of matrix ensembles and show that they all belong to the universality class $\uclass{\mu}$, as defined in \defref{univ-class}, for some suitably chosen probability measure $\mu$. \sref{related} discusses related work and recent progress, both in the general area of high-dimensional analysis of estimation problems and along the specific theme of universality, in the literature. Our main result, \thref{RLS}, is proved in \sref{proof_strategy}. A key technical component of our proof, namely the universality of vector approximate message passing algorithms, is discussed in \sref{univ-vamp}. Additional technical details, as well as miscellaneous auxiliary results, are delegated to the appendix.

\paragraph{Notations:} We conclude this section by collecting some notations used throughout this paper. 
\begin{description}[font=\normalfont\emph,leftmargin=0cm,itemsep=0ex]
\item [Some common sets:] We will use $\N$ and $\R$ to denote the set of positive integers and the set of real numbers, respectively. We define $\W \explain{def}{=} \N \cup \{0\}$ as the set of non-negative integers. For each $\dim \in \N$, $[\dim]$ denotes the set $\{1, 2, 3, \dotsc, \dim\}$ and $\ortho(\dim)$ denotes the set of $\dim \times \dim$ orthogonal matrices. \item [Asymptotics:] Given a sequence $a_{\dim}$ and a non-negative sequence $b_{\dim}$ indexed by $\dim \in \N$ we say $a_\dim \ll b_\dim$ or $a_{\dim} = o(b_\dim)$ if $a_\dim/b_\dim \rightarrow 0$. Similarly we say $a_{\dim} \lesssim b_{\dim}$ or $a_\dim = O(b_\dim)$ if there exist fixed constants $c \geq 0$ and $N_0 \in \N$, such that $|a_\dim| \leq c\, b_\dim$ for all $\dim \geq N_0$.
\item[Linear Algebra:] For a vector $v \in \R^{k}$, we use $\|v\|_1, \|v\|, \|v\|_\infty$ to denote its $\ell_1$, $\ell_2$ and $\ell_\infty$ norms, respectively and $\|v\|_0$ to denote the number of non-zero coordinates (or sparsity) of $v$. For a matrix $Q \in \R^{k \times k}$, we use $\|Q\|_{\op}, \|Q\|$ to denote the operator (spectral) norm and Frobenius norm of $Q$ respectively. On the other hand $\|Q\|_\infty \explain{def}{=} \max_{i,j \in [k]} |Q_{ij}|$ denotes the entry-wise $\ell_\infty$ norm. For a symmetric matrix $Q \in \R^{k \times k}$, $\lambda_{\min}(Q)$ and $\lambda_{\max}(Q)$ denote the smallest and largest eigenvalues of $Q$. $1_k$ denotes the vector $(1, 1, \dotsc, 1)$ in $\R^k$, $0_k$ denotes the vector $(0, 0, \dotsc, 0)$ in $\R^k$, and $e_1, e_2, \dotsc, e_{k}$ denote the standard basis vectors in $\R^k$. When the context makes the dimension clear, we will write $1_k$ as $1$ and $0_k$ as $0$. Analogously, $I_{k}$ denotes the $k \times k$ identity matrix.  We reserve the bold-face font for matrices and vectors whose dimensions diverge as $\dim$ (the dimension of the signal) grows to $\infty$. For example, the signal $\vbeta_\star \in \R^{\dim}$, the sensing matrix $\mX \in \R^{\ssize \times \dim}$ and the RLS estimator $\vbeta_{\opt} \in \R^{\dim}$ are bold-faced. 
\item[Gaussian Distributions and Hermite Polynomials:] The 
Gaussian distribution on $\R^{k}$ with mean vector $\mu \in \R^k$ and covariance matrix $\Sigma \in \R^{k \times k}$ is denoted by $\gauss{{\mu}}{\Sigma}$. For each $i \in \W$, $H_i : \R \mapsto \R$ denotes the univariate, normalized Hermite polynomial of degree $i$. The univariate Hermite polynomials are orthonormal polynomials for the standard Gaussian measure $\gauss{0}{1}$ on $\R$. This means that for $Z \sim \gauss{0}{1}$, $\E H_i^2(Z) = 1$ for each $i \in \W$ and $\E[H_i(Z) H_j(Z)] = 0$ for $i, j \in \W$ and $i \neq j$. The first few Hermite polynomials are $H_0(z) = 1, \; H_1(z) = z, \; H_2(z) = (z^2-1)/\sqrt{2}$. The multivariate Hermite polynomials generalize uni-variate Hermite problems to higher dimensions. For a degree vector $r = (r_1, \dotsc, r_k) \in \W^k$, the $k$-variate degree-$r$ Hermite polynomial $\hermite{r} : \R^k \mapsto \R$ is defined as:
\begin{align} \label{eq:hermite-notation}
    \hermite{r}(z_1, \dotsc, z_k) & \explain{def}{=} \prod_{i=1}^k \hermite{r_i}(z_i),
\end{align}
where the polynomials that appear on the RHS are the usual uni-variate Hermite polynomials. The $k$-variate Hermite polynomials are orthonormal polynomials for the standard Gaussian measure $\gauss{0}{I_k}$ on $\R^k$. We refer the reader to \citet[Chapter 11]{o2014analysis} for additional background on Hermite polynomials. 
\item[Other Distributions:] For a finite set $A$, $\unif{A}$ denotes the uniform distribution on $A$. For e.g., $\unif{\{\pm 1\}}$ and $\unif{\{\pm 1\}^\dim}$ denote the uniform distributions on $\{-1,1\}$ and the $\dim$-dimensional Boolean hypercube $\{-1,1\}^\dim$, respectively. We will use $\unif{\ortho(\dim)}$ to denote the Haar measure on the orthogonal group $\ortho(\dim)$. For any $x \in \R$, the probability measure $\delta_x$ on $\R$ denotes the point mass at $x$. 
\end{description}

\section{The Universality Class: Examples}
\label{sec:universality_examples}

In this section, we give several examples of matrix ensembles that lie in a spectral universality class, as defined in \defref{univ-class}. Our discussions in this section serve two purposes: (1) they showcase the wide applicability of the notion of universality class introduced in this work; and (2) they also demonstrate, on several different matrix ensembles, how to verify the key technical condition \eqref{eq:univ-class-moment-convergence} in \defref{univ-class}.

\subsection{Signed Incoherent Tight Frames}
\label{sec:sitf}

We start by considering matrices associated with \emph{incoherent tight frames}. Our motivation for these matrices comes from the work of \citet{monajemi2013deterministic}, who provide several examples of \emph{deterministic} incoherent tight frames that exhibit universality when used as sensing matrices in noiseless compressed sensing. 

Recall \revised{(see \emph{e.g.,} \citep{casazza2012introduction} for an overview on frame theory)} that a \emph{tight frame}\footnote{\revised{Strictly speaking, the only requirement for $\mF_{\ssize,\dim}$ to be a tight frame is $\mF_{\ssize,\dim} \mF \tran_{\ssize,\dim} = \mI_{\ssize}$. The additional requirement $(\mF_{\ssize, \dim} \tran\mF_{\ssize, \dim})_{ii}  = \ssize/\dim \quad \forall \; i \; \in \; [\dim]$ makes $\mF_{\ssize,\dim}$ an \emph{equal norm} tight frame. However, since all tight frames considered in the work of \citet{monajemi2013deterministic} and this paper are equal normed, we will omit the equal norm qualifier for brevity.}} is a deterministic $\ssize \times \dim$ matrix $\mF_{\ssize,\dim}$ with $\ssize \leq \dim$ that satisfies:
\begin{align}\label{eq:tight-frame}
    \mF_{\ssize,\dim} \mF \tran_{\ssize,\dim} = \mI_{\ssize}, \quad (\mF_{\ssize, \dim} \tran  \mF_{\ssize, \dim})_{ii}  = \frac{\ssize}{\dim} \quad \forall \; i \; \in \; [\dim].
\end{align}
We call \revised{(a sequence of)} tight frames $\mF_{\ssize,\dim}$ \emph{incoherent} if:
\begin{align} \label{eq:incoherent}
    \|\mF_{\ssize, \dim} \tran  \mF_{\ssize, \dim} -(\ssize/\dim) \cdot \mI_{\dim} \|_{\infty} \lesssim \dim^{-1/2+\epsilon} \quad \forall \; \epsilon \; > 0.
\end{align}
These two conditions imply that the rows of a tight frame are orthonormal whereas the columns have equal norms and approximately pairwise orthogonal. 

\begin{definition}
A \emph{signed} incoherent tight frame is a matrix of the form:
\begin{align} \label{eq:sign-frame}
    \mX_{\mathtt{SF}} = \mF_{\ssize,\dim} \cdot \mS.
\end{align}
where $\mF_{\ssize,\dim}$ is an incoherent tight frame and $\mS = \diag{(s_{1:\dim})}, \; s_{1:\dim} \explain{i.i.d.}{\sim} \unif{\{\pm 1\}}$ is a uniformly random sign matrix. 
\end{definition}

Matrices of the form \eqref{eq:sign-frame} can be viewed as natural semi-random analogs of the deterministic matrices considered by \citet{monajemi2013deterministic}. The following lemma shows that signed incoherent tight frames lie in the spectral universality class corresponding to a Bernoulli distribution.

\begin{lemma}\label{lem:frames} Let $ \mX_{\mathtt{SF}} = \mF_{\ssize,\dim} \cdot \mS$ be a $\ssize \times \dim$ signed incoherent tight frame with converging aspect ratio $\ssize /\dim \rightarrow \alpha \in (0,1]$. Then, $ \mX_{\mathtt{SF}} \in \uclass{\mathrm{Bern}(\alpha)}$, where $\mathrm{Bern}(\alpha)$ denotes the Bernoulli distribution with mean $\alpha$.
\end{lemma}
\begin{proof}
We need to check that $\mF_{\ssize,\dim}$ satisfies the requirements of \defref{univ-class}. Observe that \eqref{eq:tight-frame} and \eqref{eq:incoherent} guarantee that:
\begin{align*}
\|\mF_{\ssize,\dim}\|_{\op} = 1, \quad \Tr(\mF_{\ssize,\dim} \tran \mF_{\ssize,\dim}) = \ssize, \quad  
    \| \mF_{\ssize,\dim} \tran \mF_{\ssize,\dim} - \dim^{-1} \cdot \Tr(\mF_{\ssize,\dim} \tran \mF_{\ssize,\dim}) \cdot   \mI_{\dim}\|_{\infty}\lesssim \dim^{-1/2+\epsilon}, 
\end{align*}
for any $\epsilon > 0$. Furthermore, for any $k \in \N$, we can compute:
\begin{align*}
    (\mF_{\ssize,\dim} \tran \mF_{\ssize,\dim})^k & = \mF_{\ssize,\dim} \tran (\mF_{\ssize,\dim} \mF_{\ssize,\dim}\tran)^{k-1} \mF_{\ssize,\dim} \explain{\eqref{eq:tight-frame}}{=} \mF_{\ssize,\dim}\tran\mF_{\ssize,\dim}.
\end{align*}
Hence,  for any $k \in \N$ and any $\epsilon>0$: 
\begin{align*}
     \| (\mF_{\ssize,\dim} \tran \mF_{\ssize,\dim})^k - \dim^{-1} \cdot \Tr[(\mF_{\ssize,\dim} \tran \mF_{\ssize,\dim})^k] \cdot   \mI_{\dim}\|_{\infty}\lesssim \dim^{-1/2+\epsilon}.
\end{align*}
Finally note that:
\begin{align*}
    \frac{\Tr[(\mF_{\ssize,\dim} \tran \mF_{\ssize,\dim})^k]}{\dim} & =  \frac{\Tr(\mF_{\ssize,\dim} \tran \mF_{\ssize,\dim})}{\dim} \explain{\eqref{eq:tight-frame}}{=} \frac{\ssize}{\dim} \rightarrow \alpha.
\end{align*}Consequently, $\mX_{\mathtt{SF}} \in \uclass{\mu_\alpha}$ with $\mu_\alpha = \mathrm{Bern}(\alpha)$, since
\begin{align*}
    \int \lambda^k  \mu_\alpha(\diff \lambda) & = \alpha \quad \forall \; k \; \in \; \N.
\end{align*}
This concludes the proof of the lemma. 
\end{proof}
Incoherent tight frames encompass two important classes of matrices considered by \citet{monajemi2013deterministic}, which we discuss below.

\paragraph{1. Signed spikes+orthogonal matrices.} Assume that $\dim = 2\ssize $ is even. A spikes+orthogonal matrix is a $\ssize \times \dim$ matrix of the form
\begin{align}\label{eq:SO}
    \mF_{\mathtt{SO}} \explain{def}{=} \frac{1}{\sqrt{2}} \cdot  \begin{bmatrix}\mI_{\ssize} & \mO_{\ssize} \end{bmatrix}
\end{align}
where $\mO_{\dim}$ is any \emph{deterministic} $\ssize \times \ssize$ orthogonal matrix that satisfies the \emph{delocalization} estimate:
\begin{align*}
    \|\mO_{\ssize}\|_{\infty}  \lesssim \ssize^{-1/2 + \epsilon} \quad \forall \; \epsilon > 0.
\end{align*}
Examples of orthogonal matrices which satisfy the above property include discrete cosine/sine transform (DCT/DST) matrices, Hadamard-Walsh matrices and Discrete Fourier Transform (DFT) matrices (in the complex case). Thus, the \texttt{SpikeSine} ensemble considered in \sref{universality_examples} is just a special case of \eqref{eq:SO}.

Observe that spikes+orthogonal matrices satisfy:
\begin{align*}
    \mF_{\mathtt{SO}} \mF_{\mathtt{SO}}\tran & = \mI_{\ssize}, \quad\text{and}\quad
    \mF_{\mathtt{SO}} \tran \mF_{\mathtt{SO}}  = \frac{1}{2} \begin{bmatrix} \mI_{\ssize} & \mO_{\ssize} \\ \mO_{\ssize} \tran & \mI_{\dim}\end{bmatrix},
\end{align*}
and hence these matrices are incoherent tight frames [cf. \eqref{eq:tight-frame}, \eqref{eq:incoherent}]. Applying \lemref{frames}, we can conclude that the semi-random analog $\mX_{\mathtt{SO}} \explain{def}{=} \mF_{\mathtt{SO}} \cdot \mS$ (where $\mS$ is a uniformly random sign diagonal matrix) of the deterministic spikes+orthogonal matrices investigated by \citet{monajemi2013deterministic} lie in the universality class $\uclass{\mathrm{Bern(1/2)}}$. 

\paragraph{2. Signed Equiangular Tight Frames (ETFs).} A $\ssize \times \dim$ matrix $\mF_{\mathtt{ETF}}$ is an equiangular tight frame if it is a tight frame [\emph{i.e.}, it satisfies \eqref{eq:tight-frame}] and has equiangular columns. That is,
\begin{align}\label{eq:ETF}
    |(\mF_{\mathtt{ETF}} \tran \mF_{\mathtt{ETF}})_{ij}| & = c_{\ssize,\dim} \quad \forall \; i,j \; \in [\dim], \; i \neq j,
\end{align}
for some constant $c_{\ssize,\dim}$ (possibly dependent on $\ssize,\dim$). \citet{monajemi2013deterministic} empirically observed that many well-known ETFs (in the complex case) such as:
\begin{enumerate}
    \item Paley ETF constructed by \citet{bandeira2013road},
    \item Grassmannian ETF constructed by \citet{strohmer2003grassmannian}, 
    \item Delsarte-Goethals ETFs constructed by \citet{calderbank2010construction}. 
    \item Linear Chirp ETF constructed by \citet{applebaum2009chirp}
\end{enumerate}
exhibit universality properties in the context of noiseless compressed sensing. These matrices are complex ETFs that satisfy the natural complex analog of the conditions \eqref{eq:tight-frame} and \eqref{eq:ETF}:
\begin{align}\label{eq:CETF} 
    \mF_{} \mF \herm = \mI_{\ssize}, \quad (\mF_{\ssize, \dim} \herm  \mF_{\ssize, \dim})_{ii}  = \frac{\ssize}{\dim} \quad \forall \; i \; \in \; [\dim], \quad |(\mF\herm \mF)_{ij}| = c_{\ssize,\dim} \quad \forall \; i,j \; \in [\dim], \; i \neq j.
\end{align}
In the above display $\mF \herm$ denotes the Hermitian transpose of the matrix $\mF \in \C^{\ssize \times \dim}$. It is well-known that for \emph{any} ETF (real or complex) the constant $c_{\ssize,\dim}$ in \eqref{eq:ETF} and \eqref{eq:CETF} must be \citep[Theorem 2.3]{strohmer2003grassmannian}:
\begin{align}\label{eq:ETF-const}
    c_{\ssize,\dim} = \sqrt{\frac{\ssize}{\dim^2}\frac{\dim - \ssize}{(\dim-1)}} \leq \frac{1}{\sqrt{\dim}},
\end{align}
where the inequality follows from the assumption $1 \leq \ssize \leq \dim$. The characterization \eqref{eq:ETF-const} can be readily seen by computing the Frobenious norm of $\mF_{\mathtt{ETF}} \tran \mF_{\mathtt{ETF}}$ in two ways:
\begin{align*}
    \|\mF_{\mathtt{ETF}} \tran \mF_{\mathtt{ETF}}\|^2 & = \sum_{i,j=1}^\dim |\mF_{\mathtt{ETF}} \tran \mF_{\mathtt{ETF}})_{ij}|^2 \explain{\eqref{eq:tight-frame},\eqref{eq:ETF}}{=} \dim \cdot \frac{\ssize^2}{\dim^2} + \dim(\dim -1) \cdot c_{\ssize,\dim}^2, \\
     \|\mF_{\mathtt{ETF}} \tran \mF_{\mathtt{ETF}}\|^2 & = \Tr( \mF_{\mathtt{ETF}} \tran \mF_{\mathtt{ETF}} \mF_{\mathtt{ETF}} \tran \mF_{\mathtt{ETF}})  \explain{\eqref{eq:tight-frame}}{=} \Tr(\mF_{\mathtt{ETF}} \tran \mF_{\mathtt{ETF}})  \explain{\eqref{eq:tight-frame}}{=} \ssize.
\end{align*}
Equating these expression yields \eqref{eq:ETF-const}. Hence any ETF is necessarily incoherent in the sense of \eqref{eq:incoherent}. Consequently, \lemref{frames} shows that signed (real) ETFs $\mX_{\mathtt{ETF}} \explain{def}{=} \mF_{\mathtt{ETF}} \mS$ (where $\mS$ is a uniformly random sign diagonal matrix), which are real-valued, semi-random analogs of the deterministic complex ETFs studied by \citet{monajemi2013deterministic}, lie in the universality class $\uclass{\mathrm{Bern}(\alpha)}$, where $\alpha$ is the limiting aspect ratio $\ssize/\dim$ of the ETF $\mF_{\mathtt{ETF}}$. 

\subsection{Masked Orthogonal Sensing Matrices}
\label{sec:mask}

A \emph{masked orthogonal sensing matrix} is a $\ssize \times \dim$  matrix with integer aspect ratio $L = \dim/\ssize \in \N$  of the form:
\begin{align} \label{eq:masked-design}
    \mX_{\mathtt{Mask}} = \begin{bmatrix} \mD_1 \mO & \mD_2 \mO &\hdots & \mD_L \mO  \end{bmatrix} \mS
\end{align}
where:
\begin{enumerate}
    \item $L$ is the number of masks, which is assumed to be held fixed as $\dim \rightarrow \infty$.
    \item The matrices $\mD_{1:L}$ represent the masks. For each $\ell \in [L]$, $\mD_{\ell} = \diag(d_{\ell,1}, \dotsc, d_{\ell, \ssize})$ is a diagonal matrix whose entries $d_{\ell,1}, \dotsc, d_{\ell, \ssize}$ are i.i.d. copies of a symmetric and bounded random variable $\serv{D}$ (that is, $\serv{D} \explain{d}{=} - \serv{D}$ and $|\serv{D}| \leq K$ for some finite constant $K$). Furthermore, $\mD_{1}, \mD_{2}, \dotsc, \mD_{L}$ are sampled independently of each other.
    \item $\mO$ is deterministic $\ssize \times \ssize$ delocalized orthogonal matrix which satisfies:
    \begin{align}\label{eq:delocalization-mask}
        \|\mO\|_{\infty} \lesssim \ssize^{-1/2+\epsilon} \quad \forall \; \epsilon  >  0.
    \end{align}
    \item $\mS = \diag(s_{1:\dim}), \; s_{1:\dim} \explain{i.i.d.}{\sim} \unif{\{\pm 1\}}$ is a uniformly random sign diagonal matrix. 
\end{enumerate}
Note that the ensemble defined in \eqref{eq:ensemble_Hadamard} is just a special case of \eqref{eq:masked-design}, with the orthogonal matrix $\mO$ being the Hadamard-Walsh matrix, which satisfies \eqref{eq:delocalization-mask}. The following lemma identifies the universality class containing masked orthogonal sensing matrices. 

\begin{lemma} Let $\delta_0$ denote the Dirac measure at $0$, and let $\nu_{\serv{R}}$ denote the law of the random variable
\begin{align}\label{eq:DR}
    \serv{R} \explain{def}{=} \sum_{\ell = 1}^L \serv{D}_{\ell}^2,
\end{align}
where $\serv{D}_{1:L}$ are $L$ i.i.d. copies of the random variable $\serv{D}$ that was used to generate the diagonal masks $\mD_{1:L}$. Then, with probability 1, $ \mX_{\mathtt{Mask}}$ lies in the spectral universality class $\uclass{\mu_{L}}$  where $\mu_{L} \explain{def}{=} (1-1/L) \cdot \delta_0 + 1/L \cdot \nu_{\serv{R}}$.
\end{lemma}
\begin{proof} Since $\mX_{\mathtt{Mask}} = \mJ \mS$ for:
\begin{align*}
    \mJ \explain{def}{=} \begin{bmatrix} \mD_1 \mO & \mD_2 \mO &\hdots & \mD_L \mO  \end{bmatrix},
\end{align*}
it suffices to verify that $\mJ$ satisfies the requirements of \defref{univ-class} (with probability 1). Note that
\begin{align} \label{eq:R-def}
    \mJ  \mJ \tran  & = \sum_{\ell=1}^L  \mD_{\ell}^2 \explain{def}{=} \mR = \diag(r_1, \dotsc, r_{\ssize}),
\end{align}
where $r_a = \sum_{\ell=1}^L d_{\ell, a}^2$ for $a \in [M]$. Hence, for any $k \geq 1$ we have:
\begin{align} 
    (\mJ \tran \mJ)^k  = \mJ \tran (\mJ \mJ\tran)^{k-1} \mJ  &=  \mJ \tran \mR^{k-1} \mJ \nonumber \\&=  \begin{bmatrix} \mO \tran \mD_{1} \mR^{k-1} \mD_{1} \mO & \mO \tran \mD_1 \mR^{k-1} \mD_2 \mO & \hdots &   \mO \tran \mD_1 \mR^{k-1} \mD_L \mO \\ \vdots & \vdots & \vdots & \vdots \\  \mO \tran \mD_{L} \mR^{k-1} \mD_{1} \mO & \mO \tran \mD_L \mR^{k-1} \mD_2 \mO & \hdots &   \mO \tran \mD_L \mR^{k-1} \mD_L \mO \end{bmatrix}.  \label{eq:JTJ-k}
\end{align}
Observe that for any $s,t \in [L]$ and any $i,j \in [\ssize]$, the $(i,j)$ entry of the matrix $\mO \tran \mD_s \mR^{k-1} \mD_t \mO$ is of the form:
\begin{align*}
    (\mO \tran \mD_{s} \mR^{k-1} \mD_{t} \mO)_{ij} & = \sum_{a=1}^\ssize O_{ai} O_{aj} ( d_{s,a}   d_{t,a}  r_{a}^{k-1}).
\end{align*}
Hence $(\mO \tran \mD_{s} \mR^{k-1} \mD_{t} \mO)_{ij}$ is a linear combination of i.i.d. random variables $\{d_{s,a}  d_{t,a}  r_{a}^{k-1} : a \in [\ssize]\}$ which are uniformly bounded:
\begin{align*}
    \abs{d_{s,a}   d_{t,a}  r_{a}^{k-1}} & \explain{\eqref{eq:R-def}}{=} \abs{d_{s,a}   d_{t,a}  \left(\textstyle\sum_{\ell=1}^L d_{\ell, a}^2 \right)^{k-1}} \leq L^{k-1}  K^{2k}, 
\end{align*}
where the last inequality follows from the fact that $\{d_{s,a}: s \in [L], \; a \in [\ssize]\}$ are i.i.d. copies of a random variable $\serv{D}$ which satisfies $|\serv{D}| \leq K$. Hence, by Hoeffding's Inequality:
\begin{align*}
    \P \left( |(\mO \tran \mD_{s} \mR^{k-1} \mD_{t} \mO)_{ij} - \E[(\mO \tran \mD_{s} \mR^{k-1} \mD_{t} \mO)_{ij}] | > C_{K,L} \sqrt{\dim \ln(\dim)} \cdot \|\mO\|_{\infty}^2 \right) & \leq 1/\dim^4,
\end{align*}
for some constant $C_{K,L}$ determined by $K,L$. Hence by a union bound over $i,j \in [\ssize]$ and $s,t \in [L]$:
\begin{align*}
    \P \left( \|(\mJ \tran \mJ)^k -  \E[(\mJ \tran \mJ)^k] \|_{\infty} > C_{K,L} \sqrt{\dim \ln(\dim)} \cdot \|\mO\|_{\infty}^2\right) \leq 1/\dim^2.
\end{align*}
Since $\|\mO\|_{\infty} \lesssim \dim^{-1/2+\epsilon}$, using the Borel-Cantelli lemma we obtain:
\begin{align*}
    \P \left( \|(\mJ \tran \mJ)^k -  \E[(\mJ \tran \mJ)^k] \|_{\infty}  \lesssim \dim^{-1/2+\epsilon} \right) & = 1 \quad \forall \; \epsilon > 0, \; k \in \N.
\end{align*}
Taking a union bound over $\epsilon \in \mathbb{Q}$ (the set of rationals) and $k \in \N$ yields:
\begin{align*}
    \P \left( \|(\mJ \tran \mJ)^k -  \E[(\mJ \tran \mJ)^k] \|_{\infty}  \lesssim \dim^{-1/2+\epsilon} \quad \forall \; \epsilon > 0, \; k \in \N  \right) & = 1.
\end{align*}
Next, we compute $\E[(\mJ \tran \mJ)^k]$. Recall the random variables $\serv{D}_1, \dotsc, \serv{D}_{L}$ and $\serv{R}$ defined in \eqref{eq:DR}. In light of \eqref{eq:JTJ-k}, we begin by noting that for any $s,t \in [T]$, $\mD_{s} \mR^{k-1} \mD_{t}$ is a diagonal matrix whose diagonal entries are i.i.d. copies of the random variable $\serv{D}_s  \serv{D}_t \serv{R}^{k-1}$. Combining this with the fact that $\mO$ is orthogonal, we obtain:
\begin{align*}
    \E[(\mJ \tran \mJ)^k] & = \begin{bmatrix} \E[\serv{D}_1^2 \cdot \serv{R}^{k-1}] \cdot \mI_{\ssize} & \E[\serv{D}_1 \serv{D}_2 \cdot \serv{R}^{k-1}] \cdot \mI_{\ssize} & \hdots &   \E[\serv{D}_1 \serv{D}_L \cdot \serv{R}^{k-1}] \cdot \mI_{\ssize} \\ \vdots & \vdots & \vdots & \vdots \\  \E[ \serv{D}_L \serv{D}_1 \cdot \serv{R}^{k-1}] \cdot \mI_{\ssize} & \E[\serv{D}_L \serv{D}_2 \cdot \serv{R}^{k-1}] \cdot \mI_{\ssize} & \hdots &   \E[ \serv{D}_L^2 \cdot \serv{R}^{k-1}] \cdot \mI_{\ssize} \end{bmatrix}.
\end{align*}
Since $\serv{D}_{1:L}$ are i.i.d. copies of a symmetric random variable $\E[\serv{D}_s \serv{D}_t \cdot \serv{R}^{k-1}] = 0$ whenever $s \neq t$. Furthermore, by symmetry we have that $\E[\serv{D}_1^2 \cdot \serv{R}^{k-1}] = \E[\serv{D}_2^2 \cdot \serv{R}^{k-1}] = \dotsb = \E[\serv{D}_L^2 \cdot \serv{R}^{k-1}]$. Hence, for any $\ell \in [L]$:
\begin{align*}
    \E[\serv{D}_\ell^2 \cdot \serv{R}^{k-1}] & = \frac{1}{L} \E\left[ \sum_{\ell=1}^L \serv{D}_\ell^2 \cdot \serv{R}^{k-1} \right]  = \frac{\E[\serv{R}^k]}{L} \explain{(a)}{=} \int \lambda^k \; \mu_L(\diff \lambda),
\end{align*}
where step (a) follows by recalling the definition of $\mu_L$ from the statement of the lemma. Hence, we have shown:
\begin{align*}
     \P \left( \left\|(\mJ \tran \mJ)^k -  \mI_{\dim} \cdot \int \lambda^k \; \mu_L(\diff \lambda) \right\|_{\infty}  \lesssim \dim^{-1/2+\epsilon} \quad \forall \; \epsilon > 0, \; k \in \N  \right) & = 1.
\end{align*}
Notice that by the triangle inequality:
\begin{align*}
    \left|\frac{\Tr[(\mJ \tran \mJ)^k]}{\dim} - \int \lambda^k \; \mu_L(\diff \lambda) \right|  & \leq \left\|(\mJ \tran \mJ)^k -  \mI_{\dim} \cdot \int \lambda^k \; \mu_L(\diff \lambda) \right\|_{\infty}.
\end{align*}
This implies that with probability $1$, for all $k \in \N$ and all $\epsilon > 0$:
\begin{align*}
\left\|(\mJ \tran \mJ)^k -  \frac{\Tr[(\mJ \tran \mJ)^k]}{\dim}  \mI_{\dim}  \right\|_{\infty} &\leq 2\left\|(\mJ \tran \mJ)^k -  \mI_{\dim} \cdot \int \lambda^k \; \mu_L(\diff \lambda) \right\|_{\infty}  \lesssim \dim^{-1/2+\epsilon},  \\ \frac{\Tr[(\mJ \tran \mJ)^k]}{\dim} &\rightarrow \int \lambda^k \; \mu_L(\diff \lambda).
\end{align*}
Thus, $\mX_{\mathtt{Mask}} \in \uclass{\mu_L}$ with probability 1. 
\end{proof}

\subsection{Sign and Permutation Invariant Matrices}
\label{sec:sp_inv}

Next, we consider matrices whose right singular vectors are \emph{sign and permutation invariant}. As shown in the following lemma, these matrices lie in the spectral universality class corresponding to their limiting spectral measure. 

\begin{lemma} \label{lem:sign-perm-inv-ensemble} Consider a $\ssize \times \dim$ sensing matrix $\mX$ with singular value decomposition $\mX = \mU \mSigma \mV \tran$ where:
\begin{enumerate}
    \item  $\mU \in \mathbb{O}(\ssize)$ is an arbitrary deterministic orthogonal matrix.
    \item $\mSigma$ is a deterministic rectangular diagonal matrix that satisfies:
    \begin{align} \label{eq:singular-value-condition}
        \|\mSigma\|_{\op} \lesssim 1, \quad \frac{\Tr((\mSigma \tran \mSigma)^k)}{\dim} \rightarrow  \int \lambda^k \mu(\diff \lambda) \quad \forall \; k \in \N,
    \end{align}
    for some compactly supported probability measure $\mu$. 
    \item $\mV$ is an orthogonal matrix of the form $\mV = \mS \mO \mP$ where:
   \begin{enumerate}
       \item $\mO \in \mathbb{O}(\dim)$ is a delocalized deterministic orthogonal matrix that satisfies $\|\mO\|_{\infty}\lesssim \dim^{-1/2+\epsilon}$ for any $\epsilon > 0$. 
       \item $\mS = \diag(s_{1:\dim})$ is a uniformly random sign diagonal matrix with $s_{1:\dim} \explain{i.i.d.}{\sim} \unif{\{\pm 1\}}$. 
       \item $\mP$ is a uniformly random $\dim \times \dim$ permutation matrix independent of $\mS$.
       \end{enumerate}
       \end{enumerate}
Then, $\mX \in \uclass{\mu}$.
\end{lemma}

\begin{example}[Randomized Partial Hadamard-Walsh Matrix]
An important example of a sign and permutation invariant matrix is the randomized $\ssize \times \dim$ partial Hadamard-Walsh matrix $\mX_{\mathtt{PHWT}}$ with converging aspect ratio $\ssize/\dim \rightarrow \alpha \in (0,1]$ which is constructed by picking $\ssize$ rows of the $\dim \times \dim$ Hadamard-Walsh matrix uniformly at random and then randomly signing the columns of the resulting matrix. More concretely, $\mX_{\mathtt{PHWT}}$ is given by:
\begin{align*}
    \mX_{\mathtt{PHWT}} = [\mI_{\ssize}, \;  \vzero_{\ssize, \dim-\ssize}] \cdot \mP \tran \mH_{\dim} \mS
\end{align*}
where $\mP$ is a uniformly random $\dim \times \dim$ permutation matrix, $\mH_{\dim}$ is the $\dim \times \dim$ Hadamard-Walsh matrix. This type of matrix is commonly used as a structured dimension-reduction map in numerical linear algebra and high-dimensional data analysis (see, \emph{e.g.} \citep{nguyen2009fast,halko2011finding}). The singular value decomposition of $\mX_{\mathtt{PHWT}}$ is given by $\mX_{\mathtt{PHWT}} = \mU \mSigma \mV \tran$ with $\mU = \mI_{\ssize}$, $\mSigma = [\mI_{\ssize}, \;  \vzero_{\ssize, \dim-\ssize}]$, and $\mV = \mS \mH_{\dim} \tran \mP$. Since the Hadamard-Walsh matrix is delocalized in the sense $\|\mH_{\dim}\|_{\infty} = \dim^{-1/2}$, $\mX_{\mathtt{PHWT}}$ satisfies all the requirements of \lemref{sign-perm-inv-ensemble} and $\mX_{\mathtt{PHWT}} \in \uclass{\mathtt{Bern}(\alpha)}$ where $\mathrm{Bern}(\alpha)$ denotes the Bernoulli distribution with mean $\alpha$.
\end{example}

\begin{example}[Rotationally Invariant Ensembles]\label{ex:rot-inv-ensemble}
\lemref{sign-perm-inv-ensemble} also holds in the situation when the singular value decomposition of $\mX$ is given by $\mX = \mU \mSigma \mV \tran$ where $\mU, \mSigma, \mV$ are mutually independent random matrices such that $\mSigma$ satisfies the requirement \eqref{eq:singular-value-condition} with probability $1$, and the right singular vectors are Haar distributed $\mV \sim \unif{\Ortho(\dim)}$. This is because the Haar measure on $\Ortho(\dim)$ is distributionally invariant to left or right multiplication by any deterministic orthogonal matrix. Hence, $\mX = \mU \mSigma \mV \tran$ has the same distribution as $\tilde{\mX} = \mU \mSigma \tilde{\mV} \tran$ where $\tilde{\mV} \explain{def}{=} \mS \mV \mP$ and $\mS$, $\mP$ are uniformly random sign and permutation matrices independent of $\mV$. Moreover, with probability 1, $\norm{\mV}_\infty \lesssim N^{-1/2+\epsilon}$ for any $\epsilon > 0$. Thus, $\tilde{\mX} \in \uclass{\mu}$ by \lemref{sign-perm-inv-ensemble}, and therefore $\mX$ also lies in $\uclass{\mu}$. In particular, the rotationally invariant ensemble $\mX_{\mathtt{Haar}}(\mLambda)$ defined in \eqref{eq:ensemble_Haar} lies in the universality class $\uclass{\mu}$, provided that $\mLambda^{1/2}$ satisfies the condition \eqref{eq:singular-value-condition}.
\end{example}

\begin{remark} During the preparation of this manuscript, a recent independent work of \citet{wang2022universality} obtains universality results for the dynamics of AMP algorithms for sign and permutation invariant matrices by using a different proof technique. See \sref{related} for detailed discussions.
\end{remark}
\begin{proof}[Proof of \lemref{sign-perm-inv-ensemble}] Observe we can write $\mX = \mJ \mS$ where $\mJ = \mU \mSigma \mP\tran \mO \tran$. In order to prove the lemma, we need to verify that $\mJ$ satisfies the requirements of \defref{univ-class}. Observe that $\|\mJ\|_{\op} = \|\mSigma\|_{\op} \lesssim 1$ as required. Furthermore for any $k \in \N$, we have:  $(\mJ\tran \mJ)^k \explain{}{=} \mO \mP (\mSigma\tran \mSigma)^k \mP\tran \mO \tran$. In particular:
\begin{align*}
    \frac{\Tr[(\mJ\tran \mJ)^k]}{\dim} & = \frac{\Tr[(\mSigma\tran \mSigma)^k ]}{\dim} \rightarrow  \int \lambda^k \mu(\diff \lambda) \quad \forall \; k \in \N,
\end{align*}
as required by \defref{univ-class}. Furthermore, we note that since $\mP$ is a uniformly random permutation matrix:
\begin{align}\label{eq:invariant-formula}
   \E[(\mJ\tran \mJ)^k] =  \E\left[ \mO \mP (\mSigma\tran \mSigma)^k \mP\tran \mO \tran\right] & = \frac{\Tr[(\mSigma\tran \mSigma)^k]}{\dim} \cdot \mI_{\dim} = \frac{\Tr[(\mJ\tran \mJ)^k]}{\dim} \cdot \mI_{\dim}
\end{align}
Furthermore, a concentration inequality for random permutations due to \citet{bercu2015concentration} (stated as \factref{permutation} in \appref{permutation} for convenience) shows that: 
\begin{align}\label{eq:bercu}
    &\P\left( \|(\mJ\tran \mJ)^k - \E [(\mJ\tran \mJ)^k]\|_{\infty} > K  \|\mO\|_\infty^2  \|\mSigma \tran \mSigma\|_{\op}^k \cdot \left( \sqrt{\dim \ln(\dim)} + \ln(\dim) \right) \right) \leq 4/\dim^2 \quad \forall \; k \in \N,
\end{align}
for some absolute constant $K$. Using the Borel-Cantelli lemma and the fact that $\|\mO\|_{\infty} \lesssim \dim^{-1/2+\epsilon/3}$ we obtain:
\begin{align*}
    P\left( \|(\mJ\tran \mJ)^k - \E [(\mJ\tran \mJ)^k]\|_{\infty} \lesssim  \dim^{-1/2+\epsilon}\right) & = 1 \quad \forall \; k \in \N, \; \epsilon >0.
\end{align*}
Taking a union bound over $k \in \N$, $\epsilon \in \mathbb{Q}$ (the set of rationals) and recalling \eqref{eq:invariant-formula} we obtain:
\begin{align*}
    \P\left( \bigg\| (\mJ \tran \mJ)^k - \frac{\Tr[(\mJ\tran \mJ)^k]}{\dim} \cdot \mI_{\dim}  \bigg\|_{\infty} \lesssim  \dim^{-1/2 + \epsilon} \quad \forall \; k \in \N, \; \epsilon > 0 \right) & = 1.
\end{align*}
Hence, $\mJ$ satisfies the requirements of \defref{univ-class} with probability $1$.
\end{proof}
\subsection{Linear Transformations of I.I.D. Matrices}
The last class of matrices we consider are \emph{left linear transformations of i.i.d. matrices}. These are matrices of the form \begin{align} \label{eq:tiid}
    \mX_{\mathtt{tiid}}(\mT) &\explain{def}{=} \mT \mZ
\end{align} 
where:
\begin{enumerate}
    \item $\mT$ is a $\ssize \times \ssize$ deterministic matrix with bounded operator norm $\|\mT\|_{\op} \lesssim 1$ whose spectral measure $\pi_{\dim}$ (defined below) converges to a compactly supported probability distribution $\pi$:
    \begin{align*}
        \pi_{\dim} \explain{def}{=} \frac{1}{\ssize} \sum_{i=1}^\ssize \delta_{\lambda_i(\mT \tran \mT)} \explain{d}{\rightarrow} \pi,
    \end{align*}
    where $\lambda_1(\mT \tran \mT) \geq \dotsb \geq \lambda_{\ssize} (\mT \tran \mT)$ denote the eigenvalues of $\mT \tran \mT$. 
    \item $\mZ$ is a $\ssize \times \dim$ matrix with converging aspect ratio $\ssize/\dim \rightarrow \alpha$. The rescaled entries of $\mZ$, $\hat{Z}_{ij} \explain{def}{=} \sqrt{\dim} Z_{ij}$ are i.i.d. and satisfy: $\E\hat{Z}_{ij} = 0$, $\E |\hat{Z}_{ij}|^2 = 1$, and have finite moments of all orders. Furthermore, we assume that the entries of $\mZ$ are symmetrically distributed in the sense $Z_{ij} \explain{d}{=} - Z_{ij}$.
\end{enumerate}
This matrix model captures the following important sensing matrices considered in prior works, in a unified manner:
\begin{enumerate}
\item When $\mT = \mI_{\ssize}$, the matrix $\mX_{\mathtt{tiid}}(\mI_{\ssize})$ has symmetrically distributed i.i.d. entries. This model is a frequently studied generalization of the i.i.d. Gaussian sensing model.
\item  When $\mT = \diag(t_1, \dotsc, t_{\ssize})$ is a diagonal matrix, $\mX_{\mathtt{tiid}}(\mT)$ specializes to the \emph{elliptic model}, which has been used to model deviations from the peculiar geometry of high-dimensional i.i.d. matrices \citep{diaconis1984asymptotics,karoui2011geometric,el2018impact}. Specifically, while the $\ell_2$ norms of the rows of an i.i.d. matrix are approximately equal and concentrate to a deterministic value, the row norms in the elliptic model are approximately $|t_1|, \dotsc, |t_{\ssize}|$, and thus, can be widely different.
\item The case of general $\mT$ is less studied, but matrices of this form arise in the analysis of sketching algorithms (see \emph{e.g.}, \citet{liu2019ridge} for an application to ridge regression).
\end{enumerate}
The following lemma shows that $\mX \in \uclass{\pi \boxtimes \mu_{\mathrm{MP}^\alpha}}$, where $\pi \boxtimes \mu_{\mathrm{MP}}^\alpha$ denotes the \emph{free multiplicative convolution} of $\pi$ and $\mu_{\mathrm{MP}}^\alpha$, with the latter being the Marchenko-Pastur distribution \citep{marvcenko1967distribution} with aspect ratio $\alpha$. The measure $\pi \boxtimes \mu_{\mathrm{MP}}^\alpha$ is defined via its Stieltjes transform $m: \C_+ \explain{def}{=} \{z \in \C : \Im(z) > 0\} \mapsto \C_+$:
\begin{align*}
    m(z) \explain{def}{=} \int \frac{\pi \boxtimes\mu_{\mathrm{MP}}^\alpha(\diff \lambda) }{\lambda - z}.
\end{align*}
For any $z \in \C_+$, the Stieltjes transform $m(z)$ of $\pi \boxtimes \mu_{\mathrm{MP}}^\alpha$ is the unique solution in $\C_+$ of the fixed point equation \citep{marvcenko1967distribution}:
\begin{align*}
    \frac{1}{m(z)} & = - z + \alpha \int \frac{\lambda}{1+ \lambda m(z)} \; \pi(\diff \lambda).
\end{align*}
Since a probability measure is uniquely defined by its Stieltjes transform, the above description provides an implicit definition for $\pi \boxtimes \mu_{\mathrm{MP}}^\alpha$.
\begin{lemma}\label{lem:iid} Let $\mX_{\mathtt{tiid}}(\mT) = \mT \mZ$ be a left linear transformation of an i.i.d. matrix $\mZ$ satisfying the hypotheses stated above. Then $\mX_{\mathtt{tiid}}(\mT)  \in \uclass{\pi \boxtimes \mu_{\mathrm{MP}}^\alpha}$. 
\end{lemma}
\begin{proof} This proof of this lemma is provided in \appref{iid} and follows by combining some well-known random matrix theory results on i.i.d. matrices \citep{yin1986limiting,bai2008limit,knowles2017anisotropic}
\end{proof}

Note that the precise distribution of the entries of the i.i.d. matrix $\mZ$ in \eqref{eq:tiid} does not play any role in determining the spectral universality class containing $\mX_{\mathtt{tiid}}(\mT)$. In particular, \lemref{iid} shows that any matrix $\mX_{\mathtt{tiid}}(\mT)$ of the form \eqref{eq:tiid} lies is the same universality class as the \emph{correlated} Gaussian matrix $\mX_{\mathtt{Gauss}}(\mT) = \mT \mG$ where $\mG$ is a i.i.d. Gaussian matrix $G_{ij} \explain{i.i.d.}{\sim} \gauss{0}{1/\dim}$. This \emph{Gaussian universality} of i.i.d. matrices and their simple transformations has been studied in prior work:
\begin{enumerate}
    \item When $\mT = \mI_{\ssize}$, recall that the matrix $\mX_{\mathtt{tiid}}(\mI_{\ssize})$ has symmetrically distributed i.i.d. entries. A long line of work (see, \emph{e.g.}, \citep{korada2011applications,karoui2013asymptotic,panahi2017universal,bayati2015universality, montanari2017universality, oymak2018universality, el2018impact, chen2021universality,celentano2021high,han2022universality}) has shown that i.i.d. sensing matrices  exhibit universality and behave like i.i.d. Gaussian sensing matrices for many inference problems, even without the symmetric distribution requirement. 
    \item When $\mT = \diag(t_1, \dotsc, t_{\ssize})$ is a diagonal matrix, recall that $\mX_{\mathtt{tiid}}(\mT)$ specializes to the \emph{elliptic model}. The Gaussian universality result for the elliptic model was obtained by \citet{el2018impact} (again, without the symmetric distribution requirement).
    \item The case of general $\mT$ does not appear to have been studied in prior work in the context of RLS estimators with general strongly convex regularizers. In this situation, \lemref{iid} and \thref{RLS} show that the performance of RLS estimators on the sensing matrix $\mX_{\mathtt{tiid}}(\mT)$ depends on $\mT$ only via its limiting spectral measure $\pi_{\dim}$. In particular, the singular vectors of $\mT$ do not play a role. 
\end{enumerate}
Hence, \lemref{iid} implies that, via the set of unified deterministic conditions stated in \defref{univ-class}, our universality results not only explain the observed universality of nearly deterministic matrices but also capture the well-understood Gaussian universality of matrices with \emph{symmetric} i.i.d. entries \emph{and} their left linear transformations.

\section{Related Work}\label{sec:related}
\paragraph{Results for Gaussian and Rotationally Invariant Matrices.} The precise analysis of high-dimensional signal estimation has already been the subject of a vast literature (see, \emph{e.g.}, 
\citep{donoho2009message,donoho2010counting,bayati2011lasso, bayati2011dynamics, chandrasekaran2012convex,karoui2013asymptotic,amelunxen2014living,thrampoulidis2015regularized, donoho2016high, dobriban2018high,weng2018overcoming,reeves2019replica, barbier2019optimal, sur2019likelihood,sur2019modern,ma2019optimization,candes2020phase,celentano2020lasso,mignacco2020role,lu2020phase,bu2020algorithmic,li2021minimum}). Historically, sharp asymptotic characterizations were first obtained using statistical physics techniques, especially the non-rigorous replica method (see, \emph{e.g.}, \citep{guo2005randomly,takeda2006analysis,rangan2009asymptotic}). In terms of rigorous development, the seminal works of \citet{donoho2005neighborly,donoho2006high} and \citet{donoho2005sparse,donoho2009counting,donoho2005neighborliness} established the phase transition boundary for the basis pursuit estimator for noiseless compressed sensing with Gaussian sensing matrices using ideas and tools from high-dimensional polytope geometry. Subsequently, several frameworks have been developed to obtain precise asymptotic performance characterizations for high-dimensional inference problems driven by Gaussian or rotationally invariant matrices. These include frameworks based on high-dimensional convex geometry \citep{chandrasekaran2012convex,amelunxen2014living}, comparison inequalities for Gaussian processes \citep{rudelson2008sparse,stojnic2013framework,thrampoulidis2015regularized}, the leave-one-out technique for i.i.d. matrices \citep{karoui2013asymptotic,el2013robust,el2018impact}, and approximate message passing (AMP) algorithms for Gaussian matrices \citep{bolthausen2014iterative, donoho2009message, bayati2011dynamics, donoho2013information, javanmard2013state, berthier2020state, gerbelot2021graph} and rotationally invariant matrices \citep{ma2017orthogonal,rangan2019vector,takeuchi2017rigorous,fan2020approximate,takeuchi2020convolutional,takeuchi2021bayes,liu2022memory}; see \citep{feng2022unifying} for a recent review on AMP algorithms. Although Gaussian and rotationally invariant matrix models are just idealized matrix ensembles primarily chosen for their mathematical tractability, performance characterizations derived from these matrix models can often accurately describe the behavior of matrices that do not satisfy these mathematically convenient properties. This phenomenon is called universality and has been investigated in several works, which we will discuss next. 
\paragraph{Gaussian Universality.} A large body of work has shown that matrices with i.i.d. entries behave like i.i.d. Gaussian matrices in the context of spin glasses (see, \emph{e.g.}, \citep{chatterjee2005simple,carmona2006universality}), random matrix theory (see, \emph{e.g.}, \citep{tao2014random} for a survey), and statistical inference (see, \emph{e.g.}, \citep{korada2011applications,karoui2013asymptotic,panahi2017universal,bayati2015universality, montanari2017universality, oymak2018universality, el2018impact, chen2021universality,celentano2021high,han2022universality}). More recently, a line of work \citep{mei2022generalization,hu2020universality,liang2022precise,montanari2022universality,gerace2022gaussian} has shown that, in the context of inference problems, sensing matrices with independent rows (with possible correlations within a row) behave like Gaussian matrices with independent rows and matching row means and covariance matrices. These works rely on a proof technique known as  Lindeberg’s swapping trick \citep{lindeberg1922neue} or its variants. In a nutshell, this method gradually replaces the independent rows from one matrix ensemble with those from another ensemble (with matching moments). Universality holds if the macroscopic properties of interest remain stable in the swapping process. We note that the sensing ensembles we consider in this work (such as the \texttt{SpikeSine},  \texttt{Mask}, \texttt{RandDCT} ensembles introduced in \sref{demonstration}) all have dependent rows. This seems to preclude the direct application of the standard Lindeberg method in establishing universality.  

\paragraph{Beyond Gaussian Universality.} The behavior of sensing matrices we consider in this work is \emph{not} accurately described by a suitable Gaussian matrix, in general. A \emph{different spectral universality principle} governs the behavior of these matrices which can be described as follows---if the eigenvectors of the sample covariance matrix $\mX \tran \mX$ are sufficiently ``generic'', the sensing matrix $\mX$ has the same asymptotic properties as the (right) rotationally invariant matrix $\mX_{\mathtt{Haar}}(\mLambda)$ defined in \eqref{eq:ensemble_Haar}, where $\mLambda$ is chosen to match the spectrum of $\mX$. This general phenomenon has been \emph{empirically} observed in various contexts. Examples include the work of \citet{marinari1994replica} and \citet{parisi1995mean}, who observed the Sine model, an Ising spin-glass model with a fully deterministic coupling matrix exhibits similar thermodynamic properties as the Random Orthogonal Model (ROM), an Ising model with a rotationally invariant coupling matrix. Another well-known example are the empirical observations of \citet{donoho2009observed} in compressed sensing and various subsequent works \citep{monajemi2013deterministic,oymak2014case,abbara2020universality,ma2021spectral}. This general universality principle has been rigorously established in certain cases. We take this opportunity to review these related investigations in the following paragraphs. 

\paragraph{Compressed Sensing.} \citet{donoho2010counting} have provided a proof for their empirical universality observations \citep{donoho2009observed} in the performance of the Basis Pursuit estimator in noiseless compressed sensing  when the sensing matrix is generic in a suitably defined sense \revised{and the signal is non-negative (coordinate-wise)}. The proof of \citet{donoho2010counting} relies on results from the theory of random polytopes \citep{wendel1962problem,cover1965geometrical,winder1966partitions}, which makes it difficult to extend their approach beyond linear programming-based estimators and the setting of noiseless compressed sensing. 

\paragraph{Random Matrix Theory.} In a different context, similar universality results have appeared in the context of random matrix theory. Given two deterministic $\dim \times \dim$ matrices $\mA$ and $\mB$, Voiculescu \citep{voiculescu1991limit,voiculescu1992free} establishes that $\mA$ and $\mU \tran \mB \mU$ are asymptotically freely independent, when $\mU$ is a random matrix drawn from the Haar distribution on $\mathbb{O}(N)$.  Consequently, the limiting spectral measure of $\mA + \mU\tran \mB \mU$  is determined using the free additive convolution (more generally the limiting spectral measure of any polynomial in $\mA$ and $\mU\tran \mB \mU$ can be determined from the individual spectral measures). A remarkable extension of this result was established by \citet{tulino2010capacity} who proved that if $\mA$ and $\mB$ are random diagonal matrices with i.i.d. entries, $\mA$ is also asymptotically freely independent of $\mF^\star \mB \mF$, where $\mF$ is the $\dim \times \dim$ Fourier matrix. Note that the conjugating matrix $\mF$ is completely deterministic, in contrast to the Haar matrix $\mU$; however, $\mF$ is sufficiently ``generic'' or  ``random-like". These results were, in turn substantially generalized in the work of  \citet{farrell2011limiting} and \citet{anderson2014asymptotically}, who showed that conjugation by any delocalized orthogonal matrix with sign and permutation symmetries induces freeness.  More recently, these freeness results have been leveraged to characterize the performance of the sub-sampled Hadamard-Walsh sketch for ordinary least squares (OLS) regression by \citet{dobriban2019asymptotics} and \citet{lacotte2020optimal}. {These works leverage the explicit formula available for the OLS estimator to relate the performance of the OLS estimator to the spectral measure of a random matrix, which is analyzed using the freeness results of \citet{farrell2011limiting}. This approach does not seem to extend to general RLS estimators since they do not have a convenient explicit closed-form formula like the OLS estimator.}

\paragraph{Linearized AMP Algorithms.} In joint work with Milad Bakshizadeh \citep{dudeja2020universality}, the first author established the universality of a \emph{linearized} version of AMP for the phase retrieval problem. Specifically, this work shows that linearized AMP algorithms for phase retrieval have the same limiting dynamics when the sensing matrix is a randomly sub-sampled Hadamard-Walsh matrix or a randomly sub-sampled Haar matrix. This proof relied on the fact that linearized AMP algorithms can be formulated as a sequence of matrix multiplications, and the proof technique did not apply to general, non-linear AMP algorithms.

\paragraph{AMP algorithms for Semi-Random Matrices.} In recent work \citep{dudeja2022universality}, the authors of this manuscript have obtained a universality principle for a subclass of \emph{non-linear} AMP algorithms called vector approximate message passing (VAMP) algorithms \citep{ma2017orthogonal,rangan2019vector,takeuchi2017rigorous,ccakmak2019memory}. This work identified a notion of \emph{semi-random matrices} such that VAMP algorithms driven by any matrix in this class have the same limiting dynamics. The class of semi-random matrices includes many highly structured matrices, like the ones studied in this paper. As an application, our prior work provided indirect evidence for the previously discussed empirical universality observations of \citet{marinari1994replica} and \citet{parisi1995mean} by showing that a natural iterative algorithm proposed by \citet{ccakmak2019memory} to compute the magnetization of Ising models has the same dynamics in the Sine model and the random orthogonal model (ROM). VAMP algorithms play a crucial role in this paper as well. At the heart of our results in the current paper is a universality principle for VAMP algorithms (\thref{VAMP}), which generalizes the result obtained in our prior work in \citep{dudeja2022universality} in several important ways. First, we now allow VAMP algorithms to use side information in their updates. This feature is necessary to capture inference problems like regularized linear regression, where the signal and the noise are treated as side information. Furthermore, unlike in \citep{dudeja2022universality} where the VAMP algorithms are restricted to be memory-free, we now allow VAMP algorithms to use the entire history of previous iterates in their update rules. This extension enables us to implement general first-order methods (like the proximal method) using VAMP algorithms and hence obtain a universality principle for the proximal method and, subsequently, the RLS estimator. The techniques used to obtain these extensions are discussed in more detail in \sref{univ-vamp}.

\paragraph{Universality of AMP for Sign and Permutation Invariant Matrices.} 
After our earlier work \citep{dudeja2022universality} appeared on arXiv and during the preparation of the current manuscript, a parallel work of \citet{wang2022universality} used a different proof technique to analyze the universality of AMP algorithms for both symmetric and nonsymmetric matrices. By unfolding the AMP iterates in terms of a \emph{tensor network} and through a subsequent elegant argument based on this expansion, these authors obtain a universality principle for AMP algorithms driven by general i.i.d. ensembles and by matrices of the type $\mX = \mJ \mS$, where $\mS$ is a random sign diagonal matrix and $\mJ$ is a matrix whose SVD is of the form
\begin{equation}\label{eq:sp_2}
    \mJ = \mU \mSigma (\mO\mP)\tran.
\end{equation}
Here, the left singular basis $\mU$ is an arbitrary deterministic orthogonal matrix, $\mSigma$ is a deterministic rectangular diagonal matrix consisting of the singular values, and the right singular basis $\mO \mP$ is composed of a delocalized deterministic orthogonal matrix $\mO$ and a uniformly random permutation matrix $\mP$ (that is independent of $\mS$). We note that these matrices are exactly the sign and permutation invariant ensembles studied in \sref{sp_inv}, and they form a sub-class of the spectral universality class introduced in this paper.

A vital feature of the matrices in \eqref{eq:sp_2} is that the singular values of $\mJ$ are matched to the corresponding right singular vectors using a uniformly random permutation. \revised{This randomness from the permutation matrix $\mP$ seems to play an important role in the proof of \citet{wang2022universality}.} In contrast, at the expense of a longer combinatorial argument, the approaches initiated in \citep{dudeja2022universality} and further extended in this paper do not rely on the randomness of $\mJ$. In particular, as one of the main contributions of this work, we identify in \defref{univ-class} fully deterministic and easy-to-verify conditions on $\mJ$ that guarantee universality. This allows the universality principle to capture additional structured sensing ensembles whose singular values and singular vectors are not randomly matched. Examples of these include the \texttt{SpikeSine} ensemble in \eqref{eq:ensemble_ss} (see further generalizations to signed incoherent tight frames in \sref{sitf}), and the \texttt{Mask} ensemble studied in \sref{mask}.

\section{Proof of \thref{RLS}} \label{sec:proof_strategy}
This section is devoted to the proof of \thref{RLS}. As mentioned previously, we take an \emph{algorithmic approach} to prove this result: we derive a universality principle for RLS estimators via a universality principle for iterative algorithms which construct explicit and arbitrarily accurate approximations to the RLS estimator. 

\paragraph{Roadmap.} We begin by presenting a roadmap of the proof:
\begin{enumerate}
    \item In \sref{GFOM}, we introduce a broad class of iterative algorithms called \emph{General First Order Methods} (GFOMs)  \citep{celentano2020estimation}. This class not only includes the proximal method \eqref{eq:prox-method}, but many other iterative algorithms (e.g., gradient descent, accelerated gradient methods, and approximate message passing) for linear regression. We also introduce a restricted sub-class of GFOMs called \emph{Vector Approximate Message Passing Algorithms} (VAMP) \citep{ma2017orthogonal,takeuchi2017rigorous,rangan2019vector}, whose dynamics admit a simple asymptotic characterization.
    \item Since VAMP algorithms are simpler to analyze; we first prove a universality principle for VAMP algorithms. This result is presented as \thref{VAMP} in \sref{vamp-dynamics} and its proof is deferred to \sref{univ-vamp}. 
\end{enumerate}
Assuming this universality principle for VAMP algorithms, we provide a self-contained derivation of a universality principle for GFOMs and RLS estimators in \sref{GFOM-univ-proof} and \sref{RLS-univ-proof} respectively:
\begin{enumerate}
\setcounter{enumi}{2}
    \item In \sref{GFOM-univ-proof} we show that any GFOM can be implemented by non-linear post-processing of a suitably designed VAMP algorithm and use this to derive a universality result for GFOMs  (\thref{GFOM}).
    \item Finally, in \sref{RLS-univ-proof}, we prove \thref{RLS} by showing the proximal gradient method is a GFOM and hence, exhibits universality. Then, we argue that since the proximal method can construct arbitrarily accurate approximations for the RLS estimator, the RLS estimator must also exhibit universality. 
\end{enumerate}

\subsection{General First Order Methods and Vector Approximate Message Passing}
\label{sec:GFOM}
\subsubsection{General First Order Methods (GFOMs)}
The notion of general first-order methods (GFOMs) was introduced in the work of \citet{celentano2020estimation}. While the work of \citet{celentano2020estimation} is concerned with GFOMs driven by i.i.d. Gaussian matrices, this notion is well-defined even for non-i.i.d. matrices. We will find it convenient to work with a generalized definition of  a general first-order method (GFOM), which is specified by:
\begin{enumerate}
    \item The total number of iterations $T \in \N$.
    \item An ordered collection random matrices $\mM_{1:T}$. 
    \item A matrix $\dim \times \auxdim$ matrix $\auxmat$ of auxiliary information with rows $\auxvec_1, \auxvec_2, \dotsc, \auxvec_\dim \in \R^{\auxdim}$.
    \item A collection of nonlinerities $\nonlin_1, \nonlin_2, \dotsc \nonlin_T$ and $\eta_1, \eta_2, \dotsc, \eta_T$ where for each $i \in [T]$, $f_i, \eta_i: \R^{i-1 + \auxdim} \mapsto \R$.
\end{enumerate}
A first order method maintains an iterate $\iter{\vz}{t} \in \R^{\dim}$ for $t \in [T]$ which is updated as follows:
\begin{align}\label{eq:GFOM}
    \iter{\vz}{t} = \mM_t \cdot \nonlin_t(\iter{\vz}{1}, \iter{\vz}{2}, \dotsc, \iter{\vz}{t-1}; \auxmat) + \eta_t(\iter{\vz}{1}, \iter{\vz}{2}, \dotsc, \iter{\vz}{t-1}; \auxmat),
\end{align}
where the non-linearities $\nonlin_{1:T}$ and $\eta_{1:T}$ act entry-wise on their arguments. Hence, $\nonlin_t(\iter{\vz}{1}, \dotsc, \iter{\vz}{t-1}; \auxmat)$ is a vector in $\R^\dim$ with entries:
\begin{align*}
    (\nonlin_t(\iter{\vz}{1}, \iter{\vz}{2}, \dotsc, \iter{\vz}{t-1}; \auxmat))_i & = \nonlin_t(\iter{z}{1}_i, \iter{z}{2}_i, \dotsc, \iter{z}{t-1}_i ; \auxvec_i) \quad \forall \; i \; \in \; [\dim]. 
\end{align*}
This class of iterative algorithms not only includes the proximal method \eqref{eq:prox-method}, but many other iterative algorithms (e.g., gradient descent, accelerated gradient methods, and approximate message passing) for linear regression. Hence, we will seek to obtain an abstract universality result for a GFOM and use it to obtain \thref{RLS}. We will make the following assumption on the auxiliary information matrix $\auxmat$ to obtain our universality result:
\begin{assumption}[I.I.D. Auxiliary Information]\label{assump:side-info} The rows of the auxiliary information matrix $\auxmat$ are i.i.d. copies of a random vector $\serv{A} \in \R^\auxdim$ with  $\E \|\serv{A}\|^p < \infty$ for each $p \in \N$. Furthermore, the distribution of $\serv{A}$ is uniquely determined by its moments. 
\end{assumption}

To characterize the dynamics of a GFOM (or the limiting behavior of the GFOM iterates), we will use the following notion of convergence of high-dimensional vectors:

\begin{definition}[Convergence of Empirical Distributions]\label{def:PW2} A collection of $k$ random vectors $(\iter{\vv}{1}, \dotsc, \iter{\vv}{k})$ in $\R^{\dim}$ and the auxiliary information matrix $\auxmat$ converge with respect to the Wasserstein-$2$ metric to a random vector $(\serv{V}_1, \serv{V}_2, \dotsc, \serv{V}_k ; \serv{A}) \in \R^{k+\auxdim}$ in probability as $\dim \rightarrow \infty$, if for any continuous test function $h: \R^{k+\auxdim} \rightarrow \R$  (independent of $\dim$) that satisfies:
\begin{align*}
    |h(x;\auxvec) - h(y;\auxvec)| & \leq L \cdot \|x - y\| \cdot  (1 + \|x\| + \|y\| + \|\auxvec\|^D) \; \forall \; x, y \; \in \; \R^k, \; \auxvec \; \in \; \R^{\auxdim}, \\
    |h(x; \auxvec)| & \leq L \cdot (1 + \|x\|^D + \|\auxvec\|^D)
\end{align*}
for some finite constants $L \geq 0$ and $D \in \N$, we have,
\begin{align*}
    \frac{1}{\dim} \sum_{i=1}^\dim h(\iter{v}{1}_i, \iter{v}{2}_i, \dotsc, \iter{v}{k}_i; \auxvec_i) \explain{P}{\rightarrow} \E h(\serv{V}_1, \serv{V}_2, \dotsc, \serv{V}_k; \serv{A}).
\end{align*}
We denote convergence in this sense using the notation $(\iter{\vv}{1}, \iter{\vv}{2}, \dotsc, \iter{\vv}{k}; \auxmat) \explain{\pw}{\longrightarrow} (\serv{V}_1, \serv{V}_2, \dotsc, \serv{V}_k; \serv{A})$. 
\end{definition}

\subsubsection{Vector Approximate Message Passing (VAMP) Algorithms}
To derive a universality principle for GFOMs, we first prove a universality result for a restricted class of first-order methods called vector approximate message passing (VAMP), whose update rule takes the simpler form:
\begin{align}\label{eq:VAMP}
    \iter{\vz}{t} = \mM_t \cdot \nonlin_t(\iter{\vz}{1}, \iter{\vz}{2}, \dotsc, \iter{\vz}{t-1}; \auxmat).
\end{align}
In addition, further requirements are imposed on the matrices $\mM_{1:T}$ and the non-linearities $\nonlin_{1:T}$ (to be introduced momentarily). These algorithms were introduced in the work of \citet{ma2017orthogonal}, \citet{rangan2019vector} and \citet{takeuchi2017rigorous}, who consider the situation when the matrices $\mM_{1:T}$ are rotationally invariant and show that the dynamics of VAMP algorithms admit a simple asymptotic characterization called the \emph{state evolution}. However, many authors (\emph{see e.g.,} \citep{ma2017orthogonal,ccakmak2019memory,abbara2020universality}) have empirically observed that the state evolution appears to hold even when $\mM_{1:T}$ are not rotationally invariant, but highly structured and have limited randomness. 

\paragraph{Restrictions on the Matrix Ensemble.} One of the key contributions of this paper is the identification of nearly deterministic conditions on $\mM_{1:T}$ that guarantee the validity of the state evolution. These conditions are stated in the following notion of a \emph{semi-random ensemble}. 
\begin{definition}[Semi-random Ensemble] \label{def:semirandom} A semi-random ensemble is a collection of  matrices $\mM_{1:T}$ of the form $\mM_i = \mS \mPsi_i \mS$ where:
\begin{enumerate}
    \item $\mS$ is a diagonal matrix consisting of i.i.d. random signs: $\mS = \diag(s_{1:\dim})$, $s_{1:\dim} \explain{i.i.d.}{\sim} \unif{\{\pm 1\}}$. 
    \item $\mPsi_{1:T}$ are deterministic $\dim \times \dim$ matrices which satisfy:
    \begin{enumerate}
        \item $\max_{i \in [T]}\|\mPsi_i\|_{\op} \lesssim 1$.
        \item The matrix $\hat{\Omega} \in \R^{T \times T}$ with entries $\hat{\Omega}_{ij} \explain{def}{=} \Tr(\mPsi_i \mPsi_j \tran ) /\dim$ converges to a matrix $\Omega \in \R^{T \times T}$ as $\dim \rightarrow \infty$. 
        \item For any fixed $\epsilon > 0$ (independent of $\dim$), 
        \begin{align*} \max_{i,j \in [T]}  \| \mPsi_i \mPsi_j \tran - \dim^{-1} \cdot \Tr(\mPsi_i \mPsi_j \tran ) \cdot \mI_{\dim} \|_\infty \lesssim \dim^{-1/2 + \epsilon}.
        \end{align*}
        \item Lastly, for any fixed $\epsilon > 0$ (independent of $\dim$), $ \max_{i \in [T]}\|\mPsi_i\|_\infty\lesssim \dim^{-1/2+\epsilon}$.
        \end{enumerate}
\end{enumerate}
We call the matrix $\hat{\Omega}$ the empirical covariance matrix of the semi-random ensemble and the matrix $\Omega$ the limiting covariance matrix of the semi-random ensemble $\mM_{1:k}$.  
\end{definition}
\begin{remark} Observe that if $\mM_{1:T}$ form a semi-random ensemble, then because of the requirement (2d) in \defref{semirandom}, for any $i \in [T]$, $\Tr(\mM_i)/\dim = \Tr(\mPsi_i)/\dim \rightarrow 0$. In particular, the matrices $\mM_{1:T}$ are asymptotically ``trace-free''. This is a well-known hallmark of VAMP algorithms and plays an important role in the analysis of these algorithms in the rotationally invariant case \citep{ma2017orthogonal,takeuchi2017rigorous,rangan2019vector}. 
\end{remark}

Before we can formally introduce the restrictions imposed on the non-linearities $\nonlin_{1:T}$ used in a VAMP algorithm \eqref{eq:VAMP}, we will need to introduce the \emph{state evolution} associated with a VAMP algorithm, which characterizes the asymptotic dynamics of the algorithm in the high-dimensional limit $\dim \rightarrow \infty$.
\paragraph{State Evolution of a VAMP Algorithm.} Every VAMP algorithm is associated with $T$ mean zero Gaussian random variables:
\begin{subequations} \label{eq:SE-VAMP}
\begin{align}
    (\serv{Z}_1, \serv{Z}_2, \dotsc, \serv{Z}_T) \sim \gauss{0}{\Sigma_T},
\end{align} 
which will describe the  asymptotic behavior of the VAMP iterates $\iter{\vz}{1}, \dotsc, \iter{\vz}{T}$. The  covariance matrix $\Sigma_T$ is determined using the following recursion:
\begin{align} 
   (\Sigma_{T})_{s,t+1} \explain{def}{=} \E[\serv{Z}_{t+1} \serv{Z}_s] & =  \Omega_{t+1,s} \cdot \E[\nonlin_{t+1}(\serv{Z}_1, \dotsc, \serv{Z}_t; \serv{A}) \nonlin_{s}(\serv{Z}_1, \dotsc, \serv{Z}_{s-1}; \serv{A})] \quad \forall \; s \; \leq \; t+1.
\end{align}
In the above display $\serv{A}$ is the auxiliary information random variable from \assumpref{side-info}, independent of $\serv{Z}_{1}, \dotsc, \serv{Z}_{T}$ and $\Omega$ is the limiting covariance matrix of the semi-random ensemble $\mM_{1:T}$ (\defref{semirandom}) driving the VAMP algorithm \eqref{eq:VAMP}. We will also find it useful to define another $T \times T$ covariance matrix $\Phi_T$ whose entries are given by:
\begin{align}
    (\Phi_{T})_{s,t} & = \E[\nonlin_{s}(\serv{Z}_1, \dotsc, \serv{Z}_{s-1}; \serv{A}) \nonlin_{t}(\serv{Z}_1, \dotsc, \serv{Z}_{t-1}; \serv{A})] \quad \forall \; s,t \; \in \; [T].
\end{align}
We will refer to  $\serv{Z}_1, \serv{Z}_2, \dotsc, \serv{Z}_T$  as the \emph{Gaussian state evolution random variables}, $\Sigma_T$ as the \emph{Gaussian state evolution covariance} and $\Phi_T$ as the \emph{non-Gaussian state evolution covariance}. We can now formally introduce the restrictions imposed on the non-linearities $\nonlin_{1:T}$ that can be used in a VAMP algorithm.
\end{subequations}

\paragraph{Restriction on Non-linearities.} The state evolution associated with a VAMP algorithm determines the asymptotic behavior of its iterates as $\dim \rightarrow \infty$, provided the non-linearities $\nonlin_{1:T}$ are ``divergence-free'' in the following sense. 
\begin{assumption}[Divergence-Free Non-Linearities]\label{assump:div-free} For each $t \in [T]$, the non-linearity $\nonlin_t$ satisfies:
\begin{align*}
    \E[ \serv{Z}_s \nonlin_t(\serv{Z_1}, \dotsc, \serv{Z}_{t-1}; \serv{A})] & = 0 \quad \forall \; s \; \in \;  [t-1],
\end{align*}
where $\serv{A}$ is the auxiliary information random variable from \assumpref{side-info} and $\serv{Z}_1, \dotsc, \serv{Z}_T$ are the Gaussian state evolution random variables from \eqref{eq:SE-VAMP}. Furthermore, $\serv{A}$ and $(\serv{Z}_1, \dotsc, \serv{Z}_T)$ are independent.
\end{assumption}
This requirement is another well-known hallmark of VAMP algorithms and was introduced in the works of \citet{ma2017orthogonal}, \citet{rangan2019vector} and \citet{takeuchi2017rigorous} in the context of VAMP algorithms driven by rotationally invariant matrices. This concludes our formal definition of VAMP algorithms. 

\subsection{A Universality Principle for Vector Approximate Message Passing}
\label{sec:vamp-dynamics}

We now state our result characterizing the asymptotic dynamics of a VAMP algorithm in terms of the Gaussian state evolution random variables $\serv{Z}_{1:T}$.

\begin{theorem}\label{thm:VAMP} Let $T \in \N$ be fixed (independent of $\dim$). Consider $T$ iterations of the VAMP algorithm in \eqref{eq:VAMP}. Suppose that:
\begin{enumerate}
    \item $\mM_{1:T}$ form a semi-random ensemble (\defref{semirandom}). 
    \item The auxiliary information matrix $\auxmat$ satisfies \assumpref{side-info}. 
    \item The non-linearities $\nonlin_{1:T}$ are continuous functions independent of $\dim$ and satisfy the divergence-free assumption (\assumpref{div-free}). Furthermore, they are uniformly Lipschitz and polynomially bounded in the sense that there are finite constants $L \in (0,\infty)$ and $\degree \in \N$ such that:
    \begin{align*}
        |\nonlin_i(z; \auxvec) - \nonlin_i(z^\prime; \auxvec)| & \leq L \cdot  \|z - z^\prime\|_2 \quad \forall \; i \in [T], \; \auxvec \in \R^{\auxdim}, \; z, z^\prime \in \R^{i-1}, \\
        |\nonlin_i(z;\auxvec)| & \leq L \cdot (1 + \|z\|^{\degree} + \|\auxvec\|^{\degree}) \quad \forall \; i \in [T], \; \auxvec \in \R^{\auxdim}, \; z \in \R^{i-1}.
    \end{align*}
\end{enumerate}
Then,
\begin{align*}
    (\iter{\vz}{1}, \iter{\vz}{2}, \dotsc, \iter{\vz}{T} ; \auxmat) \explain{\pw}{\longrightarrow} (\serv{Z}_1, \serv{Z}_2, \dotsc, \serv{Z}_T; \serv{A}).
\end{align*}
In the above display, $\serv{A}$ is the auxiliary information random variable from \assumpref{side-info} and $\serv{Z}_1, \dotsc, \serv{Z}_T$ are the Gaussian state evolution random variables from \eqref{eq:SE-VAMP}. Furthermore, $\serv{A}$ and $(\serv{Z}_1, \dotsc, \serv{Z}_T)$ are independent. 
\end{theorem}

\thref{VAMP} can be interpreted as a universality result since it shows that the limiting empirical distribution of the VAMP iterates depends on the semi-random ensemble $\mM_{1:T}$ only via its limiting covariance matrix $\Omega$. Hence, the dynamics of a VAMP algorithm are asymptotically identical on two semi-random ensembles which have the same limiting covariance matrix. 

Observe that the dynamics of VAMP algorithms are universal under very weak conditions on $\mM_{1:T}$ (for instance, compare the requirements on $\mM_{1:T}$ in \thref{VAMP} with the requirements on the covariance matrix $\mX\tran \mX$ of the sensing matrix $\mX$ in \thref{RLS}), which makes it easier to prove a universality principle for VAMP algorithms under nearly deterministic conditions on $\mM_{1:T}$. This is why VAMP algorithms play a central role in our work. We postpone the proof of \thref{VAMP} to \sref{univ-vamp} and instead derive the universality results for GFOMs and RLS estimators from this result first.

\subsection{A Universality Principle for General First Order Methods}
\label{sec:GFOM-univ-proof}
In this section, we derive a universality principle for a general first-order method (GFOM) of the form \eqref{eq:GFOM}. To do so, we take inspiration from an argument of \citet{celentano2020estimation}. Specifically, we show that given a GFOM of the form \eqref{eq:GFOM}, one can design a VAMP algorithm (of the form \eqref{eq:VAMP}) such that the iterates of the GFOM can be obtained by a non-linear transformation (post-processing) of the VAMP iterates. We then argue that since the dynamics of VAMP are universal (by \thref{VAMP}), the dynamics of the GFOM must also universal. In order to have universal dynamics, GFOMs \eqref{eq:GFOM} require stronger requirements on the driving matrix ensemble $\mM_{1:T}$ than those stated in \defref{semirandom} (semi-random ensemble). This is because the matrices used in the VAMP algorithm that implements a given GFOM are products of the matrices $\mM_{1:T}$ that drive the given GFOM (this will be made clear in the proof of \thref{GFOM}, which will be stated and proved momentarily). To ensure that the products of $\mM_{1:T}$ are semi-random (in the sense of \defref{semirandom}),  we require that $\mM_{1:T}$ satisfy the following notion of \emph{strongly semi-random matrices}. 

\begin{definition}[Strongly Semi-Random Matrices]\label{def:strong-semirandom} A collection of matrices $\mM_{1:T}$ is strongly semi-random if the matrices are of the form $\mM_i = \mS \mPsi_i \mS$ where:
\begin{enumerate}
    \item $\mS$ is a diagonal matrix consisting of i.i.d. random signs: $\mS = \diag(s_{1:\dim})$, $s_{1:\dim} \explain{i.i.d.}{\sim} \unif{\{\pm 1\}}$. 
    \item $\mPsi_{1:T}$ are deterministic $\dim \times \dim$ matrices which satisfy $\max_{i \in T}\|\mPsi_i\|_{\op} \lesssim 1$.  
    \item Furthermore, the matrices $\mPsi_{1:T}$ have the property that for any subsets $B, B^\prime \subset [T]$ there exists a constant $\Omega_{B,B^\prime}$ (independent of $\dim$) such that:
\begin{align} \label{eq:ssr-empirical}
    \hat{\Omega}_{B,B^\prime} \explain{def}{=} \frac{\Tr( \mPsi_B \mPsi_{B^\prime} \tran )}{\dim} \rightarrow \Omega_{B,B^\prime} \quad \text{ as $\dim \rightarrow \infty$.}
\end{align}
and
\begin{align} \label{eq:ssr-Id-approx}
   \max_{B, B^\prime \subset [T]} \| \mPsi_B \mPsi_{B^\prime} \tran - \hat{\Omega}_{B,B^\prime} \mI_{\dim} \|_\infty \lesssim \dim^{-1/2 + \epsilon} \; \forall \; \epsilon > 0.
\end{align}
\end{enumerate}
In the above display, for any subset $B \subset [T]$ with sorted elements $b_1 < b_2 < \dotsb < b_{|B|}$ we defined the matrix $\mPsi_B$ as:
\begin{align*}
    \mPsi_B \explain{def}{=} \mPsi_{b_{|B|}} \cdot \mPsi_{b_{|B|-1}} \cdot \dotsb \cdot \mPsi_{b_{1}}.
\end{align*}
When $B = \emptyset$, we define $\mPsi_{\emptyset} \explain{def}{=} \mI_\dim$. The constants $\{\hat{\Omega}_{B,B^\prime} : B, B^\prime \subset [T]\}$ are called the empirical moments of $\mM_{1:T}$ and the constants $\{{\Omega}_{B,B^\prime} : B, B^\prime \subset [T]\}$ are called the limiting moments. 
\end{definition}
\noindent Before stating our universality result for GFOMs, we clarify some aspects of \defref{strong-semirandom} in the remarks below.
\begin{remark}[semi-random v.s. strongly semi-random ensembles] Notice that requirements \eqref{eq:ssr-empirical} and \eqref{eq:ssr-Id-approx} can be viewed as stronger analogs of the requirements (3b) and (3c) imposed in the definition of a semi-random ensemble (\defref{semirandom}). Indeed, setting $B, B^\prime$ as the singleton sets $B = \{i\}$ and $B^\prime = \{j\}$ for some $i,j \in [T]$ yields requirements (3b) and (3c) in \defref{semirandom}. This is why we call a matrix ensemble that satisfies the requirements of \defref{strong-semirandom} \emph{strongly} semi-random. However, an important caveat is that a strongly semi-random ensemble $\mM_{1:T} = \mS \mPsi_{1:T} \mS$ need not be semi-random in the sense of \defref{semirandom}. The reason is that \defref{semirandom} requires that the matrices $\mPsi_{1:T}$ satisfy the delocalization estimate $\max_{i \in [T]}\|\mPsi_i\|_\infty\lesssim \dim^{-1/2+\epsilon}$ for any $\epsilon > 0$. However, this requirement is not imposed in \defref{strong-semirandom}. 
\end{remark}

\begin{remark}[strongly semi-random matrices and the spectral universality class] \label{rem:ssr}  For regularized linear regression, many natural first-order methods are driven by a matrix ensemble $\mM_{1:T}$ where for each $i \in [T]$, $\mM_i$ is a matrix polynomial in $\mX\tran \mX$ (recall $\mX$ is the sensing matrix). That is, $\mM_i = p_i(\mX\tran \mX)$ for some polynomial $p_i : \R \mapsto \R$. In \sref{RLS-univ-proof}, we will show that the proximal method has this form. In this situation, if $\mX = \mJ \mS$ lies in a spectral universality class $\uclass{\mu}$ (\defref{univ-class}) for some compactly supported probability measure $\mu$ on $[0,\infty)$, then $\mM_{1:T}$ is strongly semi-random. To see this, notice that requirements (1) and (2) in \defref{strong-semirandom} are immediately satisfied thanks to requirements (1) and (2) in \defref{univ-class}. In order to verify requirements \eqref{eq:ssr-empirical} and \eqref{eq:ssr-Id-approx}, consider the special case where $\mM_i = (\mX\tran \mX)^{k_i}$ for $k_1, \dotsc, k_T \in \W$ (that is, $p_i$'s are monomials). In this case, for any $B, B^\prime \subset [T]$, $\mPsi_{B} \mPsi_{B^\prime} \tran = (\mJ\tran \mJ)^{\kappa(B) + \kappa(B^\prime)}$ where for any $B \subset [T]$, we defined $\kappa(B) = \sum_{i \in B} k_i$. Hence, \eqref{eq:ssr-empirical} and \eqref{eq:ssr-Id-approx} are satisfied thanks to requirements (3) and (4) in \defref{univ-class}. Since any polynomial can be expressed as a linear combination of monomials, this argument extends to the general case when $p_i$'s are arbitrary polynomials. 
\end{remark}

The following theorem provides a universality principle for GFOMs.

\begin{theorem}\label{thm:GFOM} Let $T \in \N$ be fixed (independent of $\dim$). Consider $T$ iterations of a general first order method of the form \eqref{eq:GFOM}. Suppose that:
\begin{enumerate}
    \item $\mM_{1:T} = \mS \mPsi_{1:T} \mS$ form a strongly semi-random ensemble (\defref{strong-semirandom}) with limiting moments $\{\Omega_{B,B^\prime} : B, B^\prime \subset [T]\}$.
    \item The auxiliary information matrix $\auxmat$ satisfies \assumpref{side-info}. 
    \item The non-linearities $\nonlin_{1:T}, \eta_{1:T}$ are continuous functions independent of $\dim$. Furthermore, they are uniformly Lipschitz and polynomially bounded in the sense that there are finite constants $L \in (0,\infty)$ and $\degree \in \N$ such that:
    \begin{align*}
        |\nonlin_i(z; \auxvec) - \nonlin_i(z^\prime; \auxvec)| \vee |\eta_i(z; \auxvec) - \eta_i(z^\prime; \auxvec)| & \leq L \cdot  \|z - z^\prime\|_2 \quad \forall \; i \in [T], \; \auxvec \in \R^{\auxdim}, \; z, z^\prime \in \R^{i-1}, \\
        |\nonlin_i(z;\auxvec)| \vee |\eta_i(z;\auxvec)| & \leq L \cdot (1 + \|z\|^{\degree} + \|\auxvec\|^{\degree}) \quad \forall \; i \in [T], \; \auxvec \in \R^{\auxdim}, \; z \in \R^{i-1}.
    \end{align*}
\end{enumerate}
Then, there exist random variables $(\serv{Z}_1, \serv{Z}_2, \dotsc, \serv{Z}_T; \serv{A})$ whose joint distribution is completely determined by $\nonlin_{1:T}$, $\eta_{1:T}$ and $\{\Omega_{B,B^\prime} : B, B^\prime \subset [T]\}$ such that:
\begin{align*}
    (\iter{\vz}{1}, \iter{\vz}{2}, \dotsc, \iter{\vz}{T}; \auxmat) \explain{\pw}{\longrightarrow} (\serv{Z}_1, \serv{Z}_2, \dotsc, \serv{Z}_T; \serv{A}).
\end{align*}
In the above display, $\serv{A}$ is the auxiliary information random variable from \assumpref{side-info}.
\end{theorem}

\thref{GFOM} is a universality result since it shows that the limiting empirical distribution of the GFOM iterates depends on the strongly semi-random ensemble $\mM_{1:T}$ only via its limiting moments $\{\Omega_{B,B^\prime} : B, B^\prime \subset [T]\}$. Hence, the dynamics of a GFOM are asymptotically identical on two strongly semi-random ensembles which have the same limiting moments. 

\begin{proof}[Proof Sketch of \thref{GFOM}] To illustrate the main idea behind the proof of \thref{GFOM}, we design a VAMP algorithm that implements the first two iterations of a given GFOM and describe how the universality principle for VAMP algorithms (\thref{VAMP}) implies \thref{GFOM}. The complete proof of \thref{GFOM} is provided in \appref{VAMP-to-GFOM}. Consider a $T$-iteration GFOM driven by a strongly semi-random ensemble $\mM_{1:T} = \mS \mPsi_{1:T} \mS$ with empirical moments $\{\hat{\Omega}_{B,B^\prime}: B, B^\prime \subset [T]\}$ and  limiting moments $\{\Omega_{B,B^\prime}: B, B^\prime \subset [T]\}$ (cf. \defref{strong-semirandom}):
\begin{align}\label{eq:GFOM-to-implement}
    \iter{\vz}{t} = \mM_t \cdot \nonlin_t(\iter{\vz}{1}, \iter{\vz}{2}, \dotsc, \iter{\vz}{t-1}; \auxmat) + \eta_t(\iter{\vz}{1}, \iter{\vz}{2}, \dotsc, \iter{\vz}{t-1}; \auxmat) \quad t \in [T].
\end{align}
The idea is to construct the VAMP algorithm implementing \eqref{eq:GFOM-to-implement} inductively.
\paragraph{Iteration 1.} Observe that the first iteration of the GFOM \eqref{eq:GFOM-to-implement}
\begin{align} \label{eq:GFOM-it-1}
    \iter{\vz}{1} & = \mS \mPsi_1 \mS \nonlin_1(\auxmat) + \eta_1(\auxmat) \explain{(a)}{=} \mS \cdot (\mPsi_1 - \hat{\Omega}_{\{1\},\emptyset} \mI_{\dim}) \cdot \mS \nonlin_1(\auxmat) +  \hat{\Omega}_{\{1\},\emptyset} \cdot \nonlin_1(\auxmat) + \eta_1(\auxmat).
\end{align}
In the above display $\hat{\Omega}_{\{1\},\emptyset} \explain{def}{=} \Tr(\mPsi_1)/\dim$ is an empirical moment of the strongly semi-random ensemble $\mM_{1:T} = \mS \mPsi_{1:T} \mS$ (cf. \defref{strong-semirandom}). The rational behind the centering done in step (a) is that we cannot directly use the matrix $\mM_1 = \mS \mPsi_1 \mS$ in the VAMP algorithm since it need not satisfy the delocalization requirement $\|\mPsi_1\|_{\infty} \lesssim \dim^{-1/2+\epsilon}$ imposed on a semi-random matrix (\defref{semirandom}). On the other hand, the centered matrix $\mPsi_1 - \hat{\Omega}_{\{1\},\emptyset} \mI_{\dim}$ does satisfy this requirement (recall \defref{strong-semirandom}). We construct the first iteration of the VAMP algorithm as:
\begin{align} \label{eq:VAMP-implement-iter1}
    \iter{\vw}{1} & = \mQ_1 g_1(\auxmat)
\end{align}
where:
\begin{enumerate}
    \item  $\mQ_1 \explain{def}{=} \mS \cdot \mXi_1 \cdot \mS$ with $\mXi_1 \explain{def}{=} \mPsi_1 - \hat{\Omega}_{\{1\},\emptyset} \mI_{\dim}$. Observe that $\mQ_1$ is a semi-random matrix (\defref{semirandom}) since the conditions imposed on $\mPsi_{1:T}$ in \defref{strong-semirandom} guarantee:
    \begin{subequations} \label{eq:semi-random-argument}
    \begin{align}
        \|\mXi_1\|_{\infty}  &= \|\mPsi_1 - \hat{\Omega}_{\{1\},\emptyset} \mI_{\dim}\|_{\infty} \lesssim \dim^{-1/2+\epsilon}, \\
        \|\mXi_1 \mXi_1\tran - \dim^{-1} \cdot \Tr(\mXi_1 \mXi_1\tran) \cdot \mI_{\dim}\|_{\infty} & =
        \|\mXi_1 \mXi_1\tran - (\hat{\Omega}_{\{1\},\{1\}} - \hat{\Omega}_{\{1\},\emptyset}^2) \cdot \mI_{\dim}\|_{\infty} \\ & = \left\|(\mPsi_1 \mPsi_1 \tran- \hat{\Omega}_{\{1\},\{1\}} \cdot \mI_{\dim}) -  \hat{\Omega}_{\{1\}, \emptyset} \cdot (\mPsi_1  + \mPsi_1\tran - 2\hat{\Omega}_{\{1\},\emptyset} \mI_{\dim}  )  \right\|_{\infty} \\
        & \leq \left\|\mPsi_1 \mPsi_1 \tran- \hat{\Omega}_{\{1\},\{1\}} \cdot \mI_{\dim} \right\|_\infty + 2 |\hat{\Omega}_{\{1\}, \emptyset}| \cdot \left\|\mPsi_1 - \hat{\Omega}_{\{1\},\emptyset} \mI_{\dim}  \right\|_{\infty} \\
        & \lesssim \dim^{-1/2+\epsilon}.
    \end{align}
    \end{subequations}
    \item The non-linearity $g_1:\R^{\auxdim} \mapsto \R$ is given by $g_1(\auxvec) \explain{def}{=} \nonlin_1(\auxvec)$. 
\end{enumerate}
Define a post-processing function $H_1: \R^{1+\auxdim} \mapsto \R$ as:
\begin{align*}
        H_1(w_1; \auxvec) \explain{def}{=} w_1 +  \Omega_{\{1\},\emptyset} \cdot \nonlin_1(\auxvec) + \eta_1(\auxvec),
\end{align*}
where $\Omega_{\{1\},\emptyset} \explain{def}{=} \lim_{\dim \rightarrow \infty} \Tr(\mPsi_1)/\dim$ is a limiting moment of the strongly semi-random ensemble $\mM_{1:T} \explain{def}{=} \mS \mPsi_{1:T} \mS$ (cf. \defref{strong-semirandom}). Recalling \eqref{eq:GFOM-it-1}, the above definitions ensure that 
\begin{align}\label{eq:post-process-iter1}
    \iter{\vz}{1} &=  \iter{\vw}{1} +  \hat{\Omega}_{\{1\},\emptyset} \cdot \nonlin_1(\auxmat) + \eta_1(\auxmat) \explain{(a)}{\approx}  \iter{\vw}{1} +  {\Omega}_{\{1\},\emptyset} \cdot \nonlin_1(\auxmat) + \eta_1(\auxmat) = H_1(\iter{\vw}{1}; \auxmat),
\end{align}
where the approximation in step (a) follows from the fact the convergence $\hat{\Omega}_{\{1\},\emptyset} \rightarrow {\Omega}_{\{1\},\emptyset}$ (see \appref{VAMP-to-GFOM} for the formal justification of this approximation). Hence, we have constructed the first iteration of the desired VAMP algorithm which implements one iteration of the given GFOM using a simple post-processing step.

\paragraph{Iteration 2.} Next, we consider the second iteration of the GFOM:
\begin{align}
    \iter{\vz}{2} = \mS \mPsi_{2} \mS \cdot \nonlin_{2}(\iter{\vz}{1}; \auxmat) + \eta_{2}(\iter{\vz}{1}; \auxmat)  &\explain{\eqref{eq:VAMP-implement-iter1}}{\approx} \mS \mPsi_{2} \mS \cdot \nonlin_{2}(H_1(\iter{\vw}{1}; \auxmat); \auxmat); \auxmat) + \eta_{2}(H_1(\iter{\vw}{1}; \auxmat); \auxmat) \nonumber \\
    & = \mS \mPsi_{2} \mS \cdot \tilde{\nonlin}_{2}(\iter{\vw}{1}; \auxmat) + \tilde{\eta}_{2}(\iter{\vw}{1}; \auxmat) \label{eq:GFOM-to-implement-iter2}
\end{align}
where in the last equation we defined the composite non-linearities $\tilde{\nonlin}_{2}, \tilde{\eta}_{2} : \R^{1 + \auxdim} \mapsto \R$ as:
\begin{align*}
    \tilde{\nonlin}_{2}(w_1; \auxvec) \explain{def}{=} \nonlin_{2}(H_1(w_1; \auxvec); \auxvec), \quad
    \tilde{\eta}_{2}(w_1; \auxvec) \explain{def}{=} \nonlin_{2}(H_1(w_1; \auxvec); \auxvec).
\end{align*}
Note that we cannot use the composite non-linearity $\tilde{\nonlin}_2$ directly in the VAMP algorithm since it need not be divergence free (\assumpref{div-free}). However, this can be addressed by correcting $\tilde{\nonlin}_2$ with a linear function to ensure it becomes divergence-free. Indeed, if $\serv{W}_1$ denotes the Gaussian state evolution random variable (recall \eqref{eq:SE-VAMP}) corresponding to the first iteration of the VAMP algorithm \eqref{eq:VAMP-implement-iter1} and $\serv{A}$ is the auxiliary information random vector (\assumpref{side-info}) independent of $\serv{W}_1$, then the function:
\begin{align}\label{eq:g_2-div-free}
    g_2(w_1; \auxmat) \explain{def}{=} \tilde{\nonlin}_{2}(w_1; \auxvec) - \beta_2 w_1, \quad \beta_2 \explain{def}{=} \frac{\E[\tilde{\nonlin}_{2}(\serv{W}_1; \serv{A}) \cdot \serv{W}_1]}{\E[\serv{W}_1^2]},
\end{align}
is divergence-free in the sense of \assumpref{div-free}. Hence, we can express  \eqref{eq:GFOM-to-implement-iter2} as:
\begin{align}
    \iter{\vz}{2} &\approx \mS \mPsi_{2} \mS \cdot g_{2}(\iter{\vw}{1}; \auxmat) + \beta_2 \cdot  \mS \mPsi_2 \mS \cdot \iter{\vw}{1} +  \tilde{\eta}_{2}(\iter{\vw}{1}; \auxmat) \nonumber\\
    & \explain{\eqref{eq:VAMP-implement-iter1}}{=} \mS \mPsi_{2} \mS \cdot g_{2}(\iter{\vw}{1}; \auxmat) + \beta_2 \cdot  \mS \mPsi_2 \mPsi_1 \mS \cdot g_{1}(\auxmat) - \beta_2 \hat{\Omega}_{\{1\},\emptyset} \cdot \mS \mPsi_2 \mS \cdot g_{1}(\auxmat)  +   \tilde{\eta}_{2}(\iter{\vw}{1}; \auxmat) \nonumber\\
    & \explain{(a)}{=} \mS (\mPsi_{2} -\hat{\Omega}_{\{2\},\emptyset} \cdot \mI_{\dim}) \mS \cdot g_{2}(\iter{\vw}{1}; \auxmat) + \beta_2 \cdot  \mS (\mPsi_2 \mPsi_1 - \hat{\Omega}_{\{1,2\},\emptyset}\cdot \mI_{\dim}) \mS \cdot g_{1}(\auxmat) + \hat{\Omega}_{\{2\},\emptyset} \cdot g_2(\iter{\vw}{1}; \auxmat)  \nonumber \\& \hspace{0.7cm} - \beta_2 \cdot  \hat{\Omega}_{\{1\},\emptyset} \cdot \mS (\mPsi_2-\hat{\Omega}_{\{2\},\emptyset}\cdot \mI_{\dim}) \mS \cdot g_{1}(\auxmat)  + \beta_2 \cdot (\hat{\Omega}_{\{1,2\},\emptyset} - \hat{\Omega}_{\{1\},\emptyset} \cdot \hat{\Omega}_{\{2\},\emptyset} ) \cdot g_1(\auxmat)  +   \tilde{\eta}_{2}(\iter{\vw}{1}; \auxmat), \label{eq:VAMP-implement-iter2-prelim} 
\end{align}
where step (a) follows from appropriately centering the matrices $\mPsi_2, \mPsi_2 \mPsi_1$ so that their centered versions satisfy the requirements of \defref{semirandom} (analogous to the argument used in \eqref{eq:GFOM-it-1}). In light of \eqref{eq:VAMP-implement-iter2-prelim}, we construct the next two iterations of the VAMP algorithm as:
\begin{align} \label{eq:VAMP-implement-iter2}
    \iter{\vw}{2} & = \mQ_2 \cdot g_2(\iter{\vw}{1}; \auxmat), \quad \iter{\vw}{3}  = \mQ_3 \cdot g_3(\iter{\vw}{1}, \iter{\vw}{2}; \auxmat),
\end{align}
where:
\begin{enumerate}
    \item The matrices $\mQ_2$ and $\mQ_3$ are defined as: $$\mQ_2 \explain{def}{=} \mS (\mPsi_2-\hat{\Omega}_{\{2\},\emptyset}\cdot \mI_{\dim}) \mS \quad \mQ_3 \explain{def}{=} \mS (\mPsi_2 \mPsi_1 - \hat{\Omega}_{\{1,2\},\emptyset} \cdot \mI_{\dim} - \hat{\Omega}_{\{1\},\emptyset} \cdot (\mPsi_2 - \hat{\Omega}_{\{2\},\emptyset}\cdot \mI_{\dim})) \mS.$$ A generalization of the argument in \eqref{eq:semi-random-argument} shows that $\mQ_{1:3}$ form a semi-random ensemble (complete details are provided in \appref{VAMP-to-GFOM}). 
    \item The non-linearity $g_2$ is as defined in \eqref{eq:g_2-div-free} and $g_3$ is defined as $g_3(w_1, w_2; \auxvec) = g_1(\auxvec)$.
\end{enumerate}
Define the post-processing function $H_2: \R^{3 + \auxdim} \mapsto \R$ as:
\begin{align*}
    H(w_1, w_2, w_3; \auxvec) & \explain{def}{=} w_2 + \beta_2 w_3 + {\Omega}_{\{2\},\emptyset} \cdot g_2(w_1 ; \auxvec) + \beta_2 \cdot (\Omega_{\{1,2\},\emptyset} - \Omega_{\{1\},\emptyset} \cdot \Omega_{\{2\},\emptyset} ) \cdot g_1(\auxvec)  +   \tilde{\eta}_{2}(w_1; \auxvec).
\end{align*}
Using these definitions and the approximations (justified formally in \appref{VAMP-to-GFOM}): $$\hat{\Omega}_{\{1\},\emptyset} \approx \Omega_{\{1\},\emptyset}, \quad \hat{\Omega}_{\{2\},\emptyset} \approx {\Omega}_{\{2\},\emptyset}, \quad \hat{\Omega}_{\{1,2\},\emptyset} \approx \Omega_{\{1,2\},\emptyset},$$ \eqref{eq:VAMP-implement-iter2-prelim} can be expressed as:
\begin{align} \label{eq:post-process-iter2}
    \iter{\vz}{2} & \approx H_2(\iter{\vw}{1}, \iter{\vw}{2}, \iter{\vw}{3}; \auxmat).
\end{align}
Hence we have designed a VAMP algorithm (cf. \eqref{eq:VAMP-implement-iter1} and \eqref{eq:VAMP-implement-iter2}) which implement two iterations of the given GFOM in 3 iterations, combined with post-processing. 

\paragraph{Dynamics of the GFOM.}  Using this reduction, the dynamics of the GFOM can be inferred from the dynamics of the VAMP algorithm constructed in \eqref{eq:VAMP-implement-iter1} and \eqref{eq:VAMP-implement-iter2}. Applying \thref{VAMP} to the constructed VAMP algorithm we obtain: 
\begin{align*}
    (\iter{\vw}{1}, \iter{\vw}{2}, \iter{\vw}{3};  \auxmat) \explain{\pw}{\longrightarrow} (\serv{W}_1, \serv{W}_2, \serv{W}_3; \serv{A}),
\end{align*}
where  $\serv{A}$ is the auxiliary information random vector from \assumpref{side-info} and $\serv{W}_{1:3}$ (independent of $\serv{A}$) are the Gaussian state evolution random variables corresponding to the VAMP algorithm (recall \eqref{eq:SE-VAMP}). Furthermore, the constructed post-processing functions $H_1, H_2$ can be shown to have the necessary continuity properties (see \appref{VAMP-to-GFOM} for details) to guarantee that:
\begin{align*}
    (\iter{\vz}{1}, \iter{\vz}{2};  \auxmat) \explain{\eqref{eq:post-process-iter1}\eqref{eq:post-process-iter2}}{\approx} (H_1(\iter{\vw}{1}; \auxmat), \; H_2(\iter{\vw}{1}, \iter{\vw}{2}, \iter{\vw}{3}; \auxmat) \; ;  \auxmat) \explain{\pw}{\longrightarrow} (\underbrace{H_1(\serv{W}_1; \serv{A}), \; H_2(\serv{W}_1, \serv{W}_2, \serv{W}_3; \serv{A})}_{\explain{def}{=}(\serv{Z}_1, \serv{Z}_2)} \; ;  \serv{A}).
\end{align*}
As is apparent from the reduction, the covariance matrix of $\serv{W}_1, \serv{W}_2$ as well as the postprocessing functions $H_1, H_2$ are completely determined by the non-linearities $\nonlin_1, \nonlin_2$ and the limiting moments $\{\Omega_{B,B^\prime}\}$ of the strongly semi-random ensemble used in the given GFOM. Consequently, $(\iter{\vz}{1}, \iter{\vz}{2};  \auxmat) \explain{\pw}{\longrightarrow} (\serv{Z}_1, \serv{Z}_2, ; \serv{A})$ where the joint law of $\serv{Z}_1, \serv{Z}_2,  \serv{A}$ is completely determined by the non-linearities $\nonlin_1, \nonlin_2$ and the limiting moments $\{\Omega_{B,B^\prime}\}$ of the strongly semi-random ensemble used in the given GFOM. These arguments extend to higher iterations and can be used to prove \thref{GFOM}. The complete proof is provided in \appref{VAMP-to-GFOM}. 
\end{proof}

\subsection{Proof of \thref{RLS}: Universality of RLS Estimators and Proximal Method}
\label{sec:RLS-univ-proof}

We now prove the universality result for the  proximal method and RLS estimators for the regularized linear regression problem (\thref{RLS}) using the general universality result for first-order methods stated in \thref{GFOM}. The proof of \thref{RLS} relies on the following fact, which collects some useful results regarding the proximal method.
\begin{fact}\label{fact:prox}[{\citealp[Theorem 6.42 and Theorem 10.21]{beck2017}}] Under \assumpref{convex-reg}, we have:
\begin{enumerate}
    \item The proximal operator $\eta(\cdot ; \stepsize): \R \mapsto \R$ (defined in \eqref{eq:prox}) is a 1-Lipschitz function for any $\stepsize>0$. 
    \item The algorithm \eqref{eq:prox-method} with any step size $\stepsize \in (0, 1/\|\mX\|^2_{\op}]$ has the convergence guarantee:
    \begin{align*}
        L(\iter{{\vbeta}}{t}; \mX, \vy) - \min_{\vbeta \in \R^{\dim}}  L(\vbeta; \mX, \vy) & \leq  \frac{\|\iter{{\vbeta}}{1} -  {\vbeta}_\opt(\mX, \vbeta_\star, \vepsilon)\|^2}{2\cdot \stepsize \cdot (t-1) \cdot \dim} \quad \forall \; t \; \in \; \N.
    \end{align*}
    In the above display,  $L(\cdot ; \mX, \vy)$ and ${\vbeta}_\opt(\mX, \vbeta_\star, \vepsilon)$ are as defined in \eqref{eq:RLS}.
\end{enumerate}
\end{fact}

\begin{proof}[Proof of \thref{RLS}] The definition of $\uclass{\mu}$ (\defref{univ-class}) guarantees the existence of a constant $C > 0$ and a $\dim_0 \in \N$ such that:
\begin{align}\label{eq:good-evnt}
    \|\mX\|_{\op} \vee \|\tilde{\mX}\|_{\op} \leq C \quad \forall\;  \dim \geq \dim_0.
\end{align}
Throughout the proof we will assume that $\dim \geq \dim_0$ so that \eqref{eq:good-evnt} holds. In order to show claim (1) of \thref{RLS}, our strategy will be to implement the proximal method using a GFOM and to use \thref{GFOM} to show its universality. In order to do so, we will find it useful to reparameterize the proximal iterates.
\paragraph{Step 1: Reparameterizing the Proximal Iterates.} Recalling \eqref{eq:prox-method}, that the proximal method on the sensing matrix $\mX$ can be rewritten as:
\begin{subequations} \label{eq:prox-method-recall}
\begin{align}
 \iter{{\vbeta}}{1} &= \eta(\mX \tran \vy ; \stepsize) = \eta\big(\mX \tran \mX \vbeta_\star + \mX\tran \vepsilon; \stepsize \big) \explain{(a)}{=}  \eta(\mX \tran \mX \vbeta_\star + \sqrt{\mX\tran \mX} \vw ; \stepsize), \quad \vw \sim \gauss{\vzero}{\noisestd^2\mI_{\dim}}  \\
\iter{{\vbeta}}{t+1} &= \eta\big( (\mI_{\dim} - \stepsize \mX\tran \mX) \cdot \iter{{\vbeta}}{t} + \stepsize \cdot (\mX \tran \mX \vbeta_\star + \sqrt{\mX\tran \mX} \vw) ; \stepsize\big).
\end{align}
In the step marked (a) we observed that $\mX\tran \vepsilon \sim \gauss{\vzero}{\noisestd^2\mX\tran \mX}$.  Analogous expressions hold for $\{\iter{\tilde{\vbeta}}{t}: t \in \W\}$,  the proximal method iterates on the sensing matrix $\tilde{\mX}$. 
\end{subequations}
We will find it useful to introduce the re-parameterization for the proximal method on $\mX$
\begin{align*}
    \iter{\vz}{1} &= \mX \tran \mX \vbeta_\star + \sqrt{\mX\tran \mX} \vw, \quad 
    \iter{\vz}{t+1}  = - \stepsize \mX\tran \mX \cdot \iter{{\vbeta}}{t} + \iter{{\vbeta}}{t} + \stepsize \iter{\vz}{1}  \quad \forall \; t \geq 1,
\end{align*}
and the analogous reparameterization $\{\iter{\tilde{\vz}}{t} : t \in \W\}$ for the proximal method on $\tilde{\mX}$. 
Observe that the original proximal iterates can be recovered from the reparameterized iterates by:
\begin{align}\label{eq:reparameterization}
    \iter{\vbeta}{t} = \eta(\iter{\vz}{t}; \stepsize), \quad \iter{\tilde{\vbeta}}{t} = \eta(\iter{\tilde{\vz}}{t}; \stepsize) \quad \forall \; t \; \geq \; 1.
\end{align}
Hence, the reparameterized iterates follow the update rules:
\begin{align}\label{eq:prox-reparam}
    \iter{\vz}{1} =  \mX \tran \mX \vbeta_\star + \sqrt{\mX\tran \mX} \vw, \quad \iter{\vz}{t+1} = - \stepsize \mX\tran \mX \cdot \eta(\iter{{\vz}}{t}; \stepsize) + \eta(\iter{{\vz}}{t}; \stepsize) - \stepsize \iter{\vz}{1} \quad \forall \; t \; \geq 1, 
\end{align}
The update rules for $\{\iter{\tilde{\vz}}{t} : t \in \W\}$, the reparameterized iterates for $\tilde{\mX}$, are analogous.  We record the following useful estimates $\|\iter{{\vz}}{t}\|$ and $\|\iter{\tilde{\vz}}{t}\|$. For any fixed $t \in \N$ (independent of $\dim$) we have:
\begin{align}\label{eq:norm-estimate}
   \limsup_{\dim \rightarrow \infty} \frac{\E[\|\iter{{\vz}}{t} \|^2  ]}{\dim} < \infty, \quad \limsup_{\dim \rightarrow \infty} \frac{\E[\|\iter{\tilde{\vz}}{t} \|^2 \ ]}{\dim} < \infty.
\end{align}
The above estimates are readily obtained by induction. For $t=1$, we have:
\begin{align*}
   \frac{\E[\|\iter{{\vz}}{1} \|^2 ]}{\dim} & =  \frac{\E[\| \mX \tran \mX \vbeta_\star  \|^2]}{\dim} +  \frac{\E[\| \sqrt{\mX\tran \mX} \vw  \|^2 ]}{\dim} \explain{(a)}{\leq} \frac{C^4 \E[\|\vbeta_\star\|^2]}{\dim} +  \frac{C^2 \E[\|\vw\|^2]}{\dim} \lesssim 1.
\end{align*}
In the above display (a) follows from the fact that $\|\mX\|_{\op} \leq C$. For $t\geq 2$, we have:
\begin{align*}
    \frac{\E[\|\iter{{\vz}}{t} \|^2 ]}{\dim}  &\explain{\eqref{eq:prox-reparam}}{\leq}  \frac{2 \stepsize^2 \E[\|\iter{\vz}{1}\|^2]}{\dim}+ \frac{2\E[\|(\mI_{\dim} - \stepsize {\mX}\tran {\mX}) \cdot \eta(\iter{{\vz}}{t-1};\stepsize) \|^2]}{\dim}\\
    &  \explain{(a)}{\leq}  \frac{2 \stepsize^2 \E[\|\iter{\vz}{1}\|^2  ]}{\dim}+ \frac{2(1+C^2\stepsize)^2\E[\|\eta(\iter{{\vz}}{t-1};\stepsize) \|^2]}{\dim}\\
    & \explain{(b)}{\leq}  \frac{2 \stepsize^2 \E[\|\iter{\vz}{1}\|^2 ]}{\dim}+ 4 (1+C^2\stepsize)^2 \eta^2(0;\stepsize) +  \frac{4(1+C^2\stepsize)^2\E[\|\iter{{\vz}}{t-1} \|^2]}{\dim} \explain{(c)}{\lesssim} 1.
\end{align*}
In the above display (a) follows from the fact that $\|\mX\|_{\op} \leq C$. Step (b) follows from the fact that $\eta$ is 1-Lipschitz (cf. \factref{prox}), and (c) follows from the induction hypothesis. This proves the estimate claimed in \eqref{eq:norm-estimate} by induction. 
\paragraph{Step 2: Polynomial Approximation of $\sqrt{\mX\tran \mX}$.} In order to implement the reparameterized iterations \eqref{eq:prox-reparam} using a GFOM, we will find it convenient to approximate $ \sqrt{\mX\tran \mX}$ using a polynomial of $\mX\tran \mX$. By the Weierstrass approximation theorem, there is a sequence of polynomials $p_k: \R \mapsto \R$ indexed by $k \in \N$ such that $p_k$ is a polynomial of degree at most $k$ and:
\begin{subequations}\label{eq:sqrt-approx}
\begin{align}
    \Delta_k \explain{def}{=} \sup_{\lambda \in [0,C^2]} |p_k(\lambda) - \sqrt{\lambda}|  \rightarrow 0 \text{ as } k \rightarrow \infty.
\end{align}
where $C$ is the constant from \eqref{eq:good-evnt}. Hence:
\begin{align}
    \|\sqrt{\mX\tran \mX} - p_k(\mX\tran \mX) \|_{\op} &\leq \Delta_k, \quad   \|\sqrt{\tilde{\mX}\tran \tilde{\mX}} - p_k(\tilde{\mX}\tran \tilde{\mX}) \|_{\op}  \leq \Delta_k,
\end{align}
where $\Delta_{k} \rightarrow 0$ as $k \rightarrow \infty$. 
\end{subequations}
We introduce a family of iterations indexed by $k \in \N$ to approximate the reparameterized proximal iterates \eqref{eq:prox-reparam} on $\mX$:
\begin{align} \label{eq:prox-method-approx}
 \iter{{\vz}}{1,k} &= \mX \tran \mX \vbeta_\star + p_k({\mX\tran \mX}) \vw, \quad \iter{{\vz}}{t+1,k} = - \stepsize \mX\tran \mX \cdot \eta(\iter{{\vz}}{t,k}; \stepsize) + \eta(\iter{{\vz}}{t,k}; \stepsize) - \stepsize \iter{\vz}{1,k} \quad \forall \; t \; \geq 1. 
\end{align}
Analogously, we define $\{\iter{\tilde{\vz}}{t,k}: t \in \W\}$, a family of iterations indexed by $k \in \N$ to approximate the reparameterized proximal iterates $\{\iter{\tilde{\vz}}{t}: t \in \W\}$ on $\tilde{\mX}$. Next, we obtain estimates on the approximation error introduced in the iterations. For any fixed $t \in \N$, we have:
\begin{align}\label{eq:poly-approx-guarantee}
    \lim_{k \rightarrow \infty} \limsup_{\dim \rightarrow \infty} \frac{\E[\|\iter{{\vz}}{t} - \iter{{\vz}}{t,k} \|^2  ]}{\dim} =0, \quad \lim_{k \rightarrow \infty} \limsup_{\dim \rightarrow \infty} \frac{\E[\|\iter{\tilde{\vz}}{t}- \iter{\tilde{\vz}}{t,k} \|^2  ]}{\dim} =0.
\end{align}
This can be verified using induction. For $t=1$ we have:
\begin{align*}
     \lim_{k \rightarrow \infty} \limsup_{\dim \rightarrow \infty}\frac{\E[\|\iter{{\vz}}{1} - \iter{{\vz}}{1,k} \|^2 ]}{\dim} & \explain{\eqref{eq:prox-reparam},\eqref{eq:prox-method-approx}}{=}  \lim_{k \rightarrow \infty} \limsup_{\dim \rightarrow \infty}\frac{\E[\| \sqrt{\mX\tran \mX} \vw - p_k(\mX\tran \mX) \vw \|^2  ]}{\dim} \\&\explain{\eqref{eq:sqrt-approx}}{\leq}  \lim_{k \rightarrow \infty} \limsup_{\dim \rightarrow \infty} \Delta_k^2 \frac{\E\|\vw\|^2}{\dim} = 0.
\end{align*}
For $t \geq 2$ we have:
\begin{align*}
      \frac{\E[\|\iter{{\vz}}{t} - \iter{{\vz}}{t,k} \|^2  ]}{\dim} & \explain{\eqref{eq:prox-reparam},\eqref{eq:prox-method-approx}}{\leq }  \frac{\E[\|(\mI_{\dim} - \stepsize \mX\tran \mX) \cdot (\eta(\iter{{\vz}}{t-1,k};\stepsize)-\eta(\iter{{\vz}}{t-1};\stepsize)) - \stepsize \cdot( \iter{\vz}{1,k}-\iter{{\vz}}{1})\|^2  ]}{\dim} \\
      & \explain{(a)}{\leq} \frac{2(1+C^2\stepsize)^2\E[\|\eta(\iter{{\vz}}{t-1,k};\stepsize)-\eta(\iter{{\vz}}{t-1};\stepsize)\|^2  ]}{\dim} + \frac{2\stepsize^2\E[\|( \iter{\vz}{1,k}-\iter{{\vz}}{1})\|^2 ]}{\dim} \\
      & \explain{(b)}{\leq} \frac{2(1+C^2\stepsize)^2\E[\|\iter{{\vz}}{t-1,k}-\iter{{\vz}}{t-1}\|^2  ]}{\dim} + \frac{2\stepsize^2\E[\|( \iter{\vz}{1,k}-\iter{{\vz}}{1})\|^2  ]}{\dim} 
\end{align*}
In the above display, (a) follows from  the observation that $\|\mX\|_{\op} \leq C$  and (b) follows from the fact that $\eta$ is 1-Lipschitz (\factref{prox}). The above recursive bound immediately yields \eqref{eq:poly-approx-guarantee} by induction. 
\paragraph{Step 3: Implementing approximate proximal iterations using a GFOM.} We claim that the approximate proximal iterations \eqref{eq:prox-method-approx} can be implemented using a GFOM of the form given in \eqref{eq:GFOM}. Indeed, by introducing an extra iterate $\iter{\vz}{0}, \iter{\tilde{\vz}}{0}$ \eqref{eq:prox-method-approx} can be expressed as:
\begin{align*}
    \iter{\vz}{0,k} & = p_k({\mX\tran \mX}) \vw, \quad  \iter{{\vz}}{1,k} = \mX \tran \mX \vbeta_\star + \iter{\vz}{0,k}, \quad \iter{{\vz}}{t+1,k} = - \stepsize \mX\tran \mX  \eta(\iter{{\vz}}{t,k}; \stepsize) + \eta(\iter{{\vz}}{t,k}; \stepsize) - \stepsize \iter{\vz}{1,k}, \\
 \iter{\tilde{\vz}}{0,k} & = p_k({\tilde{\mX}\tran \tilde{\mX}}) \tilde{\vw}, \quad \iter{\tilde{\vz}}{1,k} =  \tilde{\mX} \tran \tilde{\mX} \vbeta_\star + \iter{\tilde{\vz}}{0,k}, \quad \iter{\vz}{t+1,k} = - \stepsize \tilde{\mX}\tran \tilde{\mX}  \eta(\iter{\tilde{\vz}}{t,k}; \stepsize) + \eta(\iter{\tilde{\vz}}{t,k}; \stepsize) - \stepsize \iter{\tilde{\vz}}{1,k}.
\end{align*}
These iterations are an instance of the general GFOM iteration \eqref{eq:GFOM} driven by matrices:
\begin{align*}
    \mM_0 = p_k({\mX\tran \mX}), \quad  \mM_1 = \mX\tran \mX, \quad \dotsc,  \quad \mM_{T} =  \mX\tran \mX,  \\
    \tilde{\mM}_0 = p_k({\tilde{\mX}\tran \tilde{\mX}}), \quad  \tilde{\mM}_1 = \tilde{\mX}\tran \tilde{\mX}, \quad \dotsc,  \quad \tilde{\mM}_{T} =  \tilde{\mX}\tran \tilde{\mX}.
\end{align*}
and auxiliary information matrices $$\mA = [\vw, \vbeta_\star], \quad \tilde{\mA} = [\tilde{\vw}, \vbeta_\star].$$ Since $\mX, \tilde{\mX} \in \uclass{\mu}$, $\mM_{0:T}$ and $\tilde{\mM}_{0:T}$ are strongly semi-random matrices with the same limiting moments (see \remref{ssr}). Furthermore, recall that $\vw \explain{d}{=} \tilde{\vw} \sim \gauss{\vzero}{\sigma^2\mI_{\dim}}$. Hence, by \thref{GFOM} for any fixed $k,T \in \N$:
\begin{align} \label{eq:GFOM-univ-conclusion}
    (\iter{\vz}{0,k}, \iter{\vz}{1,k}, \dotsc, \iter{\vz}{T,k}; \vbeta_\star)  \asymeq (\iter{\tilde{\vz}}{0,k}, \iter{\tilde{\vz}}{1,k}, \dotsc, \iter{\tilde{\vz}}{T,k}; \vbeta_\star).
\end{align}
\paragraph{Step 3: Universality of Proximal Iterates.} We now prove the first claim made in the theorem:
\begin{align} \label{eq:prox-pw2}
     (\iter{\vbeta}{1}, \dotsc, \iter{\vbeta}{T}; \vbeta_\star)  \asymeq (\iter{\tilde{\vbeta}}{1}, \dotsc, \iter{\tilde{\vbeta}}{T}; \vbeta_\star).
\end{align}
Since $\iter{\vbeta}{t} = \eta(\iter{\vz}{t}; \stepsize), \; \iter{\tilde{\vbeta}}{t} = \eta(\iter{\tilde{\vz}}{t}; \stepsize)$ and $\eta$ is 1-Lipschitz, it suffices to show:
\begin{align} \label{eq:pw2_sufficient}
     (\iter{\vz}{1}, \dotsc, \iter{\vz}{T}; \vbeta_\star)  \asymeq (\iter{\tilde{\vz}}{1}, \dotsc, \iter{\tilde{\vz}}{T}; \vbeta_\star).
\end{align}
Indeed, if \eqref{eq:pw2_sufficient} holds, then for any test function  $h:\R^{T+1} \mapsto \R$ be any test function which satisfies the regularity assumptions of \defref{PW2-eq}, we have:
\begin{align}
     &\frac{1}{\dim} \sum_{i=1}^\dim h(\iter{\beta}{1}_i, \dotsc, \iter{\beta}{T}_i; (\beta_\star)_i)  -  \frac{1}{\dim} \sum_{i=1}^\dim h(\iter{\tilde{\beta}}{1}_i, \dotsc, \iter{\tilde{\beta}}{T}_i; (\beta_\star)_i) \nonumber \\& \hspace{2.5cm}= \frac{1}{\dim} \sum_{i=1}^\dim h(\eta(\iter{z}{1}_i; \stepsize), \dotsc, \eta(\iter{z}{T}_i; \stepsize); (\beta_\star)_i)  -  \frac{1}{\dim} \sum_{i=1}^\dim h(\eta(\iter{\tilde{z}}{1}_i;\stepsize), \dotsc, \eta(\iter{\tilde{z}}{T}_i;\stepsize); (\beta_\star)_i). \label{eq:pw2_sufficient_composite}
\end{align}
Since $\eta( ; \stepsize)$ is 1-Lipschitz (\factref{prox}, item 1), the composite test function $\hat{h}: \R^{T+1} \mapsto \R$ defined as: $$\hat{h}(z_1, \dotsc, z_T ;  \beta_\star) \explain{def}{=} h(\eta(z_1 ; \stepsize), \dotsc, \eta(z_T; \stepsize); \beta_\star)$$
also satisfies the regularity conditions of \defref{PW2-eq}. Hence, using \eqref{eq:pw2_sufficient}, we conclude that the RHS of \eqref{eq:pw2_sufficient_composite} converges in probability to $0$. This means that \eqref{eq:pw2_sufficient} implies \eqref{eq:prox-pw2}. In order to prove that \eqref{eq:pw2_sufficient}, we again consider a test function $h:\R^{T+1} \mapsto \R$ which satisfies the regularity assumptions of \defref{PW2-eq}. Define:
\begin{align*}
    H_\dim & \explain{def}{=} \frac{1}{\dim} \sum_{i=1}^\dim h(\iter{z}{1}_i, \dotsc, \iter{z}{T}_i; (\beta_\star)_i), \quad \iter{H}{k}_\dim \explain{def}{=} \frac{1}{\dim} \sum_{i=1}^\dim h(\iter{z}{1,k}_i, \dotsc, \iter{z}{T,k}_i; (\beta_\star)_i), \\
    \tilde{H}_\dim & \explain{def}{=} \frac{1}{\dim} \sum_{i=1}^\dim h(\iter{\tilde{z}}{1}_i, \dotsc, \iter{\tilde{z}}{T}_i; (\beta_\star)_i), \quad \iter{\tilde{H}}{k}_\dim \explain{def}{=} \frac{1}{\dim} \sum_{i=1}^\dim h(\iter{\tilde{z}}{1,k}_i, \dotsc, \iter{\tilde{z}}{T,k}_i; (\beta_\star)_i).
\end{align*}
We need to show that $H_\dim - \tilde{H}_\dim \explain{P}{\rightarrow} 0$. To this end,consider any $\epsilon > 0$. We have:
\begin{align*}
   &\limsup_{\dim \rightarrow \infty} \P(|H_\dim - \tilde{H}_\dim| > 3 \epsilon) \\ & \qquad \qquad \qquad\leq \limsup_{k \rightarrow \infty}  \limsup_{\dim \rightarrow \infty} \left( \P(|H_\dim - \iter{H}{k}_\dim| >  \epsilon) + \P(|\iter{H}{k}_\dim-\iter{\tilde{H}}{k}_\dim| >  \epsilon) +  \P(|\tilde{H}_\dim - \iter{\tilde{H}}{k}_\dim| >  \epsilon)\right) \\
   &\qquad \qquad \qquad \explain{\eqref{eq:GFOM-univ-conclusion}}{=} \limsup_{k \rightarrow \infty}  \limsup_{\dim \rightarrow \infty} \left( \P(|H_\dim - \iter{H}{k}_\dim| >  \epsilon)  +  \P(|\tilde{H}_\dim - \iter{\tilde{H}}{k}_\dim| >  \epsilon)\right) \\
   & \qquad \qquad \qquad \leq \epsilon^{-1} \cdot \limsup_{k \rightarrow \infty}  \limsup_{\dim \rightarrow \infty}  \left( \E |H_\dim - \iter{H}{k}_\dim| + \E |\tilde{H}_\dim - \iter{\tilde{H}}{k}_\dim|   \right).
\end{align*}
Using the regularity of $h$ (cf. \defref{PW2-eq}) and the Cauchy-Schwarz Inequality:
\begin{align}\label{eq:recycle}
    &\limsup_{k \rightarrow \infty}  \limsup_{\dim \rightarrow \infty}  \left( \E |H_\dim - \iter{H}{k}_\dim|  \right)^2 \\&\leq  \limsup_{k \rightarrow \infty}  \limsup_{\dim \rightarrow \infty} 3L^2 \left(1 + \E[ \serv{B}^{2D}_\star] +  \frac{1}{\dim}\sum_{t=1}^T (\E[\|\iter{\vz}{t}\|^2 ] +  \E[\|\iter{\vz}{t,k}\|^2 ]) \right) \cdot \left( \frac{1}{\dim} \sum_{t=1}^T \E[\|\iter{\vz}{t} - \iter{\vz}{t,k} \|^2]\right) \nonumber\\
    & \leq \limsup_{k \rightarrow \infty}  \limsup_{\dim \rightarrow \infty} 3L^2 \left(1 + \E[ \serv{B}^{2D}_\star]+ \frac{1}{\dim}\sum_{t=1}^T (3\E[\|\iter{\vz}{t}\|^2 ] + 2  \E[\|\iter{\vz}{t,k}-\iter{\vz}{t}\|^2]) \right) \left( \frac{1}{\dim} \sum_{t=1}^T \E[\|\iter{\vz}{t} - \iter{\vz}{t,k} \|^2 ]\right) \nonumber\\
    &\explain{\eqref{eq:norm-estimate},\eqref{eq:poly-approx-guarantee}}{=} 0. \nonumber
\end{align}
The same argument shows that:
\begin{align*}
    \lim_{k \rightarrow \infty}  \limsup_{\dim \rightarrow \infty} \E |\tilde{H}_\dim - \iter{\tilde{H}}{k}_\dim|  = 0,
\end{align*}
and hence $H_\dim - \tilde{H}_\dim \explain{P}{\rightarrow} 0$, which proves claim (1) in the statement of the theorem. To prove claim (2), we argue that since the proximal method can construct arbitrarily accurate approximations for the RLS estimator, the RLS estimator must also exhibit universality. 
\paragraph{Step 4: Universality for RLS Estimator.} Observe that the additional assumptions made for claim (2) of the theorem guarantee that the objectives $L(\cdot \; ; \mX, \mX \vbeta_\star + \vepsilon): \R^\dim \mapsto \R$ and $L(\cdot \; ; \tilde{\mX}, \tilde{\mX} \vbeta_\star + \vepsilon): \R^\dim \mapsto \R$ are  $(\kappa\cdot \dim^{-1})$-strongly convex. Under these conditions the minimizers:
\begin{align*}
    \vbeta_{\opt} \explain{def}{=}  \arg\min_{\beta \in \R^\dim} L(\vbeta \; ; \mX, \mX \vbeta_\star + \vepsilon),  \quad \tilde{\vbeta}_{\opt} \explain{def}{=}  \arg\min_{\beta \in \R^\dim} L(\vbeta \; ; \tilde{\mX}, \tilde{\mX} \vbeta_\star + \vepsilon),
\end{align*}
are unique. As before, we use $\iter{\vbeta}{t}$ and $\iter{\tilde{\vbeta}}{t}$ to denote the proximal iterates on sensing matrices $\mX$ and $\tilde{\mX}$ respectively. In order to prove the second claim of the theorem we run the proximal method iterations with the special step size $\stepsize = 1/2C^2$, where $C$ is the constant from \eqref{eq:good-evnt}.  This step-size satisfies the requirement stated in \factref{prox} and hence for any $t \in \N$:
\begin{align}
    L(\iter{\vbeta}{t} \; ; \mX, \mX \vbeta_\star + \vepsilon) - L(\vbeta_{\opt} \; ; \mX, \mX \vbeta_\star + \vepsilon) & \leq \frac{C^2\|\iter{{\vbeta}}{1} -  {\vbeta}_\opt(\mX, \vbeta_\star, \vepsilon)\|^2}{(t-1) \cdot \dim} \nonumber \\&\leq \frac{2C^2}{(t-1)} \cdot \left( \frac{\|\iter{{\vbeta}}{1}\|^2 + \|{\vbeta}_\opt\|^2}{\dim} \right)\label{eq:prox-gap}
\end{align}
On the other hand, by the definition of $(\kappa\cdot \dim^{-1})$-strong convexity the sub-optimality gap of the proximal iterates can be lower bounded by:
\begin{align*}
    &L(\iter{\vbeta}{t} \; ; \mX, \mX \vbeta_\star + \vepsilon) - L(\vbeta_{\opt} \; ; \mX, \mX \vbeta_\star + \vepsilon) \\ &\qquad \qquad\qquad \qquad\geq \ip{\partial_{\vbeta}L(\vbeta_\opt \; ; \mX, \mX \vbeta_\star + \vepsilon) }{\iter{\vbeta}{t} -\vbeta_\opt } + \frac{\kappa \| \iter{\vbeta}{t} -\vbeta_\opt \|^2}{2\dim} = \frac{\kappa \| \iter{\vbeta}{t} -\vbeta_\opt \|^2}{2\dim}, 
\end{align*}
where the last equality follows by the subgradient optimality condition at $\vbeta_\opt$. Rearranging the above display and using \eqref{eq:prox-gap} we obtain:
\begin{align}\label{eq:prox-min-gap}
     \frac{\| \iter{\vbeta}{t} -\vbeta_\opt \|^2}{\dim} &\leq  \frac{4C^2}{\kappa(t-1)} \cdot \left( \frac{\|\iter{{\vbeta}}{1}\|^2 + \|{\vbeta}_\opt\|^2}{\dim} \right).
\end{align}
An analogous estimate holds for $\| \iter{\tilde{\vbeta}}{t} -\tilde{\vbeta}_\opt \|^2/\dim$.
Next, we upper bound $\|{\vbeta}_\opt\|^2$. Again by the definition of strong convexity:
\begin{align*}
       \frac{\|\mX \vbeta_\star + \vepsilon\|^2}{\dim} + \reg(0)=  L(\vzero \; ; \mX, \mX \vbeta_\star + \vepsilon) &\geq  L(\vbeta_{\opt} \; ; \mX, \mX \vbeta_\star + \vepsilon)  + \frac{\kappa \|\vbeta_\opt \|^2}{2\dim} \\
       & \geq \rho_{\min} + \frac{\kappa \|\vbeta_\opt \|^2}{2\dim},
\end{align*}
where in the last step we defined $\rho_{\min} \explain{def}{=} \min_{x\in \R} \rho(x) > -\infty$ (cf. \assumpref{convex-reg}). Hence,
\begin{align}\label{eq:optimum-norm-bound}
    \limsup_{\dim \rightarrow \infty} \frac{\E \|\vbeta_\opt \|^2}{\dim} \leq \frac{2}{\kappa} \cdot \left( \rho(0) - \rho_{\min} + \limsup_{\dim \rightarrow \infty}\frac{\E \|\mX \vbeta_\star + \vepsilon\|^2}{\dim}\right) \leq \frac{2(\rho(0) - \rho_{\min} + C^2 \E \serv{B}_\star^2 + \noisestd^2)}{\kappa}.
\end{align}
Using the above estimate in \eqref{eq:prox-min-gap} we obtain:
\begin{align}\label{eq:prox-min-gap-final}
     \lim_{t \rightarrow \infty} \limsup_{\dim \rightarrow \infty} \frac{\E \| \iter{\vbeta}{t} -\vbeta_\opt \|^2}{\dim} &= 0, \quad
       \lim_{t \rightarrow \infty} \limsup_{\dim \rightarrow \infty} \frac{\E \| \iter{\tilde{\vbeta}}{t} -\tilde{\vbeta}_\opt \|^2}{\dim} =0.
\end{align}
We now prove the second claim made in the theorem: $(\vbeta_\opt; \vbeta_\star)  \asymeq (\tilde{\vbeta}_\opt; \vbeta_\star)$. Let $h:\R^{2} \mapsto \R$ be any test function which satisfies the regularity assumptions of \defref{PW2-eq}. Define:
\begin{align*}
    H_\dim & \explain{def}{=} \frac{1}{\dim} \sum_{i=1}^\dim h((\beta_\opt)_i; (\beta_\star)_i), \quad \iter{H}{t}_\dim \explain{def}{=} \frac{1}{\dim} \sum_{i=1}^\dim h(\iter{\beta}{t}_i; (\beta_\star)_i), \\
    \tilde{H}_\dim & \explain{def}{=} \frac{1}{\dim} \sum_{i=1}^\dim h((\tilde{\beta}_{\opt})_i; (\beta_\star)_i), \quad \iter{\tilde{H}}{t}_\dim \explain{def}{=} \frac{1}{\dim} \sum_{i=1}^\dim h(\iter{\tilde{\beta}}{t}_i; (\beta_\star)_i).
\end{align*}
We need to show that $H_\dim - \tilde{H}_\dim \explain{P}{\rightarrow} 0$. To this end,consider any $\epsilon > 0$. We have:
\begin{align*}
   &\limsup_{\dim \rightarrow \infty} \P(|H_\dim - \tilde{H}_\dim| > 3 \epsilon) \\ &\qquad\qquad\qquad \leq \limsup_{t \rightarrow \infty}  \limsup_{\dim \rightarrow \infty} \left( \P(|H_\dim - \iter{H}{t}_\dim| >  \epsilon) + \P(|\iter{H}{t}_\dim-\iter{\tilde{H}}{t}_\dim| >  \epsilon) +  \P(|\tilde{H}_\dim - \iter{\tilde{H}}{t}_\dim| >  \epsilon)\right) \\
   &\qquad\qquad\qquad \explain{(a)}{\leq} \limsup_{t \rightarrow \infty}  \limsup_{\dim \rightarrow \infty} \left( \P(|H_\dim - \iter{H}{t}_\dim| >  \epsilon)  +  \P(|\tilde{H}_\dim - \iter{\tilde{H}}{t}_\dim| >  \epsilon)\right) \\
   &\qquad\qquad\qquad \leq \epsilon^{-1} \cdot \limsup_{t \rightarrow \infty}  \limsup_{\dim \rightarrow \infty}  \left( \E |H_\dim - \iter{H}{t}_\dim| + \E |\tilde{H}_\dim - \iter{\tilde{H}}{t}_\dim|   \right) \\
   &\qquad\qquad\qquad \explain{(b)}{=} 0.
\end{align*}
In the above display (a) follows from the fact that for any $t \in \N$, $(\iter{\vbeta}{t}; \vbeta_\star) \asymeq (\iter{\tilde{\vbeta}}{t}; \vbeta_\star)$ (this was the first claim of the theorem) and step (b) follows by repeating the arguments used in display \eqref{eq:recycle} (we use estimates \eqref{eq:optimum-norm-bound},\eqref{eq:prox-min-gap-final} instead of \eqref{eq:norm-estimate},\eqref{eq:poly-approx-guarantee}). This proves the second claim of the theorem and concludes the proof of \thref{RLS}. 
\end{proof}

\section{Proof of the Universality Principle for VAMP Algorithms}
\label{sec:univ-vamp}

The universality principle for VAMP algorithms given in \thref{VAMP} generalizes a result obtained in our prior work \citep[Theorem 1]{dudeja2022universality} in two important ways:
\begin{enumerate}
    \item \thref{VAMP} allows the VAMP algorithm to use an auxiliary information matrix $\mA$. This feature is important for VAMP algorithms designed for inference problems like linear regression where the signal and the noise are treated as auxiliary information. In contrast, our prior work considered VAMP algorithms for random optimization problems, which do not need any auxiliary information. 
    \item \thref{VAMP} allows the VAMP algorithm to use the entire history $\iter{\vz}{1}, \dotsc, \iter{\vz}{t-1}$ in the update rule for $\iter{\vz}{t}$ (cf. \eqref{eq:VAMP}). This flexibility ensures that VAMP algorithms can implement general first order methods like the proximal gradient method. In contrast, the VAMP algorithms studied in our prior work only used $\iter{\vz}{t-1}$ in the update rule for $\iter{\vz}{t}$.
\end{enumerate}
We now introduce some key ideas involved in the proof of the universality principle for VAMP algorithms (\thref{VAMP}) in the form of some intermediate results whose proofs are deferred to the appendix. The proof of \thref{VAMP} is provided at the end of this section.

\subsection{Key Intermediate Results}

\subsubsection{Simplifying Assumptions} 
We will prove \thref{VAMP} under two additional simplifying assumptions and subsequently argue that the result continues to hold without these simplifying assumptions using suitable reductions and approximation arguments. We discuss these simplifying assumptions in the paragraphs below.

\paragraph{Orthogonalization.} Our first simplifying assumption is that the non-linearities $\nonlin_{1:T}$ and the semi-random ensemble $\mM_{1:T}$ used in the VAMP algorithm satisfy the following orthogonality condition. 
\begin{sassumption}[Orthogonality] \label{sassump:orthogonality}A VAMP iteration of the form \eqref{eq:VAMP} is orthogonalized if $\Omega$, the limiting covariance matrix of the semi-random ensemble $\mM_{1:T}$ and the non-linearities $\nonlin_{1:T}$ satisfy:
\begin{align*}
    \Omega_{st} &\in \{0, 1\} \quad \forall \; s,t \; \in \; [T], \\
    \Omega_{tt} & = 1 \quad \forall \; t \; \in \; [T], \\ 
    \E[\nonlin_s(\serv{G}_1, \dotsc, \serv{G}_{s-1}; \serv{A})\nonlin_t(\serv{G}_1, \dotsc, \serv{G}_{t-1}; \serv{A})] & \in \{0,1\} \quad \forall \; s,t \; \in \; [T], \\
    \E \nonlin_t^2(\serv{G}_1, \dotsc, \serv{G}_{t-1}; \serv{A}) & = 1  \quad \forall \; t \; \in \; [T], \\
    \Omega_{st} \cdot  \E[\nonlin_s(\serv{G}_1, \dotsc, \serv{G}_{s-1}; \serv{A})\nonlin_t(\serv{G}_1, \dotsc, \serv{G}_{t-1}; \serv{A})] & = 0 \quad \forall \; s,t \; \in \; [T], \; s \neq t.
\end{align*}
In the above display $\serv{A}$ is the auxiliary information random variable from \assumpref{side-info} and $\serv{G}_{1:T}$ are i.i.d. $\gauss{0}{1}$ random variables independent of $\serv{A}$.
\end{sassumption}

The motivation behind the condition stated above is that if \sassumpref{orthogonality} holds, then it is immediate from \eqref{eq:SE-VAMP} that the Gaussian state evolution covariance associated with the VAMP algorithm is given by $\Sigma_T = I_T$, and the state evolution random variables $\serv{Z}_1, \dotsc, \serv{Z}_T$ are i.i.d. $\gauss{0}{1}$.  This simple form of the state evolution makes the proof of \thref{VAMP} tractable in this case. Furthermore, in \appref{simplifications} (specifically, \lemref{orthogonalization}) we argue that the iterates of any VAMP algorithm driven by a semi-random ensemble $\mM_{1:T}$ and non-linearities $\nonlin_{1:T}$ (that need not satisfy \sassumpref{orthogonality}) can be expressed as a linear combination of the iterates of an \emph{orthogonalized} VAMP algorithm that satisfies \sassumpref{orthogonality}. The semi-random ensemble used in the orthogonalized VAMP algorithm is obtained by applying the Gram-Schmidt process on the semi-random ensemble $\mM_{1:T}$ used in the original VAMP algorithm (by viewing $\mM_{1:T}$ as vectors in $\R^{\dim^2}$). Similarly, the non-linearities used in the orthogonalized VAMP algorithm are obtained by applying the Gram-Schmidt process on $\nonlin_{1:T}$, the non-linearities used in the original VAMP algorithm (by viewing them as vectors in the Gaussian Hilbert space corresponding to the Gaussian state evolution random variables associated with the given VAMP algorithm).

\paragraph{Balancing Semi-Random Matrices.} Our second simplifying assumption is that the semi-random ensemble $\mM_{1:T}$ driving the VAMP algorithm is balanced in the following sense.
\begin{sassumption}[Balanced Semi-Random Ensemble]\label{sassump:balanced} A semi-random ensemble $\mM_{1:T}$ with limiting covariance matrix $\Omega$ is balanced if for any $s,t \in [T]$
\begin{align*}
    (\mM_{s} \mM_t \tran)_{11} = (\mM_{s} \mM_t \tran)_{22} = \dotsb = (\mM_{s} \mM_t \tran)_{\dim\dim} = \Omega_{st}.
\end{align*}
\end{sassumption}
In \appref{simplifications} (see \lemref{balancing}), we argue that given a semi-random ensemble $\mM_{1:T}$, one can construct a balanced semi-random ensemble $\hat{\mM}_{1:T}$ such that:
\begin{align*}
    \max_{t \in [T]} \|\hat{\mM}_{t} - \mM_t \|_{\op} \ll 1.
\end{align*}
This approximation guarantee is sufficient to ensure that using $\hat{\mM}_{1:T}$ instead of $\mM_{1:T}$ in the VAMP algorithm does not change the limiting dynamics of the VAMP algorithm. We record the two simplifications introduced above in the following proposition, which is proved in \appref{simplifications}. 
\begin{proposition}\label{prop:orthogonalization} It suffices to prove \thref{VAMP} when  \sassumpref{balanced} and \sassumpref{orthogonality} hold in addition to the other assumptions required by \thref{VAMP}.\footnote{In other words, the proof of this proposition will show that if \thref{VAMP} holds under the additional \sassumpref{orthogonality} and \sassumpref{balanced}, then it must also hold without these additional simplifying assumptions.}
\end{proposition}
\subsubsection{Multivariate Vector Approximate Message Passing (MVAMP)} \label{sec:MVAMP-to-VAMP}
As mentioned previously, one of the important aspects in which \thref{VAMP} generalizes the main result in our prior work \citep[Theorem 1]{dudeja2022universality} is that it allows the VAMP algorithm to use the entire history $\iter{\vz}{1:t}$ in the update rule for $\iter{\vz}{t+1}$. In order to obtain this generalization, we use a simple reduction which shows that such long-memory VAMP algorithms can be implemented using memory-free VAMP algorithms with \emph{matrix-valued} iterates. We call such VAMP algorithms Multivariate VAMP (MVAMP) algorithms and introduce them formally below.

\paragraph{MVAMP Algorithms.} A Multivariate VAMP (MVAMP) algorithm is specified using:
\begin{enumerate}
    \item A positive integer $\order \in \N$, known as the order of the MVAMP algorithm. 
    \item The total number of iterations $T$.
    \item A collection of $\order$ matrices $\mM_1, \mM_2, \dotsc, \mM_{\order} \in \R^{\dim \times \dim}$. 
    \item A matrix $\dim \times \auxdim$ matrix $\auxmat$ of auxiliary information with rows $\auxvec_1, \auxvec_2, \dotsc, \auxvec_\dim \in \R^{\auxdim}$.
    \item A collection of $\order$ nonlinerities $\nonlin_1, \nonlin_2, \dotsc \nonlin_\order$ where each $f_i: \R^{\order + \auxdim} \mapsto \R$.  
\end{enumerate}
The MVAMP algorithm maintains $\order$ iterates $\iter{\vz}{t,1}, \iter{\vz}{t,2}, \dotsc, \iter{\vz}{t,\order} \in \R^\dim$, which are updated as follows:
\begin{subequations} \label{eq:mvamp}
\begin{align}
    \iter{\vz}{t,i} & = \mM_i \nonlin_i(\iter{\vz}{t-1,1}, \iter{\vz}{t-1,2}, \dotsc, \iter{\vz}{t-1,\order}; \auxmat)  \; \forall \;  i \in [\order],\; \forall \; t \; \in [T].
\end{align}
The algorithm is initialized with:
\begin{align}
     \iter{\vz}{0,i} \explain{i.i.d.}{\sim} \gauss{\vzero}{\mI_\dim}.
\end{align}
\end{subequations}
For each $i \in [\order]$ and for each $t \in [T]$, we define the matrices $\iter{\mZ}{t,\cdot}$ and $\iter{\mZ}{\cdot,i}$ as follows:
\begin{align*}
    \iter{\mZ}{t,\cdot} & = \begin{bmatrix} \iter{\vz}{t,1} & \iter{\vz}{t,2} & \hdots & \iter{\vz}{t,\order} \end{bmatrix}, \\
    \iter{\mZ}{\cdot,i} & = \begin{bmatrix} \iter{\vz}{0,i} & \iter{\vz}{1,i} & \hdots & \iter{\vz}{T,i} \end{bmatrix}.
\end{align*}
We denote the row $j$ of these matrices by $\iter{z}{t,\cdot}_j \in \R^{\order}$ and $\iter{z}{\cdot,i}_j \in \R^{T+1}$ for each $j \in [\dim]$.

\paragraph{Advantage of MVAMP Algorithms.} Our motivation for implementing the VAMP algorithm \eqref{eq:VAMP} using the MVAMP algorithms of the form \eqref{eq:mvamp} is that these algorithms have two convenient properties which make their analysis easier. First, as mentioned previously, they are memory-free in the sense that the update equation for $\iter{\mZ}{t,\cdot}$ only depends on the previous iterate $\iter{\mZ}{t-1,\cdot}$. Secondly, the non-linearities and the semi-random matrices used in the MVAMP do not change at each iteration, unlike the VAMP algorithm in \eqref{eq:VAMP}.

\paragraph{Implementing VAMP using MVAMP.}  The class of MVAMP algorithms is rich enough to implement any VAMP algorithm. In order to demonstrate this, we consider a $T$-iteration VAMP algorithm driven by a semi-random ensemble $\mM_{1:T}$ and non-linearities $\nonlin_{1:T}$:
\begin{align} \label{eq:VAMP-to-MVAMP}
    \iter{\vz}{t} = \mM_t \cdot \nonlin_t(\iter{\vz}{1}, \iter{\vz}{2}, \dotsc, \iter{\vz}{t-1}; \auxmat) \quad \forall \; t \in [T].
\end{align}
This VAMP algorithm can be implemented using $T$ iterations of a MVAMP algorithm of order $\order = T$ with update equations:
\begin{align*}
    \iter{\vw}{t,i} \explain{def}{=} \mM_i \nonlin_i( \iter{\vw}{t-1,1}, \iter{\vw}{t-1,2} \dotsc, \iter{\vw}{t-1,i-1}; \auxmat) \quad \forall \; t,i \; \in [T].
\end{align*}
In order to show that the above MVAMP algorithm implements the given VAMP algorithm, we can compute the first two iterations of the MVAMP algorithm:
\begin{align}\label{eq:VAMP-to-MVAMP-iter1}
    \iter{\mW}{1,\cdot} & = \begin{bmatrix}\mM_1 \nonlin_1(\auxmat) &  \times & \times & \hdots &\times \end{bmatrix} \explain{\eqref{eq:VAMP-to-MVAMP}}{=} \begin{bmatrix}\iter{\vz}{1}  &  \times & \times & \hdots &\times \end{bmatrix}.
\end{align}
In the above display $\times$ denotes columns of $ \iter{\mW}{1,\cdot}$ whose explicit expressions are not important for the argument. Similarly, we can compute the second iteration of the MVAMP algorithm:
\begin{align*}
     \iter{\mW}{2,\cdot} = \begin{bmatrix}\mM_1 \nonlin_1(\auxmat) &  \mM_2 \nonlin_2(\iter{\vw}{1,1}; \auxmat) & \times & \hdots &\times \end{bmatrix} &\explain{\eqref{eq:VAMP-to-MVAMP-iter1}}{=} \begin{bmatrix}\iter{\vz}{1}  &  \mM_2 \nonlin_2(\iter{\vz}{1}; \auxmat) & \times & \hdots &\times \end{bmatrix}  \\&\explain{\eqref{eq:VAMP-to-MVAMP}}{=} \begin{bmatrix}\iter{\vz}{1}  &  \iter{\vz}{2}  & \times & \hdots &\times \end{bmatrix}.
\end{align*}
By an induction argument for any $t \in [T]$:
\begin{align*}
    \iter{\mW}{t,\cdot} & = \begin{bmatrix}\iter{\vz}{1}  &  \iter{\vz}{2}   & \hdots &\iter{\vz}{t} &\times &\hdots & \times \end{bmatrix}.
\end{align*}
Hence $\iter{\mW}{T,\cdot} = [\iter{\vz}{1}, \iter{\vz}{2}, \dotsc, \iter{\vz}{T}]$, which shows that the MVAMP algorithm implements the given VAMP algorithm. Since the final iterate of the MVAMP encodes the entire $T$-iteration trajectory of the given VAMP algorithm, in order to characterize the limiting joint empirical distribution of the VAMP iterates $\iter{\vz}{1}, \dotsc, \iter{\vz}{T}$, it is sufficient to characterize the limiting joint empirical distribution of the last iterate of the MVAMP algorithm $\iter{\vw}{T,1}, \dotsc, \iter{\vw}{T,T}$.

\subsubsection{Polynomial Approximation} \label{sec:poly-approx}
In order to analyze the limiting empirical distribution of $\iter{\vz}{T,1}, \dotsc, \iter{\vz}{T,k}$, the MVAMP iterate at time $T$ (recall \eqref{eq:mvamp}), we use the method of moments. This involves expressing the key quantity of interest:
\begin{align*}
    \frac{1}{\dim} \sum_{\ell=1}^\dim h(\iter{z}{T,1}_{\ell}, \dotsc, \iter{z}{T,\order}_{\ell} ; \auxvec_\ell)
\end{align*}
as a polynomial in the semi-random ensemble $\mM_{1:\order}$ used in the MVAMP iterations. In the above display $h:\R^{\order+\auxdim} \mapsto \R$ is a test function that satisfies the regularity assumptions stated in \defref{PW2}. In order to do so, we approximate the test function $h$ and the non-linearities $\nonlin_{1:k}$ used in the MVAMP algorithm \eqref{eq:mvamp} by polynomials. The following lemma constructs the polynomial approximations for $\nonlin_{1:k}$ and $h$. The approximations of $\nonlin_{1:k}$ are constructed so that the divergence-free property from \assumpref{div-free} and the orthogonality property stated in \sassumpref{orthogonality} are maintained. 

\begin{lemma}[Approximation] \label{lem:low-degree-approx}Let $\nonlin_{1:k} : \R^{k+\auxdim} \mapsto \R$ be a collection of polynomially bounded, continuous  non-linearities and $h: \R^{k + \auxdim} \mapsto \R$ be a polynomially bounded, continuous test function which satisfy:
\begin{align*}
        \max_{i \in [\order]} |\nonlin_i(z;\auxvec)| & \leq L \cdot (1 + \|z\|^{\degree} + \|\auxvec\|^{\degree}), \quad |h(z;\auxvec)|  \leq L \cdot (1 + \|z\|^{\degree} + \|\auxvec\|^{\degree}) \quad \forall \; \auxvec \in \R^{\auxdim}, \; z \in \R^{\order}.
 \end{align*}
Additionally, suppose that for each $i,j \in [k]$ we have:
\begin{align*}
        \E \serv{Z}_i \nonlin_j( \serv{Z}_1, \dotsc, \serv{Z}_k; \serv{A}) &= 0, \quad
        \E \nonlin_i( \serv{Z}_1, \dotsc, \serv{Z}_k; \serv{A})\nonlin_j( \serv{Z}_1, \dotsc, \serv{Z}_k; \serv{A})  \in \{0,1\}, \quad
        \E \nonlin_i^2( \serv{Z}_1, \dotsc, \serv{Z}_k; \serv{A})   = 1.
\end{align*}
Then, for any $\epsilon \in (0,1)$ there exists an integer $D_\epsilon \in \N$, a constant $L_{\epsilon} \in (0,\infty)$ and functions $\hat{\nonlin}^\epsilon_{1:k}, \hat{h}^\epsilon : \R^{k + \auxdim} \mapsto \R$ such that:
\begin{enumerate}
    \item For each $\auxvec \in \R^{\auxdim}$, $\hat{\nonlin}^\epsilon_{1}(z_1, \dotsc, z_k; \auxvec)$, $\dotsc$,  $\hat{\nonlin}^\epsilon_{k}(z_1, \dotsc, z_k; \auxvec)$, and $\hat{h}^\epsilon(z_1, \dotsc, z_k; \auxvec)$ are polynomials of degree at most $D_\epsilon$ in $z_1, \dotsc, z_k$. 
    \item $\hat{\nonlin}^\epsilon_{1}(z_1, \dotsc, z_k; \auxvec)$, $\dotsc$,  $\hat{\nonlin}^\epsilon_{k}(z_1, \dotsc, z_k; \auxvec)$, and $\hat{h}^\epsilon(z_1, \dotsc, z_k; \auxvec)$ satisfy:
    \begin{align*}
       \max_{i \in [\order]} \E[(\hat{\nonlin}^\epsilon_i( \serv{Z}_1, \dotsc, \serv{Z}_k; \serv{A}) - \nonlin_i( \serv{Z}_1, \dotsc, \serv{Z}_k; \serv{A}))^2] & \leq \epsilon^2, \quad
        \E[(\hat{h}^\epsilon( \serv{Z}_1, \dotsc, \serv{Z}_k; \serv{A}) - h( \serv{Z}_1, \dotsc, \serv{Z}_k; \serv{A}))^2]  \leq \epsilon^2.
    \end{align*}
    \item  For each $i,j \in [k]$,
    \begin{align*}
         \E[\serv{Z}_i \hat{\nonlin}^\epsilon_j( \serv{Z}_1, \dotsc, \serv{Z}_k; \serv{A})]  = 0, \;
         \E[\hat{\nonlin}^\epsilon_i( \serv{Z}_1, \dotsc, \serv{Z}_k; \serv{A})  \hat{\nonlin}^\epsilon_j( \serv{Z}_1, \dotsc, \serv{Z}_k; \serv{A})]  = \E[ \nonlin_i( \serv{Z}_1, \dotsc, \serv{Z}_k; \serv{A})\nonlin_j( \serv{Z}_1, \dotsc, \serv{Z}_k; \serv{A})].
    \end{align*}
    \item The functions $\hat{\nonlin}^\epsilon_{1}(z_1, \dotsc, z_k; \auxvec)$, $\dotsc$,  $\hat{\nonlin}^\epsilon_{k}(z_1, \dotsc, z_k; \auxvec)$, and $\hat{h}^\epsilon(z_1, \dotsc, z_k; \auxvec)$ and polynomially bounded continuous functions which satisfy:
    \begin{align*}
        \max_{i \in [\order]} |\hat{\nonlin}^\epsilon_i(z;\auxvec)| & \leq L_\epsilon \cdot (1 + \|z\|^{\degree_\epsilon} + \|\auxvec\|^{\degree_\epsilon}), \quad  |\hat{h}^\epsilon(z;\auxvec)| & \leq L_\epsilon \cdot (1 + \|z\|^{\degree_\epsilon} + \|\auxvec\|^{\degree_\epsilon}) \quad \forall  \; \auxvec \in \R^{\auxdim}, \; z \in \R^{\order}.
    \end{align*}
\end{enumerate}
In the above equations $\serv{Z}_{1:k} \explain{i.i.d.}{\sim} \gauss{0}{1}$ and $\serv{A}$ is the auxiliary information random variable from \assumpref{side-info}, independent of $\serv{Z}_{1:k}$. 
\end{lemma}
The proof of the above approximation lemma is provided in \appref{approximation}. 

\subsubsection{Dynamics of MVAMP via Method of Moments}
Using the method of moments, we show that the expectation of the joint empirical moments of the last MVAMP iterate $\iter{\mZ}{T,\cdot} = [\iter{\vz}{T,1}, \dotsc, \iter{\vz}{T,\order}]$ converges to the corresponding joint moment of the Gaussian vector $(\serv{Z}_1, \dotsc, \serv{Z}_k) \sim \gauss{0}{I_k}$. \begin{theorem}\label{thm:moments}
Consider the MVAMP iterations \eqref{eq:mvamp}. Suppose that:
\begin{enumerate}
    \item The matrices $\mM_{1:\order} = \mS \mPsi_{1:k} \mS$ form a balanced (\sassumpref{balanced}) semi-random ensemble (\defref{semirandom}) with limiting covariance matrix $\Omega$ such that:
    \begin{subequations}\label{eq:normalized-semirandom}
    \begin{align}
        \Omega_{ii} & = 1 \; \forall \; i \in [\order], \quad
    (\mPsi_i \mPsi_j \tran)_{\ell \ell} = \Omega_{ij} \; \forall \; \ell \in [\dim], \; i,j \in [\order].
    \end{align}
    \end{subequations}
    \item The auxiliary information matrix $\auxmat$ satisfies \assumpref{side-info}. 
    \item For each $\auxvec \in \R^\auxdim$, non-linearities $\nonlin_{1:\order}(z;\auxvec)$ are polynomials in $z \in \R^\order$ of degree at most $\degree$  which satisfy:
\begin{subequations}\label{eq:poly-nonlinearity}
\begin{align}
    \E[ |\nonlin_{i}(\serv{Z}; \serv{A})|^p] < \infty, \quad  \E[\nonlin_i^2(\serv{Z};\serv{A})]  = 1 &\quad \forall \; i \in [\order], \; p \in \N, \\
    \E[\serv{Z}_i \nonlin_j(\serv{Z};\serv{A})] = 0 &\quad \forall \; i,j \in [\order], \\ 
    \Omega_{ij} \cdot \E[\nonlin_i(\serv{Z};\serv{A}) \nonlin_j(\serv{Z};\serv{A})]  = 0 &\quad \forall \; i,j \in [\order], i \neq j.
\end{align}
\end{subequations}
In the above display, $\serv{A} \in \R^{\auxdim}$ is the random vector from \assumpref{side-info}, $\serv{Z} = (\serv{Z}_1, \serv{Z}_2, \dotsc, \serv{Z}_k) \sim \gauss{0}{I_k}$ is independent of $\serv{A}$ and $\Omega$ is the limiting covariance matrix corresponding to the semi-random ensemble $\mM_{1:k}$.
\end{enumerate} 
Then, for any fixed (independent of $\dim$) $T \in \N$, $r \in \W^\order$ with $\|r\|_1 \geq 1$, and any function $h: \R^{\auxdim} \mapsto \R$ with $\E[|h(\serv{A})|^p] < \infty \; \forall \; p \in \N$ we have,
\begin{align*}
    \lim_{\dim \rightarrow \infty} \E \left[ \frac{1}{\dim}\sum_{j=1}^\dim h(\auxvec_j) \cdot  \hermite{r}\left(\iter{z}{T,\cdot}_j\right) \right] & = 0. 
\end{align*}
In the above display, $\hermite{r}$ denotes the multivariate Hermite polynomial, introduced in \eqref{eq:hermite-notation}. 
\end{theorem}

The above result generalizes the key technical result from our prior work \citep[Theorem 3]{dudeja2022universality}, which was restricted to MVAMP algorithms of order $\order = 1$ that did not use any auxiliary information $\auxmat$. 

\begin{remark}\label{rem:moments} The conclusion of \thref{moments} can be alternatively stated as:
\begin{align} \label{eq:moments-alt}
    \lim_{\dim \rightarrow \infty} \E \left[ \frac{1}{\dim}\sum_{j=1}^\dim h(\auxvec_j) \cdot  \hermite{r}\left(\iter{z}{T,\cdot}_j\right) \right] & = \E[\hermite{r}(\serv{Z}_1, Z_1, \dotsc, \serv{Z}_k) \cdot h(\serv{A})].
\end{align}
In the above display, $\serv{A} \in \R^{\auxdim}$ is the random vector from \assumpref{side-info}, $\serv{Z} = (\serv{Z}_1, \serv{Z}_2, \dotsc, \serv{Z}_k) \sim \gauss{0}{I_k}$ is independent of $\serv{A}$. This is because:
\begin{align*}
    \E[\hermite{r}(\serv{Z}_1, Z_1, \dotsc, \serv{Z}_k) \cdot h(\serv{A})] \explain{(a)}{=} \E[\hermite{r}(\serv{Z}_1, Z_1, \dotsc, \serv{Z}_k)] \cdot \E[h(\serv{A})] \explain{(b)}{=} 0.
\end{align*}
In the above display, step (a) uses the fact that $\serv{Z} = (\serv{Z}_1, \serv{Z}_2, \dotsc, \serv{Z}_k)$ is independent of $\serv{A}$ and step (b) follows from the fact that $\E[\hermite{r}(\serv{Z}_1, Z_1, \dotsc, \serv{Z}_k)]=0$ when $\|r\|_1 \geq 1$. More generally, for a function $P: \R^{\order + \auxdim} \mapsto \R$ such that for any $\auxvec \in \R^{\auxdim}$, $P(z; \auxvec)$ is a polynomial of degree at most $\degree$ in $z \in \R^{\order}$, we have:
\begin{align*}
    \lim_{\dim \rightarrow \infty} \E \left[ \frac{1}{\dim}\sum_{j=1}^\dim  P(\iter{z}{T,\cdot}_j ; \auxvec_{i}) \right] & = \E[P(\serv{Z}_1, Z_1, \dotsc, \serv{Z}_k; \serv{A})]
\end{align*}
This follows from \eqref{eq:moments-alt} and the fact that the polynomial (in $z$) $P(z ; \auxvec)$ can be expressed as a linear combination of the multivariate Hermite polynomials $\{\hermite{r} : r \in \W^k, \|r\|_1 \leq D\}$.
\end{remark}
We also record the following useful corollary of \thref{moments}, which is a generalization of \remref{moments}.
\begin{corollary}\label{cor:moments} Let $h: \R^{k+\auxdim} \rightarrow \R$ be any continuous test function which satisfies $$|h(z; \auxvec)| \leq L \cdot (1+\|z\|^D + \|a\|^D) \; \forall \; z \; \in \; \R^{k}, \; \auxvec\;  \in \; \R^{\auxdim}$$ for some fixed constants $L\geq 0$ and $\degree \in \N$. Then, under the assumptions of \thref{moments} we have,
\begin{align*}
   \lim_{\dim \rightarrow \infty} \E\left[ \frac{1}{\dim} \sum_{i=1}^\dim h(\iter{z}{T,1}_i, \iter{z}{T,2}_i, \dotsc, \iter{z}{T,k}_i; \auxvec_i)  \right] & = \E h(\serv{Z}_1, Z_1, \dotsc, \serv{Z}_k ; \serv{A}).
\end{align*}
In the above display, $\serv{A} \in \R^{\auxdim}$ is the random vector from \assumpref{side-info}, $\serv{Z} = (\serv{Z}_1, \serv{Z}_2, \dotsc, \serv{Z}_k) \sim \gauss{0}{I_k}$ is independent of $\serv{A}$.
\end{corollary}
The proofs of \thref{moments} and \corref{moments} are provided in \appref{moment-method}. 

\subsubsection{Concentration Analysis of MVAMP Iterates}
The final ingredient need to prove the universality principle for VAMP algorithms is the following concentration estimate for the MVAMP algorithm. 

\begin{theorem}\label{thm:variance} Consider the MVAMP iterations \eqref{eq:mvamp} under the assumptions of \thref{moments}. For any fixed (independent of $\dim$) $T \in \N$, $r \in \W^\order$, and any function $h: \R^{\auxdim} \mapsto \R$ with $\E[|h(\serv{A})|^p] < \infty \; \forall \; p \in \N$ we have,
\begin{align*}
    \lim_{\dim \rightarrow \infty} \Var \left[ \frac{1}{\dim}\sum_{j=1}^\dim h(\auxvec_j) \cdot  \hermite{r}\left(\iter{z}{T,\cdot}_j\right) \right] & = 0. 
\end{align*}
In the above display, $\hermite{r}$ denotes the multivariate Hermite polynomial, as defined in \eqref{eq:hermite-notation}.
\end{theorem}

The variance bound above is derived using the Efron-Stein inequality, using arguments similar to our previous work \citep[Theorem 4]{dudeja2022universality}, which provided a similar concentration estimate for MVAMP algorithm of order $\order = 1$ that does not use any auxiliary information. The proof of \thref{variance} adapts and extends the arguments used in \citep[Theorem 4]{dudeja2022universality} for general $\order \geq 1$ and also  accounts for the variance that arises due to the randomness of the side information. \appref{concentration} is devoted to the proof of this result. 
\subsection{Proof of \thref{VAMP}}
We have now introduced all the key ideas used to obtain the universality principle for VAMP algorithms stated in \thref{VAMP}. The proof of \thref{VAMP} is presented below.

\begin{proof}[Proof of \thref{VAMP}] Consider $T$ iterations of a VAMP algorithm driven by a semi-random ensemble $\mM_{1:T}$, which satisfies all the assumptions of \thref{VAMP}:
\begin{align}\label{eq:original-VAMP}
    \iter{\vz}{t} = \mM_t \cdot \nonlin_t(\iter{\vz}{1}, \iter{\vz}{2}, \dotsc, \iter{\vz}{t-1}; \auxmat).
\end{align}
As a consequence of \propref{orthogonalization}, we can without loss of generality assume that:
\begin{enumerate}
    \item $\mM_{1:T}$ form a  balanced (\sassumpref{balanced}) semi-random ensemble (\defref{semirandom}) with limiting covariance matrix $\Omega$.
    \item $\Omega$ and the non-linearities $\nonlin_{1:T}$ satisfy the orthogonality conditions stated in \sassumpref{orthogonality}.
\end{enumerate}
Let $\serv{Z}_{1:T} \explain{i.i.d.}{\sim} \gauss{0}{1}$ and let $\serv{A}$ be the auxiliary information random variable from \assumpref{side-info}, sampled independently of $\serv{Z}_{1:T}$. Due to the orthogonality conditions of \sassumpref{orthogonality}, the Gaussian state evolution random variables associated with the VAMP algorithm in \eqref{eq:original-VAMP} are i.i.d. $\gauss{0}{1}$ random variables, and can be taken as $\serv{Z}_{1:T}$. Hence, to prove \thref{VAMP}, we need to show that for any test function $h:\R^{T+\auxdim} \mapsto \R$ that satisfies the regularity assumptions stated in \defref{PW2} we have:
\begin{align*}
    H_\dim \explain{def}{=} \frac{1}{\dim} \sum_{\ell=1}^\dim h(\iter{z}{1}_\ell, \iter{z}{2}_\ell,  \dotsc, \iter{z}{T}_\ell; \auxvec_\ell) \explain{P}{\rightarrow} \E h(\serv{Z}_1, \dotsc, \serv{Z}_T; \serv{A}). 
\end{align*}
\paragraph{Step 1: Implementing VAMP using MVAMP.} As described in \sref{MVAMP-to-VAMP}, we can implement $T$ iterations of the VAMP algorithm in \eqref{eq:original-VAMP} using an $T$ iterations of a MVAMP algorithm of order $\order = T$:
\begin{align*}
    \iter{\vw}{t,i} & = \mM_i f_i(\iter{\vw}{t-1,1}, \dotsc, \iter{\vw}{t-1,i-1}; \auxmat) \; \quad \forall \;  i \; \in \; [T].
\end{align*}
This MVAMP algorithm satisfies: 
\begin{align*}
    (\iter{\vw}{T,1}, \dotsc, \iter{\vw}{T,T}) & = (\iter{\vz}{1}, \dotsc, \iter{\vz}{T}).
\end{align*}
Hence, we now need to show that:
\begin{align} \label{eq:poly-approx-goal}
      H_\dim \explain{}{=} \frac{1}{\dim} \sum_{\ell=1}^\dim h(\iter{w}{T,1}_\ell, \iter{w}{T,2}_\ell,  \dotsc, \iter{w}{T,T}_\ell; \auxvec_\ell) \explain{P}{\rightarrow} \E h(\serv{Z}_1, \dotsc, \serv{Z}_T; \serv{A}). 
\end{align}
\paragraph{Step 2: Polynomial Approximation.} In order to obtain the conclusion \eqref{eq:poly-approx-goal} using \thref{moments}, we will approximate the non-linearities $f_{1:T}$ and the test function $h$ by polynomials using \lemref{low-degree-approx}. For every $\epsilon \in (0,1)$ let $\hat{f}_{1:T}^\epsilon$ and $\hat{h}^\epsilon$ be the approximating polynomials of $\nonlin_{1:T}$ and $h$ constructed in \lemref{low-degree-approx}. We construct the following approximating family of MVAMP iterations indexed by $\epsilon \in (0,1)$:
\begin{align} \label{eq:approx-MVAMP}
     \iter{\hat{\vw}}{t,i}(\epsilon) & = \mM_i \hat{f}^\epsilon_i(\iter{\hat{\vw}}{t-1,1}(\epsilon), \dotsc, \iter{\hat{\vw}}{t-1,i-1}(\epsilon); \auxmat) \; \quad \forall \;  i \; \in \; [T].
\end{align}
Additionally we introduce the random variable $\hat{H}_\dim(\epsilon)$ defined as:
\begin{align} \label{eq:hat-H}
    \hat{H}_\dim(\epsilon)  \explain{def}{=} \frac{1}{\dim} \sum_{\ell=1}^\dim \hat{h}^\epsilon(\iter{\hat{w}}{T,1}_\ell(\epsilon), \iter{\hat{w}}{T,2}_\ell(\epsilon),  \dotsc, \iter{\hat{w}}{T,T}_\ell(\epsilon); \auxvec_\ell).
\end{align}
Observe that for any $\epsilon \in (0,1)$, the iteration \eqref{eq:approx-MVAMP} satisfies all the requirements of \thref{moments} and \thref{variance}. Hence (cf. \remref{moments}),
\begin{align} \label{eq:H-approx-P-conv}
   \hat{H}_\dim(\epsilon)  \explain{P}{\rightarrow} \E \hat{h}^\epsilon(\serv{Z}_1, \dotsc, \serv{Z}_T ; \serv{A}). 
\end{align}
Furthermore, we claim that:
\begin{align} \label{eq:key-claim}
    \lim_{\epsilon \rightarrow 0} \limsup_{\dim \rightarrow \infty} \E |\hat{H}_\dim(\epsilon) - H_\dim| & = 0. 
\end{align}
Before proving this claim, we prove \eqref{eq:poly-approx-goal} (and hence, \thref{VAMP}) by showing that for any $\eta>0$,
\begin{align}\label{eq:poly-approx-goal-2}
    \lim_{\dim \rightarrow \infty} \P( | H_\dim - \E[{h}(\serv{Z}_1, \dotsc, \serv{Z}_T ; \serv{A})] | > 3\eta) & = 0.
\end{align}
Consider any $\epsilon \in (0,\eta)$. By \lemref{low-degree-approx} (item 2), we have:
\begin{align*}
    \left| \E[{h}(\serv{Z}_1, \dotsc, \serv{Z}_T ; \serv{A})]- \E[\hat{h}^\epsilon(\serv{Z}_1, \dotsc, \serv{Z}_T ; \serv{A})] \right| & \leq \left( \E[(\hat{h}^\epsilon( \serv{Z}_1, \dotsc, \serv{Z}_k; \serv{A}) - h( \serv{Z}_1, \dotsc, \serv{Z}_k; \serv{A}))^2] \right)^{1/2} \leq \epsilon < \eta.
\end{align*}
Hence,
\begin{align*}
    &\limsup_{\dim \rightarrow \infty}\P( | H_\dim - \E[{h}(\serv{Z}_1, \dotsc, \serv{Z}_T ; \serv{A})] | > 3\eta)  \\   & \hspace{4.5cm} \leq \limsup_{\dim \rightarrow \infty}  \left\{\P( | H_\dim - \hat{H}_\dim(\epsilon) | > \eta)  +    \P( |\hat{H}_\dim(\epsilon) - \E[\hat{h}^\epsilon(\serv{Z}_1, \dotsc, \serv{Z}_T ; \serv{A})]  | > \eta) \right\}\\
    & \hspace{4.5cm}\explain{\eqref{eq:H-approx-P-conv}}{=}   \limsup_{\dim \rightarrow \infty} \P( | H_\dim - \hat{H}_\dim(\epsilon) | > \eta) \\
    &\hspace{4.5cm} \leq  \limsup_{\dim \rightarrow \infty} \frac{\E| H_\dim - \hat{H}_\dim(\epsilon)|}{\eta}.
    \end{align*} 
Observe that in the above display, $\epsilon \in (0,\eta)$ was arbitrary. Taking $\epsilon \rightarrow 0$ and  using \eqref{eq:key-claim} yields the desired conclusion \eqref{eq:poly-approx-goal-2}. Now, to finish the proof of \thref{VAMP} we only need to prove the claim \eqref{eq:key-claim}. 
\paragraph{Step 3: Proof of \eqref{eq:key-claim}.} In order to prove \eqref{eq:key-claim} we introduce the random variable:
\begin{align*}
    \widetilde{H}_{\dim}(\epsilon) \explain{def}{=} \frac{1}{\dim} \sum_{\ell=1}^\dim {h}(\iter{\hat{w}}{T,1}_\ell(\epsilon), \iter{\hat{w}}{T,2}_\ell(\epsilon),  \dotsc, \iter{\hat{w}}{T,T}_\ell(\epsilon); \auxvec_\ell).
\end{align*}
Note the distinction between $\widetilde{H}_\dim(\epsilon)$ defined above and $\hat{H}_{\dim}(\epsilon)$ introduced in \eqref{eq:hat-H}: the random variable $\widetilde{H}_\dim(\epsilon)$ is obtained by applying the original test function $h$ to the iterates $\iter{\hat{\vw}}{T,1:T}(\epsilon)$ whereas, the random variable $\hat{H}(\epsilon)$ in \eqref{eq:hat-H} is obtained by applying the polynomial approximation $\hat{h}^\epsilon$ of $h$ to the iterates $\iter{\hat{\vw}}{T,1:T}(\epsilon)$. We bound $\E |\hat{H}_\dim(\epsilon) - H_\dim | \leq \E |\hat{H}_\dim(\epsilon) - \widetilde{H}_{\dim}(\epsilon) | + \E | \widetilde{H}_{\dim}(\epsilon) -  H_\dim|$ and analyze the two terms separately.
\paragraph{Step 3a: Analysis of $\E |\hat{H}_\dim(\epsilon) - \widetilde{H}_{\dim}(\epsilon) |$.} By Jensen's Inequality:
\begin{align*}
 &\lim_{\epsilon \rightarrow 0}  \limsup_{\dim \rightarrow \infty} \E |\hat{H}_\dim(\epsilon) - \widetilde{H}_{\dim}(\epsilon) | \\&\qquad\quad\leq \lim_{\epsilon \rightarrow 0}  \limsup_{\dim \rightarrow \infty} \left( \E\left[ \frac{1}{\dim} \sum_{\ell=1}^\dim ({h}(\iter{\hat{w}}{T,1}_\ell(\epsilon),   \dotsc, \iter{\hat{w}}{T}_\ell(\epsilon); \auxvec_\ell)-\hat{h}^\epsilon(\iter{\hat{w}}{T,1}_\ell(\epsilon), \dotsc, \iter{\hat{w}}{T}_\ell(\epsilon); \auxvec_\ell))^2 \right] \right)^{1/2} \\
 &\qquad\quad \explain{(a)}{=} \lim_{\epsilon \rightarrow 0}  \E[({h}(\serv{Z}_1,   \dotsc, \serv{Z}_T; \serv{A})-\hat{h}^\epsilon(\serv{Z}_1, \dotsc, \serv{Z}_T; \serv{A}))^2 ] \explain{(b)}{=} 0.
\end{align*}
In the above display, step (a) follows from \corref{moments} and step (b) follows from \lemref{low-degree-approx} (item 2).
\paragraph{Step 3b: Analysis of $\E | \widetilde{H}_{\dim}(\epsilon) -  H_\dim|$.} By Cauchy-Schwarz Inequality we have:
\begin{align*}
    &\lim_{\epsilon \rightarrow 0}  \limsup_{\dim \rightarrow \infty} \E | \widetilde{H}_{\dim}(\epsilon) -  H_\dim|\leq \lim_{\epsilon \rightarrow 0} \limsup_{\dim \rightarrow \infty} \frac{1}{\dim} \sum_{\ell=1}^\dim \E[ |{h}(\iter{\hat{w}}{T,1}_\ell(\epsilon),   \dotsc, \iter{\hat{w}}{T,T}_\ell(\epsilon); \auxvec_\ell) - {h}(\iter{{w}}{T,1}_\ell,   \dotsc, \iter{{w}}{T,T}_\ell; \auxvec_\ell)| ] \\
    & \leq \lim_{\epsilon \rightarrow 0} \limsup_{\dim \rightarrow \infty} 2L \cdot \left(1 + \E \|\serv{A}\|^{2D} + \sum_{i=1}^T \frac{\E \|  \iter{{\vw}}{T,i}  \|^2 + \E \| \iter{\hat{\vw}}{T,i}(\epsilon) \|^2}{\dim} \right)^{1/2} \cdot \left( \sum_{i=1}^T \frac{\E \| \iter{\hat{\vw}}{T,i}(\epsilon) - \iter{{\vw}}{T,i}  \|^2}{\dim} \right)^{1/2}.
\end{align*}
The final step in the above display relies the regularity assumption on $h$ from \defref{PW2}. Hence, the claim \eqref{eq:key-claim} follows if we can show:
\begin{align}
    \lim_{\epsilon \rightarrow 0} \limsup_{\dim \rightarrow \infty} \frac{ \E \| \iter{\hat{\vw}}{T,i}(\epsilon) \|^2}{\dim}  &< \infty \quad \forall \; i \; \in \; [T],  \label{eq:key-subclaim1}\\
     \lim_{\epsilon \rightarrow 0} \limsup_{\dim \rightarrow \infty} \frac{ \E \| \iter{\hat{\vw}}{t,i}(\epsilon) - \iter{{\vw}}{t,i}  \|^2}{\dim}  &= 0 \quad \forall \;i \; \in [t], \; t\; \in \; [T]. \label{eq:key-subclaim2}
\end{align}
The claim \eqref{eq:key-subclaim1} is immediate since \thref{moments} shows that $\lim_{\dim \rightarrow \infty} \E  \| \iter{\hat{\vw}}{T,i}(\epsilon) \|^2/\dim  = 1$. The claim \eqref{eq:key-subclaim2} can be shown inductively for each $t \in [T]$. For $t=1$ we have:
\begin{align*}
    \lim_{\epsilon \rightarrow 0} \limsup_{\dim \rightarrow \infty} \frac{ \E \| \iter{\hat{\vw}}{1,i}(\epsilon) - \iter{{\vw}}{1,i}  \|^2}{\dim} & \leq  \lim_{\epsilon \rightarrow 0} \limsup_{\dim \rightarrow \infty}  \|\mPsi_i\|_{\op}^2 \frac{\E[\|\nonlin_1(\auxmat) - \hat{\nonlin}_1^\epsilon(\auxmat)\|^2]}{\dim} \\&= (\lim_{\epsilon \rightarrow 0}   \E[(\nonlin_1(\serv{A}) - \hat{\nonlin}_1^\epsilon(\serv{A}))^2]) \cdot (\limsup_{\dim \rightarrow \infty} \|\mPsi_i\|_{\op}^2) = 0,
\end{align*}
where the last step follows from \lemref{low-degree-approx}. Assume as the induction hypothesis that the claim \eqref{eq:key-subclaim2} holds for some $t<T$. We verify the claim for $t+1$ as follows:
\begin{align*}
     &\lim_{\epsilon \rightarrow 0} \limsup_{\dim \rightarrow \infty} \frac{ \E \| \iter{\hat{\vw}}{t+1,i}(\epsilon) - \iter{{\vw}}{t+1,i}  \|^2}{\dim}  =  \lim_{\epsilon \rightarrow 0} \limsup_{\dim \rightarrow \infty} \frac{ \E \| \mM_i \cdot (\hat{f}^\epsilon_i(\iter{\hat{\mW}}{t,1:i-1}(\epsilon); \auxmat) - {f}_i(\iter{{\mW}}{t,1:i-1}; \auxmat)) \|^2}{\dim} \\
     & \leq 2\lim_{\epsilon \rightarrow 0} \limsup_{\dim \rightarrow \infty} \|\mPsi_i\|_{\op}^2  \left( \frac{ \E \| \hat{f}^\epsilon_i(\iter{\hat{\mW}}{t,1:i-1}(\epsilon); \auxmat) - {f}_i(\iter{\hat{\mW}}{t,1:i-1}(\epsilon); \auxmat) \|^2}{\dim} \right.\\ & \hspace{8cm}\left.+ \frac{\E \| {f}_i(\iter{\hat{\mW}}{t,1:i-1}(\epsilon); \auxmat) - {f}_i(\iter{{\mW}}{t,1:i-1}; \auxmat) \|^2}{\dim}  \right). 
\end{align*}
Appealing to \corref{moments}, we obtain:
\begin{align*}
  &\lim_{\epsilon \rightarrow 0} \limsup_{\dim \rightarrow \infty} \frac{ \E \| \iter{\hat{\vw}}{t+1,i}(\epsilon) - \iter{{\vw}}{t+1,i}  \|^2}{\dim} \\
     & \explain{(a)}{\leq} 2\lim_{\epsilon \rightarrow 0} \limsup_{\dim \rightarrow \infty} \|\mPsi_i\|_{\op}^2 \left(  \E[(f_i(\serv{Z}_{1:i-1}; \serv{A})- \hat{f}^\epsilon_i(\serv{Z}_{1:i-1}; \serv{A}))^2] + \frac{ \E \| {f}_i(\iter{\hat{\mW}}{t,1:i-1}(\epsilon); \auxmat) - {f}_i(\iter{{\mW}}{t,1:i-1}; \auxmat) \|^2}{\dim}  \right) \\
     &\explain{(b)}{=} 2\lim_{\epsilon \rightarrow 0} \limsup_{\dim \rightarrow \infty} \|\mPsi_i\|_{\op}^2 \cdot \left(  \frac{ \E \| {f}_i(\iter{\hat{\mW}}{t,1:i-1}(\epsilon); \auxmat) - {f}_i(\iter{{\mW}}{t,1:i-1}; \auxmat) \|^2}{\dim}   \right) \\
     & \explain{(c)}{\leq } \lim_{\epsilon \rightarrow 0} \limsup_{\dim \rightarrow \infty} 2L^2 \cdot  \|\mPsi_i\|_{\op}^2 \cdot \left(\sum_{j=1}^{i-1} \frac{ \E \| \iter{\hat{\vw}}{t,j}(\epsilon) - \iter{{\vw}}{t,j}  \|^2}{\dim}\right) \explain{(d)}{=} 0. 
\end{align*}
In the above display step (a) follows from \corref{moments}, step (b) follows from the approximation guarantee provided by \lemref{low-degree-approx}. Step (c) uses the fact that $f_i$ is a Lipschitz function and step (d) follows from the induction hypothesis. This completes the proof of claim \eqref{eq:key-claim} and hence, \thref{VAMP} is also proved.
\end{proof}

\section{Discussion and Future Work}
For many high-dimensional inference problems, the statistical properties of estimators appear to exhibit broad universality with respect to the underlying sensing matrix. In particular, the asymptotic performance of estimators seems to be determined only by the spectrum of the sensing matrix as long as the singular vectors are sufficiently generic. In this work, we formalized this universality heuristic in the context of the regularized linear regression by introducing the notion of a spectral universality class. This universality class consists of matrices that share the same limiting spectrum and satisfy a set of deterministic conditions which formalize the heuristic notion of ``generic'' singular vectors. Our results show that, for all sensing matrices in a given spectral universality class, the statistical properties of regularized least squares estimators and the dynamics of the proximal method (or, more generally, first-order methods) are asymptotically identical. 

\paragraph{}In addition to satisfying a set of deterministic conditions (introduced in \defref{univ-class}), our universality results require the sensing matrix $\mX$ to be sign-invariant in the sense that $\mX = \mJ \mS$  for a deterministic matrix $\mJ$ and a uniformly random sign diagonal matrix $\mS$. Consequently, fully deterministic matrices are precluded from satisfying our assumptions. An exciting avenue for future work is to relax this sign invariance requirement. We conclude this paper with a discussion of our preliminary findings regarding this aspect. 
\paragraph{Sign invariance from symmetry.} When the inference problem has an underlying sign symmetry, the sign invariance requirement on the sensing matrix is unnecessary. Concretely, consider the regularized linear regression problem where the entries of the signal vector $\vbeta_\star$ are i.i.d. copies of a random variable $\serv{B_\star}$ which satisfies $\serv{B_\star} \explain{d}{=} -\serv{B_\star}$. Consider the RLS estimator:
\begin{align} \label{eq:RLS-conclusion}
    \vbeta_\opt(\mX, \vbeta_\star, \vepsilon) \explain{def}{=} \argmin_{\vbeta\in \R^\dim} \frac{1}{2\dim}  \|\mX \vbeta_\star + \vepsilon - \mX \vbeta\|^2 + \frac{1}{\dim} \sum_{i=1}^\dim \reg(\beta_i)
\end{align}
for an \emph{even regularizer} $\reg$ (that is, $\rho(x) = \rho(-x)$ for any $x \in \R$). For any sign diagonal matrix $\mS$, by introducing the change of variables $\vbeta \leftrightarrow \mS \vbeta$ in \eqref{eq:RLS-conclusion}, we obtain:
\begin{align} \label{eq:change-of-coordinates}
    \vbeta_\opt(\mJ, \vbeta_\star, \vepsilon) & = \mS \vbeta_\opt(\mJ \mS, \mS \vbeta_\star, \vepsilon).
\end{align}
This means that:
\begin{align} 
    \|\vbeta_\opt(\mJ, \vbeta_\star, \vepsilon) - \vbeta_\star\|^2 \explain{\eqref{eq:change-of-coordinates}}{=}  \|\mS \vbeta_\opt(\mJ \mS, \mS \vbeta_\star, \vepsilon) - \vbeta_\star\|^2 &= \| \vbeta_\opt(\mJ \mS, \mS \vbeta_\star, \vepsilon) - \mS \vbeta_\star\|^2 \nonumber \\& \explain{d}{=} \| \vbeta_\opt(\mJ \mS,  \vbeta_\star, \vepsilon) -  \vbeta_\star\|^2, \label{eq:distributional-eq}
\end{align}
where the final distributional equality follows from the fact that $\mS \vbeta_\star \explain{d}{=} \vbeta_\star$ (since the entries of $\vbeta_\star$ are i.i.d. and symmetrically distributed). Taking $\mS$ to be a uniformly random sign diagonal matrix, we conclude from \eqref{eq:distributional-eq} that the asymptotic behavior of the mean square error $\|\vbeta_\opt(\mJ, \vbeta_\star, \vepsilon) - \vbeta_\star\|^2$ for a deterministic sensing matrix $\mJ$ is identical to the asymptotic behavior of the mean square error $\|\vbeta_\opt(\mX, \vbeta_\star, \vepsilon) - \vbeta_\star\|^2$ for the sign invariant sensing matrix $\mX = \mJ \mS$. The latter can be analyzed using our universality result. 
\begin{figure}
\centering
         \includegraphics[width=0.9\textwidth,trim={3.1cm 0.1cm 3.1cm 0.1cm},clip]{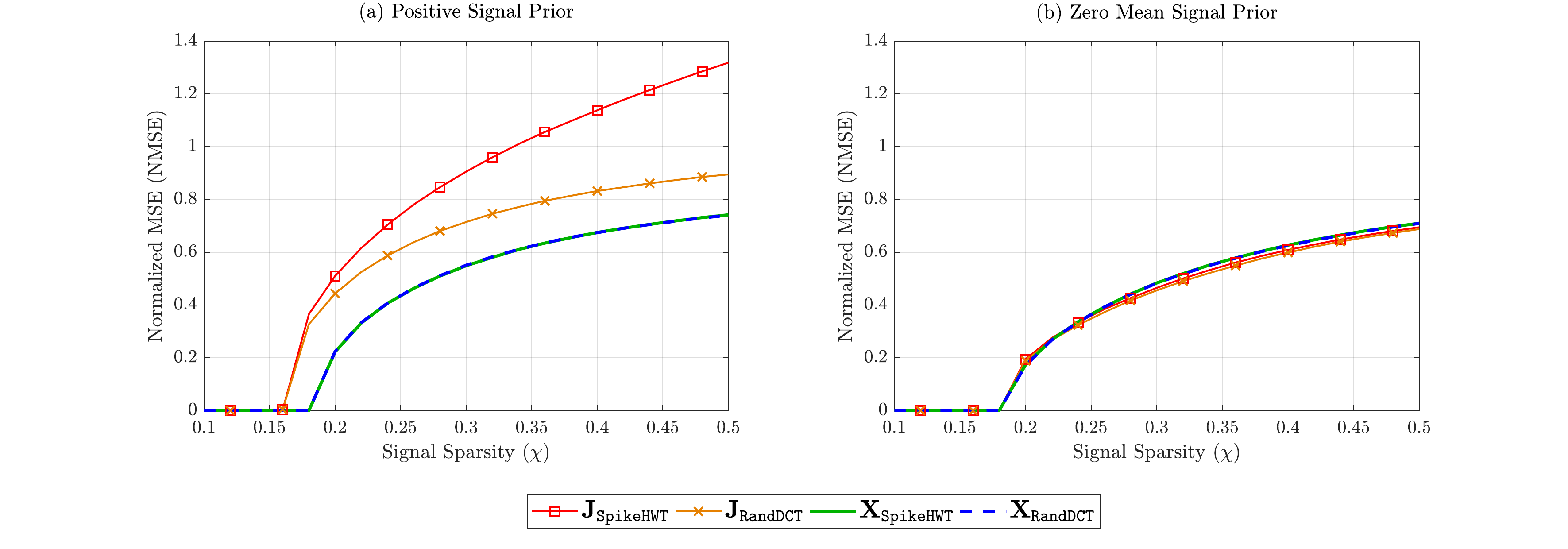}
\caption{Normalized mean square error (NMSE) of the RLS estimator with elastic net penalty v.s. signal sparsity for $\mX_{\texttt{SpikeHWT}}$ and $\mX_{\texttt{RandDCT}}$ ensembles [cf. \eqref{eq:signed-ensemble}] and their unsigned versions $\mJ_{\texttt{SpikeHWT}}$ and $\mJ_{\texttt{RandDCT}}$ [cf. \eqref{eq:unsigned-ensemble}] for a positive signal prior [Panel (a)] and a zero mean signal prior [Panel (b)].} 
\label{fig:signplot}
\end{figure}
\paragraph{Empirical breakdown of universality without sign invariance.} In the absence of additional assumptions, we have empirically observed deviations from universality for matrices that are not sign invariant. In order to demonstrate this, we conducted a numerical simulation for the noiseless linear regression problem with $\vy = \mX \vbeta_\star \in \R^{\ssize}$ ($\ssize = 2^{20})$ where the signal vector $\vbeta_\star \in \R^\dim$ ($\dim = 2 \ssize = 2^{21}$) was sampled from the i.i.d. prior:
\begin{align}\label{eq:sign-role-prior}
    (\vbeta_\star)_i \explain{i.i.d.}{\sim} (1-\chi) \cdot \delta_0 + \chi \cdot  \left( \frac{2}{3} \delta_{12} + \frac{1}{3} \delta_{20}  \right).
\end{align}
In the above display, $\chi \in [0,1]$ is a parameter which controls the sparsity of the signal vector. We plot the normalized mean square error (NMSE):
\begin{align}\label{eq:nmse}
    \mathrm{NMSE} \explain{def}{=} \frac{\| \vbeta_\opt - \vbeta_\star\|^2}{\|\vbeta_\star\|^2}
\end{align}
of the RLS estimator \eqref{eq:RLS-conclusion} with a small elastic net \eqref{eq:reg_net} regularization ($\lambda_1 = 0.01, \; \lambda_2 = 10^{-3} \cdot \lambda_1$) as we vary the signal sparsity $\chi$. We considered 4 different sensing matrices. Two of these were the sign-invariant \texttt{SpikeHWT} (defined below) and \texttt{RandDCT} (introduced in \sref{demonstration}), both of which are covered by our universality result:
\begin{align} \label{eq:signed-ensemble}
     \mX_{\texttt{SpikeHWT}} \explain{def}{=} \frac{1}{\sqrt{2}} \begin{bmatrix} \mI_{\ssize} & \mH_\ssize \end{bmatrix} \cdot \mS, \quad \mX_\texttt{RandDCT} \explain{def}{=}  \diag\{\underbrace{1, \ldots, 1}_{\ssize}, \underbrace{0, \ldots, 0}_{\ssize}\} \cdot (\mP \mQ_N \mS),
\end{align}
In the above display $\mS$ is a uniformly random sign diagonal matrix, $\mP$ is a uniformly random permutation matrix, $\mH_{\ssize}$ is the $\ssize \times \ssize$ Hadamard-Walsh matrix, and $\mQ_{\dim}$ is the $\dim \times \dim$ DCT matrix. We also considered the unsigned versions of \texttt{SpikeHWT} and \texttt{RandDCT}, which are not covered by our results:
\begin{align}\label{eq:unsigned-ensemble}
    \mJ_{\texttt{SpikeHWT}} \explain{def}{=} \frac{1}{\sqrt{2}} \begin{bmatrix} \mI_{\ssize} & \mH_\ssize \end{bmatrix}, \quad \mJ_\texttt{RandDCT} \explain{def}{=}  \diag\{\underbrace{1, \ldots, 1}_{\ssize}, \underbrace{0, \ldots, 0}_{\ssize}\} \cdot (\mP \mQ_N ).
\end{align}
The results of the simulation are shown in \fref{signplot}a. As predicted by the universality principle in \thref{RLS}, the NMSE curves for $\mX_{\texttt{SpikeHWT}}$ and $\mX_\texttt{RandDCT}$ coincide. However, the NMSE curves of $\mJ_{\texttt{SpikeHWT}}$ and $\mJ_\texttt{RandDCT}$ appear to be different, suggesting that $\mJ_{\texttt{SpikeHWT}}, \mJ_\texttt{RandDCT}$ might not lie in the same universality class as $\{\mX_{\texttt{SpikeHWT}}, \mX_\texttt{RandDCT}\}$.
\paragraph{Obstacles to universality without sign invariance.} One possible explanation for the breakdown of universality observed in \fref{signplot}a is that the action of the matrices $\mJ_{\texttt{SpikeHWT}}, \mJ_\texttt{RandDCT}$ on the all ones vector $\vone_\dim  = (1, 1, \dotsc, 1) \tran$ is very different from the action of $\{\mX_{\texttt{SpikeHWT}}, \mX_\texttt{RandDCT}\}$ on $\vone_\dim$. Indeed, \thref{VAMP} shows that for $\mX \in \{\mX_{\texttt{SpikeSine}}, \mX_\texttt{RandDCT}\}$\footnote{This can be seen by viewing $(\mX \tran \mX - \tfrac{1}{2}\mI_{\dim}) \cdot \vone_\dim$ as a single iteration VAMP algorithm and applying \thref{VAMP}.}
\begin{align} \label{eq:action-on-one}
    \mX\tran \mX  \cdot \vone_\dim & = (\mX \tran \mX - \tfrac{1}{2}\mI_{\dim}) \cdot \vone_\dim + \tfrac{1}{2} \cdot \vone_\dim \explain{\pw}{\longrightarrow} \gauss{1/2}{1/4} \quad \forall \; \mX \in \{\mX_{\texttt{SpikeHWT}}, \mX_\texttt{RandDCT}\}.
\end{align}
On the other hand, since the first row of the DCT matrix and the first row and column of the Hadamard-Walsh matrix are $\vone /\|\vone\|$ with the remaining rows/columns orthogonal to $\vone$, one can compute:
\begin{subequations} \label{eq:action-on-one-2}
\begin{align}
    \mJ_{\texttt{SpikeHWT}}\tran \; \mJ_{\texttt{SpikeHWT}} \cdot \vone_{\dim} & = \tfrac{1}{2} \cdot \vone_\dim + \tfrac{\sqrt{\ssize}}{2} \cdot (\ve_1 + \ve_{\ssize+1}), \\
    \mJ_{\texttt{RandDCT}}\tran \; \mJ_{\texttt{RandDCT}} \cdot \vone_{\dim} & = \begin{cases} \vone_\dim & \text{ with probability } 1/2 \\ \vzero & \text{ with probability } 1/2  \end{cases},
\end{align}
\end{subequations}
where $\ve_{1:\dim}$ denote the standard basis vectors of $\R^{\dim}$. Hence, $\mJ_{\texttt{SpikeHWT}}\tran \mJ_{\texttt{SpikeHWT}}  \vone_{\dim}$ and $\mJ_{\texttt{RandDCT}}\tran  \mJ_{\texttt{RandDCT}}  \vone_{\dim}$ do not have a limiting Gaussian distribution as in \eqref{eq:action-on-one}. Since the prior in \eqref{eq:sign-role-prior} has a non-zero mean, the signal vector $\vbeta_\star$ has a non-trivial projection along $\vone_\dim$. Hence the discrepancy highlighted in \eqref{eq:action-on-one} and \eqref{eq:action-on-one-2} can be one possible explanation for the breakdown of universality in \fref{signplot}a. To further test this hypothesis, we repeated our simulation keeping the same setup, but generating the signal entries from the zero-mean prior obtained by centering the prior in \eqref{eq:sign-role-prior}:
\begin{align}\label{eq:sign-role-prior-zeromean}
    (\vbeta_\star)_i \explain{i.i.d.}{\sim} (1-\chi) \cdot \delta_0 + \chi \cdot  \left( \frac{2}{3} \delta_{-8/3} + \frac{1}{3} \delta_{16/3}  \right).
\end{align}
This ensures that the signal $\vbeta_\star$ is asymptotically orthogonal to $\vone_\dim$. However, since the prior in \eqref{eq:sign-role-prior-zeromean} is not symmetric, the symmetrization argument from \eqref{eq:distributional-eq} does not apply to it and hence, this situation is not covered by the universality results of this paper. The simulation results are shown in \fref{signplot}b. We found that the deviations from universality were significantly reduced in this case, providing some evidence for our hypothesis. An important direction for future work is to understand if the discrepancy highlighted in \eqref{eq:action-on-one} and \eqref{eq:action-on-one-2} is the only obstacle to universality. Furthermore, obtaining a universality result that replaces the sign invariance assumption with easily verifiable deterministic conditions on the sensing matrix and the signal vector would also be interesting.

\subsection*{Acknowledgements}
SS gratefully acknowledges support from a Harvard FAS Dean's competitive fund award. The work of YML is supported by a Harvard FAS Dean's competitive fund award for promising scholarship, and by the US National Science Foundation under grant CCF-1910410.

\bibliographystyle{plainnat}
\bibliography{refs}

\begin{thebibliography}{137}
\providecommand{\natexlab}[1]{#1}
\providecommand{\url}[1]{\texttt{#1}}
\expandafter\ifx\csname urlstyle\endcsname\relax
  \providecommand{\doi}[1]{doi: #1}\else
  \providecommand{\doi}{doi: \begingroup \urlstyle{rm}\Url}\fi

\bibitem[Abbara et~al.(2020)Abbara, Baker, Krzakala, and
  Zdeborov{\'a}]{abbara2020universality}
Alia Abbara, Antoine Baker, Florent Krzakala, and Lenka Zdeborov{\'a}.
\newblock On the universality of noiseless linear estimation with respect to
  the measurement matrix.
\newblock \emph{Journal of Physics A: Mathematical and Theoretical},
  53\penalty0 (16):\penalty0 164001, 2020.

\bibitem[Amelunxen et~al.(2014)Amelunxen, Lotz, McCoy, and
  Tropp]{amelunxen2014living}
Dennis Amelunxen, Martin Lotz, Michael~B McCoy, and Joel~A Tropp.
\newblock Living on the edge: Phase transitions in convex programs with random
  data.
\newblock \emph{Information and Inference: A Journal of the IMA}, 3\penalty0
  (3):\penalty0 224--294, 2014.

\bibitem[Anderson and Farrell(2014)]{anderson2014asymptotically}
Greg~W Anderson and Brendan Farrell.
\newblock Asymptotically liberating sequences of random unitary matrices.
\newblock \emph{Advances in Mathematics}, 255:\penalty0 381--413, 2014.

\bibitem[Applebaum et~al.(2009)Applebaum, Howard, Searle, and
  Calderbank]{applebaum2009chirp}
Lorne Applebaum, Stephen~D Howard, Stephen Searle, and Robert Calderbank.
\newblock Chirp sensing codes: Deterministic compressed sensing measurements
  for fast recovery.
\newblock \emph{Applied and Computational Harmonic Analysis}, 26\penalty0
  (2):\penalty0 283--290, 2009.

\bibitem[Bai and Yin(2008)]{bai2008limit}
Zhi-Dong Bai and Yong-Qua Yin.
\newblock Limit of the smallest eigenvalue of a large dimensional sample
  covariance matrix.
\newblock In \emph{Advances In Statistics}, pages 108--127. World Scientific,
  2008.

\bibitem[Bai et~al.(2009)Bai, Jiang, Yao, and Zheng]{bai2009corrections}
Zhidong Bai, Dandan Jiang, Jian-Feng Yao, and Shurong Zheng.
\newblock Corrections to {LRT} on large-dimensional covariance matrix by {RMT}.
\newblock \emph{The Annals of Statistics}, 37\penalty0 (6B):\penalty0
  3822--3840, 2009.

\bibitem[Bandeira et~al.(2013)Bandeira, Fickus, Mixon, and
  Wong]{bandeira2013road}
Afonso~S Bandeira, Matthew Fickus, Dustin~G Mixon, and Percy Wong.
\newblock The road to deterministic matrices with the restricted isometry
  property.
\newblock \emph{Journal of Fourier Analysis and Applications}, 19\penalty0
  (6):\penalty0 1123--1149, 2013.

\bibitem[Bapat and Sunder(1985)]{bapat1985majorization}
Ravindra~B Bapat and Vaikalathur~S Sunder.
\newblock On majorization and {S}chur products.
\newblock \emph{Linear algebra and its applications}, 72:\penalty0 107--117,
  1985.

\bibitem[Barbier et~al.(2019)Barbier, Krzakala, Macris, Miolane, and
  Zdeborov{\'a}]{barbier2019optimal}
Jean Barbier, Florent Krzakala, Nicolas Macris, L{\'e}o Miolane, and Lenka
  Zdeborov{\'a}.
\newblock Optimal errors and phase transitions in high-dimensional generalized
  linear models.
\newblock \emph{Proceedings of the National Academy of Sciences}, 116\penalty0
  (12):\penalty0 5451--5460, 2019.

\bibitem[Bayati and Montanari(2011{\natexlab{a}})]{bayati2011dynamics}
Mohsen Bayati and Andrea Montanari.
\newblock The dynamics of message passing on dense graphs, with applications to
  compressed sensing.
\newblock \emph{IEEE Transactions on Information Theory}, 57\penalty0
  (2):\penalty0 764--785, 2011{\natexlab{a}}.

\bibitem[Bayati and Montanari(2011{\natexlab{b}})]{bayati2011lasso}
Mohsen Bayati and Andrea Montanari.
\newblock The lasso risk for {G}aussian matrices.
\newblock \emph{IEEE Transactions on Information Theory}, 58\penalty0
  (4):\penalty0 1997--2017, 2011{\natexlab{b}}.

\bibitem[Bayati et~al.(2015)Bayati, Lelarge, and
  Montanari]{bayati2015universality}
Mohsen Bayati, Marc Lelarge, and Andrea Montanari.
\newblock Universality in polytope phase transitions and message passing
  algorithms.
\newblock \emph{The Annals of Applied Probability}, 25\penalty0 (2):\penalty0
  753--822, 2015.

\bibitem[Bean et~al.(2013)Bean, Bickel, El~Karoui, and Yu]{bean2013optimal}
Derek Bean, Peter~J Bickel, Noureddine El~Karoui, and Bin Yu.
\newblock Optimal {M}-estimation in high-dimensional regression.
\newblock \emph{Proceedings of the National Academy of Sciences}, 110\penalty0
  (36):\penalty0 14563--14568, 2013.

\bibitem[Beck(2017)]{beck2017}
Amir Beck.
\newblock \emph{First-Order Methods in Optimization}.
\newblock Society for Industrial and Applied Mathematics, Philadelphia, PA,
  2017.
\newblock \doi{10.1137/1.9781611974997}.
\newblock URL \url{https://epubs.siam.org/doi/abs/10.1137/1.9781611974997}.

\bibitem[Bercu et~al.(2015)Bercu, Delyon, and Rio]{bercu2015concentration}
Bernard Bercu, Bernard Delyon, and Emmanuel Rio.
\newblock \emph{Concentration inequalities for sums and martingales}.
\newblock Springer, 2015.

\bibitem[Berthier et~al.(2020)Berthier, Montanari, and
  Nguyen]{berthier2020state}
Raphael Berthier, Andrea Montanari, and Phan-Minh Nguyen.
\newblock State evolution for approximate message passing with non-separable
  functions.
\newblock \emph{Information and Inference: A Journal of the IMA}, 9\penalty0
  (1):\penalty0 33--79, 2020.

\bibitem[Bloemendal et~al.(2014)Bloemendal, Erd{\H{o}}s, Knowles, Yau, and
  Yin]{alex2014isotropic}
Alex Bloemendal, L{\'a}szl{\'o} Erd{\H{o}}s, Antti Knowles, Horng-Tzer Yau, and
  Jun Yin.
\newblock Isotropic local laws for sample covariance and generalized {W}igner
  matrices.
\newblock \emph{Electronic Journal of Probability}, 19:\penalty0 1--53, 2014.

\bibitem[Bolthausen(2014)]{bolthausen2014iterative}
Erwin Bolthausen.
\newblock An iterative construction of solutions of the {TAP} equations for the
  {S}herrington--{K}irkpatrick model.
\newblock \emph{Communications in Mathematical Physics}, 325\penalty0
  (1):\penalty0 333--366, 2014.

\bibitem[Bu et~al.(2020)Bu, Klusowski, Rush, and Su]{bu2020algorithmic}
Zhiqi Bu, Jason~M Klusowski, Cynthia Rush, and Weijie~J Su.
\newblock Algorithmic analysis and statistical estimation of {SLOPE} via
  approximate message passing.
\newblock \emph{IEEE Transactions on Information Theory}, 67\penalty0
  (1):\penalty0 506--537, 2020.

\bibitem[{\c{C}}akmak and Opper(2019)]{ccakmak2019memory}
Burak {\c{C}}akmak and Manfred Opper.
\newblock Memory-free dynamics for the {T}houless-{A}nderson-{P}almer equations
  of {I}sing models with arbitrary rotation-invariant ensembles of random
  coupling matrices.
\newblock \emph{Physical Review E}, 99\penalty0 (6):\penalty0 062140, 2019.

\bibitem[Calderbank et~al.(2010)Calderbank, Howard, and
  Jafarpour]{calderbank2010construction}
Robert Calderbank, Stephen Howard, and Sina Jafarpour.
\newblock Construction of a large class of deterministic sensing matrices that
  satisfy a statistical isometry property.
\newblock \emph{IEEE journal of selected topics in signal processing},
  4\penalty0 (2):\penalty0 358--374, 2010.

\bibitem[Cand{\`e}s and Sur(2020)]{candes2020phase}
Emmanuel~J Cand{\`e}s and Pragya Sur.
\newblock The phase transition for the existence of the maximum likelihood
  estimate in high-dimensional logistic regression.
\newblock \emph{The Annals of Statistics}, 48\penalty0 (1):\penalty0 27--42,
  2020.

\bibitem[Carmona and Hu(2006)]{carmona2006universality}
Philippe Carmona and Yueyun Hu.
\newblock Universality in {S}herrington–{K}irkpatrick's spin glass model.
\newblock \emph{Annales de l'Institut Henri Poincare (B) Probability and
  Statistics}, 42\penalty0 (2):\penalty0 215--222, 2006.
\newblock ISSN 0246-0203.
\newblock \doi{https://doi.org/10.1016/j.anihpb.2005.04.001}.
\newblock URL
  \url{https://www.sciencedirect.com/science/article/pii/S0246020305000634}.

\bibitem[Casazza et~al.(2012)Casazza, Kutyniok, and
  Philipp]{casazza2012introduction}
Peter~G Casazza, Gitta Kutyniok, and Friedrich Philipp.
\newblock Introduction to finite frame theory.
\newblock \emph{Finite Frames: Theory and Applications}, page~1, 2012.

\bibitem[Celentano et~al.(2020{\natexlab{a}})Celentano, Montanari, and
  Wei]{celentano2020lasso}
Michael Celentano, Andrea Montanari, and Yuting Wei.
\newblock The lasso with general {G}aussian designs with applications to
  hypothesis testing.
\newblock \emph{arXiv preprint arXiv:2007.13716}, 2020{\natexlab{a}}.

\bibitem[Celentano et~al.(2020{\natexlab{b}})Celentano, Montanari, and
  Wu]{celentano2020estimation}
Michael Celentano, Andrea Montanari, and Yuchen Wu.
\newblock The estimation error of general first order methods.
\newblock In \emph{Conference on Learning Theory}, pages 1078--1141. PMLR,
  2020{\natexlab{b}}.

\bibitem[Celentano et~al.(2021)Celentano, Cheng, and
  Montanari]{celentano2021high}
Michael Celentano, Chen Cheng, and Andrea Montanari.
\newblock The high-dimensional asymptotics of first order methods with random
  data.
\newblock \emph{arXiv preprint arXiv:2112.07572}, 2021.

\bibitem[Chandrasekaran et~al.(2012)Chandrasekaran, Recht, Parrilo, and
  Willsky]{chandrasekaran2012convex}
Venkat Chandrasekaran, Benjamin Recht, Pablo~A Parrilo, and Alan~S Willsky.
\newblock The convex geometry of linear inverse problems.
\newblock \emph{Foundations of Computational mathematics}, 12\penalty0
  (6):\penalty0 805--849, 2012.

\bibitem[Chatterjee(2005)]{chatterjee2005simple}
Sourav Chatterjee.
\newblock A simple invariance theorem.
\newblock \emph{arXiv preprint math/0508213}, 2005.

\bibitem[Chen and Lam(2021)]{chen2021universality}
Wei-Kuo Chen and Wai-Kit Lam.
\newblock Universality of approximate message passing algorithms.
\newblock \emph{Electronic Journal of Probability}, 26:\penalty0 1--44, 2021.

\bibitem[Cover(1965)]{cover1965geometrical}
Thomas~M. Cover.
\newblock Geometrical and statistical properties of systems of linear
  inequalities with applications in pattern recognition.
\newblock \emph{IEEE Transactions on Electronic Computers}, EC-14\penalty0
  (3):\penalty0 326--334, 1965.
\newblock \doi{10.1109/PGEC.1965.264137}.

\bibitem[Diaconis and Freedman(1984)]{diaconis1984asymptotics}
Persi Diaconis and David Freedman.
\newblock Asymptotics of graphical projection pursuit.
\newblock \emph{The Annals of Statistics}, pages 793--815, 1984.

\bibitem[Dobriban and Liu(2019)]{dobriban2019asymptotics}
Edgar Dobriban and Sifan Liu.
\newblock Asymptotics for sketching in least squares regression.
\newblock \emph{Advances in Neural Information Processing Systems}, 32, 2019.

\bibitem[Dobriban and Wager(2018)]{dobriban2018high}
Edgar Dobriban and Stefan Wager.
\newblock High-dimensional asymptotics of prediction: Ridge regression and
  classification.
\newblock \emph{The Annals of Statistics}, 46\penalty0 (1):\penalty0 247--279,
  2018.

\bibitem[Dobson and Barnett(2018)]{dobson2018introduction}
Annette~J Dobson and Adrian~G Barnett.
\newblock \emph{An introduction to generalized linear models}.
\newblock Chapman and Hall/CRC, 2018.

\bibitem[Donoho and Montanari(2016)]{donoho2016high}
David Donoho and Andrea Montanari.
\newblock High dimensional robust {M}-estimation: Asymptotic variance via
  approximate message passing.
\newblock \emph{Probability Theory and Related Fields}, 166\penalty0
  (3):\penalty0 935--969, 2016.

\bibitem[Donoho and Tanner(2009{\natexlab{a}})]{donoho2009counting}
David Donoho and Jared Tanner.
\newblock Counting faces of randomly projected polytopes when the projection
  radically lowers dimension.
\newblock \emph{Journal of the American Mathematical Society}, 22\penalty0
  (1):\penalty0 1--53, 2009{\natexlab{a}}.

\bibitem[Donoho and Tanner(2009{\natexlab{b}})]{donoho2009observed}
David Donoho and Jared Tanner.
\newblock Observed universality of phase transitions in high-dimensional
  geometry, with implications for modern data analysis and signal processing.
\newblock \emph{Philosophical Transactions of the Royal Society A:
  Mathematical, Physical and Engineering Sciences}, 367\penalty0
  (1906):\penalty0 4273--4293, 2009{\natexlab{b}}.

\bibitem[Donoho(2005)]{donoho2005neighborly}
David~L. Donoho.
\newblock Neighborly polytopes and sparse solutions of underdetermined linear
  equations.
\newblock Technical report, Stanford University, 2005.

\bibitem[Donoho(2006)]{donoho2006high}
David~L Donoho.
\newblock High-dimensional centrally symmetric polytopes with neighborliness
  proportional to dimension.
\newblock \emph{Discrete \& Computational Geometry}, 35\penalty0 (4):\penalty0
  617--652, 2006.

\bibitem[Donoho and Tanner(2005{\natexlab{a}})]{donoho2005neighborliness}
David~L Donoho and Jared Tanner.
\newblock Neighborliness of randomly projected simplices in high dimensions.
\newblock \emph{Proceedings of the National Academy of Sciences}, 102\penalty0
  (27):\penalty0 9452--9457, 2005{\natexlab{a}}.

\bibitem[Donoho and Tanner(2005{\natexlab{b}})]{donoho2005sparse}
David~L Donoho and Jared Tanner.
\newblock Sparse nonnegative solution of underdetermined linear equations by
  linear programming.
\newblock \emph{Proceedings of the national academy of sciences}, 102\penalty0
  (27):\penalty0 9446--9451, 2005{\natexlab{b}}.

\bibitem[Donoho and Tanner(2010)]{donoho2010counting}
David~L Donoho and Jared Tanner.
\newblock Counting the faces of randomly-projected hypercubes and orthants,
  with applications.
\newblock \emph{Discrete \& computational geometry}, 43\penalty0 (3):\penalty0
  522--541, 2010.

\bibitem[Donoho et~al.(2009)Donoho, Maleki, and Montanari]{donoho2009message}
David~L Donoho, Arian Maleki, and Andrea Montanari.
\newblock Message-passing algorithms for compressed sensing.
\newblock \emph{Proceedings of the National Academy of Sciences}, 106\penalty0
  (45):\penalty0 18914--18919, 2009.

\bibitem[Donoho et~al.(2013)Donoho, Javanmard, and
  Montanari]{donoho2013information}
David~L Donoho, Adel Javanmard, and Andrea Montanari.
\newblock Information-theoretically optimal compressed sensing via spatial
  coupling and approximate message passing.
\newblock \emph{IEEE transactions on information theory}, 59\penalty0
  (11):\penalty0 7434--7464, 2013.

\bibitem[Dudeja and Bakhshizadeh(2022)]{dudeja2020universality}
Rishabh Dudeja and Milad Bakhshizadeh.
\newblock Universality of linearized message passing for phase retrieval with
  structured sensing matrices.
\newblock \emph{IEEE Transactions on Information Theory}, pages 1--1, 2022.
\newblock \doi{10.1109/TIT.2022.3182018}.

\bibitem[Dudeja et~al.(2020{\natexlab{a}})Dudeja, Bakhshizadeh, Ma, and
  Maleki]{dudeja2020analysis}
Rishabh Dudeja, Milad Bakhshizadeh, Junjie Ma, and Arian Maleki.
\newblock Analysis of spectral methods for phase retrieval with random
  orthogonal matrices.
\newblock \emph{IEEE Transactions on Information Theory}, 66\penalty0
  (8):\penalty0 5182--5203, 2020{\natexlab{a}}.

\bibitem[Dudeja et~al.(2020{\natexlab{b}})Dudeja, Ma, and
  Maleki]{dudeja2020information}
Rishabh Dudeja, Junjie Ma, and Arian Maleki.
\newblock Information theoretic limits for phase retrieval with subsampled haar
  sensing matrices.
\newblock \emph{IEEE Transactions on Information Theory}, 66\penalty0
  (12):\penalty0 8002--8045, 2020{\natexlab{b}}.

\bibitem[Dudeja et~al.(2022)Dudeja, Lu, and Sen]{dudeja2022universality}
Rishabh Dudeja, Yue~M Lu, and Subhabrata Sen.
\newblock Universality of approximate message passing with semi-random
  matrices.
\newblock \emph{arXiv preprint arXiv:2204.04281}, 2022.

\bibitem[El~Karoui(2018)]{el2018impact}
Noureddine El~Karoui.
\newblock On the impact of predictor geometry on the performance on
  high-dimensional ridge-regularized generalized robust regression estimators.
\newblock \emph{Probability Theory and Related Fields}, 170\penalty0
  (1):\penalty0 95--175, 2018.

\bibitem[El~Karoui et~al.(2013)El~Karoui, Bean, Bickel, Lim, and
  Yu]{el2013robust}
Noureddine El~Karoui, Derek Bean, Peter~J Bickel, Chinghway Lim, and Bin Yu.
\newblock On robust regression with high-dimensional predictors.
\newblock \emph{Proceedings of the National Academy of Sciences}, 110\penalty0
  (36):\penalty0 14557--14562, 2013.

\bibitem[Fan(2022)]{fan2020approximate}
Zhou Fan.
\newblock Approximate message passing algorithms for rotationally invariant
  matrices.
\newblock \emph{The Annals of Statistics}, 50\penalty0 (1):\penalty0 197--224,
  2022.

\bibitem[Farrell(2011)]{farrell2011limiting}
Brendan Farrell.
\newblock Limiting empirical singular value distribution of restrictions of
  discrete {F}ourier transform matrices.
\newblock \emph{Journal of Fourier Analysis and Applications}, 17\penalty0
  (4):\penalty0 733--753, 2011.

\bibitem[Feng et~al.(2022)Feng, Venkataramanan, Rush, and
  Samworth]{feng2022unifying}
Oliver~Y Feng, Ramji Venkataramanan, Cynthia Rush, and Richard~J Samworth.
\newblock A unifying tutorial on approximate message passing.
\newblock \emph{Foundations and Trends{\textregistered} in Machine Learning},
  15\penalty0 (4):\penalty0 335--536, 2022.

\bibitem[Fienup(1982)]{fienup1982phase}
James~R Fienup.
\newblock Phase retrieval algorithms: a comparison.
\newblock \emph{Applied optics}, 21\penalty0 (15):\penalty0 2758--2769, 1982.

\bibitem[Gerace et~al.(2022)Gerace, Krzakala, Loureiro, Stephan, and
  Zdeborov{\'a}]{gerace2022gaussian}
Federica Gerace, Florent Krzakala, Bruno Loureiro, Ludovic Stephan, and Lenka
  Zdeborov{\'a}.
\newblock {G}aussian universality of linear classifiers with random labels in
  high-dimension.
\newblock \emph{arXiv preprint arXiv:2205.13303}, 2022.

\bibitem[Gerbelot and Berthier(2021)]{gerbelot2021graph}
C{\'e}dric Gerbelot and Rapha{\"e}l Berthier.
\newblock Graph-based approximate message passing iterations.
\newblock \emph{arXiv preprint arXiv:2109.11905}, 2021.

\bibitem[Gerbelot et~al.(2020{\natexlab{a}})Gerbelot, Abbara, and
  Krzakala]{gerbelot2020asymptotic}
Cedric Gerbelot, Alia Abbara, and Florent Krzakala.
\newblock Asymptotic errors for teacher-student convex generalized linear
  models (or: How to prove {K}abashima's replica formula).
\newblock \emph{arXiv preprint arXiv:2006.06581}, 2020{\natexlab{a}}.

\bibitem[Gerbelot et~al.(2020{\natexlab{b}})Gerbelot, Abbara, and
  Krzakala]{pmlr-v125-gerbelot20a}
C\'{e}dric Gerbelot, Alia Abbara, and Florent Krzakala.
\newblock Asymptotic errors for high-dimensional convex penalized linear
  regression beyond {G}aussian matrices.
\newblock In Jacob Abernethy and Shivani Agarwal, editors, \emph{Proceedings of
  Thirty Third Conference on Learning Theory}, volume 125 of \emph{Proceedings
  of Machine Learning Research}, pages 1682--1713. PMLR, 09--12 Jul
  2020{\natexlab{b}}.
\newblock URL \url{https://proceedings.mlr.press/v125/gerbelot20a.html}.

\bibitem[Gordon(1985)]{gordon1985some}
Yehoram Gordon.
\newblock Some inequalities for {G}aussian processes and applications.
\newblock \emph{Israel Journal of Mathematics}, 50\penalty0 (4):\penalty0
  265--289, 1985.

\bibitem[Guo and Verd{\'u}(2005)]{guo2005randomly}
Dongning Guo and Sergio Verd{\'u}.
\newblock Randomly spread {CDMA}: Asymptotics via statistical physics.
\newblock \emph{IEEE Transactions on Information Theory}, 51\penalty0
  (6):\penalty0 1983--2010, 2005.

\bibitem[Halko et~al.(2011)Halko, Martinsson, and Tropp]{halko2011finding}
Nathan Halko, Per-Gunnar Martinsson, and Joel~A Tropp.
\newblock Finding structure with randomness: Probabilistic algorithms for
  constructing approximate matrix decompositions.
\newblock \emph{SIAM review}, 53\penalty0 (2):\penalty0 217--288, 2011.

\bibitem[Han and Shen(2022)]{han2022universality}
Qiyang Han and Yandi Shen.
\newblock Universality of regularized regression estimators in high dimensions.
\newblock \emph{arXiv preprint arXiv:2206.07936}, 2022.

\bibitem[He et~al.(2021)He, Jiang, Wen, and Xu]{he2021likelihood}
Yinqiu He, Tiefeng Jiang, Jiyang Wen, and Gongjun Xu.
\newblock Likelihood ratio test in multivariate linear regression: from low to
  high dimension.
\newblock \emph{Statistica Sinica}, 2021.

\bibitem[Hu and Lu(2020{\natexlab{a}})]{hu2020limiting}
Hong Hu and Yue~M Lu.
\newblock The limiting {P}oisson law of massive {MIMO} detection with box
  relaxation.
\newblock \emph{IEEE Journal on Selected Areas in Information Theory},
  1\penalty0 (3):\penalty0 695--704, 2020{\natexlab{a}}.

\bibitem[Hu and Lu(2020{\natexlab{b}})]{hu2020universality}
Hong Hu and Yue~M Lu.
\newblock Universality laws for high-dimensional learning with random features.
\newblock \emph{arXiv preprint arXiv:2009.07669}, 2020{\natexlab{b}}.

\bibitem[Hu and Lu(2022)]{hu2019asymptotics}
Hong Hu and Yue~M. Lu.
\newblock Slope for sparse linear regression: Asymptotics and optimal
  regularization.
\newblock \emph{IEEE Transactions on Information Theory}, pages 1--1, 2022.
\newblock \doi{10.1109/TIT.2022.3188753}.

\bibitem[Javanmard and Montanari(2013)]{javanmard2013state}
Adel Javanmard and Andrea Montanari.
\newblock State evolution for general approximate message passing algorithms,
  with applications to spatial coupling.
\newblock \emph{Information and Inference: A Journal of the IMA}, 2\penalty0
  (2):\penalty0 115--144, 2013.

\bibitem[Jiang et~al.(2012)Jiang, Jiang, and Yang]{jiang2012likelihood}
Dandan Jiang, Tiefeng Jiang, and Fan Yang.
\newblock Likelihood ratio tests for covariance matrices of high-dimensional
  normal distributions.
\newblock \emph{Journal of Statistical Planning and Inference}, 142\penalty0
  (8):\penalty0 2241--2256, 2012.

\bibitem[Jiang and Qi(2015)]{jiang2015likelihood}
Tiefeng Jiang and Yongcheng Qi.
\newblock Likelihood ratio tests for high-dimensional normal distributions.
\newblock \emph{Scandinavian Journal of Statistics}, 42\penalty0 (4):\penalty0
  988--1009, 2015.

\bibitem[Jiang and Yang(2013)]{jiang2013central}
Tiefeng Jiang and Fan Yang.
\newblock Central limit theorems for classical likelihood ratio tests for
  high-dimensional normal distributions.
\newblock \emph{The Annals of Statistics}, 41\penalty0 (4):\penalty0
  2029--2074, 2013.

\bibitem[Karoui(2013)]{karoui2013asymptotic}
Noureddine~El Karoui.
\newblock Asymptotic behavior of unregularized and ridge-regularized
  high-dimensional robust regression estimators: rigorous results.
\newblock \emph{arXiv preprint arXiv:1311.2445}, 2013.

\bibitem[Karoui and K{\"o}sters(2011)]{karoui2011geometric}
Noureddine~El Karoui and Holger K{\"o}sters.
\newblock Geometric sensitivity of random matrix results: consequences for
  shrinkage estimators of covariance and related statistical methods.
\newblock \emph{arXiv preprint arXiv:1105.1404}, 2011.

\bibitem[Knowles and Yin(2017)]{knowles2017anisotropic}
Antti Knowles and Jun Yin.
\newblock Anisotropic local laws for random matrices.
\newblock \emph{Probability Theory and Related Fields}, 169\penalty0
  (1):\penalty0 257--352, 2017.

\bibitem[Korada and Montanari(2011)]{korada2011applications}
Satish~Babu Korada and Andrea Montanari.
\newblock Applications of the {L}indeberg principle in communications and
  statistical learning.
\newblock \emph{IEEE transactions on information theory}, 57\penalty0
  (4):\penalty0 2440--2450, 2011.

\bibitem[Lacotte et~al.(2020)Lacotte, Liu, Dobriban, and
  Pilanci]{lacotte2020optimal}
Jonathan Lacotte, Sifan Liu, Edgar Dobriban, and Mert Pilanci.
\newblock Optimal iterative sketching methods with the subsampled randomized
  {H}adamard transform.
\newblock \emph{Advances in Neural Information Processing Systems},
  33:\penalty0 9725--9735, 2020.

\bibitem[Li and Wei(2021)]{li2021minimum}
Yue Li and Yuting Wei.
\newblock Minimum $\ell_1$-norm interpolators: Precise asymptotics and multiple
  descent.
\newblock \emph{arXiv preprint arXiv:2110.09502}, 2021.

\bibitem[Liang and Sur(2022)]{liang2022precise}
Tengyuan Liang and Pragya Sur.
\newblock A precise high-dimensional asymptotic theory for boosting and
  minimum-$\ell_1$-norm interpolated classifiers.
\newblock \emph{The Annals of Statistics}, 50\penalty0 (3):\penalty0
  1669--1695, 2022.

\bibitem[Lindeberg(1922)]{lindeberg1922neue}
Jarl~Waldemar Lindeberg.
\newblock Eine neue herleitung des exponentialgesetzes in der
  wahrscheinlichkeitsrechnung.
\newblock \emph{Mathematische Zeitschrift}, 15\penalty0 (1):\penalty0 211--225,
  1922.

\bibitem[Liu et~al.(2022)Liu, Huang, and Kurkoski]{liu2022memory}
Lei Liu, Shunqi Huang, and Brian~M Kurkoski.
\newblock Memory {AMP}.
\newblock \emph{IEEE Transactions on Information Theory}, 2022.

\bibitem[Liu and Dobriban(2019)]{liu2019ridge}
Sifan Liu and Edgar Dobriban.
\newblock Ridge regression: Structure, cross-validation, and sketching.
\newblock In \emph{International Conference on Learning Representations}, 2019.

\bibitem[Lu(2021)]{lu2021householder}
Yue~M Lu.
\newblock Householder dice: A matrix-free algorithm for simulating dynamics on
  {G}aussian and random orthogonal ensembles.
\newblock \emph{IEEE Transactions on Information Theory}, 67\penalty0
  (12):\penalty0 8264--8272, 2021.

\bibitem[Lu and Li(2020)]{lu2020phase}
Yue~M Lu and Gen Li.
\newblock Phase transitions of spectral initialization for high-dimensional
  non-convex estimation.
\newblock \emph{Information and Inference: A Journal of the IMA}, 9\penalty0
  (3):\penalty0 507--541, 2020.

\bibitem[Luo et~al.(2019)Luo, Alghamdi, and Lu]{luo2019optimal}
Wangyu Luo, Wael Alghamdi, and Yue~M Lu.
\newblock Optimal spectral initialization for signal recovery with applications
  to phase retrieval.
\newblock \emph{IEEE Transactions on Signal Processing}, 67\penalty0
  (9):\penalty0 2347--2356, 2019.

\bibitem[Ma and Ping(2017)]{ma2017orthogonal}
Junjie Ma and Li~Ping.
\newblock Orthogonal {AMP}.
\newblock \emph{IEEE Access}, 5:\penalty0 2020--2033, 2017.

\bibitem[Ma et~al.(2019)Ma, Xu, and Maleki]{ma2019optimization}
Junjie Ma, Ji~Xu, and Arian Maleki.
\newblock Optimization-based {AMP} for phase retrieval: The impact of
  initialization and $\ell_{2}$ regularization.
\newblock \emph{IEEE Transactions on Information Theory}, 65\penalty0
  (6):\penalty0 3600--3629, 2019.
\newblock \doi{10.1109/TIT.2019.2893254}.

\bibitem[Ma et~al.(2021{\natexlab{a}})Ma, Dudeja, Xu, Maleki, and
  Wang]{ma2021spectral}
Junjie Ma, Rishabh Dudeja, Ji~Xu, Arian Maleki, and Xiaodong Wang.
\newblock Spectral method for phase retrieval: an expectation propagation
  perspective.
\newblock \emph{IEEE Transactions on Information Theory}, 67\penalty0
  (2):\penalty0 1332--1355, 2021{\natexlab{a}}.

\bibitem[Ma et~al.(2021{\natexlab{b}})Ma, Xu, and Maleki]{ma2021analysis}
Junjie Ma, Ji~Xu, and Arian Maleki.
\newblock Analysis of sensing spectrum for signal recovery under a generalized
  linear model.
\newblock \emph{Advances in Neural Information Processing Systems},
  34:\penalty0 22601--22613, 2021{\natexlab{b}}.

\bibitem[Maillard et~al.(2020)Maillard, Loureiro, Krzakala, and
  Zdeborov{\'a}]{maillard2020phase}
Antoine Maillard, Bruno Loureiro, Florent Krzakala, and Lenka Zdeborov{\'a}.
\newblock Phase retrieval in high dimensions: Statistical and computational
  phase transitions.
\newblock \emph{Advances in Neural Information Processing Systems},
  33:\penalty0 11071--11082, 2020.

\bibitem[Maillard et~al.(2022)Maillard, Krzakala, Lu, and
  Zdeborov{\'a}]{maillard2022construction}
Antoine Maillard, Florent Krzakala, Yue~M Lu, and Lenka Zdeborov{\'a}.
\newblock Construction of optimal spectral methods in phase retrieval.
\newblock In \emph{Mathematical and Scientific Machine Learning}, pages
  693--720. PMLR, 2022.

\bibitem[Mar{\v{c}}enko and Pastur(1967)]{marvcenko1967distribution}
Vladimir~A Mar{\v{c}}enko and Leonid~Andreevich Pastur.
\newblock Distribution of eigenvalues for some sets of random matrices.
\newblock \emph{Mathematics of the USSR-Sbornik}, 1\penalty0 (4):\penalty0 457,
  1967.

\bibitem[Marinari et~al.(1994)Marinari, Parisi, and
  Ritort]{marinari1994replica}
Enzo Marinari, Giorgio Parisi, and Felix Ritort.
\newblock Replica field theory for deterministic models {II}. a non-random spin
  glass with glassy behaviour.
\newblock \emph{Journal of Physics A: Mathematical and General}, 27\penalty0
  (23):\penalty0 7647, 1994.

\bibitem[Mei and Montanari(2022)]{mei2022generalization}
Song Mei and Andrea Montanari.
\newblock The generalization error of random features regression: Precise
  asymptotics and the double descent curve.
\newblock \emph{Communications on Pure and Applied Mathematics}, 75\penalty0
  (4):\penalty0 667--766, 2022.

\bibitem[Mignacco et~al.(2020)Mignacco, Krzakala, Lu, Urbani, and
  Zdeborova]{mignacco2020role}
Francesca Mignacco, Florent Krzakala, Yue Lu, Pierfrancesco Urbani, and Lenka
  Zdeborova.
\newblock The role of regularization in classification of high-dimensional
  noisy {G}aussian mixture.
\newblock In \emph{International Conference on Machine Learning}, pages
  6874--6883. PMLR, 2020.

\bibitem[Monajemi et~al.(2013)Monajemi, Jafarpour, Gavish, Collaboration,
  Donoho, Ambikasaran, Bacallado, Bharadia, Chen, Choi,
  et~al.]{monajemi2013deterministic}
Hatef Monajemi, Sina Jafarpour, Matan Gavish, Stat 330/CME~362 Collaboration,
  David~L Donoho, Sivaram Ambikasaran, Sergio Bacallado, Dinesh Bharadia, Yuxin
  Chen, Young Choi, et~al.
\newblock Deterministic matrices matching the compressed sensing phase
  transitions of {G}aussian random matrices.
\newblock \emph{Proceedings of the National Academy of Sciences}, 110\penalty0
  (4):\penalty0 1181--1186, 2013.

\bibitem[Mondelli and Venkataramanan(2021)]{mondelli2021pca}
Marco Mondelli and Ramji Venkataramanan.
\newblock {PCA} initialization for approximate message passing in rotationally
  invariant models.
\newblock \emph{Advances in Neural Information Processing Systems}, 34, 2021.

\bibitem[Mondelli et~al.(2021)Mondelli, Thrampoulidis, and
  Venkataramanan]{mondelli2021optimal}
Marco Mondelli, Christos Thrampoulidis, and Ramji Venkataramanan.
\newblock Optimal combination of linear and spectral estimators for generalized
  linear models.
\newblock \emph{Foundations of Computational Mathematics}, pages 1--54, 2021.

\bibitem[Montanari and Nguyen(2017)]{montanari2017universality}
Andrea Montanari and Phan-Minh Nguyen.
\newblock Universality of the elastic net error.
\newblock In \emph{2017 IEEE International Symposium on Information Theory
  (ISIT)}, pages 2338--2342. IEEE, 2017.

\bibitem[Montanari and Saeed(2022)]{montanari2022universality}
Andrea Montanari and Basil~N. Saeed.
\newblock Universality of empirical risk minimization.
\newblock In Po-Ling Loh and Maxim Raginsky, editors, \emph{Proceedings of
  Thirty Fifth Conference on Learning Theory}, volume 178 of \emph{Proceedings
  of Machine Learning Research}, pages 4310--4312. PMLR, 02--05 Jul 2022.
\newblock URL \url{https://proceedings.mlr.press/v178/montanari22a.html}.

\bibitem[Nguyen et~al.(2009)Nguyen, Do, and Tran]{nguyen2009fast}
Nam~H Nguyen, Thong~T Do, and Trac~D Tran.
\newblock A fast and efficient algorithm for low-rank approximation of a
  matrix.
\newblock In \emph{Proceedings of the forty-first annual ACM symposium on
  Theory of computing}, pages 215--224, 2009.

\bibitem[O'Donnell(2014)]{o2014analysis}
Ryan O'Donnell.
\newblock \emph{Analysis of {B}oolean functions}.
\newblock Cambridge University Press, 2014.

\bibitem[Oymak and Hassibi(2014)]{oymak2014case}
Samet Oymak and Babak Hassibi.
\newblock A case for orthogonal measurements in linear inverse problems.
\newblock In \emph{2014 IEEE International Symposium on Information Theory},
  pages 3175--3179. IEEE, 2014.

\bibitem[Oymak and Tropp(2018)]{oymak2018universality}
Samet Oymak and Joel~A Tropp.
\newblock Universality laws for randomized dimension reduction, with
  applications.
\newblock \emph{Information and Inference: A Journal of the IMA}, 7\penalty0
  (3):\penalty0 337--446, 2018.

\bibitem[Panahi and Hassibi(2017)]{panahi2017universal}
Ashkan Panahi and Babak Hassibi.
\newblock A universal analysis of large-scale regularized least squares
  solutions.
\newblock \emph{Advances in Neural Information Processing Systems}, 30, 2017.

\bibitem[Parisi and Potters(1995)]{parisi1995mean}
Giorgio Parisi and Marc Potters.
\newblock Mean-field equations for spin models with orthogonal interaction
  matrices.
\newblock \emph{Journal of Physics A: Mathematical and General}, 28\penalty0
  (18):\penalty0 5267, 1995.

\bibitem[Rangan and Goyal(2001)]{rangan2001recursive}
Sundeep Rangan and Vivek~K Goyal.
\newblock Recursive consistent estimation with bounded noise.
\newblock \emph{IEEE Transactions on Information Theory}, 47\penalty0
  (1):\penalty0 457--464, 2001.

\bibitem[Rangan et~al.(2009)Rangan, Goyal, and Fletcher]{rangan2009asymptotic}
Sundeep Rangan, Vivek Goyal, and Alyson~K Fletcher.
\newblock Asymptotic analysis of {MAP} estimation via the replica method and
  compressed sensing.
\newblock \emph{Advances in Neural Information Processing Systems}, 22, 2009.

\bibitem[Rangan et~al.(2019)Rangan, Schniter, and Fletcher]{rangan2019vector}
Sundeep Rangan, Philip Schniter, and Alyson~K Fletcher.
\newblock Vector approximate message passing.
\newblock \emph{IEEE Transactions on Information Theory}, 65\penalty0
  (10):\penalty0 6664--6684, 2019.

\bibitem[Reeves and Pfister(2019)]{reeves2019replica}
Galen Reeves and Henry~D Pfister.
\newblock The replica-symmetric prediction for random linear estimation with
  {G}aussian matrices is exact.
\newblock \emph{IEEE Transactions on Information Theory}, 65\penalty0
  (4):\penalty0 2252--2283, 2019.

\bibitem[Rudelson and Vershynin(2008)]{rudelson2008sparse}
Mark Rudelson and Roman Vershynin.
\newblock On sparse reconstruction from {F}ourier and {G}aussian measurements.
\newblock \emph{Communications on Pure and Applied Mathematics: A Journal
  Issued by the Courant Institute of Mathematical Sciences}, 61\penalty0
  (8):\penalty0 1025--1045, 2008.

\bibitem[Schmitt(1992)]{schmitt1992perturbation}
Bernhard~A Schmitt.
\newblock Perturbation bounds for matrix square roots and {P}ythagorean sums.
\newblock \emph{Linear algebra and its applications}, 174:\penalty0 215--227,
  1992.

\bibitem[Stojnic(2013)]{stojnic2013framework}
Mihailo Stojnic.
\newblock A framework to characterize performance of lasso algorithms.
\newblock \emph{arXiv preprint arXiv:1303.7291}, 2013.

\bibitem[Strohmer and Heath~Jr(2003)]{strohmer2003grassmannian}
Thomas Strohmer and Robert~W Heath~Jr.
\newblock Grassmannian frames with applications to coding and communication.
\newblock \emph{Applied and computational harmonic analysis}, 14\penalty0
  (3):\penalty0 257--275, 2003.

\bibitem[Sur and Cand{\`e}s(2019)]{sur2019modern}
Pragya Sur and Emmanuel~J Cand{\`e}s.
\newblock A modern maximum-likelihood theory for high-dimensional logistic
  regression.
\newblock \emph{Proceedings of the National Academy of Sciences}, 116\penalty0
  (29):\penalty0 14516--14525, 2019.

\bibitem[Sur et~al.(2019)Sur, Chen, and Cand{\`e}s]{sur2019likelihood}
Pragya Sur, Yuxin Chen, and Emmanuel~J Cand{\`e}s.
\newblock The likelihood ratio test in high-dimensional logistic regression is
  asymptotically a rescaled chi-square.
\newblock \emph{Probability theory and related fields}, 175\penalty0
  (1):\penalty0 487--558, 2019.

\bibitem[Taheri et~al.(2021)Taheri, Pedarsani, and
  Thrampoulidis]{taheri2021fundamental}
Hossein Taheri, Ramtin Pedarsani, and Christos Thrampoulidis.
\newblock Fundamental limits of ridge-regularized empirical risk minimization
  in high dimensions.
\newblock In \emph{International Conference on Artificial Intelligence and
  Statistics}, pages 2773--2781. PMLR, 2021.

\bibitem[Takeda et~al.(2006)Takeda, Uda, and Kabashima]{takeda2006analysis}
Koujin Takeda, Shinsuke Uda, and Yoshiyuki Kabashima.
\newblock Analysis of {CDMA} systems that are characterized by eigenvalue
  spectrum.
\newblock \emph{EPL (Europhysics Letters)}, 76\penalty0 (6):\penalty0 1193,
  2006.

\bibitem[Takeuchi(2017)]{takeuchi2017rigorous}
Keigo Takeuchi.
\newblock Rigorous dynamics of expectation-propagation-based signal recovery
  from unitarily invariant measurements.
\newblock In \emph{2017 IEEE International Symposium on Information Theory
  (ISIT)}, pages 501--505. IEEE, 2017.

\bibitem[Takeuchi(2020)]{takeuchi2020convolutional}
Keigo Takeuchi.
\newblock Convolutional approximate message-passing.
\newblock \emph{IEEE Signal Processing Letters}, 27:\penalty0 416--420, 2020.

\bibitem[Takeuchi(2021)]{takeuchi2021bayes}
Keigo Takeuchi.
\newblock Bayes-optimal convolutional {AMP}.
\newblock In \emph{2021 IEEE International Symposium on Information Theory
  (ISIT)}, pages 1385--1390. IEEE, 2021.

\bibitem[Tao and Vu(2014)]{tao2014random}
Terence Tao and Van Vu.
\newblock Random matrices: the universality phenomenon for {W}igner ensembles.
\newblock \emph{Modern aspects of random matrix theory}, 72:\penalty0 121--172,
  2014.

\bibitem[Thrampoulidis et~al.(2015)Thrampoulidis, Oymak, and
  Hassibi]{thrampoulidis2015regularized}
Christos Thrampoulidis, Samet Oymak, and Babak Hassibi.
\newblock Regularized linear regression: A precise analysis of the estimation
  error.
\newblock In \emph{Conference on Learning Theory}, pages 1683--1709. PMLR,
  2015.

\bibitem[Thrampoulidis et~al.(2018)Thrampoulidis, Xu, and
  Hassibi]{thrampoulidis2018symbol}
Christos Thrampoulidis, Weiyu Xu, and Babak Hassibi.
\newblock Symbol error rate performance of box-relaxation decoders in massive
  {MIMO}.
\newblock \emph{IEEE Transactions on Signal Processing}, 66\penalty0
  (13):\penalty0 3377--3392, 2018.

\bibitem[Tulino et~al.(2010)Tulino, Caire, Shamai, and
  Verd{\'u}]{tulino2010capacity}
Antonia~M Tulino, Giuseppe Caire, Shlomo Shamai, and Sergio Verd{\'u}.
\newblock Capacity of channels with frequency-selective and time-selective
  fading.
\newblock \emph{IEEE Transactions on Information Theory}, 56\penalty0
  (3):\penalty0 1187--1215, 2010.

\bibitem[Unser and Eden(1988)]{unser1988maximum}
Michael Unser and Murray Eden.
\newblock Maximum likelihood estimation of linear signal parameters for
  {P}oisson processes.
\newblock \emph{IEEE Transactions on Acoustics, Speech, and Signal Processing},
  36\penalty0 (6):\penalty0 942--945, 1988.

\bibitem[Venkataramanan et~al.(2022)Venkataramanan, K{\"o}gler, and
  Mondelli]{venkataramanan2022estimation}
Ramji Venkataramanan, Kevin K{\"o}gler, and Marco Mondelli.
\newblock Estimation in rotationally invariant generalized linear models via
  approximate message passing.
\newblock In \emph{International Conference on Machine Learning}, pages
  22120--22144. PMLR, 2022.

\bibitem[Vershynin(2018)]{vershynin2018high}
Roman Vershynin.
\newblock \emph{High-dimensional probability: An introduction with applications
  in data science}, volume~47.
\newblock Cambridge university press, 2018.

\bibitem[Voiculescu(1991)]{voiculescu1991limit}
Dan Voiculescu.
\newblock Limit laws for random matrices and free products.
\newblock \emph{Inventiones mathematicae}, 104\penalty0 (1):\penalty0 201--220,
  1991.

\bibitem[Voiculescu et~al.(1992)Voiculescu, Dykema, and
  Nica]{voiculescu1992free}
Dan~V Voiculescu, KJ~Dykema, and Alexandru Nica.
\newblock Free random variables. a noncommutative probability approach to free
  products with applications to random matrices, operator algebras and harmonic
  analysis on free groups. crm monograph series, 1.
\newblock \emph{American Mathematical Society, Providence, RI}, 23, 1992.

\bibitem[Wang et~al.(2020)Wang, Weng, and Maleki]{wang2020bridge}
Shuaiwen Wang, Haolei Weng, and Arian Maleki.
\newblock Which bridge estimator is the best for variable selection?
\newblock \emph{The Annals of Statistics}, 48\penalty0 (5):\penalty0
  2791--2823, 2020.

\bibitem[Wang et~al.(2022)Wang, Zhong, and Fan]{wang2022universality}
Tianhao Wang, Xinyi Zhong, and Zhou Fan.
\newblock Universality of approximate message passing algorithms and tensor
  networks.
\newblock \emph{arXiv preprint arXiv:2206.13037}, 2022.

\bibitem[Wendel(1962)]{wendel1962problem}
James~G Wendel.
\newblock A problem in geometric probability.
\newblock \emph{Mathematica Scandinavica}, 11\penalty0 (1):\penalty0 109--111,
  1962.

\bibitem[Weng et~al.(2018)Weng, Maleki, and Zheng]{weng2018overcoming}
Haolei Weng, Arian Maleki, and Le~Zheng.
\newblock Overcoming the limitations of phase transition by higher order
  analysis of regularization techniques.
\newblock \emph{The Annals of Statistics}, 46\penalty0 (6A):\penalty0
  3099--3129, 2018.

\bibitem[Winder(1966)]{winder1966partitions}
Robert~O Winder.
\newblock Partitions of {N}-space by hyperplanes.
\newblock \emph{SIAM Journal on Applied Mathematics}, 14\penalty0 (4):\penalty0
  811--818, 1966.

\bibitem[Yang et~al.(2011)Yang, Lu, Sbaiz, and Vetterli]{yang2011bits}
Feng Yang, Yue~M Lu, Luciano Sbaiz, and Martin Vetterli.
\newblock Bits from photons: Oversampled image acquisition using binary
  {P}oisson statistics.
\newblock \emph{IEEE Transactions on image processing}, 21\penalty0
  (4):\penalty0 1421--1436, 2011.

\bibitem[Yin(1986)]{yin1986limiting}
Yong~Q Yin.
\newblock Limiting spectral distribution for a class of random matrices.
\newblock \emph{Journal of multivariate analysis}, 20\penalty0 (1):\penalty0
  50--68, 1986.

\bibitem[Zhong et~al.(2022)Zhong, Su, and Fan]{ebpca}
Xinyi Zhong, Chang Su, and Zhou Fan.
\newblock Empirical bayes {PCA} in high dimensions.
\newblock \emph{Journal of the Royal Statistical Society: Series B (Statistical
  Methodology)}, 84\penalty0 (3):\penalty0 853--878, 2022.
\newblock \doi{https://doi.org/10.1111/rssb.12490}.
\newblock URL
  \url{https://rss.onlinelibrary.wiley.com/doi/abs/10.1111/rssb.12490}.

\end{thebibliography}
\appendix
\section{Proof of \thref{GFOM}}\label{appendix:VAMP-to-GFOM}
In this appendix, we prove the universality of GFOMs (\thref{GFOM}) using the universality principle for VAMP algorithms (\thref{VAMP}). 

\begin{proof}[Proof of \thref{GFOM}] Consider a $T$-iteration GFOM driven by a strongly semi-random ensemble $\mM_{1:T} = \mS \mPsi_{1:T} \mS$ with empirical moments $\{\hat{\Omega}_{B,B^\prime}: B,B^\prime \subset [T]\}$ and limiting moments $\{{\Omega}_{B,B^\prime}: B,B^\prime \subset [T]\}$:
\begin{align}\label{eq:GFOM-recall}
    \iter{\vz}{t} = \mM_t \cdot \nonlin_t(\iter{\vz}{1}, \iter{\vz}{2}, \dotsc, \iter{\vz}{t-1}; \auxmat) + \eta_t(\iter{\vz}{1}, \iter{\vz}{2}, \dotsc, \iter{\vz}{t-1}; \auxmat) \quad \forall \; t \; \in \; [T].
\end{align}
We begin by making the following claim.
\paragraph{Claim.} \emph{For any $T$-iteration GFOM \eqref{eq:GFOM-recall}, there exists a strictly increasing function $\tau: [T] \cup \{0\} \mapsto \N$ with $\tau(0) = 0$, a $\tau(T)$-iteration VAMP algorithm of the form:
\begin{align}\label{eq:VAMP-GFOM}
    \iter{\vw}{i} & = \mQ_i \cdot  g_i( \iter{\vw}{1}, \dotsc, \iter{\vw}{i-1}; \auxmat), \quad i \in [\tau(T)]
\end{align}
and post-processing functions $H_{1:T}$ such that:
\begin{enumerate}
    \item The matrices $\mQ_1, \mQ_{2}, \dotsc, \mQ_{\tau(T)-1}, \mQ_{\tau(T)}$ form a semi-random ensemble (\defref{semirandom}) and have the property that for any $t \in [T]$ and any $\tau(t-1)<i \leq \tau(t)$, $\mQ_i$ is of the form: 
    \begin{subequations} \label{eq:Q-form}
    \begin{align}
        \mXi_i & = \sum_{B \subset [t]} \hat{\alpha}_{i,B} \cdot   (\mPsi_B - \hat{\Omega}_{B,\emptyset} \cdot  \mI_{\dim}) , \\
        \mQ_i & = \mS \cdot \mXi_i \cdot \mS,
    \end{align}
    \end{subequations}
    for some coefficients $\{\hat{\alpha}_{i,B}\}$ that are determined by the empirical moments $\{\hat{\Omega}_{B,B^\prime}: B, B^\prime \subset [T]\}$ of the strongly semi-random ensemble $\mM_{1:T}$. Furthermore, the coefficients $\{\hat{\alpha}_{i,B}\}$ converge to limiting values $\{{\alpha}_{i,B}\}$ as $\dim \rightarrow \infty$. The limiting coefficients $\{{\alpha}_{i,B}\}$ and the limiting covariance matrix associated with the semi-random ensemble  $\mQ_1, \mQ_{2}, \dotsc, \mQ_{\tau(T)-1}, \mQ_{\tau(T)}$ are determined by the limiting moments $\{\Omega_{B,B^\prime} : B, B^\prime \subset [t]\}$ of the strongly semi-random ensemble $\mM_{1:T}$. 
    \item The non-linearities $g_{1}, g_2 \dotsc, g_{\tau(T)}$ are continuous and for each $i \in [\tau(T)]$, $g_i : \R^{i-1+\auxdim} \mapsto \R$ is uniformly Lipschitz and polynomially bounded in the sense that there are finite constants $L^\prime \in (0,\infty)$ and $\degree^\prime \in \N$ such that:
    \begin{align*}
        |g_i(z; \auxvec) - g_i(z^\prime; \auxvec)| & \leq L^\prime \cdot  \|z - z^\prime\| \quad \forall \; \auxvec \in \R^{\auxdim}, \; z, z^\prime \in \R^{i-1}, \\
        |g_i(z;\auxvec)| & \leq L^\prime \cdot (1 + \|z\|^{\degree^\prime} + \|\auxvec\|^{\degree^\prime}) \quad \forall \; \auxvec \in \R^{\auxdim}, \; z \in \R^{i-1}.
    \end{align*}
    The non-linearities $g_{1:\tau(T)}$ are divergence-free (\assumpref{div-free}) with respect to $\serv{W}_1,  \dotsc, \serv{W}_{\tau(T)}$, the Gaussian state evolution random variables associated with the VAMP algorithm \eqref{eq:VAMP-GFOM}. Furthermore, they are determined completely by $f_{1:T}$ and the limiting moments $\{\Omega_{B,B^\prime} : B, B^\prime \subset [t]\}$.
    \item The post-processing functions $H_{1:T}$ are continuous and for each $t \in [T]$ $H_t: \R^{\tau(t) + \auxdim} \mapsto \R$ satisfies:
    \begin{subequations}\label{eq:post-processing-regularity}
     \begin{align} 
        |H_t(w; \auxvec) - H_t(w^\prime; \auxvec)| & \leq L^\prime \cdot  \|w - w^\prime\| \quad \forall \; \auxvec \in \R^{\auxdim}, \; z, z^\prime \in \R^{\tau(t)}, \\
        |H_t(w; \auxvec)| & \leq L^\prime \cdot (1 + \|w\| + \|\auxvec\|^{\degree^\prime}) \quad \forall\;  \auxvec \in \R^{\auxdim}, \; w \in \R^{\tau(t)},
    \end{align}
    \end{subequations}
    for some finite constants $L^\prime \in (0,\infty)$ and $\degree^\prime \in \N$. Furthermore, the post-processing functions are determined completely by $f_{1:T}$ and the limiting moments $\{\Omega_{B,B^\prime} : B, B^\prime \subset [t]\}$.
    \item The VAMP algorithm implements the GFOM \eqref{eq:GFOM-recall} in the sense that for any $t \in [T]$, the error:
    \begin{align*}
        \iter{\vDelta}{t} \explain{def}{=} \iter{\vz}{t} - H_t(\iter{\vw}{1}, \dotsc, \iter{\vw}{\tau(t)}; \auxmat)
    \end{align*}
    satisfies $\|\iter{\vDelta}{t}\|^2/\dim \explain{P}{\rightarrow} 0$. In the above display, the post-processing map $H_t$ acts entry-wise on its arguments. 
\end{enumerate}}
\emph{Proof of \thref{GFOM}.} Assuming the above claim, we now show that \thref{GFOM} follows from \thref{VAMP}. Since the VAMP algorithm constructed in \eqref{eq:VAMP-GFOM} satisfies all the requirements of \thref{VAMP}, we have $$(\iter{\vw}{1}, \dotsc, \iter{\vw}{\tau(T)}; \auxmat) \explain{\pw}{\longrightarrow} (\serv{W}_1, \serv{W}_2, \dotsc, \serv{W}_{\tau(T)}; \serv{A}).$$ For any test function $h:\R^{T + \auxdim} \mapsto \R$ which satisfies the regularity assumptions in the definition {\pw} convergence (\defref{PW2}), as a consequence of \eqref{eq:post-processing-regularity} the composite test function $\tilde{h} : \R^{\tau(T) + \auxdim} \mapsto \R$:
\begin{align*}
    \tilde{h}(w_1, \dotsc, w_{\tau(T)}; \auxvec) \explain{def}{=} h(H_1(w_1, w_{\tau(1)}; \auxvec), H_2(w_1, \dotsc, w_{\tau(2)}; \auxvec), \dotsc, H_T(w_1, \dotsc, w_{\tau(T)}; \auxvec) ; \auxvec)
\end{align*}
also satisfies the regularity assumptions of \defref{PW2}. As a consequence, the post-processed iterates
\begin{align*}
    \iter{\tilde{\vz}}{t} &\explain{def}{=} H_t(\iter{\vw}{1}, \dotsc, \iter{\vw}{\tau(t)} ; \auxmat) \quad \forall \; t \; \in \; [T]
\end{align*}
satisfy:
\begin{align} \label{eq:post-process-pw2}
    (\iter{\tilde{\vz}}{1}, \dotsc, \iter{\tilde{\vz}}{T}; \auxmat) &\explain{\pw}{\longrightarrow} (\underbrace{H_1(\serv{W}_1, \dotsc, \serv{W}_{\tau(1)}; \serv{A}), H_2(\serv{W}_2, \dotsc, \serv{W}_{\tau(2)}; \serv{A}), \dotsc, H_T(\serv{W}_1, \dotsc, \serv{W}_{\tau(T)}; \serv{A})}_{\explain{def}{=} \serv{Z}_1, \serv{Z}_2, \dotsc, \serv{Z}_T}; \serv{A}).
\end{align}
Furthermore,
\begin{align*}
   &\left| \frac{1}{\dim} \sum_{i=1}^\dim h(\iter{z}{1}_i, \dotsc, \iter{z}{T}_i; \auxvec_i) -   \frac{1}{\dim} \sum_{i=1}^\dim h(\iter{\tilde{z}}{1}_i, \dotsc, \iter{\tilde{z}}{T}_i; \auxvec_i) \right| \\
   &\hspace{3cm}\leq 2L \cdot \left(1 + \frac{1}{\dim}\sum_{t=1}^T (\|\iter{\vz}{t}\|^2 + \|\iter{\tilde{\vz}}{t}\|^2) + \frac{1}{\dim} \sum_{i=1}^\dim \|\auxvec_i\|^{2D} \right)^{1/2} \cdot \left( \frac{1}{\dim} \sum_{t=1}^T \|\iter{\vDelta}{t}\|^2 \right)^{1/2} \\
   &\hspace{3cm} \leq 2L \cdot \left(1 + \frac{1}{\dim}\sum_{t=1}^T (2\|\iter{\vDelta}{t}\|^2 + 3\|\iter{\tilde{\vz}}{t}\|^2) + \frac{1}{\dim} \sum_{i=1}^\dim \|\auxvec_i\|^{2D}\right)^{1/2} \cdot \left( \frac{1}{\dim} \sum_{t=1}^T \|\iter{\vDelta}{t}\|^2 \right)^{1/2}.
\end{align*}
By \eqref{eq:post-process-pw2} $\|\iter{\tilde{\vz}}{t}\|^2/\dim \explain{P}{\rightarrow} \E \serv{\tilde{Z}}_t^2 < \infty$. \assumpref{side-info} guarantees that $\sum_i \|\auxvec_i\|^{2D}/\dim \explain{P}{\rightarrow} \E \|\serv{A}\|^{2D}< \infty$. Since $\|\iter{\vDelta}{t}\|^2/\dim \explain{P}{\rightarrow} 0$, we obtain:
\begin{align*}
    \frac{1}{\dim} \sum_{i=1}^\dim h(\iter{z}{1}_i, \dotsc, \iter{z}{T}_i; \auxvec_i) -   \frac{1}{\dim} \sum_{i=1}^\dim h(\iter{\tilde{z}}{1}_i, \dotsc, \iter{\tilde{z}}{T}_i; \auxvec_i) \explain{P}{\rightarrow 0},
\end{align*}
which when combined with \eqref{eq:post-process-pw2} yields:
\begin{align*}
     (\iter{{\vz}}{1}, \dotsc, \iter{{\vz}}{T}; \auxmat) &\explain{\pw}{\longrightarrow} (\underbrace{H_1(\serv{W}_1, \dotsc, \serv{W}_{\tau(1)}; \serv{A}), H_2(\serv{W}_2, \dotsc, \serv{W}_{\tau(2)}; \serv{A}), \dotsc, H_T(\serv{W}_1, \dotsc, \serv{W}_{\tau(T)}; \serv{A})}_{\explain{def}{=} \serv{Z}_1, \serv{Z}_2, \dotsc, \serv{Z}_T}; \serv{A}).
\end{align*}
Observe that the distribution of the random variables on the right hand side is completely determined by the limiting covariance of the semi-random ensemble $\mQ_{1:\tau(T)}$, the non-linearities $g_{1:\tau(T)}$ used in the VAMP algorithm, and the post-processing functions $H_{1:T}$. The claim stated above guarantees that each of these are determined by the non-linearities $f_{1:T}$ of the GFOM and the limiting moments of the strongly semi-random matrix ensemble $\mM_{1:T}$ driving it. This proves the assertion made by \thref{GFOM}. 

\emph{Proof of Claim.} We now prove the claim made above. We will construct the VAMP algorithm in \eqref{eq:VAMP-GFOM} inductively. Consider the first iteration of the GFOM:
\begin{align} 
    \iter{\vz}{1} & = \mS \mPsi_1 \mS \nonlin_1(\auxmat) + \eta_1(\auxmat) \nonumber \\&= \mS \cdot (\mPsi_1 - \hat{\Omega}_{\{1\},\emptyset} \mI_{\dim}) \cdot \mS \nonlin_1(\auxmat) +  \Omega_{\{1\},\emptyset} \cdot \nonlin_1(\auxmat) + \eta_1(\auxmat) + (\hat{\Omega}_{\{1\},\emptyset}-{\Omega}_{\{1\},\emptyset}) \cdot \nonlin_1(\auxmat)\label{eq:GFOM-iter-1}
\end{align}
We set $\tau(1) = 1$ and construct the first iteration of the VAMP algorithm as:
\begin{align*}
    \iter{\vw}{1} & = \mQ_1 g_1(\auxmat)
\end{align*}
where:
\begin{enumerate}
    \item $\mXi_1 = \mPsi_1 - \hat{\Omega}_{\{1\},\emptyset} \mI_{\dim}$ and  $\mQ_1 = \mS \cdot \mXi_1 \cdot \mS$. Observe that $\mQ_1$ is of the form \eqref{eq:Q-form}. Furthermore, as verified in the proof sketch of \thref{GFOM} given in \sref{GFOM-univ-proof}, $\mQ_1$ is a semi-random in the sense of \defref{semirandom}. 
    \item The non-linearity $g_1:\R^{\auxdim} \mapsto \R$ is given by $g_1(\auxvec) \explain{def}{=} \nonlin_1(\auxvec)$. 
    \item The post-processing function $H_1: \R^{1+\auxdim} \mapsto \R$ is given by:
    \begin{align*}
        H_1(w_1; \auxvec) \explain{def}{=} w_1 +  \Omega_{\{1\},\emptyset} \cdot \nonlin_1(\auxvec) + \eta_1(\auxvec).
    \end{align*}
    \item Recalling \eqref{eq:GFOM-iter-1}, the above definitions ensure that $\iter{\vz}{1} = H_1(\iter{\vw}{1}; \auxmat) + \iter{\vDelta}{1}$ where $$\iter{\vDelta}{1} \explain{def}{=} (\hat{\Omega}_{\{1\},\emptyset}-{\Omega}_{\{1\},\emptyset}) \cdot \nonlin_1(\auxmat).$$ Since $\hat{\Omega}_{\{1\},\emptyset}-{\Omega}_{\{1\},\emptyset} \rightarrow 0$ (\defref{strong-semirandom}) and $\|\nonlin_1(\auxmat)\|^2/\dim \explain{P}{\rightarrow} \E\nonlin_1^2(\serv{A})$ (\assumpref{side-info}), we have $\|\iter{\vDelta}{1}\|^2/\dim \explain{P}{\rightarrow 0}$.
\end{enumerate}
Hence, we have constructed the first iteration of the desired VAMP algorithm which implements one iteration of the given GFOM. As the induction hypothesis, we assume that for some $t<T$, we have constructed $\tau(t)$ iterations of the VAMP algorithm:
\begin{align*}
   \iter{\vw}{i} & = \mQ_i \cdot  g_i( \iter{\vw}{1}, \dotsc, \iter{\vw}{i-1}; \auxmat) \quad i \in [\tau(t)],
\end{align*}
along with the corresponding post-processing functions $H_{1:t}$ which implement $t$ iterations of the given GFOM and satisfy the assertions made in the claim. Let: $$(\serv{W}_1, \dotsc, \serv{W}_{\tau(t)}) \sim \gauss{0}{\Sigma_{\tau(t)}}$$ denote the Gaussian state evolution random variables associated with $\tau(t)$ iterations of the constructed VAMP algorithm. We now implement iteration $t+1$ of the GFOM:
\begin{align*}
    \iter{\vz}{t+1} & = \mS \mPsi_{t+1} \mS \cdot \nonlin_{t+1}(\iter{\vz}{1}, \dotsc, \iter{\vz}{t}; \auxmat) + \eta_{t+1}(\iter{\vz}{1}, \dotsc, \iter{\vz}{t}; \auxmat) \\
    & \explain{(a)}{=} \mS \mPsi_{t+1} \mS \cdot \nonlin_{t+1}(H_1(\iter{\vw}{1}; \auxmat), \dotsc, H_t(\iter{\vw}{1}, \dotsc, \iter{\vw}{\tau(t)}; \auxmat); \auxmat) \\& \hspace{5cm}+ \eta_{t+1}(H_1(\iter{\vw}{1}; \auxmat), \dotsc, H_t(\iter{\vw}{1}, \dotsc, \iter{\vw}{\tau(t)}; \auxmat); \auxmat) + \iter{\vdelta}{t+1} \\
    & \explain{(b)}{=} \mS \mPsi_{t+1} \mS \cdot \tilde{\nonlin}_{t+1}(\iter{\vw}{1}, \dotsc, \iter{\vw}{\tau(t)}; \auxmat) + \tilde{\eta}_{t+1}(\iter{\vw}{1}, \dotsc, \iter{\vw}{\tau(t)}; \auxmat) + \iter{\vdelta}{t+1}.
\end{align*}
In the above display, in the step marked (a) we defined the error vector:
\begin{align}\label{eq:delta-error}
    \iter{\vdelta}{t+1} &\explain{def}{=} \mM_{t+1} \cdot  (\nonlin_{t+1}(\iter{\vz}{1}, \dotsc, \iter{\vz}{t}; \auxmat) - \nonlin_{t+1}(H_1(\iter{\vw}{1}; \auxmat), \dotsc, H_t(\iter{\vw}{1}, \dotsc, \iter{\vw}{\tau(t)}; \auxmat); \auxmat)) \\ &\hspace{2cm} + \eta_{t+1}(\iter{\vz}{1}, \dotsc, \iter{\vz}{t}; \auxmat) - \eta_{t+1}(H_1(\iter{\vw}{1}; \auxmat), \dotsc, H_t(\iter{\vw}{1}, \dotsc, \iter{\vw}{\tau(t)}; \auxmat); \auxmat). 
\end{align}
In the step marked (b), we defined the composite non-linearities $\tilde{\nonlin}_{t+1}, \tilde{\eta}_{t+1} : \R^{\tau(t) + \auxdim} \mapsto \R$ as:
\begin{align*}
    \tilde{\nonlin}_{t+1}(w_1, \dotsc, w_{\tau(t)}; \auxvec) \explain{def}{=} \nonlin_{t+1}(H_1(w_1; \auxvec), \dotsc, H_t(w_1, \dotsc, w_{\tau(t)}; \auxvec); \auxvec), \\
    \tilde{\eta}_{t+1}(w_1, \dotsc, w_{\tau(t)}; \auxvec) \explain{def}{=} \eta_{t+1}(H_1(w_1; \auxvec), \dotsc, H_t(w_1, \dotsc, w_{\tau(t)}; \auxvec); \auxvec).
\end{align*}
Note that the composite non-linearity $\tilde{\nonlin}_{t+1}$ might not be divergence-free (\assumpref{div-free}) with respect to $(\serv{W}_1, \dotsc, \serv{W}_{\tau(t)})$. In order to address this, we define:
\begin{align*}
    \hat{f}_{t+1}(w_1, \dotsc, w_{\tau(t)}; \auxvec) & = \tilde{\nonlin}_{t+1}(w_1, \dotsc, w_{\tau(t)}; \auxvec) - \sum_{i=1}^{\tau(t)} (\beta_{t+1})_i w_i,
\end{align*}
where the vector $\beta_{t+1} \in \R^{\tau(t)}$ is any solution to the linear equation:
\begin{align*}
    \Sigma_{\tau(t)} \cdot \beta_{t+1} & = \E \big[ \tilde{\nonlin}_{t+1}(\serv{W}_1, \dotsc, \serv{W}_{\tau(t)}; \serv{A}) \cdot (\serv{W}_1, \dotsc, \serv{W}_{\tau(t)})\tran \big],
\end{align*}
which is guaranteed to have a solution since the vector:
\begin{align*}
    \E \big[ \tilde{\nonlin}_{t+1}(\serv{W}_1, \dotsc, \serv{W}_{\tau(t)}; \serv{A}) \cdot (\serv{W}_1, \dotsc, \serv{W}_{\tau(t)})\tran \big] = \Sigma_{\tau(t)}^{1/2} \cdot \E \big[ \tilde{\nonlin}_{t+1}(\Sigma_{\tau(t)}^{1/2} \serv{G}; \serv{A}) \cdot \serv{G} \big], \; \serv{G} \sim \gauss{0}{I_{\tau(t)}}
\end{align*}
lies in the range of $\Sigma_{\tau(t)}$. Observe that by construction, the modified composite non-linearity $\hat{f}_{t+1}$ is divergence-free (\assumpref{div-free}) with respect to $(\serv{W}_1, \dotsc, \serv{W}_{\tau(t)})$. We can then express $\iter{\vz}{t+1}$ as:
\begin{align}
    \iter{\vz}{t+1}  &= \mS \mPsi_{t+1} \mS \cdot \hat{\nonlin}_{t+1}(\iter{\vw}{1}, \dotsc, \iter{\vw}{\tau(t)}; \auxmat) + \sum_{i=1}^{\tau(t)} (\beta_{t+1})_i \cdot  \mS \mPsi_{t+1} \mS \cdot \iter{\vw}{i} + \tilde{\eta}_{t+1}(\iter{\vw}{1}, \dotsc, \iter{\vw}{\tau(t)}; \auxmat) +  \iter{\vdelta}{t+1}\nonumber\\
    & = \mS \mPsi_{t+1} \mS \cdot \hat{\nonlin}_{t+1}(\iter{\vw}{1}, \dotsc, \iter{\vw}{\tau(t)}; \auxmat) + \sum_{i=1}^{\tau(t)} (\beta_{t+1})_i \cdot  \mS \mPsi_{t+1} \mXi_i \mS \cdot  g_i(\iter{\vw}{1}, \dotsc, \iter{\vw}{i-1}; \auxmat)  \label{eq:GFOM-t+1-iterate} \\ & \hspace{9.5cm} + \tilde{\eta}_{t+1}(\iter{\vw}{1}, \dotsc, \iter{\vw}{\tau(t)}; \auxmat) + \iter{\vdelta}{t+1}. \nonumber
\end{align}
We set $\tau(t+1) = 2\tau(t) + 1$ and construct the VAMP iterates $\iter{\vw}{i} = \mQ_i \cdot g_i(\iter{\vw}{1}, \dotsc, \iter{\vw}{i-1}; \auxmat)$ for $\tau(t) <i \leq \tau(t+1)$ as follows:
\begin{enumerate}
    \item We define the non-linearities $g_i$ for $\tau(t) <i \leq \tau(t+1)$ as follows:
    \begin{align} 
        g_{\tau(t+1)}(w_1, \dotsc, w_{\tau(t+1)-1}; \auxvec) &\explain{def}{=} \hat{f}_{t+1}(w_1, \dotsc, w_{\tau(t)}; \auxvec), \\
        g_{\tau(t)+j}(w_1, \dotsc, w_{\tau(t)+j-1}; \auxvec) &\explain{def}{=} g_{j}(w_1, \dotsc, w_{j-1}; \auxvec) \quad \forall \; j \; \in \; [\tau(t)]. \label{eq:new-nonlin}
    \end{align}
    The construction of $\hat{f}_{t+1}$ guarantees that $g_{\tau(t+1)}$ is divergence-free (\assumpref{div-free}). Since the non-linearities $g_{1:\tau(t)}$ are divergence-free by the induction hypothesis, $ g_{\tau(t)+j}$ is also divergence-free for any $j \leq \tau(t)$. Indeed, $j \in [\tau(t)]$ and any $i<j$ we have:
    \begin{align} \label{eq:div-free-easy}
        \E[\serv{W}_i  \cdot g_{\tau(t)+j}(\serv{W}_1, \dotsc, \serv{W}_{\tau(t) + j - 1}; \serv{A})]& \explain{\eqref{eq:new-nonlin}}{=} \E[\serv{W}_i \cdot g_j(\serv{W}_1, \dotsc, \serv{W}_{j-1}; \serv{A} )] \explain{(a)}{=} 0.
    \end{align}
    In the above display, the equality (a) follows from the fact that the non-linearity $g_{j}$ is divergence-free (induction hypothesis). On the other hand for any $i$ such that $j \leq i \leq \tau(t) + j - 1$ we can write:
    \begin{align*}
         \E[\serv{W}_i  \cdot g_{\tau(t)+j}(\serv{W}_1, \dotsc, \serv{W}_{\tau(t) + j - 1}; \serv{A})] & \explain{\eqref{eq:new-nonlin}}{=} \E[\serv{W}_i \cdot g_j(\serv{W}_1, \dotsc, \serv{W}_{j-1}; \serv{A} )] \\&= \E[ \E[\serv{W}_i | \serv{W}_1, \dotsc, \serv{W}_{j-1}, \serv{A}] \cdot g_j(\serv{W}_1, \dotsc, \serv{W}_{j-1}; \serv{A} )] \\
         & \explain{(a)}{=} \E[ \E[\serv{W}_i | \serv{W}_1, \dotsc, \serv{W}_{j-1}] \cdot g_j(\serv{W}_1, \dotsc, \serv{W}_{j-1}; \serv{A} )] \\
         & \explain{(b)}{=} 0.
    \end{align*}
    In the above display step (a) follows because the side information random variable $\serv{A}$ is independent of the Gaussian state evolution random variables $\serv{W}_{1:i}$ associated with VAMP algorithm. The equality in step (b) is obtained by observing that $\E[\serv{W}_i | \serv{W}_1, \dotsc, \serv{W}_{j-1}]$ is a linear combination of $\serv{W}_1, \dotsc, \serv{W}_{j-1}$ (since $\serv{W}_{1:i}$ are jointly Gaussian) and appealing to \eqref{eq:div-free-easy}. Hence, we have checked that the non-linearities $g_i$ for $\tau(t) <i \leq \tau(t+1)$ are divergence-free in the sense of \assumpref{div-free}.
    \item In order to define the matrices $\{\mQ_i : \tau(t) < i \leq \tau(t+1)\}$, we first recall that by the induction hypothesis, we have already constructed matrices $\mXi_{1:\tau(t)}$ given by the formulas:
    \begin{align} \label{eq:induction-xi-form}
        \mXi_i & = \sum_{B \subset [t]} \hat{\alpha}_{i,B} \cdot   (\mPsi_B - \hat{\Omega}_{B,\emptyset} \cdot  \mI_{\dim}) \quad \forall \;  i \;  \in [\tau(t)].
    \end{align}
    For each $\tau(t) < i \leq \tau(t+1)$, we define $\mQ_i = \mS \mXi_i \mS$ where $\mXi_i$ is given by:
    \begin{subequations} \label{eq:new-matrices}
    \begin{align}
        \mXi_{\tau(t+1)} &\explain{def}{=} \mPsi_{t+1} - \hat{\Omega}_{\{t+1\}, \emptyset} \cdot \mI_{\dim}, \\
        \mXi_{\tau(t) + j}  &\explain{def}{=} \mPsi_{t+1} \mXi_j - \left(\sum_{B \subset [t]} \hat{\alpha}_{j,B} \cdot (\hat{\Omega}_{B \cup \{t+1\}, \emptyset} - \hat{\Omega}_{B, \emptyset} \cdot \hat{\Omega}_{ \{t+1\}, \emptyset}) \right) \cdot \mI_{\dim} \quad \forall \;  j \;  \in [\tau(t)]\\
        & \explain{\eqref{eq:induction-xi-form}}{=} \sum_{B \subset [t]} \hat{\alpha}_{j,B} \cdot (\mPsi_{B \cup \{t+1\}} - \hat{\Omega}_{B \cup \{t+1\}, \emptyset} \mI_{\dim})  - \hat{\gamma}_{j} \cdot (\mPsi_{t+1} - \hat{\Omega}_{\{t+1\},\emptyset}) \cdot \mI_{\dim}, \nonumber
    \end{align}
    \end{subequations}
    where:
    \begin{align*}
        \hat{\gamma}_{j} \explain{def}{=} \sum_{B \subset [t]} \hat{\alpha}_{j,B} \cdot\hat{\Omega}_{B, \emptyset}.
    \end{align*}
    Observe that the newly constructed matrices $\{\mXi_i : \tau(t) < i \leq \tau(t+1)\}$ are also of the form given in \eqref{eq:induction-xi-form}. This can be used to verify that the matrices $\mQ_{1:\tau(t+1)}$ form a semi-random ensemble. Indeed, since $\mM_{1:T} = \mS\mPsi_{1:T}\mS$ form a strongly semi-random ensemble (\defref{strong-semirandom}) for any $i,j \in [\tau(t+1)]$, we have:
    \begin{align*}
        &\|\mXi_i\|_{\infty}  \leq \sum_{B \subset [t+1]} |\hat{\alpha}_{i,B}| \cdot \|\mPsi_{B} -  \hat{\Omega}_{B,\emptyset} \cdot  \mI_{\dim}\|_{\infty} \lesssim \dim^{-1/2+\epsilon}, \\  &\Tr(\mXi_i \mXi_j \tran)/\dim = \sum_{B \subset [t+1], B^\prime \subset [j]} \hat{\alpha}_{i,B} \cdot \hat{\alpha}_{j,B^\prime} \cdot (\hat{\Omega}_{B,B^\prime} - \hat{\Omega}_{B, \emptyset} \hat{\Omega}_{B^\prime, \emptyset}),  \\ &\left\| \mXi_i \mXi_j \tran - \left(\sum_{B \subset [t+1], B^\prime \subset [j]} \hat{\alpha}_{i,B} \cdot \hat{\alpha}_{j,B^\prime} \cdot (\hat{\Omega}_{B,B^\prime} - \hat{\Omega}_{B, \emptyset} \hat{\Omega}_{B^\prime, \emptyset})  \right) \cdot \mI_{\dim} \right\|_{\infty}  \lesssim \dim^{-1/2+\epsilon}.
    \end{align*}
    \item In light of the above definitions, we can express the formula obtained for the $t+1$ iteration of the GFOM in \eqref{eq:GFOM-t+1-iterate} as:
    \begin{align*}
        &\iter{\vz}{t+1} =  \mS \mPsi_{t+1} \mS \cdot \hat{\nonlin}_{t+1}(\iter{\vw}{1}, \dotsc, \iter{\vw}{\tau(t)}; \auxmat) + \sum_{i=1}^{\tau(t)} (\beta_{t+1})_i \cdot  \mS \mPsi_{t+1} \mXi_i \mS \cdot  g_i(\iter{\vw}{1}, \dotsc, \iter{\vw}{i-1}; \auxmat) \\& \hspace{10cm} + \tilde{\eta}_{t+1}(\iter{\vw}{1}, \dotsc, \iter{\vw}{\tau(t)}; \auxmat) + \iter{\vdelta}{t+1} \\
        & \explain{\eqref{eq:new-matrices}}{=} \mQ_{\tau(t+1)} \cdot g_{\tau(t+1)}(\iter{\vw}{1}, \dotsc, \iter{\vw}{\tau(t+1) - 1}; \auxmat) + \sum_{i=1}^{\tau(t)} (\beta_{t+1})_i \cdot \mQ_{\tau(t) +i} \cdot g_{\tau(t)+i}(\iter{\vw}{1}, \dotsc, \iter{\vw}{\tau(t)+i-1}; \auxmat) \\
        &\quad+ \hat{\Omega}_{t+1, \emptyset} \cdot g_{\tau(t+1)}(\iter{\vw}{1}, \dotsc, \iter{\vw}{\tau(t+1) - 1}; \auxmat) + \sum_{i=1}^{\tau(t)} (\beta_{t+1})_i \cdot \hat{\Upsilon}_{t+1,i} \cdot g_{\tau(t)+i}(\iter{\vw}{1}, \dotsc, \iter{\vw}{\tau(t)+i-1}; \auxmat) \\& \hspace{10cm} + \tilde{\eta}_{t+1}(\iter{\vw}{1}, \dotsc, \iter{\vw}{\tau(t)}; \auxmat) + \iter{\vdelta}{t+1}, \end{align*}
        where:
        \begin{align*}
            \hat{\Upsilon}_{t+1,i} \explain{def}{=} \sum_{B \subset [t]} \hat{\alpha}_{i,B} \cdot (\hat{\Omega}_{B \cup \{t+1\}, \emptyset} - \hat{\Omega}_{B, \emptyset} \cdot \hat{\Omega}_{ \{t+1\}, \emptyset}) \quad \forall \; i \; \in \; [\tau(t)].
        \end{align*}
        Hence:
        \begin{align*}
        & \iter{\vz}{t+1} = \iter{\vw}{\tau(t+1)} + \sum_{i=1}^{\tau(t)} (\beta_{t+1})_i \cdot \iter{\vw}{\tau(t) +i} + \hat{\Omega}_{t+1, \emptyset} \cdot g_{\tau(t+1)}(\iter{\vw}{1}, \dotsc, \iter{\vw}{\tau(t+1) - 1}; \auxmat) \\&\hspace{1.5cm}+ \sum_{i=1}^{\tau(t)} (\beta_{t+1})_i \cdot \hat{\Upsilon}_{t+1,i} \cdot g_{\tau(t)+i}(\iter{\vw}{1}, \dotsc, \iter{\vw}{\tau(t)+i-1}; \auxmat) + \tilde{\eta}_{t+1}(\iter{\vw}{1}, \dotsc, \iter{\vw}{\tau(t)}; \auxmat) + \iter{\vdelta}{t+1}.
    \end{align*}
    In light of the above display, we define the post-processing function $H_{t+1} : \R^{\tau(t+1) + \auxdim} \mapsto \R$ as:
    \begin{align*}
        H_{t+1}(w_1, \dotsc, w_{\tau(t+1)}; \auxvec) &\explain{def}{=} w_{\tau(t+1)} + \sum_{i=1}^{\tau(t)} (\beta_{t+1})_i \cdot w_{\tau(t) +i} + \Omega_{t+1, \emptyset} \cdot g_{\tau(t+1)}(w_1, \dotsc, w_{\tau(t+1) - 1}; \auxvec) \\&\hspace{0.5cm}+ \sum_{i=1}^{\tau(t)} (\beta_{t+1})_i \cdot \Upsilon_{t+1,i} \cdot g_{\tau(t)+i}(w_{1}, \dotsc, w_{\tau(t)+i-1}; \auxvec) + \tilde{\eta}_{t+1}(w_1, \dotsc, w_{\tau(t)}; \auxvec),
    \end{align*}
    where:
    \begin{align*}
         {\Upsilon}_{t+1,i} \explain{def}{=} \sum_{B \subset [t]} {\alpha}_{i,B} \cdot ({\Omega}_{B \cup \{t+1\}, \emptyset} - {\Omega}_{B, \emptyset} \cdot {\Omega}_{ \{t+1\}, \emptyset}) \quad \forall \; i \; \in \; [\tau(t)].
    \end{align*}
    This ensures:
    \begin{align*}
        \iter{\vz}{t+1} = H_{t+1}(\iter{\vw}{1}, \dotsc, \iter{\vw}{\tau(t+1)}; \auxmat) + \iter{\vDelta}{t+1}
    \end{align*}
    where:
    \begin{align*}
        \iter{\vDelta}{t+1} & = \iter{\vdelta}{t+1} + (\hat{\Omega}_{t+1, \emptyset}-{\Omega}_{t+1, \emptyset}) \cdot g_{\tau(t+1)}(\iter{\vw}{1}, \dotsc, \iter{\vw}{\tau(t+1) - 1}; \auxmat) \\&\hspace{3cm} + \sum_{i=1}^{\tau(t)} (\beta_{t+1})_i \cdot (\hat{\Upsilon}_{t+1,i}-{\Upsilon}_{t+1,i}) \cdot g_{\tau(t)+i}(\iter{\vw}{1}, \dotsc, \iter{\vw}{\tau(t)+i-1}; \auxmat).
    \end{align*}
    We now verify that $\|\iter{\vDelta}{t+1}\|^2/\dim \explain{P}{\rightarrow} {0}$. Note that by \thref{VAMP}:
    \begin{align*}
        \|g_{\tau}(\iter{\vw}{1}, \dotsc, \iter{\vw}{\tau - 1}; \auxmat)\|^2/\dim \explain{P}{\rightarrow} \E g_{\tau}^2(\serv{W}_1, \dotsc, \serv{W}_{\tau-1}; \serv{A}) < \infty \quad \forall \; \tau \; \in \; [\tau(t+1)].
    \end{align*}
    Furthermore, recalling the definition of $\iter{\vdelta}{t+1}$ from \eqref{eq:delta-error} we have:
    \begin{align*}
        \|\iter{\vdelta}{t+1}\|^2/\dim & \leq 2\|\mPsi_{t+1}\|_{\op}^2 \cdot  \|\nonlin_{t+1}(\iter{\vz}{1}, \dotsc, \iter{\vz}{t}; \auxmat) - \nonlin_{t+1}(\iter{\vz}{1} - \iter{\vDelta}{1}, \dotsc, \iter{\vz}{t} - \iter{\vDelta}{t} ; \auxmat)\|^2/\dim \\ &\hspace{2cm} + 2\|\eta_{t+1}(\iter{\vz}{1}, \dotsc, \iter{\vz}{t}; \auxmat) - \eta_{t+1}(\iter{\vz}{1} - \iter{\vDelta}{1}, \dotsc, \iter{\vz}{t} - \iter{\vDelta}{t} ; \auxmat)\|^2/\dim \\
        & \explain{}{\leq} 2 L^2 \cdot  (1 + \|\mPsi_{t+1}\|_{\op}^2) \cdot \sum_{i=1}^t \|\iter{\vDelta}{i}\|^2 /\dim \\
        &\explain{P}{\rightarrow} 0.
    \end{align*}
    In the above display the second-to-last step follows from the regularity assumptions imposed on $\nonlin_{t+1}, \eta_{t+1}$ in the statement of \thref{GFOM} and the last step follows from the induction hypothesis. Combining the above conclusions with the fact that  $\hat{\Omega}_{t+1, \emptyset}-{\Omega}_{t+1, \emptyset} \rightarrow 0$ and $\hat{\Upsilon}_{t+1,i}-{\Upsilon}_{t+1,i} \rightarrow 0$ yields $\|\iter{\vDelta}{t+1}\|^2/\dim \explain{P}{\rightarrow} {0}$, as desired. This completes the inductive construction of the desired VAMP algorithm which implements a given GFOM and concludes the proof of \thref{GFOM}. 
\end{enumerate}
\end{proof}

\section{Proof of \thref{moments} via Method of Moments} \label{appendix:moment-method}

This appendix is devoted to the proof of \thref{moments} and is organized as follows:
\begin{enumerate}
    \item \appref{key-results} introduces the key ideas involved in the proof of \thref{moments} in the form of a few intermediate results. 
    \item The proofs of \thref{moments} and \corref{moments} are immediate given these intermediate results, and is presented in \appref{moments-actual-proof}. 
    \item \appref{proof-hermite-expansion}, \appref{expectation-formula}, \appref{key-estimate}, and \appref{decomposition} are devoted to the proofs of the intermediate results introduced in \appref{key-results}. 
\end{enumerate}

\subsection{Key Results} \label{appendix:key-results}
\subsubsection{Unrolling the MVAMP Iterations}
We start by expressing the key quantity of interest in \thref{moments}:
\begin{align} \label{eq:key-observable}
     \frac{1}{\dim}\sum_{j=1}^\dim h(\auxvec_j) \cdot  \hermite{r}\left(\iter{z}{T,\cdot}_j\right)
\end{align}
as a polynomial of the matrices $\mM_{1:k}$ used the MVAMP iteration and the initialization $\iter{\mZ}{0,\cdot}$. This involves recursively ``unrolling'' the MVAMP iterations to express each iterate $\iter{\mZ}{t,\cdot}$ as a polynomial of the previous iterate $\iter{\mZ}{t-1,\cdot}$ and continuing this process till we reach the initialization $\iter{\mZ}{0,\cdot}$. The resulting polynomial can be expressed as a combinatorial sum over colorings of decorated $\order$-trees, defined below.  

\begin{definition}[Decorated $\order$-trees] \label{def:decorated-tree} Let $\order\in \N$ be arbitrary positive integer. A decorated $\order$-tree $F$ is given by a tuple $(V, E, \height{\cdot}, \pweight{\cdot}{}, \qweight{\cdot}{})$ where:
\begin{enumerate}
    \item $V = \{0,1, 2, 3, \dotsc, |V|-1\}$ is the set of vertices.
    \item $E$ is the set of directed edges.
\end{enumerate}
The sets $(V,E)$ are such that the directed graph given by $(V,E)$ is a directed tree. {Furthermore, the tree is non-trivial in the sense that $|E| \geq 1$.} We define the following notions:
\begin{enumerate}
     \item If $u \rightarrow v \in E$, we say $u$ is the parent of $v$ and $v$ is a child of $u$. Each vertex in a directed tree has at most one parent. 
    \item A vertex with no parent is called a root vertex. A directed tree has exactly one root vertex, denoted by $0$. 
    \item For every vertex $u$ we define $\chrn{F}(u)$ as the number of children of $u$.
    \item A vertex with no children is called a leaf. The set of all leaves is denoted by  $\leaves{F}$.
    \item A pair of non-root vertices $u, v \in V \backslash \{0\}$ are siblings if they have the same parent.
    \item The positive integer $\order$ is called the order of the tree.
\end{enumerate}
The tree is decorated with 3 functions 
\begin{align*}\height{\cdot}{} : V &\rightarrow \W = \{0,1, 2, 3, \dotsc \}, \\ \pweight{\cdot}{} : V \backslash \{0\} &\rightarrow \W^\order \backslash \{(0,0, \dotsc, 0)\}, \\ \qweight{\cdot}{} : V &\rightarrow \W^k,
\end{align*}
such that:
\begin{enumerate}
    \item  The height function $\height{\cdot}{}$ has the following properties:
    \begin{enumerate}
        \item {$\height{0}{} \geq 1$}.
        \item for any $u \rightarrow v \in E$, $\height{v} = \height{u} - 1$. 
        \item If $\height{u} = 0$, then $u$ has no children $(\chrn{F}(u)=0)$. 
        \item For every vertex $u$ with no children and $\|\qweight{u}{}\|_1 \geq 1$, we have $\height{u} = 0$.
    \end{enumerate}
    \item The functions $\pweight{\cdot}{}$ and $\qweight{\cdot}{}$ satisfy: 
    \begin{enumerate}
    \item {$\|\qweight{0}{}\|_1 \geq 1$.}
    \item For any non-leaf vertex $u \in V \backslash \leaves{F}$, we have,
    \begin{align}\label{eq:conservation-eq}
        \qweight{u}{} = \sum_{v\in V: u \rightarrow v } \pweight{v}{}.
    \end{align}
    \end{enumerate}
\end{enumerate}
\end{definition}

Next, we introduce the notion of a coloring of a decorated $\order$-tree. 

\begin{definition}[Coloring of a decorated $\order$-tree] \label{def:coloring}
A coloring of a decorated $\order$-tree $F$ with vertex set $V$ is a map $\ell : V \rightarrow [\dim]$. The set of all colorings of a $\order$-tree $F$ with vertex set $V$ is denoted by $[\dim]^{V}$.
\end{definition}

The colored decorated $\order$-trees that appear in the polynomial expansion of \eqref{eq:key-observable} satisfy certain constraints, which we collect in the following definition of \emph{valid} colorings.   

\begin{definition}[Valid Colorings of Decorated $\order$-trees] \label{def:valid-tree}  A decorated $\order$-tree $F = (V, E, \height{\cdot}, \pweight{\cdot}{}, \qweight{\cdot}{})$  and a coloring $\ell : V \rightarrow [\dim]$ are valid if for vertices $u, v$ that are siblings in the tree, we have $\ell_u \neq \ell_v$.
We denote valid decorated colored $\order$-trees by defining the indicator function $\valid$ such that $\valid(F,\ell) = 1$ iff $(F,\ell)$ is a valid colored decorated $\order$-tree and $\valid(F,\ell) = 0$ otherwise. 
\end{definition}

With these definitions, we can now present a formula for the polynomial representation of the key quantity \eqref{eq:key-observable} in the lemma below. 

\begin{lemma}[Unrolling Lemma]\label{lem:hermite-expansion} For any $T \in \N$, $r \in \W^\order$ with $\|r\|_1 \geq 1$ and $h: \R^{\auxdim} \rightarrow \R$, we have,
\begin{subequations}\label{eq:hermite-expansion}
\begin{align}
    \frac{1}{\dim}\sum_{j=1}^\dim h(\auxvec_j) \cdot \hermite{r}(\iter{z}{T,\cdot}_j)  = \sum_{\substack{F \in \trees{k}{r,T} \\ F = (V, E, \height{\cdot}, \pweight{\cdot}{}, \qweight{\cdot}{})}}\alpha(F)  \sum_{\substack{\ell \in [\dim]^V }} \valid(F,\ell)  \chi(\auxmat; F, \ell)  \gamma(\mM_{1:k} ; F, \ell)   \beta(\iter{\mZ}{0,\cdot}; F, \ell).
\end{align}
In the above display $\trees{\order}{r,T}$ denotes the set of all decorated $\order$-trees with $\qweight{0}{} = r$ and $\height{0} = T$. Furthermore, for a colored decorated $\order$-tree $(F,\ell)$ with $F = (V, E, \height{\cdot}, \pweight{\cdot}{}, \qweight{\cdot}{})$, the weights $\alpha(F)$, $\chi(\auxmat; F,\ell)$, $\gamma(\mM_{1:k} ; F, \ell)$, and $\beta(\iter{\mZ}{0,\cdot}; F, \ell)$ are defined as follows: 
\begin{align}
    \alpha(F) &\explain{def}{=}   \prod_{u \in V \backslash \leaves{F}} \frac{\sqrt{\qweight{u}{}!}}{\chrn{F}(u)!} \prod_{v : u \rightarrow v} \frac{1}{\sqrt{\pweight{v}{}!}}, \\
    \chi(\auxmat; F, \ell) &\explain{def}{=} h(\auxvec_{\ell_0}) \cdot  \prod_{\substack{u \in V \\ u \neq 0}} \fourier{\pweight{u}{}}{\qweight{u}{}}{\auxvec_{\ell_u}}, \\
    \gamma(\mM_{1:k} ; F, \ell) & \explain{def}{=} \frac{1}{\dim} \prod_{\substack{e \in E \\ e = u \rightarrow v}} \prod_{i=1}^\order (M_i)_{\ell_u, \ell_v}^{\pweight{v}{i}}, \\
    \beta(\iter{\mZ}{0,\cdot}; F, \ell) &\explain{def}{=} \prod_{\substack{v \in \leaves{F} \\ \|\qweight{v}{}\|_1 \geq 1}} \hermite{\qweight{v}{}}( \iter{z}{0, \cdot }_{\ell_v}),
\end{align}
\end{subequations}
where for any $r \in \W^\order, r^\prime \in \W^\order$, $z \in \R^{\order}$, and $\auxvec \in \R^\auxdim$, we defined the notations:
\begin{align} \label{eq:hermite-expansion-notations}
    r! \explain{def}{=} \prod_{i=1}^\order r_i! \; , \quad \hermite{r}{(z)} = \prod_{i=1}^\order \hermite{r_i}(z_i), \quad \fourier{r}{r^\prime}{\auxvec} \explain{def}{=}  \E \left[ \hermite{r^\prime}(\serv{Z}) \cdot  \prod_{i=1}^\order \hermite{r_i}(\nonlin_i(\serv{Z}; \auxvec)) \right], \; \serv{Z} \sim \gauss{0}{I_\order}.
    \end{align}
\end{lemma}
We postpone the proof of this formula to \appref{proof-hermite-expansion}. 
\subsubsection{The Expectation Formula}
The next step in the proof of \thref{moments} involves computing the expectations with respect to the signed diagonal matrix $\mS = \diag(s_{1:\dim})$ used to construct  the semi-random ensemble $\mM_{1:\order} = \mS \mPsi_{1:\order} \mS$, the Gaussian initialization $\iter{\mZ}{0,\cdot}$ and the auxiliary information $\auxmat$ in the polynomial expansion given in \lemref{hermite-expansion}.  Observe that since $\mS, \iter{\mZ}{0,\cdot}, \auxmat$ are mutually independent, taking expectations in \eqref{eq:hermite-expansion} yields:
\begin{align} 
     &\E \left[ \frac{1}{\dim}\sum_{j=1}^\dim h(\auxvec_j) \cdot \hermite{r}(\iter{z}{T,\cdot}_j) \right] \nonumber \\& \hspace{1.5cm} = \sum_{\substack{F \in \trees{k}{r,T} \\ F = (V, E, \height{\cdot}, \pweight{\cdot}{}, \qweight{\cdot}{})}}\alpha(F) \cdot \sum_{\substack{\ell \in [\dim]^V }} \valid(F,\ell) \cdot \E[\chi(\auxmat; F, \ell)] \cdot \E[\gamma(\mM_{1:\order} ; F, \ell)] \cdot \E[ \beta(\iter{\mZ}{0,\cdot}; F, \ell)]. \label{eq:intermediate-expectation}
\end{align}
Observe that in order to evaluate each of the expectations:
\begin{align*}
    \E[\chi(\auxmat; F, \ell)], \quad \E[\gamma(\mM_{1:\order} ; F, \ell)], \quad \E[ \beta(\iter{\mZ}{0,\cdot}; F, \ell)], 
\end{align*}
the repetition pattern of the coloring $\ell \in [\dim]^V$ is important---for two vertices $u,v$ that have the same color $\ell_{u} = \ell_{v}$, the corresponding random variables are identical $(s_{\ell_u}, \auxvec_{\ell_u}, \iter{z}{0,\cdot}_{\ell_u}) = (s_{\ell_v}, \auxvec_{\ell_v}, \iter{z}{0,\cdot}_{\ell_v})$. On the other hand, for vertices $u, v$ with different colors $\ell_{u} \neq \ell_{v}$,   the corresponding $(s_{\ell_u}, \auxvec_{\ell_u}, \iter{z}{0,\cdot}_{\ell_u})$ and $(s_{\ell_v}, \auxvec_{\ell_v}, \iter{z}{0,\cdot}_{\ell_v})$ are independent. The repetition pattern of a coloring in $\ell \in [\dim]^V$ can be encoded by a partition of the vertex set $V$. This motivates the following definitions.

\begin{definition}[Partitions and Configurations]\label{def:partition} Given a decorated $\order$-tree $F=(V, E, \height{\cdot}, \pweight{\cdot}{}, \qweight{\cdot}{})$, a partition $\pi$ of the vertex set $V$ is a collection of disjoint subsets (called blocks) $\{B_1,B_2, \dotsc ,B_{s}\}$ such that
\begin{align*}
    \bigcup_{j=1}^s B_j = V, \; B_j \cap B_k = \emptyset \; \forall \; j \neq k.
\end{align*}
We define $|\pi|$ to be the number of blocks in $\pi$. For every $v \in V$, we use $\pi(v)$ to denote the unique block $j \in   [|\pi|]$ such that $v \in B_j$. Without loss of generality, we will assume that for the root vertex $0$, $\pi(0) = 1$ or equivalently $0 \in B_1$. The set of all partitions of $V$ is denoted by $\part{V}$. A configuration is a pair $(F,\pi)$ consisting of a decorated $\order$-tree $F$ and a partition $\pi$ of its vertices. 
\end{definition}

\begin{definition}[Colorings consistent with a partition] Let $\pi$ be a partition of the vertex set $V$ of a decorated $\order$-tree $F=(V, E, \height{\cdot}, \pweight{\cdot}{}, \qweight{\cdot}{})$. A coloring consistent with $\pi$ is a function $\ell: V \rightarrow [\dim]$ such that, 
\begin{align*}
    \ell_u = \ell_v & \Leftrightarrow \pi(u) = \pi(v).
\end{align*}
The set of all colorings that are consistent with a partition $\pi$ is denoted by $\colorings{}{\pi}$.
\end{definition}

Next, we note that whether a colored $\order$-tree $(F,\ell)$ is valid or not (\defref{valid-tree}) can be determined by knowing the repetition pattern of $\ell$. Hence, we introduce the following definition.

\begin{definition}[Valid Configurations] A decorated $\order$-tree $F = (V, E, \height{\cdot}, \pweight{\cdot}{}, \qweight{\cdot}{})$  and a partition $\pi \in \part{V}$ form a valid configuration if for any two vertices $u, v$ that are siblings in the forest, we have $\pi(u) \neq \pi(v)$. We denote valid configurations by defining the indicator function $\valid$ such that $\valid(F,\pi) = 1$ iff $(F,\pi)$ is a valid configuration and $\valid(F,\pi) = 0$ otherwise.
\end{definition}

In light of the above definition, observe that \eqref{eq:intermediate-expectation} can be rearranged as:
\begin{align} \label{eq:rearrangement}
    &\E \left[ \frac{1}{\dim}\sum_{j=1}^\dim h(\auxvec_j) \cdot \hermite{r}(\iter{z}{T,\cdot}_j) \right]  = \nonumber \\&\sum_{\substack{F \in \trees{k}{r,T} \\ F = (V, E, \height{\cdot}, \pweight{\cdot}{}, \qweight{\cdot}{})}} \sum_{\pi \in \part{V}}\alpha(F) \cdot \valid(F,\pi) \cdot \sum_{\substack{\ell \in \colorings{}{\pi} }}  \E[\chi(\auxmat; F, \ell)] \cdot \E[\gamma(\mM_{1:\order} ; F, \ell)] \cdot  \E[ \beta(\iter{\mZ}{0,\cdot}; F, \ell)]. 
\end{align}
The following lemma presents a simplified formula for the above expectation. 

\begin{lemma}[Expectation Formula]\label{lem:expectation-formula} For any $T \in \N$, any $r \in \W^\order$ with $\|r\|_1 \geq 1$ and any function $h:\R^{\auxdim} \mapsto \R$ with $\E[|h(\serv{A})|^p] <\infty$ for all $p \in \N$ we have,
\begin{subequations}\label{eq:expectation-formula}
\begin{align}
    &\E \left[ \frac{1}{\dim}\sum_{j=1}^\dim h(\auxvec_j) \cdot \hermite{r}(\iter{z}{T,\cdot}_j) \right] \nonumber\\& \hspace{1cm} = \sum_{\substack{F \in \trees{k}{r,T} \\ F = (V, E, \height{\cdot}, \pweight{\cdot}{}, \qweight{\cdot}{})}} \sum_{\pi \in \part{V}}\alpha(F)  \cdot \relev(F,\pi) \cdot \chi(F,\pi) \cdot \beta(F,\pi) \cdot  \sum_{\substack{\ell \in \colorings{}{\pi} }} \gamma(\mPsi_{1:\order} ; F, \ell). 
\end{align}
In the above display, $\trees{\order}{r,T}$ denotes the set of all decorated $\order$-trees with $\qweight{0}{} = r$ and $\height{0} = T$. Furthermore, for a configuration $(F,\pi)$ with $F = (V, E, \height{\cdot}, \pweight{\cdot}{}, \qweight{\cdot}{})$, $\pi = \{B_1, B_2, \dotsc, B_{|\pi|}\}$ and a coloring $\ell \in \colorings{}{\pi}$ the weights $\alpha(F)$,  $\chi(F,\pi)$, $\beta(F,\pi)$, and $\gamma(\mPsi_{1:\order} ; F, \ell)$ are defined as follows:
\begin{align}
     \alpha(F) &\explain{def}{=}   \prod_{u \in V \backslash \leaves{F}} \frac{\sqrt{\qweight{u}{}!}}{\chrn{F}(u)!} \prod_{v : u \rightarrow v} \frac{1}{\sqrt{\pweight{v}{}!}}, \\
     \chi( F, \pi) &\explain{def}{=}  \E \left[ h(\serv{A}) \cdot  \prod_{\substack{u \in B_1 \\ u \neq 0}} \fourier{\pweight{u}{}}{\qweight{u}{}}{\serv{A}} \right] \cdot \prod_{i=2}^{|\pi|} \E \left[ \prod_{\substack{u \in B_i}} \fourier{\pweight{u}{}}{\qweight{u}{}}{\serv{A}} \right],  \\
     \beta(F,\pi) & = \prod_{i=1}^{|\pi|} \E\left[ \prod_{\substack{u \in B_i \cap \leaves{F} \\ \|\qweight{v}{}\|_1 \geq 1}}  \hermite{\qweight{u}{}}(\serv{Z}) \right], \\
      \gamma(\mPsi_{1:k} ; F, \ell) & \explain{def}{=} \frac{1}{\dim} \prod_{\substack{e \in E \\ e = u \rightarrow v}} \prod_{i=1}^\order (\Psi_i)_{\ell_u, \ell_v}^{\pweight{v}{i}}, 
\end{align}
\end{subequations}
where $\serv{A}$ is the auxiliary information random variable from \assumpref{side-info} and $\serv{Z}  \sim \gauss{0}{I_{\order}}$. Finally,  $\relev({F}, {\pi})$ is an indicator function which satisfies $\relev(F,\pi) = 1$ iff $(F,\pi)$ form a relevant configuration as defined below, and is zero otherwise.      
\end{lemma}
\begin{proof} The proof of this formula is provided in \appref{expectation-formula}. 
\end{proof}
\begin{definition}[Relevant Configurations] \label{def:relevant-config} A decorated $\order$-tree  $F = (V, E, \height{\cdot}, \pweight{\cdot}{}, \qweight{\cdot}{})$ and a partition $\pi = \{B_1, B_2, \dotsc, B_{|\pi|}\}$ form a relevant configuration if they satisfy the following properties:\footnote{Some of these properties have a strong/weak qualifier since we will introduce weaker or stronger versions of these properties later.}
\begin{enumerate}
     \myitem{\texttt{Strong Sibling Property}}\label{strong-sibling-property}: For vertices $u, v \in V \backslash \{0\}$ that are siblings in the forest, we have $\pi(u) \neq \pi(v)$.
    \myitem{\texttt{Weak Forbidden Weights Property}}\label{weak-forbidden-weights-property}: There are no vertices $u \in V \backslash \{0\}$ with $|B_{\pi(u)}| = 1$  such that: 
    \begin{enumerate}
        \item $\|\pweight{u}{}\|_1 = 1, \|\qweight{u}{}\|_1 = 1$ or,
        \item $\pweight{u}{} = 2 e_i, \qweight{u}{} = 0$ for some $i \in [\order]$, where $e_{1:\order}$ denote the standard basis vectors in $\R^{\order}$. 
        \item $\pweight{u}{}= e_i + e_j, \qweight{u}{} = 0$ for some $i,j \in [\order], \; i \neq j$ such that $\Omega_{ij} \neq 0$. Here, $e_{1:\order}$ denote the standard basis vectors in $\R^{\order}$ and $\Omega \in \R^{\order \times \order}$ is the limiting covariance matrix corresponding to the semi-random ensemble $\mM_{1:k}$ (cf. \defref{semirandom}). 
    \end{enumerate} 
    \myitem{\texttt{Leaf Property}}\label{original-leaf-property}: There are no leaf vertices $v \in \leaves{F}$ with $\|\qweight{v}{}\|_1 \geq 1$ and $|B_{\pi(v)}| = 1$.
    \myitem{\texttt{Parity Property}}\label{parity-property}: For each block $B$ of the partition $\pi$, the sum:
    \begin{align*}
       \|\qweight{0}{}\|_1 \cdot \vone_{0 \in B} + \left(\sum_{u \in B \backslash ( \leaves{F} \cup \{0\})} \|\pweight{u}{}\|_1 + \|\qweight{u}{}\|_1 \right) + \left( \sum_{v \in B \cap \leaves{F}} \|\pweight{v}{}\|_1 \right) 
    \end{align*}
    has even parity. In the above display $\vone_{0 \in B}$ is the indicator which is $1$ iff the root vertex $0$ lies in the block $B$ and is zero otherwise. 
\end{enumerate}
\end{definition}
\subsubsection{Improved Estimates on Polynomials Associated with a Configuration}

In light of \lemref{expectation-formula}, in order to prove \thref{moments} it suffices to show that for any \emph{relevant configuration} (\defref{relevant-config}):
\begin{align*}
    \sum_{\ell \in \colorings{}{\pi}} \frac{1}{\dim} \prod_{i=1}^\order \prod_{\substack{e \in E \\ e = u \rightarrow v}} (\mPsi_i)_{\ell_u, \ell_v}^{\pweight{v}{i}} \ll 1.
\end{align*}
Hence, we develop estimates on polynomials $\Gamma(\mPsi_{1:\order}; F,\pi)$ associated with a configuration $(F,\pi)$, defined as follows:

\begin{align} \label{eq:invariant-polynomial}
    \Gamma(\mPsi_{1:\order}; F, \pi) \explain{def}{=} \sum_{\ell \in \colorings{}{\pi}} \gamma(\mPsi_{1:\order}; F, \ell), \quad \gamma(\mPsi_{1:\order}; F, \ell) \explain{def}{=} \frac{1}{\dim} \prod_{i=1}^\order \prod_{\substack{e \in E \\ e = u \rightarrow v}} (\mPsi_i)_{\ell_u, \ell_v}^{\pweight{v}{i}}.
\end{align}
A simple estimate on $|\Gamma(\mPsi; F,\pi)|$ is as follows:
\begin{align} \label{eq:naive-estimate}
    |\Gamma(\mPsi_{1:\order}; F,\pi)| & \explain{(a)}{\leq} \sum_{\substack{\ell \in \colorings{}{\pi}}}  |\gamma(\mPsi; F,\ell)| \leq |\colorings{}{\pi}| \cdot \dim^{-1} \cdot \max_{\ell \in \colorings{}{\pi}} \gamma(\mPsi; F,\ell)   \explain{(b)}{\leq} \dim^{|\pi|-1-\frac{1}{2} \sum_{v \in V \backslash \{0\}} \|\pweight{v}{}\|_1 + \epsilon}.
\end{align}
In the above display, the step (a) uses the triangle inequality whereas step (b) uses the fact that $|\colorings{}{\pi}| \asymp \dim^{\pi}$ and the assumption that $\max_{i \in [\order]}\|\mPsi_i\|_{\infty} \lesssim \dim^{-1/2 + \epsilon}$. However, for many relevant configurations the naive estimate in \eqref{eq:naive-estimate} is insufficient to even obtain the weaker conclusion that $|\Gamma(\mPsi_{1:\order}; F,\pi)| \lesssim 1$. We refer the reader to our prior work \citep[Section 5.1.4]{dudeja2022universality} for a simple example illustrating this failure (in the situation when $\order = 1$). The key limitation of the simple estimate of \eqref{eq:naive-estimate} is the use of the triangle inequality in step (a). Many decorated $\order$-trees have certain structures called nullifying edges (introduced below) which allow us to use the determinstic constraints:
\begin{align} \label{eq:recall-semi-random}
    \max_{i,j \in [\order]} \| \mPsi_i \mPsi_j \tran - \Omega_{ij} \mI_{\dim} \|_\infty \lesssim \dim^{-1/2 + \epsilon} \; \forall \; \epsilon > 0.
\end{align}
to improve on the naive estimate of \eqref{eq:naive-estimate}. Recall that these constraints are satisfied by the matrices $\mPsi_{1:k}$ used to construct the semi-random ensemble $\mM_{1:k}$ (\defref{semirandom}). 

\begin{definition}[Nullifying Leaves and Edges]\label{def:nullifying}A pair of edges $u \rightarrow v$ and $u^\prime \rightarrow v^\prime$ is a pair of nullifying edges for a configuration $(F,\pi)$ with $F = (V, E, \height{\cdot}, \pweight{\cdot}{}, \qweight{\cdot}{})$ and $\pi = \{B_1, B_2, \dotsc, B_{|\pi|}\}$ if:
\begin{enumerate}
    \item $v\neq v^\prime$ and $v,v^\prime$ are leaves in F.
    \item $\|\pweight{v}{}\|_1 = \|\pweight{v^\prime}{}\|_1 = 1$,
    \item $B_{\pi(v)} = B_{\pi(v^\prime)} = \{v,v^\prime\}$, 
    \item $\pi(u) \neq \pi(u^\prime)$.
\end{enumerate}
In this situation, $v, v^\prime$ are referred to as a pair of nullifying leaves and the set of all nullifying leaves of a configuration $(T,\pi)$ is denoted by $\nullleaves{T,\pi}$. Note that $\nullleaves{T,\pi}$ is always even (since nullifying leaves occur in pairs) and the number of nullifying edges in a configuration is given by $|\nullleaves{T,\pi}|/2$.
\end{definition}
In order to illustrate the significance of nullifying edges consider a pair of nullifying edges $u \rightarrow v, u^\prime \rightarrow v^\prime$. Since $\|\pweight{v}{}\|_1 = \|\pweight{v^\prime}{}\|_1 = 1$ it must be that $\pweight{v}{} = e_i, \; \pweight{v^\prime}{} = e_j$ for some $i,j \in [\order]$ (where $e_{1:\order}$ denote the standard basis of $\R^\order$).  Summing over the possible colors for $v,v^\prime$ in \eqref{eq:invariant-polynomial} and noting that $\ell_u \neq \ell_u^\prime$ (cf. \defref{nullifying}) yields the expression
\begin{align*}
    \Bigg|\sum_{\substack{\ell_v, \ell_{v^\prime} \in [\dim] \\ \ell_{v} = \ell_{v^\prime}}}  (\Psi_i)_{\ell_u\ell_v} (\Psi_j)_{\ell_{u^\prime} \ell_{v^\prime}} \Bigg|  = |(\mPsi_i \mPsi_j \tran)_{\ell_{u} \ell_{u^\prime}}| \explain{\eqref{eq:recall-semi-random}}{\lesssim} \dim^{-1/2+\epsilon}.
\end{align*}
The estimate above is an improvement of the naive estimate obtained by the triangle inequality and the assumption $\|\mPsi\|_{\infty} \lesssim \dim^{-1/2+\epsilon}$ in \eqref{eq:naive-estimate}:
\begin{align*}
   \Bigg|\sum_{\substack{\ell_v, \ell_{v^\prime} \in [\dim] \\ \ell_{v} = \ell_{v^\prime}}}  (\Psi_i)_{\ell_u\ell_v} (\Psi_j)_{\ell_{u^\prime} \ell_{v^\prime}} \Bigg| \leq \sum_{\substack{\ell_v, \ell_{v^\prime} \in [\dim] \\ \ell_{v} = \ell_{v^\prime}}}  |(\Psi_i)_{\ell_u\ell_v}|  \cdot |(\Psi_j)_{\ell_{u^\prime} \ell_{v^\prime}}| \lesssim \dim \cdot \dim^{-1 + \epsilon} \lesssim \dim^{\epsilon}.
\end{align*}
Hence, if a configuration $(F,\pi)$ has $|\nullleaves{F,\pi}|$ nullifying leaves (or $|\nullleaves{F,\pi}|/2$ pairs of nullifying edges), one can expect to improve upon the naive estimate in \eqref{eq:naive-estimate} by a factor of $\dim^{|\nullleaves{F,\pi}|/4}$. The following proposition formalizes this intuition and is our key estimate on polynomials $\Gamma({\mPsi}_{1:\order}; F, \pi)$ associated with a configuration $(F,\pi)$ (cf. \eqref{eq:invariant-polynomial}).

\begin{proposition}[Improved Estimate] \label{prop:key-estimate} Consider a configuration $(F,\pi)$ with $F = (V, E, \height{\cdot}, \pweight{\cdot}{}, \qweight{\cdot}{})$ and $\pi = \{B_1, B_2, \dotsc, B_{|\pi|}\}$. We have,
\begin{align*}
   |\Gamma({\mPsi}_{1:\order}; F, \pi)| \explain{def}{=}  \left|   \sum_{\ell \in \colorings{}{\pi}} \frac{1}{\dim} \prod_{\substack{e \in E \\ e = u \rightarrow v}} \prod_{i=1}^\order (\Psi_i)_{\ell_u\ell_v}^{\pweight{v}{i}}  \right| & \lesssim \dim^{-\exponent(F,\pi) + \epsilon} \; \forall \; \epsilon > 0,
\end{align*}
where,
\begin{align*}
    \exponent(F,\pi) \explain{def}{=}  \frac{|\nullleaves{F,\pi}|}{4} + 1 - |\pi| + \frac{1}{2} \sum_{v \in V \backslash \{0\}}  \|\pweight{v}{}\|_1.
\end{align*}
\end{proposition}
The above result is a natural generalization of an estimate  \citep[Proposition 2]{dudeja2022universality} obtained in our prior work, which provided an upper bound on $\Gamma(\mPsi_{1:k}; F,\pi)$ in the special situation when $F$ is a decorated tree with order $\order = 1$. The proof of the above result is provided in \appref{key-estimate}. 
\subsubsection{Decomposition into Simple Configurations}
A direct application of the improved estimate in \propref{key-estimate} is often still not enough to show that the polynomial $\Gamma(\mPsi_{1:k}; F,\pi)$ (recall \eqref{eq:invariant-polynomial}) associated with a relevant configuration $(F,\pi)$ satisfies the desired estimate $\Gamma(\mPsi_{1:k}; F,\pi) \ll 1$ required to prove \thref{moments}. We again refer the reader to our prior work \citep[Section 5.1.5]{dudeja2022universality} for a simple example illustrating the failure of the improved estimate in \propref{key-estimate} in the case when $F$ is a tree of order $\order = 1$. However, it turns out that one can address this issue by simplifying the polynomial $\Gamma(\mPsi_{1:k}; F,\pi)$ using the deterministic constraints:
\begin{align} \label{eq:simplifying-constraint}
    (\mPsi_i \mPsi_j \tran)_{\ell \ell} & = \Omega_{ij} \; \forall \; \ell \in [\dim], \; i,j \in [\order].
\end{align}
before applying the improved estimate of \propref{key-estimate}. The constraints \eqref{eq:simplifying-constraint} were imposed on the matrix $\mPsi_{1:k}$ in the statement of \thref{moments}. In the following definitions, we introduce two structures, which, when present in a configuration $(F,\pi)$, allow one to leverage the constraints \eqref{eq:simplifying-constraint} to simplify the polynomial $\Gamma(\mPsi_{1:k}; F,\pi)$ associated with $(F,\pi)$. 

\begin{definition}[Removable Edge] \label{def:removable-edge}  An edge $u \rightarrow v$ is called a removable edge for configuration $(F,\pi)$ with $F = (V,E, \height{\cdot}, \pweight{\cdot}{}, \qweight{\cdot}{})$  and $\pi=\{B_1, B_2, \dotsc, B_{|\pi|}\}$ if:
\begin{enumerate}
    \item $v \in \leaves{F}$.
    \item $\pweight{v}{} = e_i + e_j$ for some $i,j \in [\order]$ with $i \neq j$ and $\Omega_{ij} = 0$. Here $e_{1:\order}$ denote the standard basis vectors in $\R^{\order}$ and $\Omega$ is the limiting covariance matrix corresponding to the semi-random ensemble $\mM_{1:k}$ (cf. \defref{semirandom}). 
    \item $|B_{\pi(v)}| =1$.
\end{enumerate}
\end{definition}

\begin{definition}[Removable Edge Pair] \label{def:removable-pair}  A pair of edges $u \rightarrow v, u^\prime \rightarrow v^\prime$ is called a removable edge pair for configuration $(F,\pi)$ with $F = (V,E, \height{\cdot}, \pweight{\cdot}{}, \qweight{\cdot}{})$  and $\pi=\{B_1, B_2, \dotsc, B_{|\pi|}\}$ if:
\begin{enumerate}
    \item $v,v^\prime \in \leaves{F}$, $v \neq v^\prime$.
    \item $\|\pweight{v}{}\|_1 = \|\pweight{v^\prime}{}\|_1 = 1$.
    \item $B_{\pi(v)} = B_{\pi(v^\prime)} = \{v,v^\prime\}$.
    \item $\pi(u) = \pi(u^\prime)$.
\end{enumerate}
\end{definition}
The definition of removable edge pair is a natural generalization of a notion \citep[Definition 11]{dudeja2022universality} introduced in our prior work to $\order$-trees of arbitrary order $\order \geq 1$ (our prior work considered the case $\order = 1$). On the other hand, the notion of a removable edge does not have a counter part in decorated trees of $\order = 1$ and is important only when order of tree is at least $2$ ($\order \geq 2$). 

In order to understand how removable edges and removable edge pairs can be used to simplify the polynomial $\Gamma(\mPsi_{1:k} ; F,\pi)$ consider the situation when a configuration $(F,\pi)$ has a removable edge pair $u \rightarrow v, u^\prime \rightarrow v^\prime$ with $\pweight{v}{} = e_i$ and $\pweight{v^\prime}{} = e_j$ for some $i,j \in [k]$ (here $e_{1:k}$ denote the standard basis of $\R^k$). Observe that the evaluation of $\Gamma(\mPsi; F,\pi)$ (cf. \eqref{eq:invariant-polynomial}) involves summing over the possible colors for $v,v^\prime$, which yields an expression of the form:
\begin{align*}
   \sum_{\substack{\ell_v, \ell_{v^\prime} \in [\dim] \\ \ell_{v} = \ell_{v^\prime}}}  (\Psi_i)_{\ell_u\ell_v} (\Psi_j)_{\ell_{u^\prime} \ell_{v^\prime}}   = (\mPsi_i \mPsi_j \tran)_{\ell_{u} \ell_{u^\prime}} \explain{(d)}{=} \Omega_{ij},
\end{align*}
where the equality (d) follows from \eqref{eq:simplifying-constraint} and the fact that $\ell_{u} = \ell_{u^\prime}$ for a pair of removable edges $u \rightarrow v, u^\prime \rightarrow v^\prime$. Since the sum over the possible colors for $v,v^\prime$ can be evaluated explicitly, the block $\{v, v^\prime\}$ can be effectively removed from the configuration $(F,\pi)$, thus simplifying its structure. A similar simplification occurs if a removable edge (\defref{removable-edge}) is present in the configuration. By eliminating every removable edge and every pair of removable edges in a relevant configuration $(F,\pi)$ one can express the corresponding polynomial $\Gamma(\mPsi; F, \pi)$ as a linear combination of polynomials associated with \emph{simple configurations}, which we introduce next.   

\begin{definition}\label{def:simple-configuration} A decorated $\order$-tree $F = (V, E, \height{\cdot}, \pweight{\cdot}{}, \qweight{\cdot}{})$ and a partition $\pi = \{B_1, B_2, \dotsc, B_{|\pi|}\}$ of $V$ form a simple configuration if they satisfy:
\begin{enumerate}
    \myitem{\texttt{Modified Leaf Property}}\label{modified-leaf-property}: Each leaf $v \in \leaves{F}$ with $|B_{\pi(v)}| = 1$ satisfies $\|\pweight{v}{}\|_1 \geq 4$.
   \myitem{\texttt{Paired Leaf Property}}\label{paired-leaf-property}: Any pair of distinct leaves $v, v^\prime \in \leaves{F}$  with $\|\pweight{v}{}\|_1 = \|\pweight{v^\prime}{}\|_1 = 1$ and $B_{\pi(v)} = B_{\pi(v^\prime)} = \{v,v^\prime\}$ satisfies $\pi(u) \neq \pi(u^\prime)$, where $u,u^\prime$ are the parents of $v,v^\prime$ respectively. 
    \myitem{\texttt{Strong Forbidden Weights Property}}\label{strong-forbidden-weights-property}: There are no vertices $u \in V \backslash \{0\}$ such that $|B_{\pi(u)}| = 1, \|\pweight{u}{}\|_1 = 1, \|\qweight{u}{}\|_1 = 1$ or $|B_{\pi(u)}| = 1, \|\pweight{u}{}\|_1 = 2, \|\qweight{u}{}\|_1 = 0$.
    \item{\ref{parity-property}} as described in \defref{relevant-config}. 
\end{enumerate}
\end{definition}

Simple configurations are maximally ``simplified'' in the sense that they do not have any removable edges or removable edge pairs (which could have been used to further simplify the structure of the configuration). Indeed, the \ref{modified-leaf-property} rules out the presence of a removable edge and the \ref{paired-leaf-property} rules out the presence of a removable edge pair. The following is the formal statement of our decomposition result, which shows that for any relevant configuration $(F,\pi)$ (\defref{relevant-config}), the polynomial $\Gamma(\mPsi_{1:\order}; F,\pi)$ (cf. \eqref{eq:invariant-polynomial}) associated with the relevant configuration can be expressed as a linear combination of a few (independent of dimension $\dim$) simple configurations (cf. \defref{simple-configuration}). 

\begin{proposition}[Decomposition Result]\label{prop:decomposition} For any relevant configuration $(F_0, \pi_0)$, there exists a collection $\calS$ of simple configurations with $|\calS| \leq |\pi| !$ and a map $\xi: \calS \mapsto [-1,1]$ such that:
\begin{align*}
    \Gamma(\mPsi_{1:\order}; F_0, \pi_0) & = \sum_{(F,\pi) \in \calS} \xi(F,\pi) \cdot  \Gamma(\mPsi_{1:\order}; F, \pi)
\end{align*}
\end{proposition}
The proof of this result is provided in \appref{decomposition}. 

\subsubsection{Universality of Simple Configurations}
The last ingredient in the proof of \thref{moments} is the following result, which shows that the limiting behavior of the polynomial $\Gamma(\mPsi_{1:\order}; F,\pi)$ (cf. \eqref{eq:invariant-polynomial}) associated with a simple configuration (\defref{simple-configuration}) $(F,\pi)$ is identical for any collection of matrices $\mPsi_{1:k}$ that satisfy the requirements of \thref{moments}. 
\begin{proposition}\label{prop:universality-simple} For any simple configuration $(F,\pi)$, 
$\lim_{\dim \rightarrow \infty} \Gamma(\mPsi_{1:\order}; F,\pi)  = 0$. 
\end{proposition}

The proof of \propref{universality-simple} relies on the following graph-theoretic result on the structure of simple configurations from our prior work \citep{dudeja2022universality} along with the improved estimate on polynomials associated with a configuration given in \propref{key-estimate}.

\begin{fact}[{\citep[Proposition 5]{dudeja2022universality}}]\label{fact:graph-theory} Let $F = (V, E, \height{\cdot}, \pweight{\cdot}{}, \qweight{\cdot}{})$ be a decorated $k$-tree with order $k=1$ and let $\pi$ be a partition of its vertex set $V$ such that $(F,\pi)$ form a simple configuration. Then, we have,
\begin{align*}
    \frac{|\nullleaves{F,\pi}|}{4} + 1 - |\pi| + \frac{1}{2} \sum_{v \in V \backslash \{0\}}  \pweight{v}{} \geq \frac{1}{4}.
\end{align*}
\end{fact}
\propref{universality-simple} follows immediately, given the above fact. 
\begin{proof}[Proof of \propref{universality-simple}] Using the estimate on $\Gamma(\mPsi_{1:\order}; F,\pi)$ stated in \propref{key-estimate} we have,
\begin{align*}
   |\Gamma({\mPsi}_{1:\order}; F, \pi)| & \lesssim \dim^{-\exponent(F,\pi) + \epsilon} \; \forall \; \epsilon > 0,
\end{align*}
where,
\begin{align*}
    \exponent(F,\pi) \explain{def}{=}  \frac{|\nullleaves{F,\pi}|}{4} + 1 - |\pi| + \frac{1}{2} \sum_{v \in V \backslash \{0\}}  \|\pweight{v}{}\|_1.
\end{align*}
Observe that the claim of the proposition follows if we show that $\exponent(F,\pi) \geq 1/4$. In order to show this, we will appeal to \factref{graph-theory}. A minor difficulty is that \factref{graph-theory} only applies when $F$ is a decorated $\order$-tree with $\order = 1$ and not for arbitrary $\order$. In order to address this issue, we consider the the decorated $1$-tree $\widetilde{F} = (V, E, \height{\cdot}, \widetilde{p}(\cdot), \widetilde{q}(\cdot))$ defined as follows:
\begin{enumerate}
    \item $\widetilde{F}$ has the same tree structure as $F$ that is, the same vertex set $V$ and the same edge set $E$ and the same height function $\height{\cdot}$. 
    \item For each $v \in V \backslash \{0\}$, we set $\widetilde{p}(v) = \|\pweight{v}{}\|_1$. 
    \item Similarly, for each $v \in V$, we set $\widetilde{q}(v) = \|\qweight{v}{}\|_1$. 
\end{enumerate}
It is straightforward to verify that $\widetilde{F} = (V, E, \height{\cdot}, \widetilde{p}(\cdot), \widetilde{q}(\cdot))$ satisfies all the requirements of \defref{decorated-tree} to be a decorated $1$-tree. Observe that $\pi$ is also partition of the vertex of $\widetilde{F}$. Moreover, it is also immediate from the assumption that $(F,\pi)$ was a simple configuration, $(\widetilde{F}, \pi)$ is also simple (in the sense of \defref{simple-configuration}). Hence by \factref{graph-theory}:
\begin{align*}
     \frac{|\nullleaves{\widetilde{F},\pi}|}{4} + 1 - |\pi| + \frac{1}{2} \sum_{v \in V \backslash \{0\}}  \widetilde{p}(v) = \frac{|\nullleaves{\widetilde{F},\pi}|}{4} + 1 - |\pi| + \frac{1}{2} \sum_{v \in V \backslash \{0\}}  \|\pweight{v}{}\|_1  \geq \frac{1}{4}.
\end{align*}
It follows immediately from the definition of nullifying edges (\defref{nullifying}) that an edge pair $(u \rightarrow v, u^\prime \rightarrow v^\prime)$ is a pair of nullifying edges in the configuration $(F,\pi)$  iff it is a pair of nullifying edges in the configuration $(\widetilde{F},\pi)$. Hence, $\nullleaves{\widetilde{F},\pi} = \nullleaves{{F},\pi}$. Hence we have shown that $\eta(F,\pi) \geq 1/4$, as desired. This concludes the proof.  
\end{proof}

\subsection{Proof of \thref{moments} and \corref{moments}} \label{appendix:moments-actual-proof}
We have now introduced all the key ingredients used to obtain \thref{moments} and \corref{moments}, which we prove below. 

\begin{proof}[Proof of \thref{moments}] Recall that in \lemref{expectation-formula} we computed:
\begin{align}\label{eq: E-formula-recall}
    &\E \left[ \frac{1}{\dim}\sum_{j=1}^\dim h(\auxvec_j) \cdot \hermite{r}(\iter{z}{T,\cdot}_j) \right]  = \sum_{\substack{F \in \trees{k}{r,T} \\ F = (V, E, \height{\cdot}, \pweight{\cdot}{}, \qweight{\cdot}{})}} \sum_{\pi \in \part{V}}\alpha(F)  \relev(F,\pi) \chi(F,\pi)  \beta(F,\pi) \cdot  \Gamma(\mPsi_{1:\order}; F,\pi).
\end{align}
Recall that $\relev(F,\pi)$ was the indicator function which is $1$ iff $(F,\pi)$ is a relevant configuration (cf. \defref{relevant-config}) and zero otherwise. Hence it, suffices to compute $\lim_{\dim \rightarrow \infty}\Gamma(\mPsi_{1:\order}; F,\pi)$ for relevant configuration. Recall from \propref{decomposition} that for any relevant configuration $(F,\pi)$ there exists a collection $\calS$ of simple configurations with $|\calS| \leq |\pi| !$ and a map $a: \calS \mapsto [-1,1]$ such that:
\begin{align*}
    \Gamma(\mPsi_{1:\order}; F_0, \pi_0) & = \sum_{(F,\pi) \in \calS} a(F,\pi) \cdot  \Gamma(\mPsi_{1:\order}; F, \pi)
\end{align*}
Since \propref{universality-simple} showed that $\lim_{\dim \rightarrow \infty}\Gamma(\mPsi_{1:\order}; F,\pi) = 0$ for any simple configuration, we also have $\lim_{\dim \rightarrow \infty}\Gamma(\mPsi_{1:\order}; F,\pi) = 0$ for any relevant configuration. As a consequence,
\begin{align*}
   \lim_{\dim \rightarrow \infty} \E \left[ \frac{1}{\dim}\sum_{j=1}^\dim h(\auxvec_j) \cdot \hermite{r}(\iter{z}{T,\cdot}_j) \right] & = 0,
\end{align*}
as claimed. 
\end{proof}

Next, we provide the proof for \corref{moments}.

\begin{proof}[Proof of \corref{moments}]
The proof follows the argument employed in the \citep[Proof of Corollary 2, Appendix E.4]{dudeja2022universality}. We briefly summarize the argument here for completeness. Let $\iter{\serv{Z}}{\dim}_1, \dotsc, \iter{\serv{Z}}{\dim}_k, \iter{\serv{A}}{\dim}$ denote the random variables with the law:
\begin{align*}
   {\mu}_\dim & \explain{def}{=} \E \left[ \frac{1}{\dim} \sum_{i=1}^\dim \delta_{\iter{z}{T,1}_i, \iter{z}{T,2}_i, \dotsc, \iter{z}{T,k}_i, \auxvec_i} \right].
\end{align*}
By choosing $h(\auxvec) = \auxvec_1^{d_1} \cdot \cdot \auxvec^{d_\auxdim}_{\auxdim}$ for $d_{1:\auxdim} \in \W^{\auxdim}$ in \thref{moments}, one obtains the conclusion that the random variables $(\iter{\serv{Z}}{\dim}_1, \dotsc, \iter{\serv{Z}}{\dim}_k, \iter{\serv{A}}{\dim})$ converge to $(\serv{Z}_1, Z_1, \dotsc, \serv{Z}_k,  \serv{A})$ in moments. \assumpref{side-info} guarantees that the distribution of $(\serv{Z}_1, Z_1, \dotsc, \serv{Z}_k,  \serv{A})$ is uniquely determined by its moments. Hence, we also have that $(\iter{\serv{Z}}{\dim}_1, \dotsc, \iter{\serv{Z}}{\dim}_k, \iter{\serv{A}}{\dim})$ converges in distribution to $(\serv{Z}_1, Z_1, \dotsc, \serv{Z}_k,  \serv{A})$. Since the test function $h$ is bounded by a polynomial and all moments of  $(\serv{Z}_1, Z_1, \dotsc, \serv{Z}_k,  \serv{A})$ one obtains $\E[h(\iter{\serv{Z}}{\dim}_1, \dotsc, \iter{\serv{Z}}{\dim}_k ; \iter{\serv{A}}{\dim})] \rightarrow \E[h(\serv{Z}_1, Z_1, \dotsc, \serv{Z}_k ;  \serv{A})]$ using the continuous mapping theorem and a uniform integrability argument. 
\end{proof}

\subsection{Proof of the Unrolling Lemma} \label{appendix:proof-hermite-expansion}
This subsection provides a proof of the Unrolling Lemma (\lemref{hermite-expansion}). Consider arbitrary $t \in \N$, $\qweight{0}{} = (\qweight{0}{1}, \qweight{0}{2}, \dotsc, \qweight{0}{k}) \in \W^\order$ with $\|\qweight{0}{}\|_1 \geq 1$ \footnote{We make unusual choice of using the variable name $\qweight{0}{}$ (instead of $r$ used in the statement of \lemref{hermite-expansion}) since the unrolling process will lead to the introduction of vectors $\qweight{1}{}, \qweight{2}{}, \dotsc \in \W^\order$.} and $j \in [N]$.
We begin by expressing
\begin{align*}
     \hermite{\qweight{0}{}}\left(\iter{z}{t,\cdot}_j\right) \explain{def}{=} \prod_{i=1}^\order \hermite{\qweight{0}{i}}\left( \iter{z}{t,i}_j \right)
\end{align*}
as a polynomial of the initialization $\iter{\mZ}{0,\cdot}$.  The expansion relies on the following property of Hermite polynomials.
\begin{fact}\label{fact:hermite-property} Let $\vu \in \R^\dim$ be such that $\|\vu\| = 1$. For any $q \in \W$ and any $\vx \in \R^\dim$, we have:
\begin{align*}
    \hermite{q}(\ip{\vu}{\vx}) & = \sum_{\substack{\valpha \in \W^\dim \\ \|\valpha\|_1 = q}} \sqrt{\binom{q}{\valpha}} \cdot \vu^{\valpha} \cdot \hermite{\valpha}(\vx).
\end{align*}
In the display above,
\begin{align*}
    \binom{q}{\valpha} \explain{def}{=} \frac{q!}{\alpha_1 ! \alpha_2! \dotsb \alpha_\dim!}, \; \vu^{\valpha} \explain{def}{=} \prod_{i=1}^\dim u_i^{\alpha_{i}}, \; \hermite{\valpha}(\vx) \explain{def}{=} \prod_{i=1}^\dim  \hermite{\alpha_i}(x_i).
\end{align*}
\end{fact}
The property stated above is easily derived using the well known generating formula for Hermite polynomials, see for e.g. \citep[Appendix F]{dudeja2022universality} for a proof. Using the formula in \factref{hermite-property}, we obtain:
\begin{align}
     \hermite{\qweight{0}{}}\left(\iter{z}{t,\cdot}_j\right) & = \prod_{i=1}^\order \hermite{\qweight{0}{i}} \left( \ip{\mM_i \tran \ve_j}{\nonlin_i(\iter{\mZ}{t-1,\cdot}; \auxmat)} \right) \nonumber\\
     & = \prod_{i=1}^\order  \sum_{\substack{\valpha_i \in \W^\dim \\ \|\valpha_i\|_1 = \qweight{0}{i}}} \sqrt{\binom{\qweight{0}{i}}{\valpha_i}} \cdot (\mM_i \tran \ve_j)^{\valpha_i} \cdot \hermite{\valpha_i} (\nonlin_i(\iter{\mZ}{t-1,\cdot}; \auxmat))  \nonumber \\
     & = \sum_{\substack{\valpha_{1:\order} \in \W^\dim \\ \|\valpha_i\|_1 = \qweight{0}{i}}} \prod_{i=1}^\order \left( \sqrt{\binom{\qweight{0}{i}}{\valpha_i}} \cdot (\mM_i \tran \ve_j)^{\valpha_i} \cdot \hermite{\valpha_i} (\nonlin_i(\iter{\mZ}{t-1,\cdot}; \auxmat))\right) \label{eq:expansion-intermediate-1}.
\end{align}
In the above display, $\ve_{1:\dim}$ denote the standard basis vectors in $\R^\dim$. Consider the following procedure to pick $\valpha_{1:\order} \in \W^\dim$ with $\|\valpha_i\|_1 = \qweight{0}{i}$:
\begin{enumerate}
    \item First, we pick $c = \|\valpha_1 + \valpha_2 + \dotsc + \valpha_\order \|_0$. Since $\|\valpha_i\|_1 = \qweight{0}{i}$ and $\|\qweight{0}{}\|_1 \geq 1$, $c \in [\|\qweight{0}{}\|_1]$.
    \item Next, we pick $\ell_{1:c} \in [\dim]$ with $\ell_1 < \ell_2< \dotsb < \ell_c$. These will be the locations of the non-zero coordinates of the vector $\valpha_1 + \valpha_2 + \dotsc + \valpha_\order$. 
    \item Then, we pick vectors $\pweight{1}{}, \pweight{2}{}, \dotsc, \pweight{c}{} \in \W^\order$ which satisfy:
    \begin{align*}
        \|\pweight{\cdot}{i}\|_1 \explain{def}{=} \sum_{v=1}^c \pweight{v}{i} = \qweight{0}{i}, \; \|\pweight{v}{}\|_1 \geq 1 \; \forall \; v \in [c].
    \end{align*}
    For each $i \in [\order]$, the vector $\pweight{\cdot}{i} \explain{def}{=} (\pweight{1}{i}, \pweight{2}{i}, \dotsc, \pweight{c}{i})$ will specify the values of the entries of $\valpha_i$ on the indices $\{\ell_1, \ell_2, \dotsc, \ell_c\}$.
    \item Finally, for each $i \in [\order]$ we set the entries of the vector $\valpha_i$ as follows:
    \begin{align*}
        (\alpha_i)_{\ell_v} &= \pweight{v}{i} \; \forall \; v \; \in \; [c], \\
        (\alpha_i)_{\ell} & = 0 \; \forall \; \ell \; \not\in \; \{\ell_1, \ell_2, \dotsc, \ell_c\}. 
    \end{align*}
\end{enumerate}
Observe that the vectors $\valpha_{1:\order} \in \W^\dim$ constructed this way satisfy $\|\valpha_i\|_1 = \qweight{0}{i}$ for each $i \in [\order]$. Moreover, any collection of vectors $\valpha_{1:\order}$ that satisfy $\|\valpha_i\|_1 = \qweight{0}{i}$ for each $i \in [\order]$ can be obtained using the above procedure. Hence, \eqref{eq:expansion-intermediate-1} can be rewritten as:
\begin{align}
    &\hermite{\qweight{0}{}}\left(\iter{z}{t,\cdot}_j\right)  = \sum_{c=1}^{\|\qweight{0}{}\|_1} \sum_{\substack{\pweight{1}{}, \dotsc, \pweight{c}{} \in \W^\order \\  \|\pweight{\cdot}{i}\|_1 = \qweight{0}{i} \\ \|\pweight{v}{}\|_1 \geq 1 }} \sum_{\substack{\ell_{1:c} \in [\dim] \\ \ell_1 < \ell_2 < \dotsb < \ell_c}} \prod_{i=1}^\order \left( \sqrt{\binom{\qweight{0}{i}}{\pweight{\cdot}{i}}} \cdot \prod_{v=1}^c  (M_i)_{j,\ell_v}^{\pweight{v}{i}} \cdot \hermite{\pweight{v}{i}}(\nonlin_i(\iter{z}{t-1,\cdot}_{\ell_v}; \auxvec_{\ell_v}))   \right) \nonumber \\
    & = \sum_{c=1}^{\|\qweight{0}{}\|_1} \sum_{\substack{\pweight{1}{}, \dotsc, \pweight{c}{} \in \W^\order \\  \|\pweight{\cdot}{i}\|_1 = \qweight{0}{i} \\ \|\pweight{v}{}\|_1 \geq 1 }} \left( \frac{1}{c!} \prod_{i=1}^\order \sqrt{\binom{\qweight{0}{i}}{\pweight{\cdot}{i}}}  \right) \cdot \sum_{\substack{\ell_{1:c} \in [\dim] \\ \ell_u \neq \ell_v \forall u \neq v}} \prod_{v=1}^c \left(\prod_{i=1}^\order (M_i)_{j,\ell_v}^{\pweight{v}{i}} \right) \cdot \underbrace{\left( \prod_{i=1}^{\order} \hermite{\pweight{v}{i}}(\nonlin_i(\iter{z}{t-1,\cdot}_{\ell_v}; \auxvec_{\ell_v})) \right)}_{(\star)}. \label{eq:expansion-intermediate-2}
\end{align}
Next, we consider the Hermite decomposition of the function $(\star)$. In the above display. For any $p \in \W^\order$ and any $\auxvec \in \R^{\auxdim}$ consider the function:
\begin{align*}
    z \in \R^\order \mapsto \prod_{i=1}^\order \hermite{p_i}(\nonlin_i(z; \auxvec)).
\end{align*}
We consider the Hermite decomposition of this function:
\begin{align} \label{eq:hermite-decomposition}
    \prod_{i=1}^\order \hermite{p_i}(\nonlin_i(z; \auxvec)) & = \sum_{q \in \W^\order} \fourier{p}{q}{\auxvec} \cdot \hermite{q}(z), 
\end{align}
where
\begin{align*}
     \fourier{p}{q}{\auxvec} & \explain{def}{=} \E \left[ \hermite{q}(\serv{Z}) \cdot  \prod_{i=1}^\order \hermite{p_i}(\nonlin_i(\serv{Z}; \auxvec)) \right], \; \serv{Z} \sim \gauss{0}{I_\order}. 
\end{align*}
Applying this decomposition to the function $(\star)$ in \eqref{eq:expansion-intermediate-2} we obtain:
\begin{align*}
     &\hermite{\qweight{0}{}}\left(\iter{z}{t,\cdot}_j\right) \\&=  \sum_{c=1}^{\|\qweight{0}{}\|_1} \sum_{\substack{\pweight{1}{}, \dotsc, \pweight{c}{} \in \W^\order \\  \|\pweight{\cdot}{i}\|_1 = \qweight{0}{i} \\ \|\pweight{v}{}\|_1 \geq 1 }}  \alpha_0(\qweight{0}{}, \pweight{1:c}{}) \sum_{\substack{\ell_{1:c} \in [\dim] \\ \ell_u \neq \ell_v \forall u \neq v}} \prod_{v=1}^c \left\{\left(\prod_{i=1}^\order (M_i)_{j,\ell_v}^{\pweight{v}{i}} \right) \sum_{\qweight{v}{} \in \W^\order} \fourier{\pweight{v}{}}{\qweight{v}{}}{\auxvec_{\ell_v}}  \hermite{\qweight{v}{}}(\iter{z}{t-1,\cdot}_{\ell_v}) \right\},
\end{align*}
where:
\begin{align*}
    \alpha_0(\qweight{0}{}, \pweight{1}{}, \dotsc, \pweight{c}{}) \explain{def}{=}\left( \frac{1}{c!} \prod_{i=1}^\order \sqrt{\binom{\qweight{0}{i}}{\pweight{\cdot}{i}}}  \right)
\end{align*}
Hence,
\begin{align}
    &\hermite{\qweight{0}{}}\left(\iter{z}{t,\cdot}_j\right) = \nonumber \\& \sum_{c=1}^{\|\qweight{0}{}\|_1} \sum_{\substack{\pweight{1}{}, \dotsc, \pweight{c}{} \in \W^\order \\  \|\pweight{\cdot}{i}\|_1 = \qweight{0}{i} \\ \|\pweight{v}{}\|_1 \geq 1 }}  \sum_{\qweight{1}{}, \dotsc, \qweight{c}{} \in \W^\order}\alpha_0(\qweight{0}{}, \pweight{1:c}{}) \sum_{\substack{\ell_{1:c} \in [\dim] \\ \ell_u \neq \ell_v \forall u \neq v}} \prod_{v=1}^c  \left(\prod_{i=1}^\order (M_i)_{j,\ell_v}^{\pweight{v}{i}} \right) \fourier{\pweight{v}{}}{\qweight{v}{}}{\auxvec_{\ell_v}}  \hermite{\qweight{v}{}}(\iter{z}{t-1,\cdot}_{\ell_v})  \label{eq:expansion-recursive-formula}
\end{align}
We can now apply the formula in \eqref{eq:expansion-recursive-formula} recursively to expand $ \hermite{\qweight{v}{}}(\iter{z}{t-1,\cdot}_{\ell_v})$ for every $v \in [c]$ as a polynomial in $\iter{\mZ}{t-2,\cdot}$. We continue this process to obtain a expansion of $\hermite{\qweight{0}{}}\left(\iter{z}{t,\cdot}_j\right)$ as a polynomial in $\iter{\mZ}{0,\cdot}$. This immediately yields the claimed formula in \lemref{hermite-expansion} which takes the form of a combinatorial sum over colorings of decorated $\order$-trees (defined in \defref{decorated-tree}). We highlight the following  aspects of the definition of decorated forests (\defref{decorated-tree}) and valid colored decorated forests  (\defref{valid-tree}) that play an important role in ensuring that the formula in \lemref{hermite-expansion} is correct:
\begin{enumerate}
    \item In \defref{decorated-tree}, the height function $h$ keeps track of the extent to which the iterations have been unrolled: property (a) of $h$ captures the fact that each step of unrolling expresses the coordinates  of $\iter{\mZ}{t,\cdot}$ as a polynomial in $\iter{\mZ}{t-1,\cdot}$, property (b) captures the fact that the unrolling process stops once a polynomial in $\iter{\mZ}{0,\cdot}$ is obtained and property (c) ensures that the unrolling process continues till every non-trivial polynomial in the iterates has been expressed in terms of the initialization $\iter{\mZ}{0,\cdot}$. 
    \item In \defref{valid-tree}, the second requirement (no two siblings have the same color) captures the $\ell_2 \neq \ell_3 \neq \dotsb$ constraint that appears in \eqref{eq:expansion-recursive-formula}.
\end{enumerate}

\subsection{Proof of the Expectation Formula} \label{appendix:expectation-formula}
This subsection presents the proof of the expectation formula provided in \lemref{expectation-formula}. 
\begin{proof}[Proof of \lemref{expectation-formula}] We begin by recalling \eqref{eq:rearrangement}:
\begin{align} \label{eq:E-formula-intermediate}
    &\E \left[ \frac{1}{\dim}\sum_{j=1}^\dim h(\auxvec_j) \cdot \hermite{r}(\iter{z}{T,\cdot}_j) \right]  = \nonumber \\& \hspace{0.8cm} \sum_{\substack{F \in \trees{k}{r,T} \\ F = (V, E, \height{\cdot}, \pweight{\cdot}{}, \qweight{\cdot}{})}} \sum_{\pi \in \part{V}}\alpha(F) \cdot \valid(F,\pi) \cdot \sum_{\substack{\ell \in \colorings{}{\pi} }}  \E[\chi(\auxmat; F, \ell)] \cdot \E[\gamma(\mM_{1:\order} ; F, \ell)] \cdot  \E[ \beta(\iter{\mZ}{0,\cdot}; F, \ell)]. 
\end{align}
Next, we compute the expectations $\E[\chi(\auxmat; F, \ell)]$, $\E[\gamma(\mM_{1:\order} ; F, \ell)]$ and $\E[ \beta(\iter{\mZ}{0,\cdot}; F, \ell)]$.

\noindent\underline{\textbf{Analysis of $\E[\chi(\auxmat; F, \ell)]$.}}  Recall from \lemref{hermite-expansion} that:
\begin{align*}
    \E[\chi(\auxmat; F, \ell)] & \explain{def}{=} \E \left[ h(\auxvec_{\ell_0}) \cdot  \prod_{\substack{u \in V \\ u \neq 0}} \fourier{\pweight{u}{}}{\qweight{u}{}}{\auxvec_{\ell_u}} \right] \\
    & \explain{(a)}{=}  \E \left[ \left( h(\auxvec_{\ell_0}) \prod_{\substack{u \in B_1 \\ u \neq 0}} \fourier{\pweight{u}{}}{\qweight{u}{}}{a_{\ell_u}} \right) \cdot \prod_{i=2}^{|\pi|} \left( \prod_{\substack{u \in B_i }} \fourier{\pweight{u}{}}{\qweight{u}{}}{a_{\ell_u}}  \right) \right] \\
    & \explain{(b)}{=} \E \left[ h(\serv{A}) \cdot  \prod_{\substack{u \in B_1 \\ u \neq 0}} \fourier{\pweight{u}{}}{\qweight{u}{}}{\serv{A}} \right] \cdot \prod_{i=2}^{|\pi|} \E \left[ \prod_{\substack{u \in B_i}} \fourier{\pweight{u}{}}{\qweight{u}{}}{\serv{A}} \right] \\
    & \explain{def}{=} \chi(F,\pi). 
\end{align*}
In the above display, in step (a), we grouped the vertices that lie in the same block in $\pi$ together. Step (b) uses the fact that since $\ell \in \colorings{}{\pi}$, two vertices have the same color iff they lie in the same block of the $\pi$ along with the assumption that the rows of $\auxmat$ are i.i.d. copies of the random variable $\serv{A}$ (\assumpref{side-info}). Note that if $(F,\pi)$ violates the \ref{weak-forbidden-weights-property}, then there is a vertex $u \neq 0$ with $B_{\pi(u)} = \{u\}$ such that one of the following is true:
\begin{description}
\item [Case 1:] $\|\pweight{u}{}\|_1 = \|\qweight{u}{}\|_1 = 1$. This means that $\pweight{u}{} = e_i$ and $\qweight{u}{} = e_j$ for some $i,j \in [k]$. Recalling the formula for $\fourier{\pweight{u}{}}{\qweight{u}{}}{\auxvec}$ from \eqref{eq:hermite-expansion-notations} we obtain,
\begin{align*}
    \E[\fourier{\pweight{u}{}}{\qweight{u}{}}{\serv{A}}] &\explain{}{=} \E[\nonlin_i(\serv{Z}; \serv{A})  \cdot \serv{Z}_j]  \explain{(c)}{=} 0.
\end{align*}
In the above display, $\serv{Z} \sim \gauss{0}{I_k}$ is independent of $\serv{A}$ and step (c) follows from the assumption that the non-linearities $\nonlin_{1:\order}$ are divergence-free made in the statement of \thref{moments} (cf. \eqref{eq:poly-nonlinearity}). Consequently, since $B_{\pi(u)} = \{u\}$, we have $\chi(F,\pi) = 0$ in this case. 
\item [Case 2:] $\pweight{u}{} = 2 e_i, \qweight{u}{} = 0$ for some $i \in [k]$. As before, recalling the formula for $\fourier{\pweight{u}{}}{\qweight{u}{}}{\auxvec}$ from \eqref{eq:hermite-expansion-notations} and the assumption $\E[(\nonlin_i^2(\serv{Z}; \serv{A})] = 1$ made in the statement of \thref{moments} (cf. \eqref{eq:poly-nonlinearity}), we have, $ \E[\fourier{\pweight{u}{}}{\qweight{u}{}}{\serv{A}}] = \E[(\nonlin_i^2(\serv{Z}; \serv{A})-1)  \cdot 1]/\sqrt{2} = 0$. Hence, again, $\chi(F,\pi) = 0$. 
\item [Case 3:]  $\pweight{u}{} = e_i + e_j, \qweight{u}{} = 0$ for some $i,j \in [\order]$ with $i \neq j$ and $\Omega_{ij} \neq 0$. By \eqref{eq:hermite-expansion-notations}, $$\E[\fourier{\pweight{u}{}}{\qweight{u}{}}{\serv{A}}] = \E[ \nonlin_i(\serv{Z}; \serv{A})  \nonlin_j(\serv{Z}; \serv{A})].$$
Recall that \thref{moments} assumes that $\Omega_{ij} \cdot \E[ \nonlin_i(\serv{Z}; \serv{A})  \nonlin_j(\serv{Z}; \serv{A})] = 0$ (cf. \eqref{eq:poly-nonlinearity}). Since $\Omega_{ij} \neq 0$, we have $\E[ \nonlin_i(\serv{Z}; \serv{A})  \nonlin_j(\serv{Z}; \serv{A})]$ = 0. Hence, again, $\chi(F,\pi) = 0$ in this case. 
\end{description}
To summarize, we have shown that:
\begin{align}
     \E[\chi(\auxmat; F, \ell)] & = \chi(F,\pi), \label{eq:chi-E}\\
     (F,\pi) \text{ violates \ref{weak-forbidden-weights-property}} &\implies  \E[\chi(\auxmat; F, \ell)] = 0. \label{eq:chi-0}
\end{align}

\noindent\underline{\textbf{Analysis of $\E[\beta(\iter{\mZ}{0,\cdot}; F, \ell)]$.}} Using the same argument we can also compute:
\begin{align} \label{eq:beta-E}
   \E[\beta(\iter{\mZ}{0,\cdot}; F, \ell)] = \E\left[ \prod_{\substack{v \in \leaves{F} \\ \|\qweight{v}{}\|_1 \geq 1}} \hermite{\qweight{v}{}}( \iter{z}{0, \cdot }_{\ell_v}) \right] & = \prod_{i=1}^{|\pi|} \E\left[ \prod_{\substack{v \in B_i \cap \leaves{F} \\ \|\qweight{v}{}\|_1 \geq 1}} \hermite{\qweight{v}{}}(\serv{Z}) \right] \explain{def}{=} \beta(F,\pi), \; \serv{Z} \sim \gauss{0}{I_{\order}}. 
\end{align}
Consider the situation in which $(F,\pi)$ violates the \ref{original-leaf-property}. This means that there is a leaf vertex $v \in \leaves{F}$ with $\|\qweight{v}{}\|_1 \geq 1$ and $B_{\pi(v)} = \{v\}$. Since $\E[\hermite{\qweight{v}{}}(\serv{Z})] = 0$, we have $\beta(F,\pi) = 0$. To conclude, we have shown that,
\begin{align} \label{eq:beta-0}
     (F,\pi) \text{ violates \ref{original-leaf-property}} &\implies  \E[\beta(\iter{\mZ}{0,\cdot}; F, \ell)] = 0.
\end{align}

\noindent\underline{\textbf{Analysis of $\E[\gamma(\mM_{1:k}; F,\ell)]$.}} Finally, we compute $\E[\gamma(\mM_{1:k}; F,\ell)]$. Since $\mM_{1:k}$ is a semi-random ensemble (\defref{semirandom}), $\mM_i = \mS \mPsi_i \mS$ where $\mS = \diag(s_{1:\dim})$ with $s_{1:\dim} \explain{i.i.d.}{\sim} \unif{\{\pm 1\}}$. Hence,
\begin{align*}
    \E[\gamma(\mM_{1:k}; F,\ell)]  &= \gamma(\mPsi_{1:k}; F, \ell) \cdot \E \left[ \prod_{\substack{e \in E \\ e = u \rightarrow v}} s_{\ell_u}^{\|\pweight{v}{}\|_1} s_{\ell_v}^{\|\pweight{v}{}\|_1} \right]. 
\end{align*}
We can compute:
\begin{align*}
     &\E \left[ \prod_{\substack{e \in E \\ e = u \rightarrow v}} s_{\ell_u}^{\|\pweight{v}{}\|_1} s_{\ell_v}^{\|\pweight{v}{}\|_1} \right] \\ &\qquad \qquad  \explain{(c)}{=} \E \left[\left(\prod_{v: 0 \rightarrow v} s_{\ell_0}^{\|\pweight{v}{}\|_1} \right) \cdot \left( \prod_{u \in V \backslash ({\leaves{F} \cup \{0\}})} s_{\ell_u}^{\|\pweight{u}{}\|_1} \prod_{v: u \rightarrow v} s_{\ell_u}^{\|\pweight{v}{}\|_1} \right) \cdot \left( \prod_{w \in \leaves{F}} s_{\ell_w}^{\|\pweight{w}{}\|_1} \right) \right] \\
     &\qquad \qquad \explain{(d)}{=} \E\left[ s_{\ell_0}^{\|\qweight{0}{}\|_1} \cdot \left( \prod_{u \in V \backslash ({\leaves{F} \cup \{0\}})} s_{\ell_u}^{\|\pweight{u}{}\|_1+ \|\qweight{u}{}\|_1} \right) \cdot \left( \prod_{w \in \leaves{F}} s_{\ell_w}^{\|\pweight{w}{}\|_1} \right)  \right]. 
\end{align*}
In the above display, in step (c), we reorganized the product over edges in order to collect the signs variables corresponding to the same node together. We made the distinction between the root vertex (no parent), leaf vertices (no children), and all other vertices (have a parent and one or more children). In the step marked (d), we recalled the conservation equation \eqref{eq:conservation-eq},
\begin{align*}
    \qweight{u}{i} & = \sum_{v: u \rightarrow v } \pweight{v}{i}, \; \forall \; i \; \in \; [\order] \; u \; \in \;  V \backslash \leaves{F}. 
\end{align*}
In particular,
\begin{align*}
    \|\qweight{u}{}\|_1 & = \sum_{v: u \rightarrow v } \|\pweight{v}{}\|_1 \; \forall \;  u \; \in \;  V \backslash \leaves{F},
\end{align*}
which gives the equality in step (d). Finally by grouping the vertices that lie in the same block together as before, we obtain:
\begin{align} \label{eq:gamme-E}
     \E[\gamma(\mM_{1:k}; F,\ell)] & = \begin{cases} \gamma(\mPsi_{1:k}; F,\ell)  &: \text{if $(F,\pi)$ satisfies \ref{parity-property}}, \\ 0 &: \text{otherwise}. \end{cases}
\end{align}

\noindent\underline{\textbf{Conclusion of the proof.}} Combining \eqref{eq:E-formula-intermediate} together with the conclusions obtained in \eqref{eq:chi-E}, \eqref{eq:chi-0}, \eqref{eq:beta-E}, \eqref{eq:beta-0}, and \eqref{eq:gamme-E} immediately yields the claim of the lemma. 
\end{proof}

\subsection{Proof of the Improved Estimate} \label{appendix:key-estimate}
This subsection is devoted to the proof of the improved estimated stated in \propref{key-estimate}. As mentioned previously, \propref{key-estimate} is a generalization of  \citep[Proposition 2]{dudeja2022universality} obtained in our prior work, which provided an upper bound on $\Gamma(\mPsi_{1:k}; F,\pi)$ in the special situation when $F$ is a decorated tree with order $\order = 1$. The proof of \propref{key-estimate} closely follows the proof of \citep[Proposition 2]{dudeja2022universality}. In particular, we will rely on combinatorial result from this work \citep[Lemma 8]{dudeja2022universality}, reproduced below for convenience. 

\begin{fact}[{\citep[Lemma 8]{dudeja2022universality}}]\label{fact:mobius} Let $\iter{\vu}{1}, \iter{\vu}{2}, \dotsc, \iter{\vu}{r}$ be a collection of $r\in \N$ vectors in $\R^\dim$. We have,
\begin{align*}
    \left|\sum_{\substack{\ell_{1:r} \in [\dim] \\ \ell_i \neq \ell_j \; \forall \; i \neq j}} \prod_{i=1}^r \iter{u}{i}_{\ell_i} \right| & \leq r^{2r} \cdot \max(\sqrt{\dim} U_\infty, \min(\overline{U},\dim U_\infty))^r,
\end{align*}
where,
\begin{align*}
    \overline{U} \explain{def}{=} \max_{i \in [r]} \left| \sum_{j=1}^\dim \iter{u}{i}_j \right|, \; U_\infty \explain{def}{=} \max_{i \in [r]} \|\iter{\vu}{i}\|_\infty
\end{align*}
\end{fact}
We now present the proof of \propref{key-estimate}. 

\begin{proof}[Proof of \propref{key-estimate}] Let $r$ denote the number of pairs of nullifying leaves in configuration $(F,\pi)$. We label the nullifying leaves as:
\begin{align*}
    \nullleaves{F,\pi} & = \{v_1, v_1^\prime, v_2, v_2^\prime, \dotsc, v_r, v_r^\prime\}. 
\end{align*}
where $v_i, v_i^\prime$ form a pair of nullifying leaves in the sense of \defref{nullifying}. Let $u_i$ and $u_i^\prime$ denote the parents of $v_i$ and $v_i^\prime$ respectively. Recalling \defref{nullifying}, we see that the edges $e_i \bydef u_i \rightarrow v_i, \; e_i^\prime \bydef u_i^\prime \rightarrow v_i^\prime$ form a pair of nullifying edges. Since $\|\pweight{v_i}{}\|_1 = \|\pweight{v_i^\prime}{}\|_1 = 1 $ for each $i \in [r]$, this means that $\pweight{v_i}{}, \pweight{v_i^\prime}{} \in \{e_1, e_2, \dotsc, e_\order\}$ where $e_{1:\order}$ are the standard basis vectors in $\R^k$. Let $\pweight{v_i}{} = e_{t_i}$ and $\pweight{v_i^\prime}{} = e_{t_i^\prime}$ for each $i \in [r]$ for some $t_{1}, t_{2}, \dotsc, t_{r}, t_{1}^\prime, t_2^\prime, \dotsc, t_{r}^\prime \in [\order]$. Let $a_i \in [\dim]$ denote the color assigned to block $B_i$ by $\ell$. Then we can write,
\begin{align*}
    \gamma(\mPsi_{1:\order}; F, \ell)  &= \frac{1}{\dim} \prod_{\substack{e\in E\\ e = u \rightarrow v}} \prod_{i=1}^\order (\Psi_i)_{a_{\pi(u)},a_{\pi(v)}}^{\pweight{v}{i}}, \\
    \sum_{\ell \in \colorings{}{\pi}} \gamma(\mPsi_{1:\order}; F, \ell) & = \sum_{\substack{a_1, a_2, \dotsc, a_{|\pi|}\\ a_i \neq a_j \forall i \neq j}} \frac{1}{\dim} \prod_{\substack{e\in E\\ e = u \rightarrow v}} \prod_{i=1}^\order (\Psi_i)_{a_{\pi(u)},a_{\pi(v)}}^{\pweight{v}{i}}.
\end{align*} 
Notice that, without loss of generality, we can assume that $\{(v_i, v_i^\prime): i \in [r]\}$ form the last $r$ blocks of $\pi$. That is, $B_{i+|\pi|-r} = \{v_i, v_i^\prime \}$ for each $i \in [r]$. The definition of nullifying edges (\defref{nullifying}) guarantees that for each $i \in [r]$ the color $a_{|\pi| - r + i}$ appears exactly twice in the product $\gamma(\mPsi_{1:\order}; F, \ell)$: once with the edge $u_i \rightarrow v_i$, and the second time with the edge $u_i^\prime \rightarrow v_i^\prime$. We isolate the occurrences of these colors as follows:
\begin{align*}
     \gamma(\mPsi_{1:\order}; F, \ell) & = \widetilde{\gamma}(\mPsi_{1:\order}; F, a_{1}, a_{2}, \dotsc, a_{|\pi|-r}) \cdot \prod_{i=1}^r \left\{(\Psi_{t_i})_{a_{\pi(u_i)}, a_{i+|\pi|-r}} (\Psi_{t_i^\prime})_{a_{\pi(u_i^\prime)}, a_{i+|\pi|-r}} \right\},
\end{align*}
where $\widetilde{\gamma}(\mPsi_{1:\order}; F, a_{1}, a_{2}, \dotsc, a_{|\pi|-r})$ is defined as follows:
\begin{align*}
    \widetilde{\gamma}(\mPsi_{1:\order}; F, a_{1}, a_{2}, \dotsc, a_{|\pi|-r}) &\bydef \frac{1}{\dim}\prod_{\substack{e \in E \backslash \{e_{1:r}, e^\prime_{1:r}\}\\ e = u \rightarrow v}} \prod_{i=1}^{\order} (\Psi_i)_{a_{\pi(u)},a_{\pi(v)}}^{\pweight{v}{i}}.
\end{align*}
By defining the indices $b_i = a_{i+|\pi| - r}$ for $i \in [r]$, the above expression can be written as:
\begin{align*}
     &\left|\sum_{\ell \in \colorings{}{\pi}} \gamma(\mPsi_{1:\order}; F, \ell) \right| \\&= \left|\sum_{\substack{a_1,  \dotsc, a_{|\pi| - r} \in [\dim] \\ a_i \neq a_j \forall i \neq j}} \widetilde{\gamma}(\mPsi_{1:\order}; F, a_1, \dotsc, a_{|\pi| - r}) \cdot \sum_{\substack{b_{1:r} \in [\dim] \backslash \{a_{1}, \dotsc, a_{|\pi| - r}\}\\ b_i \neq b_j \forall i \neq j}} \prod_{i=1}^r \left\{(\Psi_{t_i})_{a_{\pi(u_i)}, b_i} (\Psi_{t_i^\prime})_{a_{\pi(u_i^\prime)}, b_i} \right\} \right| \\
     & \leq \sum_{\substack{a_1,  \dotsc, a_{|\pi| - r} \in [\dim] \\ a_i \neq a_j \forall i \neq j}} |\widetilde{\gamma}(\mPsi_{1:\order}; F, a_1, \dotsc, a_{|\pi| - r})|  \left| \sum_{\substack{b_{1:r} \in [\dim] \backslash \{a_{1}, \dotsc, a_{|\pi| - r}\}\\ b_i \neq b_j \forall i \neq j}} \prod_{i=1}^r \left\{(\Psi_{t_i})_{a_{\pi(u_i)}, b_i} (\Psi_{t_i^\prime})_{a_{\pi(u_i^\prime)}, b_i} \right\} \right|.
\end{align*}
Next, we define:
\begin{align*}
    \widetilde{\lambda}(\mPsi_{1:\order}; F, a_1, \dotsc, a_{|\pi| - r}) \bydef \sum_{\substack{b_{1:r} \in [\dim] \backslash \{a_{1}, \dotsc, a_{|\pi| - r}\}\\ b_i \neq b_j \forall i \neq j}} \prod_{i=1}^r \left\{(\Psi_{t_i})_{a_{\pi(u_i)}, b_i} (\Psi_{t_i^\prime})_{a_{\pi(u_i^\prime)}, b_i} \right\},
\end{align*}
and rewrite the previously obtained bound as,
\begin{align}
     &\left|\sum_{\ell \in \colorings{}{\pi}} \gamma(\mPsi_{1:\order}; F, \ell) \right|  \leq \sum_{\substack{a_1,  \dotsc, a_{|\pi| - r} \in [\dim] \\ a_i \neq a_j \forall i \neq j}} |\widetilde{\gamma}(\mPsi_{1:\order}; F, a_1, \dotsc, a_{|\pi| - r})| \cdot |\widetilde{\lambda}(\mPsi_{1:\order}; F, a_1, \dotsc, a_{|\pi| - r})|  \nonumber \\ 
     & \leq \dim^{|\pi| - r} \cdot \left( \max_{\substack{a_1, \dotsc, a_{|\pi|-r} \in [\dim] \\ a_i \neq a_j \forall i \neq j}} |\widetilde{\gamma}(\mPsi_{1:\order}; F, a_1, \dotsc, a_{|\pi| - r})| \right) \cdot \left( \max_{\substack{a_1, \dotsc, a_{|\pi|-r} \in [\dim] \\ a_i \neq a_j \forall i \neq j}} |\widetilde{\lambda}(\mPsi_{1:\order}; F, a_1, \dotsc, a_{|\pi| - r})| \right).  \label{eq:universal-forest-intermediate-estimate}
\end{align}
To prove the claim of the proposition, we need to bound $|\widetilde{\gamma}(\mPsi_{1:\order}; F, a_{1:|\pi| - r})|$ and $|\widetilde{\lambda}(\mPsi_{1:\order}; F, a_{1:|\pi| - r})|$. We will use the following elementary bound on $|\widetilde{\gamma}|$:
\begin{subequations} \label{eq:gamma-tilde-bound}
\begin{align}
    \left| \widetilde{\gamma}(\mPsi_{1:\order}; F, a_1, \dotsc, a_{|\pi| - r})\right| & \bydef \left|  \frac{1}{\dim}\prod_{\substack{e \in E \backslash \{e_{1:k}, e^\prime_{1:k}\}\\ e = u \rightarrow v}} \prod_{i=1}^\order (\Psi_i)_{a_{\pi(u)},a_{\pi(v)}}^{\pweight{v}{i}} \right|  \leq \left(\max_{i \in [\order]}\|\mPsi_i\|_\infty\right)^{\alpha} \cdot \dim^{-1},
\end{align}
where,
\begin{align}
    \alpha \bydef \sum_{v \in V \backslash (\{0\}\cup \nullleaves{F,\pi})} \|\pweight{v}{}\|_1.
\end{align}
Hence,
\begin{align}
       \max_{\substack{a_1, \dotsc, a_{|\pi|-r} \in [\dim] \\ a_i \neq a_j \forall i \neq j}} \left| \widetilde{\gamma}(\mPsi_{1:\order}; F, a_1, \dotsc, a_{|\pi| - r})\right| & \lesssim \dim^{-\alpha/2 - 1 + \epsilon}
\end{align}
\end{subequations}
To control $|\widetilde{\lambda}(\mPsi_{1:\order}; F, a_1, \dotsc, a_{|\pi| - r})|$ we take advantage of the following property of a semi-random ensemble (\defref{semirandom}): $$ \|\mPsi_i \mPsi_j\tran - \Omega_{ij} \cdot \mI_\dim\|_\infty  \lesssim \dim^{-\frac{1}{2} + \epsilon} \; \forall \; i,j \; \in \; [\order].$$ 
This will be done by appealing to \factref{mobius} for a suitable choice of vectors $\iter{\vu}{1:r}\in \R^{\dim -(|\pi| -r)}$. We will index the entries of these vectors using the set $[\dim] \backslash \{a_1, a_2, \dotsc, a_{|\pi| - r}\}$. The entries of these vectors are defined as follows:
\begin{align*}
    \iter{u}{i}_j \bydef (\Psi_{t_i})_{a_{\pi(u_i)},j} \cdot (\Psi_{t_i^\prime})_{a_{\pi(u_i^\prime)},j} \; \forall \; j \; \in \; [\dim] \backslash \{a_1, a_2, \dotsc, a_{|\pi| - r}\}.
\end{align*}
To apply \factref{mobius}, we bound $U_\infty$ and $ \overline{U}$ as follows:
\begin{align*}
    U_\infty &\bydef \max_{i \in [r]} \|\iter{\vu}{i}\|_\infty \leq \left(\max_{i \in [\order]}\|\mPsi_i\|_\infty \right)^2 \lesssim \dim^{-1 + \epsilon},
\end{align*}
and,
\begin{align*}
    \overline{U} &\bydef \max_{i \in [r]} \left| \sum_{j\in [\dim] \backslash \{a_1, a_2, \dotsc, a_{|\pi| - r}\}} \iter{u}{i}_j\right| \\
    & = \max_{i \in [r]} \left| \sum_{j \in [\dim] \backslash \{a_1, a_2, \dotsc, a_{|\pi| - r}\}} (\Psi_{t_i})_{a_{\pi(u_i)},j} \cdot (\Psi_{t_i^\prime})_{a_{\pi(u_i^\prime)},j}\right| \\
    & \leq \max_{i \in [r]} \left(\left| \sum_{j =1}^\dim (\Psi_{t_i})_{a_{\pi(u_i)},j} \cdot (\Psi_{t_i^\prime})_{a_{\pi(u_i^\prime)},j}\right|  + \sum_{j \in \{a_1, a_2, \dotsc, a_{|\pi| - r}\}} |(\Psi_{t_i})_{a_{\pi(u_i)},j}| \cdot |(\Psi_{t_i^\prime})_{a_{\pi(u_i^\prime)},j}| \right) \\
    & \explain{}{\leq} \max_{i \in [r]} \left| (\mPsi_{t_i}\mPsi_{t_i^\prime}\tran)_{a_{\pi(u_i)},a_{\pi(u_i^\prime)}} \right| + (|\pi|-r) \cdot \max_{t \in [\order]} \|\mPsi_t\|_\infty^2 \\
    & \explain{(a)}{\leq} \max_{\substack{i,j \in [\dim] \\ i \neq j}} \max_{s,t \in [\order]}   |(\mPsi_s\mPsi_t\tran)_{ij}| + (|\pi|-r) \cdot \max_{t \in [\order]} \|\mPsi_t\|_\infty^2\\
    & \leq \max_{s,t \in [\order]} \|\mPsi_s\mPsi_t\tran - \Omega_{st} \cdot \mI_\dim \|_\infty + (|\pi|-r) \cdot  \max_{t \in [\order]}\|\mPsi_t\|_\infty^2 \\
    & \lesssim \dim^{-1/2 + \epsilon},
\end{align*}
where in step (a) we noted that since $(u_i \rightarrow v_i, u_i^\prime \rightarrow v_i^\prime)$ is a pair of nullifying edges (see \defref{nullifying}), we have $\pi(u_i) \neq \pi(u_i^\prime)$ and hence, $a_{\pi(u_i)} \neq a_{\pi(u_i^\prime)}$. Combining these bounds on $U_\infty$ and $ \overline{U}$ with \factref{mobius} we obtain:
\begin{align} \label{eq:lambda-tilde-bound}
    \max_{\substack{a_1, \dotsc, a_{|\pi|-r} \in [\dim] \\ a_i \neq a_j \forall i \neq j}} |\widetilde{\lambda}(\mPsi_{1:\order}; F, a_1, \dotsc, a_{|\pi| - r})| & \leq r^{2r} \cdot \max(\sqrt{\dim} U_\infty, \min(\overline{U},\dim U_\infty))^r  \lesssim \dim^{-r/2 + \epsilon}.
\end{align}
Plugging the bounds on $|\widetilde{\gamma}|, |\widetilde{\lambda}|$ obtained in \eref{gamma-tilde-bound} and \eref{lambda-tilde-bound} into \eref{universal-forest-intermediate-estimate} gives:
\begin{align*}
     \left|\sum_{\ell \in \colorings{}{\pi}} \gamma(\mPsi; F, \ell) \right| & \lesssim \dim^{(|\pi| - r - 1) - \frac{\alpha}{2} - \frac{r}{2} + \epsilon} \lesssim \dim^{-\exponent(F,\pi) + (2k + \alpha) \epsilon},
\end{align*}
where we defined $\exponent(F,\pi) = \alpha/2 + r/2 + 1 - (|\pi|-r)$. To conclude the proof of this proposition, we observe that $\exponent(F,\pi)$ can be expressed as follows:
\begin{align*}
    \exponent(F,\pi) & \explain{def}{=} 1 + \frac{\alpha}{2} - (|\pi| - r) + \frac{r}{2} \\& \explain{(c)}{=}   1 + \frac{\alpha}{2} - |\pi| + \frac{|\nullleaves{F,\pi}|}{2} + \frac{|\nullleaves{F,\pi}|}{4} \\
    &\explain{\eref{gamma-tilde-bound}}{=} 1 + \left(\frac{1}{2} \sum_{v \in V \backslash (\{0\} \cup \nullleaves{F,\pi})} \|\pweight{v}{}\|_1 \right) + \frac{|\nullleaves{F,\pi}|}{2} - |\pi| + \frac{|\nullleaves{F,\pi}|}{4} \\
    & \explain{(d)}{=} 1+ \left(\frac{1}{2} \sum_{v \in V \backslash  \{0\}} \|\pweight{v}{}\|_1 \right) - |\pi| + \frac{|\nullleaves{F,\pi}|}{4}.
\end{align*}
In the above display, we recalled that $|\nullleaves{F,\pi}|  = 2r$ in step (c). To obtain equality (d) we observed that for any nullifying leaf $v \in \nullleaves{F,\pi}$, we have $\|\pweight{v}{}\|_1 = 1$ (cf. \defref{nullifying}). This completes the proof of \propref{key-estimate}.
\end{proof}

\subsection{Proof of the Decomposition Result} \label{appendix:decomposition}
This subsection presents the proof of \propref{decomposition}, which shows that the polynomial $\Gamma(\mPsi_{1:k}; F,\pi)$ associated with any relevant configuration can be expressed as a linear combination of polynomials associated with a few simple configurations. We begin by observing that relevant configurations (\defref{relevant-config}) already satisfy many requirements of simple configurations (\defref{simple-configuration}). The simple configuration requirements already satisfied by a relevant configuration are collected in the following definition of semi-simple configurations.

\begin{definition} \label{def:semi-simple-config} A decorated $\order$-tree $F = (V, E, \height{\cdot}, \pweight{\cdot}{}, \qweight{\cdot}{})$ and a partition $\pi = \{B_1, B_2, \dotsc, B_{|\pi|}\}$ of $V$ form a semi-simple configuration if they satisfy:
\begin{enumerate}
     \myitem{\texttt{Weak Sibling Property}}\label{weak-sibling-property}: There are no vertices $u, v \in V \backslash \{0\}$ which are siblings and satisfy
     $B_{\pi(u)} = B_{\pi(v)} = \{u,v\}$ and $\|\pweight{u}{}\|_1 = \|\pweight{v}{}\|_1 = 1$.
     \item{} \ref{original-leaf-property} as described in \defref{relevant-config}. 
     \item \ref{weak-forbidden-weights-property} as described in \defref{relevant-config}.  
     \item{\ref{parity-property}} as described in \defref{relevant-config}. 
\end{enumerate}
\end{definition}

The following lemma verifies that relevant configurations are semi-simple. 

\begin{lemma} \label{lem:relevant-is-semisimple} A relevant configuration (\defref{relevant-config}) is semi-simple (\defref{semi-simple-config}). 
\end{lemma}
\begin{proof}
Let $(F,\pi)$ be a relevant configuration. By the definition of a relevant configuration, $(F, \pi)$ already satisfies the \ref{weak-forbidden-weights-property}, \ref{original-leaf-property}, and \ref{parity-property}. Observe that the \ref{weak-sibling-property} is weaker than the \ref{strong-sibling-property}. Since a relevant configuration satisfies the \ref{strong-sibling-property}, it also satisfies the \ref{weak-sibling-property}. Hence, a relevant configuration is semi-simple.   
\end{proof}
In fact, the following lemma shows that the only obstacles that prevent a semi-simple configuration (\defref{semi-simple-config}) from being simple (\defref{simple-configuration}) is the presence of removable edges (\defref{removable-edge}) and removable edge pairs (\defref{removable-pair}).

\begin{lemma}\label{lem:semi-simple-2-simple} Let $(F,\pi)$ be a semi-simple configuration which has no removable edges (\defref{removable-edge}) and no removable edge pairs (\defref{removable-pair}). Then $(F,\pi)$ is a simple configuration. 
\end{lemma}
\begin{proof}
Consider a semi-simple configuration $(F,\pi)$ with $F = (V, E , \height{\cdot}, \pweight{\cdot}{}, \qweight{\cdot}{})$ and $\pi = \{B_1, B_2, \dotsc, B_{|\pi|}\}$ that has no removable edges (\defref{removable-edge}) and no removable edge pairs (\defref{removable-pair}). In order to show that $(F,\pi)$ is simple, we verify each of the requirements of \defref{simple-configuration}:
\begin{enumerate}
    \item \ref{parity-property}: By the definition of semi-simple configuration, \ref{parity-property} is satisfied. 
    \item \ref{strong-forbidden-weights-property}: Suppose that for the sake of contradiction, $(F,\pi)$ does not satisfy the \ref{strong-forbidden-weights-property}. This means that there is a vertex $v \in V \backslash \{0\}$ such that $|B_{\pi(v)}| = 1, \|\pweight{v}{}\|_1 = 1, \|\qweight{v}{}\|_1 = 1$ or $|B_{\pi(v)}| = 1, \|\pweight{v}{}\|_1 = 2, \|\qweight{v}{}\|_1 = 0$. However, since semi-simple configurations satisfy the \ref{weak-forbidden-weights-property}, we must have that $|B_{\pi(v)}| = 1$, $\pweight{v}{} = e_i + e_j, \qweight{v}{} = 0$ for some $i,j \in [\order]$ such that $i \neq j$ and $\Omega_{ij} = 0$. Since $\|\qweight{v}{}\|_1 = 0$, this means that $v$ must be a leaf (if $v$ was not a leaf, it would violate the conservation equation \eqref{eq:conservation-eq} in \defref{decorated-tree}). Since $v \neq 0$, it has a parent $u$. Observe that $u \rightarrow v$ is a removable edge for $(F,\pi)$. This contradicts the assumption that $(F,\pi)$ has no removable edges.
    \item \ref{paired-leaf-property}: Comparing the definition of a removable edge pair (\defref{removable-pair}) and the \ref{paired-leaf-property}, it is immediate that the absence of  removable edge pairs ensures that the \ref{paired-leaf-property} is satisfied.
    \item \ref{modified-leaf-property}: Consider a leaf vertex $v \in \leaves{F}$ with $|B_{\pi(v)}| = 1$. By the \ref{original-leaf-property} of semi-simple configurations, we must have $\qweight{v}{} = 0$. Furthermore by the \ref{parity-property} of semi-simple configurations $\|\pweight{v}{}\|_1$ is even and hence $\|\pweight{v}{}\|_1 \in \{2, 4, 6, 8, \dotsc, \}$.  Observe that since we have already shown that $(F,\pi)$ satisfies the \ref{strong-forbidden-weights-property} we must have $\|\pweight{v}{}\|_1 \neq 2$. Hence $\|\pweight{v}{}\|_1 \geq 4$, which verifies the \ref{modified-leaf-property}. 
\end{enumerate}
\end{proof}
In light of \lemref{semi-simple-2-simple}, we describe how removable edges and removable edge pairs can be eliminated from a semi-simple configuration to transform it into a simple configuration. The following lemma shows that if a semi-simple configuration $(F,\pi)$ has a removable edge $u_\star \rightarrow v_\star$ (\defref{removable-edge}), then the polynomial $\Gamma(\mPsi_{1:\order}; F,\pi)$ (cf. \eqref{eq:invariant-polynomial}) associated with the configuration $(F,\pi)$ can be expressed as a linear combination of polynomials associated with a few other configurations whose partitions have fewer blocks.  

\begin{lemma}[Elimination of a Removable Edge]\label{lem:eliminate-removable-edge} Let $(F,\pi)$ be a semi-simple configuration (\defref{semi-simple-config}) with at least one removable edge  (\defref{removable-edge}). Then, there exist semi-simple configurations $\{ (F_i, \pi_i): i \in [|\pi|-1]\}$ such that $|\pi_i| = |\pi| - 1$ and,
\begin{align*}
    \Gamma(\mPsi_{1:\order} ; F,\pi) & = - \sum_{i=1}^{|\pi| - 1} \Gamma(\mPsi_{1:\order}; F_i, \pi_i).
\end{align*}
\end{lemma}
\begin{proof}
See \appref{eliminate-removable-edge}. 
\end{proof}

Next, we show an analog of \lemref{eliminate-removable-edge} for a removable edge pair. This result is a simple generalization of a result \citep[Lemma 10]{dudeja2022universality} in our prior work, which considered the special case when $F$ was a decorated $\order$-tree with $k=1$. 

\begin{lemma}[Elimination of a Removable Edge Pair]\label{lem:eliminate-pair} Consider a semi-simple configuration $(F,\pi)$ with a removable edge pair (cf. \defref{removable-pair}). Then there exist semi-simple configurations $\{(F_i, \pi_i): i \in \{0, 1, 2, \dotsc, |\pi| - 1\}\}$ with $|\pi_i| = |\pi| - 1$  and a constant $c \in [-1,1]$ such that,
\begin{align*}
    \Gamma(\mPsi_{1:\order}; F,\pi) = c \cdot \Gamma(\mPsi_{1:\order}; F_0,\pi_0) -  \sum_{i=1}^{|\pi| - 1}    \Gamma(\mPsi_{1:\order}; F_i,\ell).
\end{align*}
\end{lemma}
\begin{proof}
See \appref{eliminate-pair}. 
\end{proof}

We postpone the proof of these intermediate results to the end of this subsection and present the proof of \propref{decomposition}. 

\begin{proof}[Proof of \propref{decomposition}]
The basic idea is that in order to decompose a relevant configuration into a linear combination of simple configurations, we will repeatedly apply \lemref{eliminate-removable-edge} to eliminate all removable edges and \lemref{eliminate-pair} to eliminate all removable edge pairs. This leads to the algorithm shown in \fref{decomposition-algorithm}. \propref{decomposition} will follow from the analysis of this algorithm. 
\begin{figure}[htb!]
\begin{algbox}
\textbf{Decomposition Algorithm} 

\vspace{1mm}

\textit{Input}: $(F_0,\pi_0)$, a relevant configuration. \\
\textit{Output}: $\calS$: A collection of simple configurations, and a map $\xi: \calS \rightarrow [-1,1]$. \\
\textit{Initialization} : $\iter{\calS}{0}:= \{(F_0,\pi_0)\}$, $\iter{\xi}{0}(F_0,\pi_0) := 1$.
\begin{itemize}
\item For $t \in \{1, 2, 3, \dotsc \}$
\begin{enumerate}
    \item Let $\iter{\calU}{t-1} := \{(F,\pi) \in \iter{\calS}{t-1}: (F,\pi)\text{ has at least one removable edge} \}$, and let ${u_{t-1}: = |\iter{\calU}{t-1}|}$. 
    \item Let $\iter{\calV}{t-1} := \{(F,\pi) \in \iter{\calS}{t-1} \backslash \iter{\calU}{t-1}: (F,\pi)\text{ has at least one removable edge-pair} \}$, and let  ${v_{t-1}: = |\iter{\calV}{t-1}|}$.
    \item If $u_{t-1}=v_{t-1} = 0$, end for loop. Otherwise,
    \begin{itemize}
        \item Let $\iter{\calU}{t-1} = \{(F_1,\pi_1), (F_2,\pi_2), \dotsc, (F_{u_{t-1}}, \pi_{u_{t-1}})\}$ be any enumeration of $\iter{\calU}{t-1}$. 
        \item Let $\iter{\calV}{t-1} = \{(G_1,\tau_1), (G_2,\tau_2), \dotsc, (G_{v_{t-1}}, \tau_{v_{t-1}})\}$ be any enumeration of $\iter{\calV}{t-1}$. 
        \item For each $i \in [u_{t-1}]$, decompose configuration $(F_i,\pi_i)$ using \lemref{eliminate-removable-edge} to obtain configurations $\{(F_{ij},\pi_{ij}: j \in \{ 1, \dotsc, |\pi_i|-1\}\}$ such that:
        \begin{align*}
    \Gamma(\mPsi_{1:\order}; F_i, \pi_i) \explain{}{=} - \sum_{j=1}^{|\pi_i|-1} \Gamma(\mPsi_{1:\order}; F_{ij}, \pi_{ij}).
\end{align*}
    \item For each $i \in [v_{t-1}]$, decompose configuration $(G_i,\tau_i)$ using \lemref{eliminate-pair} to obtain configurations $\{(G_{ij},\tau_{ij}: j \in \{0, 1, \dotsc, |\tau_i|-1\}\}$ and a constant $c_{it} \in [-1,1]$ such that:
        \begin{align*}
    \Gamma(\mPsi_{1:\order}; G_i, \tau_i) \explain{}{=} c_{it} \cdot \Gamma(\mPsi_{1:\order}; G_{i0}, \tau_{i0}) - \sum_{j=1}^{|\tau_i|-1} \Gamma(\mPsi_{1:\order}; G_{ij}, \tau_{ij}).
\end{align*}
        \item Update:
        \begin{subequations} \label{eq:update-eq-decomposition-algo}
        \begin{align}
            \iter{\calS}{t} &:= (\iter{\calS}{t-1} \backslash (\iter{\calU}{t-1} \cup \iter{\calV}{t-1})) \cup \left( \bigcup_{i=1}^{u_{t-1}} \bigcup_{j=1}^{|\pi_{i}|-1} \{(F_{ij},\pi_{ij})\}\right)  \nonumber\\& \hspace{6cm} \cup \left( \bigcup_{i=1}^{v_{t-1}} \bigcup_{j=0}^{|\tau_{i}|-1} \{(G_{ij},\tau_{ij})\}\right) \\
            \iter{\xi}{t}(F,\pi) &:= \iter{\xi}{t-1}(F,\pi) \; \forall \; (F,\pi) \; \in \; \iter{\calS}{t-1} \backslash (\iter{\calU}{t-1} \cup \iter{\calV}{t-1}), \\
             \iter{\xi}{t}(F_{ij},\pi_{ij}) &:= -\iter{\xi}{t-1}(F_i, \pi_i) \; \forall \; j \; \in \; [|\pi_i|-1], \; \forall \; i \; \in \; [u_{t-1}], \\
            \iter{\xi}{t}(G_{i0},\tau_{i0}) &:= c_{it} \cdot \iter{\xi}{t-1}(G_i,\tau_i) \; \forall \; i \; \in \; [v_{t-1}], \\
            \iter{\xi}{t}(G_{ij},\tau_{ij}) &:= -\iter{\xi}{t-1}(G_i, \tau_i) \; \forall \; j \; \in \; [|\tau_i|-1], \; \forall \; i \; \in \; [v_{t-1}].
        \end{align}
        \end{subequations}
    \end{itemize}
\end{enumerate}
\item \textit{Return} $\calS := \iter{\calS}{t-1}$, $\xi := \iter{\xi}{t-1}$.
\end{itemize}
\vspace{1mm}
\end{algbox}
\caption{Decomposing a relevant configuration into a collection of simple configurations} \label{fig:decomposition-algorithm}
\end{figure}
The proof of the proposition follows from the following sequence of arguments:
\begin{enumerate}
    \item  We claim that for any $t \geq 0$, $\calS^{(t)}$ is a collection of semi-simple configurations. This is true for $t=0$ since $\calS^{(0)} = \{(F_0, \pi_0)\}$ where $(F_0, \pi_0)$ is a relevant configuration by assumption and \lemref{relevant-is-semisimple} shows that relevant configurations are semi-simple. For $t\geq 1$, this claim follows by induction since the configurations generated by applying \lemref{eliminate-removable-edge} or \lemref{eliminate-pair} to a semi-simple configuration are also semi-simple. 
    \item Next we claim that for any $t \geq 0$, we have,
    \begin{align} \label{eq:induction-hypothesis}
        \Gamma(\mPsi_{1:\order}; F_0, \pi_0) & = \sum_{(F,\pi) \in \iter{\calS}{t}} \iter{\xi}{t}(F,\pi) \cdot  \Gamma(\mPsi_{1:\order}; F, \pi).
    \end{align}
    This is trivially true at $t=0$ since $\calS^{(0)} = \{(F_0, \pi_0)\}$ and $a^{(0)}(F_0, \pi_0) = 1$. For $t\geq 1$, this can be verified by induction. Suppose that the claim \eqref{eq:induction-hypothesis} holds for some $t \in \W$. Using the definition of $\calU^{(t)}$  and $\calV^{(t)}$ from the decomposition algorithm in \fref{decomposition-algorithm} we can write \eqref{eq:induction-hypothesis} as:
    \begin{align}\label{eq:ind-hypo-expanded}
        \Gamma(\mPsi_{1:\order}; F_0, \pi_0) & = \sum_{(F,\pi) \in \iter{\calS}{t}\backslash (\calU^{(t)} \cup \calV^{(t)})} \iter{\xi}{t}(F,\pi) \cdot  \Gamma(\mPsi_{1:\order}; F, \pi) \nonumber \\&\qquad  \qquad \qquad + \sum_{i=1}^{u_t} \iter{\xi}{t}(F_i,\pi_i) \cdot \Gamma(\mPsi_{1:\order}; F_i, \pi_i) + \sum_{i=1}^{v_t} \iter{\xi}{t}(G_i,\tau_i) \cdot \Gamma(\mPsi_{1:\order}; G_i, \tau_i).
    \end{align}
\lemref{eliminate-removable-edge} and \lemref{eliminate-pair} guarantee that the configurations $\{(F_{ij}, \pi_{ij})\}$, $\{(G_{ij}, \tau_{ij})\}$ generated by their application satisfy:
\begin{align}\label{eq:elimination-conclusion}
    \Gamma(\mPsi_{1:\order}; F_i, \pi_i) &\explain{}{=} - \sum_{j=1}^{|\pi_i|-1} \Gamma(\mPsi_{1:\order}; F_{ij}, \pi_{ij}) \quad \forall \; j \; \in \; [|\pi_i|], \; i \; \in \; [u_t], \\
    \Gamma(\mPsi_{1:\order}; G_i, \tau_i) &\explain{}{=} c_{it} \cdot \Gamma(\mPsi_{1:\order}; G_{i0}, \tau_{i0}) - \sum_{j=1}^{|\tau_i|-1} \Gamma(\mPsi_{1:\order}; G_{ij}, \tau_{ij}) \quad \forall \; j \; \in \; [|\tau_i|], \; i \; \in \; [v_t].
\end{align}
Plugging \eqref{eq:elimination-conclusion} in \eqref{eq:ind-hypo-expanded} and using the update formulae for the map $\iter{\xi}{t+1}$ given in \eqref{eq:update-eq-decomposition-algo} shows that \eqref{eq:induction-hypothesis} also holds at step $t+1$, as claimed. 
\item Notice that the algorithm terminates at step $t$ iff all semi-simple configurations in $\calS^{(t-1)}$ do not have any removable edges or removable edge-pairs. By \lemref{semi-simple-2-simple}, this means that at termination, $\calS : = \calS^{(t-1)}$ consists of simple configurations, as claimed. 
\item Next, we need to show that the algorithm terminates. In order to track the convergence of the algorithm define the potential:
\begin{align}\label{eq:potential}
    w_t & \explain{def}{=}  \max \{ |\pi| : (F,\pi) \in \iter{\calU}{t-1} \cup \iter{\calV}{t-1} \}.
\end{align}
If the algorithm does not terminate at iteration $t$, then each configuration $(F,\pi) \in \iter{\calU}{t-1} \cup \iter{\calV}{t-1}$ is replaced by a few configurations $(F^\prime, \pi^\prime)$ with $|\pi^\prime| = |\pi| - 1$ by an application of \lemref{eliminate-removable-edge} or \lemref{eliminate-pair}. If the algorithm does not terminate at iteration $t+1$ either, then $w_{t+1} \leq w_t - 1$. Observe that $w_1 = |\pi_0|$ (unless $(F_0,\pi_0)$ was already a simple configuration, in which case the algorithm terminates at the first iteration). Hence if the algorithm does not terminate at iteration $t$, then $w_t  \leq |\pi_0| - t +1$. On the other hand, if the algorithm does not terminate at iteration $t$, we must have $w_t \geq 2$. This is because any semi-simple configuration with a removable edge $u \rightarrow v$ has at least two blocks (one containing $v$ and the other containing $u$). Similarly a configuration with a removable edge pair $u \rightarrow v, u^\prime \rightarrow v^\prime$ also has at least two blocks (the block $\{v, v^\prime\}$ and the block containing $\{u, u^\prime\}$). Hence if the algorithm does not terminate at iteration $t$, $2 \leq w_t \leq |\pi_0| - t + 1 \implies t \leq |\pi_0| - 1$. Hence, the algorithm terminates by iteration $t = |\pi_0|$.  
\item Suppose the algorithm has not terminated at iteration $t$. We can bound the cardinality of $\iter{\calS}{t}$ as follows:
\begin{align*}
     |\iter{\calS}{t}| & \explain{(a)}{\leq} |\iter{\calS}{t-1}| - |\iter{\calU}{t-1}|- |\iter{\calV}{t-1}| + \left(\sum_{(F,\pi) \in \iter{\calU}{t-1}} |\pi| - 1 \right) + \left(\sum_{(F,\pi) \in \iter{\calV}{t-1}} |\pi| \right) \\
        & \explain{(b)}{\leq} |\iter{\calS}{t-1}|  + (w_{t}-2) \cdot  |\iter{\calU}{t-1}| +  (w_{t}-1) \cdot  |\iter{\calV}{t-1}|   \\
        & \explain{(c)}{\leq} w_{t} \cdot |\iter{\calS}{t-1}| \\
        & \explain{(d)}{\leq} (|\pi_0| - t + 1) |\iter{\calS}{t-1}| \\
        & \explain{(e)}{\leq} |\pi_0| \cdot (|\pi_0| - 1) \cdot \dotsb \cdot (|\pi_0| - t + 1).
\end{align*}
In the above display, step (a) follows from the update equation \eqref{eq:update-eq-decomposition-algo} and using the fact that in the decomposition algorithm of \fref{decomposition-algorithm} each configuration $(F,\pi) \in \iter{\calU}{t-1}$ is replaced by $|\pi|-1$ new configurations by an application of \lemref{eliminate-removable-edge} and  each configuration $(F,\pi) \in \iter{\calV}{t-1}$ is replaced by $|\pi|$ new configurations by an application of \lemref{eliminate-pair}. Step (b) follows from the definition of $w_t$ in \eqref{eq:potential}, step (c) uses the fact that $|\iter{\calU}{t-1}| + |\iter{\calV}{t-1}| \leq |\iter{\calS}{t-1}|$. In step (d) we used the bound $w_t \leq |\pi_0| - t + 1$ derived previously and step (e) follows from unrolling the recursive estimate. Since the algorithm terminates by iteration $t = |\pi_0|$, we obtain that $|\calS| \leq |\pi_0|!$ using the above estimate. 
\end{enumerate}
This concludes the proof of this proposition. 
\end{proof}

\subsubsection{Proof of \lemref{eliminate-removable-edge}} \label{appendix:eliminate-removable-edge}
\begin{proof}[Proof of \lemref{eliminate-removable-edge}]
Let $F = (V, E, \height{\cdot}, \pweight{\cdot}{}, \qweight{\cdot}{})$ and $\pi = \{B_1, B_2, \dotsc, B_{|\pi|}\}$ and $e_\star = u_\star \rightarrow v_\star$ denote the removable edge (\defref{removable-edge}). Hence, $v_\star \in \leaves{F}$ with $B_{\pi(v_\star)} = \{v_\star\}$. $\pweight{v_\star}{} = e_{i_\star} + e_{j_\star}$ for some $i_\star, j_\star \in [\order]$ with $i_\star \neq j_\star$ and $\Omega_{i_\star j_\star} = 0$. Without loss of generality, we can assume that $v_\star \in B_{|\pi|}$ and $u_\star \in B_{|\pi| - 1}$. Recall that from \eqref{eq:invariant-polynomial} that:
\begin{align*} 
    \Gamma(\mPsi_{1:\order}; F, \pi) \explain{def}{=} \sum_{\ell \in \colorings{}{\pi}} \gamma(\mPsi_{1:\order}; F, \ell)\text{, where }  \gamma(\mPsi_{1:\order}; F, \ell) \explain{def}{=} \frac{1}{\dim} \prod_{i=1}^\order \prod_{\substack{e \in E \\ e = u \rightarrow v}} (\Psi_i)_{\ell_u, \ell_v}^{\pweight{v}{i}}.
\end{align*}
With every $\ell \in \colorings{}{\pi}$, we associate a vector $a \in [\dim]^{|\pi|}$ such that $a_i$ denotes the color assigned by the coloring $\ell$ to vertices in $B_i$ for each $i \in [|\pi|]$. Hence we can express $\gamma(\mPsi_{1:\order}; F,\ell)$ as:
\begin{align*}
    \gamma(\mPsi_{1:\order}; F, \ell) \explain{}{=} \frac{1}{\dim} \prod_{i=1}^\order \prod_{\substack{e \in E \\ e = u \rightarrow v}} (\Psi_i)_{a_{\pi(u)}, a_{\pi(v)}}^{\pweight{v}{i}}.
\end{align*}
Since the block $B_{|\pi|} = \{v_\star\}$ consists of exactly one vertex, which is a leaf, the color $a_{|\pi|}$ appears exactly once in the product  $\gamma(\mPsi_{1:\order}; F, \ell)$ (in the term corresponding to the edge $u_\star \rightarrow v_\star$). We can isolate the occurence of the color $a_{|\pi|}$ using the factorization:
\begin{align*}
    \gamma(\mPsi_{1:\order}; F, \ell)  & = \widetilde{\gamma}(\mPsi_{1:\order}; F, a_{1}, a_2, \dotsc, a_{|\pi|-1}) \cdot (\Psi_{i_\star})_{a_{|\pi|-1}, a_{|\pi|}} \cdot (\Psi_{j_\star})_{a_{|\pi|-1}, a_{|\pi|}},
\end{align*}
where the factor $\widetilde{\gamma}(\mPsi_{1:\order}; F, a_{1}, a_2, \dotsc, a_{|\pi|-1})$ does not depend on $a_{|\pi|}$ and is defined as:
\begin{align*}
     \gamma(\mPsi_{1:\order}; F, a_{1}, a_2, \dotsc, a_{|\pi|-1}) \explain{}{=} \frac{1}{\dim} \prod_{i=1}^\order \prod_{\substack{e \in E \backslash \{e_\star\}\\ e = u \rightarrow v}} (\Psi_i)_{a_{\pi(u)}, a_{\pi(v)}}^{\pweight{v}{i}}.
\end{align*}
Hence,
\begin{align}
     &\Gamma(\mPsi_{1:\order}; F, \pi) \explain{}{=} \sum_{\substack{a_{1}, \dotsc, a_{|\pi| - 1} \\ a_i \neq a_j \forall i \neq j}}  \gamma(\mPsi_{1:\order}; F, a_{1},  \dotsc, a_{|\pi|-1}) \cdot \sum_{a_{|\pi|} \in [\dim] \backslash \{a_1, \dotsc, a_{|\pi| - 1}\}} (\Psi_{i_\star})_{a_{|\pi|-1}, a_{|\pi|}} \cdot (\Psi_{j_\star})_{a_{|\pi|-1}, a_{|\pi|}} \nonumber \\
     & = \sum_{\substack{a_{1}, \dotsc, a_{|\pi| - 1} \\ a_i \neq a_j \forall i \neq j}}  \gamma(\mPsi_{1:\order}; F, a_{1},  \dotsc, a_{|\pi|-1})  \left( \sum_{a_{|\pi|}=1 }^{\dim} (\Psi_{i_\star})_{a_{|\pi|-1}, a_{|\pi|}}  (\Psi_{j_\star})_{a_{|\pi|-1}, a_{|\pi|}} - \sum_{i=1}^{|\pi|-1} (\Psi_{i_\star})_{a_{|\pi|-1}, a_{i}}  (\Psi_{j_\star})_{a_{|\pi|-1}, a_{i}}\right) \nonumber \\
     & =  \sum_{\substack{a_{1}, \dotsc, a_{|\pi| - 1} \\ a_i \neq a_j \forall i \neq j}}  \gamma(\mPsi_{1:\order}; F, a_{1},  \dotsc, a_{|\pi|-1})  \left( (\mPsi_{i_\star} \mPsi_{j_\star}\tran)_{a_{|\pi|-1}, a_{|\pi|-1}} - \sum_{i=1}^{|\pi|-1} (\Psi_{i_\star})_{a_{|\pi|-1}, a_{i}}  (\Psi_{j_\star})_{a_{|\pi|-1}, a_{i}}\right) \nonumber \\
     & \explain{(a)}{=} -  \sum_{\substack{a_{1}, \dotsc, a_{|\pi| - 1} \\ a_i \neq a_j \forall i \neq j}}  \gamma(\mPsi_{1:\order}; F, a_{1},  \dotsc, a_{|\pi|-1})  \sum_{i=1}^{|\pi|-1} (\Psi_{i_\star})_{a_{|\pi|-1}, a_{i}}  (\Psi_{j_\star})_{a_{|\pi|-1}, a_{i}}. \label{eq:edge-removal}
\end{align}
In the above display, step (a) follows from the assumption that $(\mPsi_{i_\star} \mPsi_{j_\star}\tran)_{\ell,\ell} = \Omega_{i_\star j_\star} = 0$ made in the statement of \thref{moments}. We introduce $|\pi|-1$ new configurations $\{(F_i, \pi_i): i \in [|\pi|-1]\}$ defined as follows:
\begin{enumerate}
    \item For each $i \in [|\pi| - 1]$, we define $F_i \explain{def}{=} F$. 
    \item For each $i \in [|\pi| - 1]$, we define $\pi_i = \{B_1, B_2, \dotsc, B_{i} \cup B_{|\pi|}, B_{i+1}, \dotsc, B_{|\pi|-1}\}$, where $B_1, B_2, \dotsc, B_{|\pi|}$ were the blocks of $\pi$. 
\end{enumerate}
Observe that \eqref{eq:edge-removal} can be written as:
\begin{align*}
    \Gamma(\mPsi_{1:\order}; F, \pi) \explain{}{=} - \sum_{i=1}^{|\pi|-1} \Gamma(\mPsi_{1:\order}; F_i, \pi_i),
\end{align*}
as desired. In order to complete the proof of this we need to verify that the configurations $\{(F_i, \pi_i): i \in [|\pi|-1]\}$ are semi-simple. Indeed,
\begin{enumerate}
    \item Since a singleton block in $\pi_i$ must also be singleton block in $\pi$ and since $F_i = F$, the configuration $(F_i, \pi_i)$ inherits the \ref{original-leaf-property} and \ref{weak-forbidden-weights-property} from $(F,\pi)$.
    \item Since $\pi_i$ is formed by merging two blocks of $\pi$, the \ref{parity-property} is maintained in $(F_i, \pi_i)$. 
    \item Finally to verify the \ref{weak-sibling-property}, we simply need to check that there is no block $B$ in $\pi_i$ of the form:
    \begin{align} \label{eq:forbidden-block-form}
        B = \{u, v\} \text{ where $u,v$ are siblings and } \|\pweight{u}{}\|_1 =\|\pweight{v}{}\|_1 = 1.
    \end{align}
    Recall that the blocks of $\pi_i$ are $\pi_i = \{B_1, B_2, \dotsc, B_{i} \cup B_{|\pi|}, B_{i+1}, \dotsc, B_{|\pi|}\}$, where $B_1, B_2, \dotsc, B_{|\pi|}$ were the blocks of $\pi$. Since $(F,\pi)$ is semi-simple and satisfies the \ref{weak-sibling-property}, the blocks  $B_1, B_2, \dotsc, B_{i},  B_{i+1}, \dotsc, B_{|\pi|-1}$ are not of the form \eqref{eq:forbidden-block-form}. Furthermore the block $B_{i} \cup B_{|\pi|}$ is also not of the form \eqref{eq:forbidden-block-form} since $v_\star \in B_{|\pi|} \subset B_{i} \cup B_{|\pi|}$ and $\|\pweight{v_\star}{}\|_1 = 2$. 
\end{enumerate}
This concludes the proof.
\end{proof}

\subsubsection{Proof of \lemref{eliminate-pair}} \label{appendix:eliminate-pair}
\begin{proof}[Proof of \lemref{eliminate-pair}]
As mentioned previously, this result is a simple generalization of \citep[Lemma 10]{dudeja2022universality}, which considered the special case when $F$ was a decorated $\order$-tree with $k=1$. Hence, we will closely follow the proof of \citep[Lemma 10]{dudeja2022universality}. Let $F = (V, E, \height{\cdot}, \pweight{\cdot}{}, \qweight{\cdot}{})$. Let $e_\star = u_\star \rightarrow v_\star, \; e_\star^\prime = u^\prime_\star \rightarrow v^\prime_\star$ denote the removable edge pair. By \defref{removable-pair}, we know that:
\begin{align*}
   v_\star, v_\star^\prime \in \leaves{F}, \quad  B_{\pi(v_\star)} = B_{\pi(v_\star^\prime)} = \{v_\star, v_\star^\prime \}, \; \pweight{v_\star}{} = e_{i_\star}, \quad \pweight{v_\star^\prime}{} = e_{i_\star^\prime}
\end{align*}
for some $i_\star, i_\star^\prime \in [\order]$. In the above display $e_{1:\order}$ denote the standard basis vectors in $\R^\order$.  Since $(F,\pi)$ is semi-simple, by the \ref{weak-sibling-property} $v_\star, v_\star^\prime$ cannot be siblings. Consequently, $u_\star \neq u_\star^\prime$. We assume (without loss of generality) that $\pi = \{B_1, B_2, \dotsc, B_{|\pi|}\}$ with,\begin{align*}
    B_{|\pi|} & = \{v_\star,  v_\star^\prime \}, \; \{u_\star, u_\star^\prime\} \subset B_{|\pi|-1}.
\end{align*}
Recall that from \eqref{eq:invariant-polynomial} that:
\begin{align*} 
    \Gamma(\mPsi_{1:\order}; F, \pi) \explain{def}{=} \sum_{\ell \in \colorings{}{\pi}} \gamma(\mPsi_{1:\order}; F, \ell)\text{, where }  \gamma(\mPsi_{1:\order}; F, \ell) \explain{def}{=} \frac{1}{\dim} \prod_{i=1}^\order \prod_{\substack{e \in E \\ e = u \rightarrow v}} (\Psi_i)_{\ell_u, \ell_v}^{\pweight{v}{i}}.
\end{align*}
With every $\ell \in \colorings{}{\pi}$, we associate a vector $a \in [\dim]^{|\pi|}$ such that $a_i$ denotes the color assigned by the coloring $\ell$ to vertices in $B_i$ for each $i \in [|\pi|]$. Hence we can express $\gamma(\mPsi_{1:\order}; F,\ell)$ as:
\begin{align*}
    \gamma(\mPsi_{1:\order}; F, \ell) \explain{}{=} \frac{1}{\dim} \prod_{i=1}^\order \prod_{\substack{e \in E \\ e = u \rightarrow v}} (\Psi_i)_{a_{\pi(u)}, a_{\pi(v)}}^{\pweight{v}{i}}.
\end{align*}
Since the block $B_{|\pi|} = \{v_\star,v_\star^\prime\}$ consists of exactly two vertices, which are both leaves, the color $a_{|\pi|}$ appears exactly twice in the product  $\gamma(\mPsi_{1:\order}; F, \ell)$ (in the terms corresponding to the edges $u_\star \rightarrow v_\star$ and $u_\star^\prime \rightarrow v_\star^\prime$). We can isolate the occurrence of the color $a_{|\pi|}$ using the factorization:
\begin{align*}
    \gamma(\mPsi_{1:\order}; F, \ell)  & = \widetilde{\gamma}(\mPsi_{1:\order}; F, a_{1}, a_2, \dotsc, a_{|\pi|-1}) \cdot (\Psi_{i_\star})_{a_{|\pi|-1}, a_{|\pi|}} \cdot (\Psi_{i_\star^\prime})_{a_{|\pi|-1}, a_{|\pi|}},
\end{align*}
where the factor $\widetilde{\gamma}(\mPsi_{1:\order}; F, a_{1}, a_2, \dotsc, a_{|\pi|-1})$ does not depend on $a_{|\pi|}$ and is defined as:
\begin{align*}
     \gamma(\mPsi_{1:\order}; F, a_{1}, a_2, \dotsc, a_{|\pi|-1}) \explain{}{=} \frac{1}{\dim} \prod_{i=1}^\order \prod_{\substack{e \in E \backslash \{e_\star,e_\star^\prime\}\\ e = u \rightarrow v}} (\Psi_i)_{a_{\pi(u)}, a_{\pi(v)}}^{\pweight{v}{i}}.
\end{align*}
Hence,
\begin{align}
     &\Gamma(\mPsi_{1:\order}; F, \pi) \explain{}{=} \sum_{\substack{a_{1}, \dotsc, a_{|\pi| - 1} \\ a_i \neq a_j \forall i \neq j}}  \gamma(\mPsi_{1:\order}; F, a_{1},  \dotsc, a_{|\pi|-1}) \cdot \sum_{a_{|\pi|} \in [\dim] \backslash \{a_1, \dotsc, a_{|\pi| - 1}\}} (\Psi_{i_\star})_{a_{|\pi|-1}, a_{|\pi|}} \cdot (\Psi_{i_\star^\prime})_{a_{|\pi|-1}, a_{|\pi|}} \nonumber \\
     & = \sum_{\substack{a_{1}, \dotsc, a_{|\pi| - 1} \\ a_i \neq a_j \forall i \neq j}}  \gamma(\mPsi_{1:\order}; F, a_{1},  \dotsc, a_{|\pi|-1})  \left( \sum_{a_{|\pi|}=1 }^{\dim} (\Psi_{i_\star})_{a_{|\pi|-1}, a_{|\pi|}}  (\Psi_{i_\star^\prime})_{a_{|\pi|-1}, a_{|\pi|}} - \sum_{i=1}^{|\pi|-1} (\Psi_{i_\star})_{a_{|\pi|-1}, a_{i}}  (\Psi_{i_\star^\prime})_{a_{|\pi|-1}, a_{i}}\right) \nonumber \\
     & =  \sum_{\substack{a_{1}, \dotsc, a_{|\pi| - 1} \\ a_i \neq a_j \forall i \neq j}}  \gamma(\mPsi_{1:\order}; F, a_{1},  \dotsc, a_{|\pi|-1})  \left( (\mPsi_{i_\star} \mPsi_{i_\star^\prime}\tran)_{a_{|\pi|-1}, a_{|\pi|-1}} - \sum_{i=1}^{|\pi|-1} (\Psi_{i_\star})_{a_{|\pi|-1}, a_{i}}  (\Psi_{i_\star^\prime})_{a_{|\pi|-1}, a_{i}}\right) \nonumber \\
     & \explain{(a)}{=}  \Omega_{i_\star i_\star^\prime} \sum_{\substack{a_{1}, \dotsc, a_{|\pi| - 1} \\ a_i \neq a_j \forall i \neq j}}  \gamma(\mPsi_{1:\order}; F, a_{1},  \dotsc, a_{|\pi|-1}) -  \sum_{\substack{a_{1}, \dotsc, a_{|\pi| - 1} \\ a_i \neq a_j \forall i \neq j}}  \gamma(\mPsi_{1:\order}; F, a_{1},  \dotsc, a_{|\pi|-1})  \sum_{i=1}^{|\pi|-1} (\Psi_{i_\star})_{a_{|\pi|-1}, a_{i}}  (\Psi_{j_\star})_{a_{|\pi|-1}, a_{i}}. \label{eq:pair-removal}
\end{align}
In the above display, step (a) follows from the assumption that $(\mPsi_{i_\star} \mPsi_{i_\star^\prime} \tran)_{\ell,\ell} = \Omega_{i_\star i_\star^{\prime}}$ made in the statement of \thref{moments}. Next we define the configurations $(F_i, \pi_i)$ for $i = 0, 1, \dotsc, |\pi| - 1$. Recall that the original decorated $\order$-tree was given by $F = (V, E, \height{\cdot}, \pweight{\cdot}{}, \qweight{\cdot}{})$ and the original partition $\pi$ was given by $\pi = \{B_1, \dotsc, B_{|\pi|}\}$. Then,
\begin{enumerate}
    \item We define the decorated $\order$-tree $F_0 = (V_0, E_0, \height[0]{\cdot}, \pweight{\cdot}{0}, \qweight{\cdot}{0})$ as follows:
    \begin{enumerate}
        \item The vertex set is given by: $V_0 = V \backslash \{v_\star, v_\star^\prime\}$.
        \item The edge set is given by: $E_0 = E \backslash \{e_\star, e_\star^\prime\}$.
    \end{enumerate}
    This defines a directed graph $(V_0, E_0)$. It is straightforward to check that since $F$ was a directed tree with root $0$, $F_0$ is also a directed tree with root $0$.    Next we define the functions $\height[0]{\cdot}, \pweight{\cdot}{0}, \qweight{\cdot}{0}$:
    \begin{enumerate}
        \setcounter{enumii}{2}
        \item We set $\height[0]{v} = \height{v}$ for any $v \in V_0$. 
        \item We set $\pweight{v}{0} = \pweight{v}{}$ for any $v \in \iter{V}{0} \backslash \{0\}$\footnote{Notice that $\pweight{v}{}, \pweight{v}{0} \in \W^{\order}$. We will use the notations $\pweight{v}{j}$ and $\pweight{v}{0j}$ to refer to the coordinate $j$ of $\pweight{v}{}$ and $\pweight{v}{0}$. }. 
        \item We set  $\qweight{v}{0} = \qweight{v}{}$ for any $v \in V_0 \backslash \{u_\star, u_\star^\prime\}$. We set $\qweight{u_\star}{0}, \qweight{u_\star^\prime}{0} \in \W^{\order}$ as follows:
        \begin{subequations} \label{eq:ensure-conservation}
        \begin{align}
            \qweight{u_\star}{0} & = \qweight{u_\star}{} - \pweight{v_\star}{}, \\
            \qweight{u_\star^\prime}{0} & = \qweight{u_\star^\prime}{} - \pweight{v_\star^\prime}{}.
        \end{align}
        \end{subequations}
    \end{enumerate}
    We check that $F_0$ is a decorated $\order$-tree (in the sense of \defref{decorated-tree}).  Observe that $u_\star, u_\star^\prime \in V_0$. Recall that $u_{\star} \neq u_\star^\prime$, and hence $F_0$ has at least one non-root vertex (since $u_\star = u_\star^\prime = 0$ is not possible). Hence $|E_0| \geq 1$, as required by \defref{decorated-tree}. Since $\height[0]{\cdot} = \height{\cdot}$ on $V_0$, $\height[0]{\cdot}$ satisfies all the requirements described in \defref{decorated-tree}. \eqref{eq:ensure-conservation} ensures $\pweight{\cdot}{0}, \qweight{\cdot}{0}$ satisfy the conservation equation \eqref{eq:conservation-eq} in \defref{decorated-tree}. Furthermore since $|E_0| \geq 1$, the root vertex $0$ has alteast one child in $F_0$ and hence by the the conservation equation \eqref{eq:conservation-eq}, $\|\qweight{0}{}\|_1 \geq 1$, as required by \defref{decorated-tree}. Hence $F_0$ satisfies all the requirements described in \defref{decorated-tree}. We set $\iter{\pi}{0} = \{B_1, B_3, \dotsc, B_{|\pi|-1}\}$. Observe that this is a valid partition (\defref{partition}) of $V_0$. 
    \item For $i \geq 1$, we set $F_i = F$ and $\pi_i = \{B_1, B_3, \dotsc, B_{i-1}, B_{i} \cup\{v_\star, v_\star^\prime\}, B_{i+1}, \dotsc, B_{|\pi|-1}\}$. Observe that $\pi_i$ is a valid partition of $V$ (\defref{partition}) which is the vertex set of $F_i$.
\end{enumerate}
Using these definitions,  \eqref{eq:pair-removal} can be written as:
\begin{align} \label{eq:pair-removal-2}
     \Gamma(\mPsi_{1:\order}; F,\pi) = \Omega_{i_\star i_\star^\prime} \cdot  \Gamma(\mPsi_{1:\order}; F_0,\pi_0) -  \sum_{i=1}^{|\pi| - 1} \Gamma(\mPsi_{1:\order}; F_i,\pi_i).
\end{align}
Since $\Omega$ (the limiting covariance matrix of the semi-random ensemble, cf. \defref{semirandom}) is psd,  $$|\Omega_{i_\star i_\star^\prime}| \leq \sqrt{\Omega_{i_\star i_\star}\Omega_{i^\prime_\star i^\prime_\star}} = 1.$$
Hence the expression \eqref{eq:pair-removal-2} is of the form claimed in the statement of the lemma. To complete the proof, we need to verify that each $(F_i, \pi_i)$ is a semi-simple configuration, and we do so next. For each $i \in \{0, 1, 2, \dotsc, |\pi|-1\}$ let us denote the components of the decorated $\order$-trees $F_i$\footnote{Recall that $\pweight{v}{i}, \qweight{v}{i} \in \W^{\order}$. We can use the notation $\pweight{v}{ij}, \qweight{v}{ij}$ to refer to the coordinate $j$ of $\pweight{v}{i},\qweight{v}{i}$ if required.} and the blocks of the partition $\pi_i$ as follows:
\begin{align*}
    F_i &= (V_i, E_i, \height[i]{\cdot}, \pweight{\cdot}{i}, \qweight{\cdot}{i}), \\
    \pi_i &= \{B_{i,1}, B_{i,2}, \dotsc, B_{i,|\pi_i|}\}. 
\end{align*}
We begin by making the following observations for any $i \geq 0$:
\begin{description}
\item [Observation 1: ] For any vertex $v \in V_i$, either $B_{i,\pi_i(v)} = B_{\pi(v)}$ (that is, the block of $v$ is unchanged) or $|B_{i,\pi_i(v)}| \geq 3$ (that is, the block of $v$ has cardinality at least $3$).  The latter scenario covers the case when the leaf vertices $\{v_\star, v_\star^\prime\}$ of the removable edge-pair are added to the block $B_{\pi(v)}$ to form $B_{i,\pi_i(v)}$.
\item [Observation 2: ] The deletion of edges $u_\star \rightarrow v_\star$ and $u_\star^\prime \rightarrow v_\star^\prime$, might result in $u_\star$ or $u_\star^\prime$ become leaves in $F_0$. Hence, $\leaves{F_0} \subset \leaves{F} \cup \{u_\star, u_\star^\prime\}$. On the other hand since $F_i = F$ for $i \geq 1$,  $\leaves{F_i} = \leaves{F}$ for $i \geq 1$.
\item [Observation 3: ] $u_\star \in \leaves{F_0}\backslash \leaves{F}$ iff $u_\star$ had exactly one child in $F$ namely, $v_\star$. Hence, by the conservation equation \eqref{eq:conservation-eq}, $u_{\star}\in \leaves{F_0}\backslash \leaves{F}$ is a leaf in $F_0$ iff $\|\qweight{u_\star}{}\|_1 = 1$. Furthermore, in this situation by \eqref{eq:ensure-conservation}, $\|\qweight{u_\star}{0}\|_1 = 0$. The same observation holds for $u_\star^\prime$. 
\end{description}

In order to check $\{(F_i, \pi_i)\}$ is a collection of semi-simple configurations, we check each of the requirements of \defref{semi-simple-config}:

\begin{enumerate}
    \item  In order to verify the \ref{weak-sibling-property}, for the sake of contradiction assume that  $\pi_i$ has a block of the form $\{u,v\}$ where $\{u,v\}$ are siblings in $F_i$ and $\|\pweight{u}{0}\|_1 = \|\pweight{v}{0}\|_1$ Observe $\{u,v\}$ are also siblings in $F$ and that $\pweight{u}{} = \pweight{u}{0}$, $\qweight{u}{} = \qweight{u}{0}$.  Furthermore, $\{u,v\}$ is also a block of $\pi$ (Observation~2). This contradicts the \ref{weak-sibling-property} of $(F,\pi)$. 
    \item In order to verify the \ref{weak-forbidden-weights-property}, for the sake of contradiction suppose that there is a $v \in V_i\backslash \{0\}$ such that $|B_{i,\pi_i(v)}| = 1$ and either:
    \begin{enumerate}
        \item $\|\pweight{v}{i}\|_1 = 1, \|\qweight{v}{i}\|_1 = 1$ or,
        \item $\pweight{v}{i} = 2 e_j, \qweight{u}{i} = 0$ for some $j \in [\order]$, where $e_{1:\order}$ denote the standard basis vectors in $\R^{\order}$. 
        \item $\pweight{u}{i}= e_{j} + e_{j^\prime}, \qweight{u}{i} = 0$ for some $j,j^\prime \in [\order], \; j \neq j^\prime$ such that $\Omega_{jj^\prime} \neq 0$. Here, $e_{1:\order}$ denote the standard basis vectors in $\R^{\order}$ and $\Omega \in \R^{\order \times \order}$ is the limiting covariance matrix corresponding to the semi-random ensemble $\mM_{1:k}$ (cf. \defref{semirandom}). 
    \end{enumerate} 
    By Observation 1, $|B_{{\pi}(v)}| = 1$ and hence $v \not\in \{u_\star, u_\star^\prime, v_\star, v_\star^\prime\}$. This means that $\pweight{v}{} = \pweight{v}{i}$ and $\qweight{v}{} = \qweight{v}{i}$. This leads to a contradiction of the \ref{weak-forbidden-weights-property} for $(F,\pi)$.
    \item In order to verify the \ref{original-leaf-property}, for the sake of contradiction, assume that there is a leaf ${v \in \leaves{F_i}}$ with $|B_{i,\pi_i(v)}| = 1$ and $\|\qweight{v}{i}\|_1 \geq 1$. By Observation 1, $|B_{\pi(v)}| = 1$. Observe that $v \not\in \{u_\star, u_\star^\prime, v_\star, v_\star^\prime\}$ since these vertices belong to blocks of size at least $2$ in $\pi$. Hence $v \in \leaves{F}$ (see Observation 2), $\qweight{v}{} = \qweight{v}{0}$. The existence of such a $v$ contradicts the \ref{original-leaf-property} of $(F,\pi)$.
    \item  Lastly, we verify the \ref{parity-property}. For the configuration $(F_i, \pi_i)$ for $i \geq 1$, we recall that $F_i = F$ and $\pi_i$ is obtained by merging some blocks of $\pi$. Notice that the \ref{parity-property} is not disturbed by merging some blocks. Hence, $(F_i, \pi_i)$ automatically satisfies \ref{parity-property} for $i \geq 1$. Now, we verify the configuration $(F_0, \pi_0)$ also satisfies \ref{parity-property}. Recall that $\pi_0 = \{B_{0,1}, B_{0,2}, \dots, B_{0,|\pi|-1}\}$ where $B_{0,i} = B_{i}$. Observe that for any $i < |\pi|-1$, since $u_\star, u_\star^\prime \not\in {B}_{0,i}$,
    \begin{align*}
         \pweight{u}{0} = \pweight{u}{}, \quad \qweight{u}{0} = \qweight{u}{} \; \forall \; u \; \in \; B_{0,i} = B_{i}. 
    \end{align*}
    Furthermore, recalling Observation 2 and the fact that $u_\star, u_\star^\prime \not\in {B}_{0,i}$, we have for any $i < |\pi| - 1$:
    \begin{subequations} \label{eq:set-relations}
    \begin{align}
        B_{0,i} \cap \leaves{F_0}   &= {B}_{i} \cap \leaves{F}, \\
        B_{0,i} \backslash (\leaves{F_0} \cup \{0\}) & = B_{i} \backslash (\leaves{F} \cup \{0\}). 
    \end{align}
    \end{subequations}
    Hence, for any $i < |\pi| -1$, we have,
    \begin{align*}
           &\|\qweight{0}{0}\|_1 \cdot \vone_{0 \in B_{0,i}} + \left(\sum_{u \in B_{0,i} \backslash ( \leaves{F_0} \cup \{0\})} \|\pweight{u}{0}\|_1 + \|\qweight{u}{0}\|_1 \right) + \left( \sum_{v \in B_{0,i} \cap \leaves{F_0}} \|\pweight{v}{0}\|_1 \right) \\
           & \qquad \qquad = \|\qweight{0}{}\|_1 \cdot \vone_{0 \in B_{i}} + \left(\sum_{u \in B_{i} \backslash ( \leaves{F} \cup \{0\})} \|\pweight{u}{}\|_1 + \|\qweight{u}{}\|_1 \right) + \left( \sum_{v \in B_{i} \cap \leaves{F}} \|\pweight{v}{}\|_1 \right).
    \end{align*}
    Note that the RHS of the above display is even because $(F,\pi)$ satisfies the \ref{parity-property}. Hence, we have verified for all blocks of $\pi_0$ except $B_{0,|\pi|-1}$. Lastly, we verify the \ref{parity-property} for $B_{0,|\pi|-1}$. Recall that $\{u_\star, u_\star^\prime\} \subset B_{0,|\pi|-1} = B_{|\pi|-1}$ and $\|\qweight{u_\star}{0}\|_1 = \|\qweight{u_\star}{}\|_1 - 1$, $\|\qweight{u_\star^\prime}{0}\|_1 = \|\qweight{u_\star^\prime}{}\|_1 - 1$ (cf. \eqref{eq:ensure-conservation}). If $u_\star, u_\star^\prime \not\in \leaves{F_0}$, then $\leaves{F} = \leaves{F_0}$ and the set equalities \eqref{eq:set-relations} continue to hold for $i=|\pi| - 1$ and we have:
    \begin{align} \label{eq:parity-property-special-block}
         &\|\qweight{0}{0}\|_1 \cdot \vone_{0 \in B_{0,|\pi| - 1}} + \left(\sum_{u \in B_{0,|\pi| - 1} \backslash ( \leaves{F_0} \cup \{0\})} \|\pweight{u}{0}\|_1 + \|\qweight{u}{0}\|_1 \right) + \left( \sum_{v \in B_{0,|\pi| - 1} \cap \leaves{F_0}} \|\pweight{v}{0}\|_1 \right) \nonumber \\
           & \quad = \|\qweight{0}{}\|_1 \cdot \vone_{0 \in B_{|\pi| - 1}} + \left(\sum_{u \in B_{|\pi| - 1} \backslash ( \leaves{F} \cup \{0\})} \|\pweight{u}{}\|_1 + \|\qweight{u}{}\|_1 \right) + \left( \sum_{v \in B_{|\pi| - 1} \cap \leaves{F}} \|\pweight{v}{}\|_1 \right) - 2.
    \end{align}
    In the above display the term $-2$ accounts for the fact that $\|\qweight{u_\star}{0}\|_1 = \|\qweight{u_\star}{}\|_1 - 1$, $\|\qweight{u_\star^\prime}{0}\|_1 = \|\qweight{u_\star^\prime}{}\|_1 - 1$. Furthermore, due to Observation~3, \eref{parity-property-special-block} holds even if $u_\star \in \leaves{F_0}$ or $u_\star^\prime \in \leaves{F_0}$ (or both) since $\|\qweight{u_\star}{0}\|_1 = 0$ or $\|\qweight{u_\star^\prime}{0}\|_1 = 0$ (or both) in these scenarios. Since $(F,\pi)$ satisfies the \ref{parity-property}, the RHS of \eqref{eq:parity-property-special-block} is even. This verifies the \ref{parity-property} for all blocks of the configuration $(F_0, \pi_0)$.
\end{enumerate}
This concludes the proof of \lemref{eliminate-pair}. 
\end{proof}

\section{Concentration Analysis}\label{appendix:concentration}
This appendix provides a proof for the concentration estimate for the MVAMP iteration stated in \thref{variance}. We begin by introducing some useful notations that we use through out this appendix.  
\paragraph{Some Additional Notation: } Observe that there are three sources of randomness in the MVAMP iterations \eqref{eq:mvamp}:
\begin{enumerate}
    \item The random sign diagonal matrix $\mS = \diag(s_{1:\dim})$ with $s_{1:\dim} \explain{i.i.d.}{\sim} \unif{\{\pm 1\}}$ used to generate the semi-random ensemble $\mM_{1:\order}$. 
    \item The $\dim \times \auxdim$ matrix of side information $\auxmat$, whose rows $\auxvec_{1:\dim} \in \R^{\auxdim}$ are i.i.d. copies of a random vector $\serv{A}$ (cf. \assumpref{side-info}). 
    \item The $\dim \times \order$ Gaussian matrix $\mG$ whose rows $g_{1:\dim}$ are sampled i.i.d. from $\gauss{0}{I_{\order}}$ which is used to generate the initialization $\iter{\mZ}{0,\cdot} := \mG$.
\end{enumerate}
In this section, we will make the dependence of the MVAMP iterates on these random variables explicit by using the notations:
\begin{align*}
    \iter{\vz}{t,i}(\mS, \auxmat, \mG), \quad \iter{\mZ}{t,\cdot}(\mS, \auxmat, \mG), \quad \iter{\mZ}{\cdot,i}(\mS, \auxmat, \mG) 
\end{align*}
when needed. 
In order to apply the Efron-Stein Inequality, we will need to develop estimates on the perturbation introduced in the iterates when a single sign, the corresponding row of $\auxmat$, and the corresponding row of $\mG$ is changed. Hence, we introduce the following setup. Let $(\hat{\mS}, \hat{\auxmat}, \hat{\mG})$ be an independent copy of $(\mS, \auxmat, \mG)$. Let $\hat{\mS} = \diag(\hat{s}_{1:\dim})$ and let the rows of $\hat{\auxmat}$ and $\hat{\mG}$ be denoted by $\hat{\auxvec}_{1:\dim} \in \R^{\auxdim}$ and $\hat{g}_{1:\dim} \in \R^{\order}$ respectively. For each $i \in [\dim]$ we define:
\begin{enumerate}
    \item A sign diagonal matrix $\iter{\mS}{i} \explain{def}{=} \diag(s_1, s_2, \dotsc, s_{i-1}, \hat{s}_i, s_{i+1}, \dotsc s_{\dim})$. 
    \item A $\dim \times \auxdim$ matrix $\iter{\auxmat}{i}$ with rows $\auxvec_1, \auxvec_2, \dotsc, \auxvec_{i-1}, \hat{\auxvec}_i, \auxvec_{i+1}, \dotsc, \auxvec_\dim$.
    \item A $\dim \times \order$ matrix $\iter{\mG}{i}$ with rows $g_1, g_2, \dotsc, g_{i-1}, \hat{g}_i, g_{i+1}, \dotsc, g_\dim$.
\end{enumerate}
Finally, for each $i \in [\dim]$, $t \in [T]$, and $j \in [\order]$, we introduce the perturbed iterations:
\begin{align} \label{eq:pert-iterate}
    \iter{\vz}{t,j | i} \explain{def}{=} \iter{\vz}{t,j}(\iter{\mS}{i}, \iter{\auxmat}{i}, \iter{\mG}{i}).
\end{align}
These are iterates generated by the sign diagonal matrix $\iter{\mS}{i}$, side information matrix $\iter{\auxmat}{i}$, and Gaussian initialization $\iter{\mG}{i}$. We use $ \iter{\vz}{t,j | \emptyset}$ to refer to the unperturbed iterations:
\begin{align} \label{eq:unpert-iterate}
     \iter{\vz}{t,j | \emptyset }  \explain{def}{=} \iter{\vz}{t,j}(\mS, \auxmat, \mG).
\end{align}
These are iterates generated by the sign diagonal matrix ${\mS}{}$, side information matrix ${\auxmat}{}$, and Gaussian initialization ${\mG}{}$. 
Finally, we define the perturbation error vectors as follows:
\begin{align} \label{eq:pert-errors}
    \iter{\vDelta}{t,j | i} \explain{def}{=} \iter{\vz}{t,j | i} - \iter{\vz}{t,j | \emptyset} \in \R^{\dim}
\end{align}

\paragraph{}The following lemma presents the estimates on the perturbation error vectors required by Efron-Stein Inequality. 

\begin{lemma}[Perturbation Bounds]\label{lem:perturbation-estimate-efron-stein} Suppose that the polynomial nonlinearities $\nonlin_{1:\order}$ satisfy the requirements stated in \thref{moments}. Then, for any fixed $T \in \N$, 
\begin{subequations}
\begin{align}
   \max_{t \leq T} \max_{j \in [\order]} \max_{i \in [\dim]} \left\|\iter{\vDelta}{t,j | i} \right\| &\leq C \cdot \rho_T(\mPsi_{1:\order}, \mS, \hat{\mS}, \auxmat, \hat{\auxmat}, \mG, \hat{\mG})^{T+1}, \label{eq:pert-claim1} \\
    \max_{t \leq T}  \max_{j \in [\order]} \left\| \sum_{i=1}^\dim \iter{\vDelta}{t,j | i}  \cdot {\iter{\vDelta}{t,j | i}} \tran  \right\|_{\op} & \leq C \cdot \rho_T(\mPsi_{1:\order}, \mS, \hat{\mS}, \auxmat, \hat{\auxmat}, \mG, \hat{\mG})^{2(T+1)^2}, \label{eq:pert-claim2}
\end{align}
\end{subequations}
where,
\begin{align}
    \rho_T(\mPsi_{1:\order}, \mS, \hat{\mS}, \auxmat, \hat{\auxmat}, \mG, \hat{\mG}) &\explain{def}{=}  (1 + \max_{j \in [\order]} \|\mPsi_i\|_{\op})  \times (1 + \|\mG\|_{\infty} + \|\hat{\mG}\|_{\infty})\times  (1 + \|\auxC(\auxmat)\|_{\infty} + \|\auxC(\hat{\auxmat})\|_{\infty}) \nonumber\\  & \qquad \qquad \qquad \qquad \qquad  \times (1 + \max_{\substack{i \in [\dim], t \leq [T]}} \|\iter{\mZ}{t,\cdot | i}\|_{\infty}^\degree + \max_{t \leq T} \|\iter{\mZ}{t,\cdot | \emptyset }\|_{\infty}^\degree). \label{eq:rho-def}
\end{align}
In the above display,
\begin{enumerate}
    \item $C$ is a finite constant that depends only on $T$, $\order$, $\degree$ (maximum degree of the polynomials $\nonlin_{1:k}(\cdot\; ;\;  \auxvec)$ from the statement of \thref{moments}).
    \item $\auxC : \R^{\auxdim} \mapsto [0,\infty)$ is a function (independent of $\dim$) that is determined by $\nonlin_{1:\order}$ and satisfies $\E[|\auxC(\serv{A})|^p] < \infty$ for each $p \in \W$. 
    \item The various norms are defined as follows:
\end{enumerate}
\begin{align*}
    \|\mG\|_{\infty} & \explain{def}{=} \max_{j \in [\order]\ell \in [\dim]} |G_{\ell j}|, \quad \|\hat{\mG}\|_{\infty} \explain{def}{=} \max_{j \in [\order]\ell \in [\dim]} |\hat{G}_{\ell j}| \\
    \|\auxC(\auxmat)\|_{\infty} & \explain{def}{=} \max_{\ell \in [\dim]} |\auxC(\auxvec_\ell)|, \quad \|\auxC(\hat{\auxmat})\|_{\infty}  \explain{def}{=} \max_{\ell \in [\dim]} |\auxC(\hat{\auxvec}_\ell)|, \\
    \|\iter{\mZ}{t,\cdot | i}\|_{\infty} & \explain{def}{=} \max_{j \in [\order],\ell \in [\dim]} |\iter{z}{t,j | i}_{\ell}|, \quad \|\iter{\mZ}{t,\cdot | \emptyset}\|_{\infty}  \explain{def}{=} \max_{j \in [\order],\ell \in [\dim]} |\iter{z}{t,j |. \emptyset}_{\ell}|.    
\end{align*}
\end{lemma}
\begin{proof} See the end of this appendix (\appref{perturbation-subsection}) for a proof.
\end{proof}

Since the above perturbation estimates are stated in terms of $ \rho_T(\mPsi_{1:\order}, \mS, \hat{\mS}, \auxmat, \hat{\auxmat}, \mG, \hat{\mG})$ defined in \eqref{eq:rho-def}, we will find the following moment estimates useful.

\begin{lemma}\label{lem:delocalization}
For any fixed $p, T \in \W$ and $\epsilon \in (0,1)$, we have,
\begin{align*}
   \E \left[ |\rho_T(\mPsi_{1:\order}, \mS, \hat{\mS}, \auxmat, \hat{\auxmat}, \mG, \hat{\mG})|^{p}\right]   & \lesssim \dim^\epsilon,
\end{align*}
where $\rho_T(\mPsi_{1:\order}, \mS, \hat{\mS}, \auxmat, \hat{\auxmat}, \mG, \hat{\mG})$ is as defined in \eqref{eq:rho-def}.
\end{lemma}
\begin{proof}
 See the end of this appendix (\appref{delocalization-subsection}) for a proof.
\end{proof}
The proof of these intermediate results are deferred to the end of this section. We now provide a proof for the variance bound claimed in \thref{variance}. 

\begin{proof}[Proof of \thref{variance}] Throughout the proof, we will use $C$ to denote a constant that depends only on $T,\degree, k$ and this constant may change from one line to the next. We will find it convenient to introduce the definition:
\begin{align*}
    H_T(\mS, \auxmat, \mG) \explain{def}{=}  \frac{1}{\dim}\sum_{\ell=1}^\dim h(\auxvec_\ell) \cdot  \hermite{r}\left(\iter{z}{T,\cdot}_\ell(\mS, \auxmat, \mG)\right) \explain{(a)}{=} \frac{1}{\dim}\sum_{\ell=1}^\dim h(\auxvec_\ell) \cdot  \hermite{r}\left(\iter{z}{T,\cdot | \emptyset }_\ell\right).
\end{align*}
Analogously for each $i \in [\dim]$, we define:
\begin{align*}
     H_T(\iter{\mS}{i}, \iter{\auxmat}{i}, \iter{\mG}{i}) &\explain{def}{=}  \frac{1}{\dim}\sum_{\substack{\ell=1 \\ \ell \neq i}}^\dim h(\auxvec_\ell) \cdot  \hermite{r}\left(\iter{z}{T,\cdot}_\ell(\iter{\mS}{i}, \iter{\auxmat}{i}, \iter{\mG}{i})\right) + \frac{1}{\dim} \cdot h(\hat{\auxvec}_i) \cdot  \hermite{r}\left(\iter{z}{T,\cdot}_\ell(\iter{\mS}{i}, \iter{\auxmat}{i}, \iter{\mG}{i})\right)  \\&\explain{(a)}{=} \frac{1}{\dim}\sum_{\substack{\ell=1 \\ \ell \neq i}}^\dim h(\auxvec_\ell) \cdot  \hermite{r}\left(\iter{z}{T,\cdot | i}_\ell\right) + \frac{1}{\dim} \cdot h(\hat{\auxvec}_i) \cdot  \hermite{r}\left(\iter{z}{T,\cdot | i}_i\right) 
\end{align*}
The equalities marked (a) follow from recalling the definitions of ${\iter{\mZ}{t,\cdot | i}}$ and ${\iter{\mZ}{t,\cdot | \emptyset}}$ from \eqref{eq:pert-iterate} and \eqref{eq:unpert-iterate}. By the Efron-Stein Inequality,
\begin{align*}
    2\Var[H_T(\mS, \auxmat, \mG)] & \leq  \sum_{i=1}^\dim \E[(H_T(\mS, \auxmat, \mG)-H_T(\iter{\mS}{i}, \iter{\auxmat}{i}, \iter{\mG}{i}))^2] \leq 2 \cdot \{ (\star) + (\dagger) \},
\end{align*}
where we defined the terms $(\star)$ and $(\dagger)$ as follows:
\begin{align*}
    (\star) & \explain{def}{=} \sum_{i=1}^\dim \E\left[ \left( \frac{1}{\dim}\sum_{\ell=1}^\dim h(\auxvec_\ell) \cdot  \left\{ \hermite{r}(\iter{z}{T,\cdot | i }_\ell) - \hermite{r}(\iter{z}{T,\cdot | \emptyset }_\ell) \right\} \right)^2 \right], \\
    (\dagger) & \explain{def}{=}  \frac{1}{\dim^2} \sum_{i=1}^{\dim} \E\left[  \hermite{r}^2\left(\iter{z}{T,\cdot | i}_i\right) \cdot ( h(\auxvec_i) - h(\hat{\auxvec_i}))^2 \right].
\end{align*}
In order to complete the proof, we show that $ \lim_{\dim \rightarrow \infty} (\star) = \lim_{\dim \rightarrow \infty} (\dagger) = 0$ by analyzing each of these terms individually.
\paragraph{Analysis of $(\dagger)$.} Consider the following estimates:
\begin{align*}
     (\dagger) & \leq \frac{2}{\dim^2} \sum_{i=1}^{\dim} \E\left[  \hermite{r}^2\left(\iter{z}{T,\cdot | i}_i\right) \cdot  h^2(\auxvec_i) \right] + \frac{2}{\dim^2} \sum_{i=1}^{\dim} \E\left[  \hermite{r}^2\left(\iter{z}{T,\cdot | i}_i\right) \cdot  h^2(\hat{\auxvec}_i) \right] \\
     & \explain{(a)}{=} \frac{2 \cdot \E[h^2(\serv{A})]}{\dim^2} \cdot \sum_{i=1}^{\dim} \E\left[  \hermite{r}^2\left(\iter{z}{T,\cdot | i}_i\right)    \right] + \frac{2}{\dim^2} \sum_{i=1}^{\dim} \E\left[  \hermite{r}^2\left(\iter{z}{T,\cdot | i}_i\right) \cdot  h^2(\hat{\auxvec}_i) \right] \\
     & \explain{(b)}{=}  \frac{2 \cdot \E[h^2(\serv{A})]}{\dim^2} \cdot \sum_{i=1}^{\dim} \E\left[  \hermite{r}^2\left(\iter{z}{T,\cdot | \emptyset}_i\right)    \right] + \frac{2}{\dim^2} \sum_{i=1}^{\dim} \E\left[  \hermite{r}^2\left(\iter{z}{T,\cdot | \emptyset}_i\right) \cdot  h^2({\auxvec}_i) \right].
\end{align*}
In the above display (a) follows by observing that $\iter{z}{T,\cdot | i}_i$ and $\auxvec_i$ are independent. Step (b) follows from observing that $(\iter{z}{T,\cdot | i}_i, \hat{\auxvec}_i) \explain{d}{=} (\iter{z}{T,\cdot | \emptyset}_i, \auxvec_i)$. By \thref{moments},
\begin{align*}
    \lim_{\dim \rightarrow \infty } \E\left[\frac{1}{\dim}  \sum_{i=1}^{\dim}   \hermite{r}^2\left(\iter{z}{T,\cdot | \emptyset}_i\right)    \right] & = \E \hermite{r}^2(\serv{Z}) = 1, \\
    \lim_{\dim \rightarrow \infty } \E\left[\frac{1}{\dim}  \sum_{i=1}^{\dim}   \hermite{r}^2\left(\iter{z}{T,\cdot | \emptyset}_i\right) \cdot  h^2({\auxvec}_i)   \right] & = \E[ \hermite{r}^2(\serv{Z})]\cdot \E[h^2(\serv{A})] = \E[h^2(\serv{A})],
\end{align*}
where $\serv{Z} \sim \gauss{0}{I_\order}$. Hence $\lim_{\dim \rightarrow \infty} (\dagger) \leq 4 \lim_{\dim \rightarrow \infty}  \E[h^2(\serv{A})]/ \dim$ and in particular, $\lim_{\dim \rightarrow \infty}(\dagger) = 0$, as required. 
\paragraph{Analysis of $(\star)$.} Recalling the definition of $\iter{\vDelta}{t,j|i}$ from \eqref{eq:pert-errors}, we obtain using Taylor's expansion:
\begin{align} \label{eq: taylor-expansion}
    \frac{1}{\dim}\sum_{\ell=1}^\dim h(\auxvec_\ell) \cdot  \left\{ \hermite{r}(\iter{z}{T,\cdot | i }_\ell) - \hermite{r}(\iter{z}{T,\cdot | \emptyset }_\ell) \right\} & \explain{}{=} \frac{1}{\dim} \sum_{\ell=1}^\dim \sum_{j=1}^\order \partial_j \hermite{r}(\iter{z}{T,\cdot| \emptyset}_{\ell}) \cdot \iter{\Delta}{T,j | i}_{\ell} + \frac{1}{\dim}\sum_{\ell=1}^\dim \iter{\epsilon}{T|i}_{\ell} \\
    & \explain{(b)}{=} \frac{1}{\dim}\sum_{j=1}^\order \ip{\partial_j \hermite{r}(\iter{\mZ}{T,\cdot| \emptyset})}{\iter{\vDelta}{T,j | i}} + \frac{1}{\dim}\sum_{\ell=1}^\dim \iter{\epsilon}{T|i}_{\ell}
\end{align}
In the above display, in the step marked (a), $\partial_j \hermite{r}$ denotes the partial derivative of $\hermite{r}(z_1, \dotsc, z_\order)$ with respect to $z_j$. In step (b), we introduced the notations $\partial_j \hermite{r}(\iter{\mZ}{T,\cdot| \emptyset})$ to denote the vector:
\begin{align*}
    \partial_j \hermite{r}(\iter{\mZ}{T,\cdot| \emptyset}) & \explain{def}{=} \big[ \partial_j \hermite{r}(\iter{z}{T,\cdot| \emptyset}_{1}),  \partial_j \hermite{r}(\iter{z}{T,\cdot| \emptyset}_{2}), \dotsc,  \partial_j \hermite{r}(\iter{z}{T,\cdot| \emptyset}_{\dim})\big]\tran.
\end{align*}
Note that because $\hermite{r}$ is a polynomial of degree $\|r\|_1$, the second order error term in the Taylor's expansion can be controlled by:
\begin{align} 
    |\iter{\epsilon}{T|i}_{\ell}| & \leq C \cdot \left(1 + \|\iter{z}{T,\cdot | i}_{\ell}\|^{\|r\|_1}_{\infty} + \|\iter{z}{T,\cdot | \emptyset}_{\ell}\|^{\|r\|_1}_{\infty} \right)\cdot \left(\sum_{j=1}^{\order} |\iter{\Delta}{T,j | i}_{\ell}|^2 \right) \nonumber \\
    & \explain{\eqref{eq:rho-def}}{\leq} C \cdot \rho_T(\mS, \hat{\mS}, \auxmat, \hat{\auxmat}, \mG, \hat{\mG})^{\|r\|_1} \cdot \left(\sum_{j=1}^{\order} |\iter{\Delta}{T,j | i}_{\ell}|^2 \right). \label{eq:taylor-thm-remainder-bound}
\end{align}
In the above display $C$ is a finite constant that depends only on $\|r\|_1$ (which is fixed). Hence we can further upper bound $(\star)$ by:
\begin{align*}
    (\star) & \leq 2 \cdot \underbrace{\E\left[ \sum_{i=1}^\dim \left(\frac{1}{\dim}\sum_{j=1}^\order \ip{\partial_j \hermite{r}(\iter{\mZ}{T,\cdot| \emptyset})}{\iter{\vDelta}{T,j | i}}  \right)^2 \right]}_{(\clubsuit)} + 2 \cdot \underbrace{\E\left[\sum_{i=1}^\dim \left(\frac{1}{\dim}\sum_{\ell=1}^\dim \iter{\epsilon}{T|i}_{\ell}\right)^2 \right]}_{(\diamondsuit)}.
\end{align*}
In order to show that $\lim_{\dim \rightarrow \infty} (\star) = 0$, we need to show that $\lim_{\dim \rightarrow \infty} (\clubsuit) = 0$ and $\lim_{\dim \rightarrow \infty} (\diamondsuit) = 0$. We first consider the term $(\clubsuit)$.

\begin{align*}
    (\clubsuit)  &\explain{def}{=} \E\left[ \sum_{i=1}^\dim \left(\frac{1}{\dim}\sum_{j=1}^\order \ip{\partial_j \hermite{r}(\iter{\mZ}{T,\cdot| \emptyset})}{\iter{\vDelta}{T,j | i}}  \right)^2 \right]  \explain{(a)}{\leq} \frac{\order}{\dim^2} \cdot \sum_{j=1}^\order \E\left[ \sum_{i=1}^\dim \ip{\partial_j \hermite{r}(\iter{\mZ}{T,\cdot| \emptyset})}{\iter{\vDelta}{T,j | i}}^2  \right] \\
    & \hspace{5cm}\leq \frac{\order}{\dim^2} \cdot \sum_{j=1}^\order \E\left[  \left\| \sum_{i=1}^\dim \iter{\vDelta}{T,j | i}{\iter{\vDelta}{T, j | i}}\tran \right\|_{\op} \cdot \|\partial_j \hermite{r}(\iter{\mZ}{T,\cdot| \emptyset}) \|^2  \right] \\
    &\hspace{5cm}\explain{(b)}{\leq} \frac{\order}{\dim^2} \cdot \sum_{j=1}^\order \E\left[  \left\| \sum_{i=1}^\dim \iter{\vDelta}{T,j | i}{\iter{\vDelta}{T, j | i}}\tran \right\|_{\op} \cdot C  \cdot \dim \cdot \big(1 + \|\iter{\mZ}{T, \cdot | \emptyset}\|_{\infty}^{\|r\|_1}\big)^2  \right] \\
    & \hspace{5cm}\explain{(c)}{\leq}\frac{C}{\dim} \cdot \sum_{j=1}^\order \E\left[ \rho_T(\mS, \hat{\mS}, \auxmat, \hat{\auxmat}, \mG, \hat{\mG})^{2(T+1)^2} \cdot   \rho_T(\mS, \hat{\mS}, \auxmat, \hat{\auxmat}, \mG, \hat{\mG})^{2\|r\|_1}    \right] \\
    &\hspace{5cm} \explain{(d)}{\lesssim} \dim^{-1+\epsilon}.
\end{align*}
In the above display, step (a) follows from the Cauchy-Schwarz Inequality, step (b) uses the fact that $\partial_j \hermite{r}$ is a polynomial of degree at most $\|r\|_1$ and hence satisfies an estimate of the form:
\begin{align*}
    |\partial_j \hermite{r}(z)|& \leq C ( 1 + \|z\|_{\infty}^{\|r\|_1}),
\end{align*}
for a constant $C$ that depends only on $\|r\|_1$. Step (c) uses estimates on the operator norm from \lemref{perturbation-estimate-efron-stein} and the definition of $\rho_T(\mS, \hat{\mS}, \auxmat, \hat{\auxmat}, \mG, \hat{\mG})$ from \eqref{eq:rho-def}. Finally step (d) relied on the moment estimates from \lemref{delocalization}. Hence, we have shown that $\lim_{\dim \rightarrow \infty} (\clubsuit) = 0$. Next, we consider the term $(\diamondsuit)$:
\begin{align*}
    (\diamondsuit) \explain{def}{=}  \E\left[\sum_{i=1}^\dim \left(\frac{1}{\dim}\sum_{\ell=1}^\dim \iter{\epsilon}{T|i}_{\ell}\right)^2 \right] &\leq \E\left[\sum_{i=1}^\dim \left(\frac{1}{\dim}\sum_{\ell=1}^\dim |\iter{\epsilon}{T|i}_{\ell}|\right)^2 \right]\\
    & \explain{\eqref{eq:taylor-thm-remainder-bound}}{\leq} \frac{C}{\dim^2} \sum_{i=1}^\dim \E\left[ \rho_T(\mS, \hat{\mS}, \auxmat, \hat{\auxmat}, \mG, \hat{\mG})^{2\|r\|_1} \cdot \left(\sum_{j=1}^\order \|\iter{\vDelta}{T,j|i}\|^2 \right)^2 \right] \\
    & \explain{(a)}{\leq} \frac{C \order}{\dim^2} \sum_{j=1}^\order\sum_{i=1}^\dim \E\left[ \rho_T(\mS, \hat{\mS}, \auxmat, \hat{\auxmat}, \mG, \hat{\mG})^{2\|r\|_1} \cdot  \|\iter{\vDelta}{T,j|i}\|^4 \right] \\
    & \explain{(b)}{\leq} \frac{C \order}{\dim} \sum_{j=1}^\order \E\left[  \rho_T(\mPsi_{1:\order}, \mS, \hat{\mS}, \auxmat, \hat{\auxmat}, \mG, \hat{\mG})^{2\|r\|_1+4T+4} \right]
    \\  &\explain{(c)}{\lesssim} \;  \dim^{-1+\epsilon}.
\end{align*}
In the above display, step (a) follows from Cauchy-Schwarz Inequality, step (b) relies on the perturbation estimate from \lemref{perturbation-estimate-efron-stein}, and step (c) uses the moment estimates from \lemref{delocalization}.
Hence, $\lim_{\dim \rightarrow \infty} (\diamondsuit) = 0$, which shows that $\lim_{\dim \rightarrow \infty} (\star) = 0$ and concludes the proof of this theorem. 
\end{proof}

\subsection{Proof of \lemref{perturbation-estimate-efron-stein}} \label{appendix:perturbation-subsection}
\begin{proof}[Proof of \lemref{perturbation-estimate-efron-stein}] We begin by noting that since the non-linearities are assumed to satisfy the requirements of \thref{moments}, \lemref{misc-PL} in \appref{misc} guarantees the existence of a function $\auxC : \R^{\auxdim} \mapsto [0,\infty)$ such that $\E[|\auxC(\serv{A})|^p] < \infty$ for each $p \in \W$ and for each $j \in [\order]$ the non-linearity $\nonlin_{j}$ satisfies the estimates:
\begin{subequations}\label{eq:misc-lemma-conclusion}
\begin{align}
    |\nonlin_j(z; \auxvec)| & \leq \auxC(\auxvec) \cdot (1 + \|z\|_{\infty}^\degree), \\
    \|\nabla_{z}\nonlin_j(z; \auxvec)\|_{\infty} & \leq \auxC(\auxvec) \cdot (1 + \|z\|_{\infty}^\degree), \\
     |\nonlin_j(z; \auxvec) - \nonlin_j(z^\prime; \auxvec)| & \leq \auxC(\auxvec) \cdot (1 + \|z\|_{\infty}^\degree + \|z^\prime\|_{\infty}^\degree ) \cdot \| z - z^\prime\|_{\infty}, \\
     |\nonlin_j(z; \auxvec) - \nonlin_j(z^\prime; \auxvec) - \ip{\nabla_z \nonlin_j(z; \auxvec)}{z-z^\prime}| & \leq \auxC(\auxvec) \cdot (1 + \|z\|_{\infty}^\degree + \|z^\prime\|_{\infty}^\degree ) \cdot \| z - z^\prime\|_{\infty}^2,
\end{align} 
\end{subequations}
for all $z, z^\prime \in \R^{\order}$. Throughout the proof, we will use $C$ to denote a constant that depends only on $T,\degree, k$ and this constant may change from one line to the next. We will also suppress the depends of $\rho_T$ defined in \eqref{eq:rho-def} on $(\mPsi_{1:\order}, \mS, \hat{\mS}, \auxmat, \hat{\auxmat}, \mG, \hat{\mG})$ for ease of notation.   

\paragraph{}Recalling the definitions of $\iter{\vz}{t,j | i}$ and $\iter{\vz}{t, j | \emptyset}$ from \eqref{eq:pert-iterate} and \eqref{eq:unpert-iterate} along with the MVAMP updates \eqref{eq:mvamp}:
\begin{align*}
    &\iter{\vDelta}{t,j | i} \explain{def}{=} \iter{\vz}{t,j | i} - \iter{\vz}{t,j | \emptyset} \\
     & = \iter{\mS}{i} \mPsi_j \iter{\mS}{i} \nonlin_j(\iter{\mZ}{t-1,\cdot | i}; \iter{\auxmat}{i}) - \mS \mPsi_j \mS \nonlin_j(\iter{\mZ}{t-1,\cdot | \emptyset}; {\auxmat}) \\
    & = \iter{\mS}{i} \mPsi_j \iter{\mS}{i} \nonlin_j(\iter{\mZ}{t-1,\cdot | i}; \iter{\auxmat}{i}) - \mS \mPsi_j \mS\nonlin_j(\iter{\mZ}{t-1,\cdot | i}; \auxmat) + \mS \mPsi_j \mS \cdot \{ \nonlin_j(\iter{\mZ}{t-1,\cdot | i}; \auxmat) - \nonlin_j(\iter{\mZ}{t-1,\cdot | \emptyset}; {\auxmat}) \}.
\end{align*}
We define,
\begin{subequations}\label{eq:alpha-beta-def}
\begin{align}
    \iter{\valpha}{t, j | i} &\explain{def}{=} \iter{\mS}{i} \mPsi_j \iter{\mS}{i} \nonlin_j(\iter{\mZ}{t-1,\cdot | i}; \iter{\auxmat}{i}) -  {\mS} \mPsi_j \iter{\mS}{i} \nonlin_j(\iter{\mZ}{t-1,\cdot | i}; \iter{\auxmat}{i}) \nonumber \\
    &= (1- \hat{s}_i s_i) \cdot \iter{z}{t,j | i}_i \cdot \ve_i, \\
    \iter{\vbeta}{t, j | i} &\explain{def}{=} {\mS} \mPsi_j \iter{\mS}{i} \nonlin_j(\iter{\mZ}{t-1,\cdot | i}; \iter{\auxmat}{i}) - \mS \mPsi_j \mS\nonlin_j(\iter{\mZ}{t-1,\cdot | i}; \auxmat) \nonumber\\&= (\hat{s}_i \nonlin_j(\iter{z}{t-1, \cdot | i}_i; \hat{\auxvec}_i)- {s}_i \nonlin_j(\iter{z}{t-1, \cdot | i}_i; {\auxvec}_i)) \cdot \mS \mPsi_j \ve_i, \\
    \iter{\vgamma}{t, j | i} &\explain{def}{=} \mS \mPsi_j \mS \cdot \{ \nonlin_j(\iter{\mZ}{t-1,\cdot | i}; \auxmat) - \nonlin_j(\iter{\mZ}{t-1,\cdot | \emptyset}; {\auxmat}) \}.
\end{align}
\end{subequations}
where, $\ve_{1:\dim}$ denote the standard basis of $\R^{\dim}$. Hence,
\begin{align}\label{eq:delta-decomp}
   \iter{\vDelta}{t,j | i} & = \iter{\valpha}{t, j | i} + \iter{\vbeta}{t, j | i} +  \iter{\vgamma}{t, j | i}. 
\end{align}
We prove each of the two claims in the lemma one by one. 
\paragraph{Proof of \eqref{eq:pert-claim1}.} By the triangle inequality:
\begin{align*}
    \|\iter{\vDelta}{t,j | i}\| & \leq \|\iter{\valpha}{t, j | i}\| + \|\iter{\vbeta}{t, j | i}\| +  \|\iter{\vgamma}{t, j | i}\|.
\end{align*}
Recalling the formula for $\iter{\valpha}{t, j | i}, \iter{\vbeta}{t, j | i}$ from \eqref{eq:alpha-beta-def}:
\begin{align*}
    \|\iter{\valpha}{t, j | i}\| & \leq 2|\iter{z}{t,j| i}_i| \explain{\eqref{eq:rho-def}}{\leq} C \rho_T. \\
    \|\iter{\vbeta}{t, j | i}\| & \leq \|\mPsi_j\|_{\op} \cdot \{|\nonlin_j(\iter{z}{t-1, \cdot | i}_i; \hat{\auxvec}_i)|+ |\nonlin_j(\iter{z}{t-1, \cdot | i}_i; {\auxvec}_i)|\}  \\
    & \explain{\eqref{eq:misc-lemma-conclusion}}{\leq} \|\mPsi_j\|_{\op} \cdot (\|\auxC(\auxmat)\|_{\infty} + \|\auxC(\hat{\auxmat})\|_{\infty}) \cdot ( 1 + \|\iter{\mZ}{t-1, \cdot | i}\|_{\infty}^\degree) \\
    & \explain{\eqref{eq:rho-def}}{\leq} \rho_T.
\end{align*}
Similarly, we can estimate $\|\iter{\vgamma}{t, j | i}\|$ as follows:
\begin{align*}
    \|\iter{\vgamma}{t, j | i}\| & \explain{def}{=} \|  \mS \mPsi_j \mS \cdot \{ \nonlin_j(\iter{\mZ}{t-1,\cdot | i}; \auxmat) - \nonlin_j(\iter{\mZ}{t-1,\cdot | \emptyset}; {\auxmat})  \}\| \\
    & \leq \|\mPsi_j\|_{\op} \cdot \| \nonlin_j(\iter{\mZ}{t-1,\cdot | i}; \auxmat) - \nonlin_j(\iter{\mZ}{t-1,\cdot | \emptyset}; {\auxmat})\| \\
    & \explain{\eqref{eq:misc-lemma-conclusion}}{\leq}  C \cdot \|\mPsi_j\|_{\op} \cdot \|\auxC(\auxmat)\|_{\infty} \cdot (1 + \|\iter{\mZ}{t-1,\cdot | i}\|_{\infty}^\degree + \|\iter{\mZ}{t-1,\cdot | \emptyset}\|_{\infty}^\degree) \cdot \left( \max_{j \in [\order]} \| \iter{\vDelta}{t-1,j | i}\| \right) \\
    & \explain{\eqref{eq:rho-def}}{\leq} C \cdot \rho_T \cdot \left( \max_{j \in [\order]} \| \iter{\vDelta}{t-1,j | i}\| \right).
\end{align*}
Hence, we have shown:
\begin{align*}
   \left( \max_{j \in [\order]} \| \iter{\vDelta}{t,j | i}\| \right) & \leq C \rho_T + C \rho_T \left( \max_{j \in [\order]} \| \iter{\vDelta}{t-1,j | i}\| \right).
\end{align*}
Unrolling this recursive estimate gives:
\begin{align*}
    \left( \max_{j \in [\order]} \| \iter{\vDelta}{t,j | i}\| \right) & \leq C \rho_T + (C\rho_T)^2 + \dotsb + (C\rho_T)^t \left( \max_{j \in [\order]} \| \iter{\vDelta}{0,j | i}\| \right) \\
    & =  C\rho_T + (C\rho_T)^2 + \dotsb + (C\rho_T)^t \cdot  \underbrace{\|g_{i} - \hat{g}_i\|_{\infty}}_{\leq \rho_T}.
\end{align*}
Hence, 
\begin{align*}
    \max_{t \leq T} \max_{j \in [\order]} \max_{i \in [\dim]} \| \iter{\vDelta}{t,j | i}\| & \leq C \rho_T^{T+1},
\end{align*}
as claimed. 
\paragraph{Proof of \eqref{eq:pert-claim2}} Applying Taylor's Theorem to $\iter{\vgamma}{t, j | i}$ yields:
\begin{align}
    \iter{\vgamma}{t, j | i} & \explain{def}{=} \mS \mPsi_j \mS \cdot \{ \nonlin_j(\iter{\mZ}{t-1,\cdot | i}; \auxmat) - \nonlin_j(\iter{\mZ}{t-1,\cdot | \emptyset}; {\auxmat})  \} \nonumber \\
    & = \mS \mPsi_j \mS \cdot \left\{ \sum_{r=1}^\order \iter{\mD}{t-1,j,r} \cdot \iter{\vDelta}{t-1, r | i} + \iter{\vepsilon}{t,j| i}  \right\} \label{eq:taylor-expand}
\end{align}
In \eqref{eq:taylor-expand}, for each $r \in [\order]$, $\iter{\mD}{t-1,j,r}$ is a diagonal matrix given by:
\begin{align*}
    \iter{\mD}{t-1,j,r} \explain{def}{=} \diag(\partial_r\nonlin_j (\iter{z}{t-1, \cdot| \emptyset}_{1}; \auxvec_1), \partial_r\nonlin_j (\iter{z}{t-1, \cdot | \emptyset}_{2}; \auxvec_2), \dotsc, \partial_r\nonlin_j (\iter{z}{t-1, \cdot | \emptyset}_{\dim}; \auxvec_{\dim})),
\end{align*}
where $\partial_r \nonlin_j(z_1, z_2, \dotsc, z_\order; a)$ denotes the derivative of $\nonlin_j(z_1, z_2, \dotsc, z_\order; a)$ with respect to $z_r$. Observe that,
\begin{align} \label{eq:D-op-norm}
    \|\iter{\mD}{t-1,j,r}\|_{\op} \explain{\eqref{eq:misc-lemma-conclusion}}{\leq} \|\auxC(\auxmat)\|_{\infty} \cdot (1 + \|\iter{\mZ}{t-1,\cdot | \emptyset}\|_{\infty}^{\degree}).
\end{align}
Furthermore, $\iter{\vepsilon}{t, j  | i}$ in \eqref{eq:taylor-expand} is the Taylor's remainder:
\begin{align*}
   \iter{\vepsilon}{t, j  | i} \explain{def}{=} \nonlin_j(\iter{\mZ}{t-1,\cdot | i}; \auxmat) - \nonlin_j(\iter{\mZ}{t-1,\cdot | \emptyset}; {\auxmat})  - \sum_{r=1}^\order \iter{\mD}{t-1,j,r} \cdot \iter{\vDelta}{t-1, r | i},
\end{align*}
which can be bounded entry-wise by:
\begin{align}\label{eq:taylor-error-bound}
     |\iter{\epsilon}{t, j  | i}_{\ell}| & \explain{\eqref{eq:misc-lemma-conclusion}}{\leq} \auxC(\auxvec_{\ell}) \cdot (1 + \|\iter{z}{t-1,\cdot | i}_{\ell}\|_{\infty}^\degree + \|\iter{z}{t-1,\cdot | \emptyset}_{\ell}\|_{\infty}^\degree ) \cdot   \sum_{r = 1}^{\order} |\iter{\Delta}{t-1,r | i}_{\ell}|^2
\end{align}
Hence, the decomposition \eqref{eq:delta-decomp} can be written as:
\begin{align} \label{delta-decomp-improved}
   \iter{\vDelta}{t,j | i} & = \iter{\valpha}{t, j | i} + \iter{\vbeta}{t, j | i} +  \left( \sum_{r=1}^\order \mS \mPsi_j \mS  \iter{\mD}{t-1,j,r} \cdot \iter{\vDelta}{t-1, r | i}  \right) + \mS \mPsi_j \mS \cdot \iter{\vepsilon}{t,j| i} 
\end{align}
For brevity, we introduce the definition:
\begin{align}\label{eq:op-short}
    \iter{O}{t,j} \explain{def}{=}  \left\| \sum_{i=1}^\dim \iter{\vDelta}{t,j | i} \cdot {\iter{\vDelta}{t,j | i}}\tran  \right\|_{\op}. 
\end{align}
For a collection of vectors $\iter{\vv}{1:\dim}$ we use the notation $[\iter{\vv}{1:\dim}]$ to denote the matrix whose columns are $\iter{\vv}{1:\dim}$. Notice that applying the triangle inequality to \eqref{delta-decomp-improved} yields:
\begin{align}
    &\iter{O}{t,j}  = \|[\iter{\vDelta}{t,j | 1:\dim}]\|_{\op}^2 \nonumber\\
    & \explain{}{\leq} C  \left( \|[\iter{\valpha}{t,j | 1:\dim}]\|_{\op}^2 +  \|[\iter{\vbeta}{t,j | 1:\dim}]\|_{\op}^2 + \|\mS \mPsi_j \mS  [\iter{\vepsilon}{t,j | 1:\dim}]\|_{\op}^2 + \sum_{r=1}^{\order}  \|\mS \mPsi_j \mS \iter{\mD}{t-1,j,r} [\iter{\vDelta}{t-1,r | 1:\dim}]\|_{\op}^2 \right). \label{eq:L-recursion-initial}
\end{align}
We bound each term that appears in the above inequality. Recalling the definitions of $\iter{\valpha}{t,j | i}, \iter{\vbeta}{t,j | i}$ from \eqref{eq:alpha-beta-def}, we obtain:
\begin{align}
     \|[\iter{\valpha}{t,j| 1:\dim }]\|_{\op}^2 & \leq 4 \max_{i \in [\dim]} |\iter{z}{t,j | i}_i|^2 \leq C \rho_T^2, \label{eq:alpha-norm-ub} \\
      \|[\iter{\vbeta}{t,j| 1:\dim }]\|_{\op}^2 & \explain{\eqref{eq:misc-lemma-conclusion}}{\leq}  \|\mPsi\|_{\op}^2 \cdot (\|\auxC(\auxmat)\|_{\infty} + \|\auxC(\hat{\auxmat})\|_{\infty})^2 \cdot (1+\|\iter{\mZ}{t-1,\cdot | i}\|_{\infty}^\degree)^2  \nonumber \\ 
      & \explain{\eqref{eq:rho-def}}{\leq} \rho_{T}^2.  \label{eq:beta-norm-ub}
\end{align}
For each $r \in [\order]$, we have the estimate:
\begin{align}
   \|\mS \mPsi_j \mS \iter{\mD}{t-1,j,r} \cdot [\iter{\vDelta}{t-1,r | 1:\dim}]\|_{\op}^2  & = \left\| \mS \mPsi_j \mS \iter{\mD}{t-1,j,r} \cdot \left(\sum_{i=1}^\dim \iter{\vDelta}{t-1,r | i} \cdot {\iter{\vDelta}{t-1,r | i}}\tran  \right) \cdot \iter{\mD}{t-1,j,r}  \mS  \mPsi_j\tran \mS \right\|_{\op} \nonumber\\
    & \leq \|\iter{\mD}{t-1,j,r} \|_{\op}^2 \cdot \|\mPsi_j\|_{\op}^2 \cdot \iter{O}{t-1,r}\nonumber \\
    & \explain{\eqref{eq:D-op-norm}}{\leq} \|\auxC(\auxmat)\|_{\infty}^2 \cdot (1 + \|\iter{\mZ}{t-1,\cdot | \emptyset}\|_{\infty}^{\degree})^2 \cdot \|\mPsi\|_{\op}^2 \cdot \iter{O}{t-1,r} \nonumber \\
    & \explain{\eqref{eq:rho-def}}{\leq} \rho_T^2 \cdot \iter{O}{t-1,r}.\label{eq:delta-norm-ub}
\end{align}
Lastly, we control:
\begin{align*}
   \|\mS \mPsi_j \mS  [\iter{\vepsilon}{t,j | 1:\dim}]\|_{\op}^2 & = \left\| \mS \mPsi_j \mS \cdot \left( \sum_{i=1}^\dim \iter{\vepsilon}{t,j | i} {\iter{\vepsilon}{t,j | i}}\tran\right) \mS \mPsi\tran_j \mS \right\|_{\op} \leq \|\mPsi_j\|_{\op}^2 \cdot \left\| \sum_{i=1}^\dim \iter{\vepsilon}{t,j | i} {\iter{\vepsilon}{t, j | i}}\tran \right\|_{\op}.
\end{align*}
Gershgorin's Circle theorem implies that the spectral norm of a symmetric matrix can be bounded by the largest row $\ell_1$ norm. Hence,
\begin{align*}
    \|\mS \mPsi_j \mS  [\iter{\vepsilon}{t,j | 1:\dim}]\|_{\op}^2 &\leq \|\mPsi_j\|_{\op}^2 \cdot \max_{\ell \in [\dim]}\left( \sum_{\ell^\prime=1}^{\dim}\sum_{i=1}^\dim |\iter{\epsilon}{t,j | i}_\ell| |\iter{\epsilon}{t, j | i}_{\ell^\prime}| \right).
\end{align*}
Recalling the estimates on $|\iter{\epsilon}{t,j| i}_\ell|$ from \eqref{eq:taylor-error-bound}, we obtain,
\begin{align}
   &\|\mS \mPsi_j \mS  [\iter{\vepsilon}{t,j | 1:\dim}]\|_{\op}^2 \nonumber \\ & \leq C\|\mPsi_j\|_{\op}^2 \|\auxC(\auxmat)\|_{\infty}^2  (1 + \|\iter{\mZ}{t-1,\cdot | i}\|_{\infty}^\degree + \|\iter{\mZ}{t-1,\cdot | \emptyset}\|_{\infty}^\degree )^2 \cdot \max_{r,r^\prime \in [\order]}\max_{\ell \in [\dim]}\left( \sum_{\ell^\prime=1}^{\dim}\sum_{i=1}^\dim |\iter{\Delta}{t-1,r | i}_\ell |^2 |\iter{\Delta}{t-1, r^\prime | i}_{\ell^\prime}|^2 \right)\nonumber\\
    & =C\rho_T^2 \cdot \max_{r,r^\prime \in [\order]}\max_{\ell \in [\dim]}\left( \sum_{i=1}^\dim |\iter{\Delta}{t-1,r | i}_\ell |^2 \|\iter{\vDelta}{t-1, r^\prime | i}\|^2 \right)\nonumber \\
    & \explain{(a)}{\leq}  C\rho_T^{2T+4} \cdot  \max_{r \in [\order]}\max_{\ell \in [\dim]}\left( \sum_{i=1}^\dim |\iter{\Delta}{t-1,r | i}_\ell |^2  \right)\nonumber \\
    & \explain{(b)}{\leq} C\rho_T^{2T+4} \cdot  \max_{r \in [\order]}\left\| \sum_{i=1}^\dim \iter{\vDelta}{t-1,r | i} {\iter{\vDelta}{t-1,r | i}} \tran   \right\|_{\op}\nonumber\\
    & \explain{(c)}{=} C \cdot \rho_T^{2T+4} \cdot \max_{r \in [\order]} \left(\iter{O}{t-1,r}\right). \label{eq:epsilon-norm-ub}
\end{align}
In the above display, the inequality marked (a) follows from the first estimate \eqref{eq:pert-claim1} claimed in this lemma. In step (b) we used the fact that the maximum diagonal entry of a symmetric matrix can be upper bounded by its operator norm and in step (c) we recalled the definition of $\iter{O}{t-1,r}$ from \eqref{eq:op-short}. Plugging in the estimates in \eref{alpha-norm-ub}, \eref{beta-norm-ub}, \eref{delta-norm-ub} and \eref{epsilon-norm-ub} into \eref{L-recursion-initial} gives the following recursive estimate for $\max_{j \in [\order]} \iter{O}{t,j}$:
\begin{align*}
    \max_{j \in [\order]} (\iter{O}{t,j})& \leq C \cdot \left(\rho_T^2 + \rho_T^{2T + 4} \cdot \max_{j \in [\order]} (\iter{O}{t-1,j}) \right)\; \forall \; t \leq T.
\end{align*}
Unrolling this recursive upper bound gives:
\begin{align*}
     &\max_{t \leq T}  \max_{j \in [\order]} \left\| \sum_{i=1}^\dim \iter{\vDelta}{t,j | i}  \cdot {\iter{\vDelta}{t,j | i}} \tran  \right\|_{\op}  \\& \qquad \qquad \leq C \cdot \left( \rho_T^2 + \rho_T^2 \cdot (\rho_T^{2T + 4}) + \rho_T^2 \cdot (\rho_T^{2T + 4})^2 + \dotsb +  (\rho_T^{2T + 4})^T \cdot  \max_{j \in [\order]} \left\| \sum_{i=1}^\dim \iter{\vDelta}{0,j | i}  \cdot {\iter{\vDelta}{0,j | i}} \tran  \right\|_{\op}   \right).
\end{align*}
Recall that $\iter{\vDelta}{0,j | i} = \iter{\vz}{0,j| i} - \iter{\vz}{0,j| \emptyset} = (\hat{G}_{ij} - G_{ij}) \cdot \ve_i$. Hence,
\begin{align*}
   \max_{j \in [\order]}  \left\| \sum_{i=1}^\dim \iter{\vDelta}{0,j | i}  \cdot {\iter{\vDelta}{0,j | i}} \tran  \right\|_{\op} & \leq \|\hat{\mG} - \mG\|_{\infty}^2 \leq \rho_T^2.
\end{align*}
Finally, we obtain,
\begin{align*}
    \max_{t \leq T}  \max_{j \in [\order]} \left\| \sum_{i=1}^\dim \iter{\vDelta}{t,j | i}  \cdot {\iter{\vDelta}{t,j | i}} \tran  \right\|_{\op} & \leq C \cdot \rho_T^{2(T+1)^2},
\end{align*}
as claimed.
\end{proof}

\subsection{Proof of \lemref{delocalization}}\label{appendix:delocalization-subsection}
\begin{proof}[Proof of \lemref{delocalization}]
Recall that,
\begin{align*}
    \rho_T(\mPsi_{1:\order}, \mS, \hat{\mS}, \auxmat, \hat{\auxmat}, \mG, \hat{\mG}) &\explain{def}{=}  (1 + \max_{j \in [\order]} \|\mPsi_i\|_{\op})  \times (1 + \|\mG\|_{\infty} + \|\hat{\mG}\|_{\infty})\times  (1 + \|\auxC(\auxmat)\|_{\infty} + \|\auxC(\hat{\auxmat})\|_{\infty}) \nonumber\\  & \qquad \qquad \qquad \qquad \qquad  \times (1 + \max_{\substack{i \in [\dim], t \leq [T]}} \|\iter{\mZ}{t,\cdot | i}\|_{\infty}^\degree + \max_{t \leq T} \|\iter{\mZ}{t,\cdot | \emptyset }\|_{\infty}^\degree).
\end{align*}
In light of \defref{semirandom}, we know that $\max_{j \in [\order]} \|\mPsi_i\|_{\op} \lesssim 1$. Hence, by Cauchy–Schwarz inequality, it suffices to show that for any $p \in \W$ and any $\epsilon \in (0,1)$:
\begin{subequations} \label{eq:deloc-goals}
\begin{align}
    \E \|\auxC(\auxmat)\|_{\infty}^p \explain{(a)}{=}  \E \|\auxC(\hat{\auxmat})\|_{\infty}^p &\lesssim \dim^{\epsilon}, \label{eq:deloc-goal-1} \\
    \E \|\mG\|_{\infty}^p \explain{(a)}{=}  \E \|\hat{\mG}\|_{\infty}^p &\lesssim \dim^{\epsilon}, \label{eq:deloc-goal-2} \\
    \E \bigg[ \max_{t \leq [T]} \|\iter{\mZ}{t,\cdot | \emptyset }\|_{\infty}^p \bigg]  + \E \bigg[ \max_{\substack{i \in [\dim], t \leq [T]}} \|\iter{\mZ}{t,\cdot | i }\|_{\infty}^p \bigg] & \lesssim \dim^{\epsilon} \label{eq:deloc-goal-3}
\end{align}
\end{subequations}
Note that in the above claims, the equalities marked (a) are immediate from the fact that $\auxmat \explain{d}{=} \hat{\auxmat}$ and $\mG \explain{d}{=} \hat{\mG}$. We consider each of the remaining claims in \eqref{eq:deloc-goals} separately. 
\paragraph{Proof of \eqref{eq:deloc-goal-1}.} Define $q \explain{def}{=} \lceil \tfrac{1}{\epsilon} \rceil$. Consider the following estimate:
\begin{align*}
    \left(\E \|\auxC(\hat{\auxmat})\|_{\infty}^p \right)^{q} & \explain{(a)}{\leq} (\E \|\auxC(\hat{\auxmat})\|_{\infty}^{pq}) 
     \leq \sum_{i=1}^{\dim} \E[|\auxC(\auxvec_i)|^{pq}] 
     \explain{(b)}{=} \dim \E[|\auxC(\serv{A})|^{pq}] 
\end{align*}
In the above display step (a) follows from Jensen's inequality and step (b) follows from the assumption that $\auxvec_{1:\dim}$, the rows of $\auxmat$ are i.i.d. copies of $\serv{A}$ (cf. \assumpref{side-info}). Recall that \lemref{perturbation-estimate-efron-stein} guarantees that $\E[|\auxC(\serv{A})|^{pq}] < \infty$. Hence, we have shown:
\begin{align*}
    \E \|\auxC(\hat{\auxmat})\|_{\infty}^p \leq  (\E[|\auxC(\serv{A})|^{pq}])^{1/q} \cdot \dim^{1/q} \lesssim \dim^{\epsilon}. 
\end{align*}
\paragraph{Proof of \eqref{eq:deloc-goal-2}.} The bound \eqref{eq:deloc-goal-2} can be derived using the same argument as above, or by using standard bounds on the maximum of Gaussian random variables. 
\paragraph{Proof of \eqref{eq:deloc-goal-3}.} As before, we have:
\begin{align*}
    &\left( \E \bigg[ \max_{t \leq [T]} \|\iter{\mZ}{t,\cdot | \emptyset }\|_{\infty}^{p} \bigg] \right)^{2q}  + \left( \E \bigg[ \max_{\substack{i \in [\dim], t \leq [T]}} \|\iter{\mZ}{t,\cdot | i }\|_{\infty}^p \bigg] \right)^{2q} \\& \hspace{7cm} \explain{(a)}{\leq}  \E \bigg[ \max_{t \leq [T]} \|\iter{\mZ}{t,\cdot | \emptyset }\|_{\infty}^{2pq} \bigg] +  \E \bigg[ \max_{\substack{i \in [\dim], t \leq [T]}} \|\iter{\mZ}{t,\cdot | i }\|_{\infty}^{2pq} \bigg] \\
    &\hspace{7cm} \leq \sum_{t=1}^{T} \left( \E[ \|\iter{\mZ}{t,\cdot | \emptyset }\|_{\infty}^{2pq}] + \sum_{i=1}^\dim \E[ \|\iter{\mZ}{t,\cdot | i }\|_{\infty}^{2pq}]  \right)\\
    &\hspace{7cm} = (\dim+1) \cdot \sum_{t=1}^{T} \E[ \|\iter{\mZ}{t,\cdot | \emptyset }\|_{\infty}^{2pq}] \\
    &\hspace{7cm} \leq  (\dim+1) \cdot \sum_{t=1}^T \sum_{j=1}^\order \sum_{\ell=1}^{\dim} \E[|\iter{z}{t,j}_\ell|^{2pq}] \\
    &\hspace{7cm} = \dim(\dim + 1)  \sum_{t=1}^T \sum_{j=1}^\order \frac{1}{\dim}\sum_{\ell=1}^{\dim} \E[(\iter{z}{t,j}_\ell)^{2pq}],
\end{align*}
where inequality (a) follows from Jensen's Inequality. By \thref{moments}, we know that:
\begin{align*}
    \lim_{\dim \rightarrow \infty } \frac{1}{\dim}\sum_{\ell=1}^{\dim} \E[(\iter{z}{t,j}_\ell)^{2pq}] & = \E[\serv{Z}^{2pq}] < \infty,
\end{align*}
where $\serv{Z} \sim \gauss{0}{1}$. Hence,
\begin{align*}
     &\E \bigg[ \max_{t \leq [T]} \|\iter{\mZ}{t,\cdot | \emptyset }\|_{\infty}^p \bigg]  + \E \bigg[ \max_{\substack{i \in [\dim], t \leq [T]}} \|\iter{\mZ}{t,\cdot | i }\|_{\infty}^p \bigg]  \\& \hspace{4.8cm} \leq 2^{1-\frac{1}{2q}} \cdot \left\{ \left( \E \bigg[ \max_{t \leq [T]} \|\iter{\mZ}{t,\cdot | \emptyset }\|_{\infty}^{p} \bigg] \right)^{2q}  + \left( \E \bigg[ \max_{\substack{i \in [\dim], t \leq [T]}} \|\iter{\mZ}{t,\cdot | i }\|_{\infty}^p \bigg] \right)^{2q} \right\}^{\frac{1}{2q}} \\
     &\hspace{5cm} \lesssim \dim^{\frac{1}{q}} \lesssim \dim^{\epsilon}.
\end{align*}
This concludes the proof.
\end{proof}
\subsection{Continuity Estimates}
\label{appendix:misc}
\begin{lemma}\label{lem:misc-PL} For any $\auxvec \in \R^{\auxdim}$ let $\nonlin(z ; \auxvec)$ be a $\order$-variate, degree $\degree$ polynomial in $z \in \R^{\order}$. Suppose that:
\begin{align*}
    \E[ |\nonlin(\serv{Z}; \serv{A})|^p] < \infty \; \forall \; p \; \in \; \W.
\end{align*}
Then, there exists a function $\auxC: \R^{\auxdim} \mapsto [0,\infty)$ with $\E[|\auxC(\serv{A})|^p] < \infty$ for each $p \in \W$ such that for any $z,z^\prime \in \R^{\order}$, $\nonlin$ satisfies the estimates:
\begin{subequations}\label{eq:misc-lemma}
\begin{align}
    |\nonlin(z; \auxvec)| & \leq \auxC(\auxvec) \cdot (1 + \|z\|_{\infty}^\degree), \label{eq:misc-lemma-1} \\
    \|\nabla_{z}\nonlin(z; \auxvec)\|_{\infty} & \leq \auxC(\auxvec) \cdot (1 + \|z\|_{\infty}^\degree), \label{eq:misc-lemma-2}\\
     |\nonlin(z; \auxvec) - \nonlin(z^\prime; \auxvec)| & \leq \auxC(\auxvec) \cdot (1 + \|z\|_{\infty}^\degree + \|z^\prime\|_{\infty}^\degree ) \cdot \| z - z^\prime\|_{\infty}, \label{eq:misc-lemma-3}\\
     |\nonlin(z; \auxvec) - \nonlin(z^\prime; \auxvec) - \ip{\nabla_z \nonlin(z; \auxvec)}{z-z^\prime}| & \leq \auxC(\auxvec) \cdot (1 + \|z\|_{\infty}^\degree + \|z^\prime\|_{\infty}^\degree ) \cdot \| z - z^\prime\|_{\infty}^2. \label{eq:misc-lemma-4}
\end{align} 
\end{subequations} 
In the above equations $\serv{A}$ is the random variable from \assumpref{side-info}, $\serv{Z} \sim \gauss{0}{I_{\order}}$ is independent of $\serv{A}$ and $\nabla_{z}\nonlin(z; \auxvec)$ denotes the gradient of the polynomial $\nonlin(z; \auxvec)$ with respect to $z$. 
\end{lemma}
\begin{proof} Consider the Hermite Decomposition of the polynomial $\nonlin(z; \auxvec)$:
\begin{align*}
    \nonlin(z; \auxvec) & = \sum_{\substack{i\in \W^\order \\ \|i\|_1 \leq \degree}} c_{i}(\auxvec) \cdot \hermite{i}(z),
\end{align*}
where $\{H_i(z): i \in \W^\order\}$ denote the $\order$-variate Hermite polynomials and the Hermite coefficients $c_i(\auxvec)$ are given by the formula:
\begin{align*}
    c_i(\auxvec) = \E[ \nonlin(\serv{Z}; \auxvec) \cdot \hermite{i}(\serv{Z})], \quad \serv{Z} \sim \gauss{0}{I_{\order}}.
\end{align*}
Observe that the coefficients $c_i(\auxvec)$ satisfy $\E[| c_i(\serv{A})|^p] < \infty$ for each $p \in \W$. Indeed,
\begin{align*}
    \E[| c_i(\serv{A})|^p] & = \E_{\serv{A}} \left[ \big| \E_{\serv{Z}}[ \nonlin(\serv{Z}; \serv{A}) \cdot \hermite{i}(\serv{Z}) ]\big|^p \right]  \leq \E\left[ \big| \nonlin(\serv{Z}; \serv{A})  \hermite{i}(\serv{Z})\big|^p \right] \leq \sqrt{\E[|\nonlin(\serv{Z}; \serv{A})|^{2p}] \E[|\hermite{i}(\serv{Z})|^{2p}]} <\infty. 
\end{align*}
Notice that there is a finite constant $C$ (determined by $\degree, \order$) such that for any $i \in \W^{\order}$ with $\|i\|_1 \leq \degree$  and any $j,j^\prime \in [k]$ we have:
\begin{align*}
    |\hermite{i}(z)| + |\partial_{z_j}\hermite{i}(z)| + ||\partial^2_{z_j z_j^\prime}\hermite{i}(z)|  & \leq C (1 + \|z\|_{\infty}^\degree) \quad \forall \; z \; \in \; \R^{\order}.
\end{align*}
This because $\{\hermite{i}(z), \; \partial_{z_j}\hermite{i}(z), \; \partial^2_{z_j z_j^\prime}\hermite{i}(z) \; : \; i \in \W^\order, \; \|i\|_1 \leq \degree, \; j,j^\prime \; \in \; [\order] \}$ is a finite collection of polynomials of degree at most $\degree$. Hence, 
\begin{align} \label{eq:misc-lemma-temp}
     |\nonlin(z; \auxvec)| + |\partial_{z_j} \nonlin(z; \auxvec)| + |\partial^2_{z_j z_j^\prime}\nonlin(z; \auxvec)| & \leq C \cdot \left( \sum_{\substack{i\in \W^\order : \|i\|_1 \leq \degree}} |c_{i}(\auxvec)| \right) \cdot (1 + \|z\|_{\infty}^\degree) 
\end{align}
We define $\auxC(\auxvec)$ as:
\begin{align*}
    \auxC(\auxvec) = Ck^2 \cdot \left( \sum_{\substack{i\in \W^\order : \|i\|_1 \leq \degree}} |c_{i}(\auxvec)| \right).
\end{align*}
The estimate in \eqref{eq:misc-lemma-temp} shows that claims \eqref{eq:misc-lemma-1} and \eqref{eq:misc-lemma-2} in the statement of the lemma hold with this choice of $\auxC(\auxvec)$. The claims \eqref{eq:misc-lemma-3} and \eqref{eq:misc-lemma-4} follow from Taylor's theorem. 
\end{proof}

\section{Reductions and Simplifications} \label{appendix:simplifications}
This appendix is devoted to the proof of \propref{orthogonalization}, which claims that it is sufficient to prove \thref{VAMP} when \sassumpref{orthogonality} and \sassumpref{balanced}. In order to prove this result, we will find it helpful to introduce the following additional simplifying assumption, which we will argue can be assumed without loss of generality. Similar simplifying assumptions have been used in prior works \citep{berthier2020state,fan2020approximate}.
\begin{sassumption}[Non-Degeneracy Condition]\label{sassump:non-degenerate} The limiting covariance matrix of the semi-random ensemble $\Omega$ (cf. \defref{semirandom}) and the state evolution covariances $\Phi_T, \Sigma_T$ defined in \eqref{eq:SE-VAMP} satisfy $\lambda_{\min{}}(\Omega) > 0$, $\lambda_{\min}(\Phi_T)>0$, and $\lambda_{\min{}}(\Sigma_T) > 0$.
\end{sassumption}

We prove \propref{orthogonalization} in three steps, introducing the various simplifying assumptions in a convenient order. These steps are stated in the following three lemmas. 

\begin{lemma}[Removing Non-Degeneracy Assumption]\label{lem:degeneracy} It suffices to prove \thref{VAMP} when \sassumpref{non-degenerate} holds in addition to the other assumptions required by \thref{VAMP}.
\end{lemma}
\begin{lemma}[Balancing Semi-Random Matrices]\label{lem:balancing} It suffices to prove \thref{VAMP} when  \sassumpref{non-degenerate} and \sassumpref{balanced} hold in addition to the other assumptions required by \thref{VAMP}.
\end{lemma}
\begin{lemma}[Orthogonalization]\label{lem:orthogonalization} It suffices to prove \thref{VAMP} when  \sassumpref{balanced} and \sassumpref{orthogonality} hold in addition to the other assumptions required by \thref{VAMP}.
\end{lemma}
Observe that \lemref{orthogonalization} is a restatement of \propref{orthogonalization}. Hence, the remainder of this appendix is devoted to the lemmas introduced above and is organized as follows:
\begin{enumerate}
    \item \appref{non-degeneracy} proves \lemref{degeneracy} by adapting a perturbation argument of \citet{berthier2020state}. 
    \item \appref{balancing} proves \lemref{balancing} by showing that a semi-random ensemble can be approximated by a balanced semi-random ensemble. 
    \item \appref{orthogonalization} proves \lemref{orthogonalization} by showing the a VAMP algorithm can be implemented using a suitably designed orthogonalized VAMP algorithm that satisfies the orthogonality conditions stated in \sassumpref{orthogonality}. 
\end{enumerate}

\subsection{Removing the Non-Degeneracy Assumption} \label{appendix:non-degeneracy}
This section is devoted to the proof of \lemref{degeneracy}. In order to prove the claim of this lemma, we need to show that if \thref{VAMP} holds under the additional non-degeneracy assumption stated as \sassumpref{non-degenerate}, then it must also hold without the non-degeneracy assumption. To this end, we consider $T$ iterations of a VAMP algorithm, which satisfies all the assumptions of \thref{VAMP}, but need not satisfy  \sassumpref{non-degenerate}:
\begin{align}\label{eq:VAMP-unpert}
    \iter{\vz}{t} = \mM_t \cdot \nonlin_t(\iter{\vz}{1}, \iter{\vz}{2}, \dotsc, \iter{\vz}{t-1}; \auxmat) \quad \forall \; t\; \in \; [T].
\end{align}
In the above display $\mM_{1:T}$ is a semi-random ensemble with $\mM_i = \mS \mPsi_i \mS$ where $\mS$ is a uniformly random sign diagonal matrix. 
Let: 
\begin{subequations}\label{eq:SE-VAMP-unpert}
\begin{align}\serv{Z}_{1}, \dotsc, \serv{Z}_{T} \sim \gauss{0}{\Sigma_T},
\end{align} $\Phi_T$ and $\Sigma_T$ denote the state evolution random variables and covariance matrices associated with the VAMP algorithm in \eqref{eq:VAMP-unpert}. Recall from \eqref{eq:SE-VAMP}, for each $t  \in \{0, 1, \dotsc, T-1\}$, these are defined recursively as follows:
\begin{align} 
    (\Phi_{T})_{s,t+1} & \explain{def}{=} \E[\nonlin_{s}(\serv{Z}_1, \dotsc, \serv{Z}_{s-1}; \serv{A}) \nonlin_{t+1}(\serv{Z}_1, \dotsc, \serv{Z}_{t}; \serv{A})] \quad \forall \; s \; \leq \; t+1, \\
   (\Sigma_{T})_{s,t+1} & \explain{def}{=}  \Omega_{s,t+1} \cdot (\Phi_{T})_{s,t+1}   \quad \forall \; s \; \leq \; t+1.
\end{align}
\end{subequations}
In the above display $\serv{A}$ is the auxiliary information random variable from \assumpref{side-info} independent of $\serv{Z}_{1}, \dotsc, \serv{Z}_{T}$. In other to prove \lemref{degeneracy}, we need to show that 
\begin{align} \label{eq:degeneracy-goal}(\iter{\vz}{1}, \iter{\vz}{2}, \dotsc, \iter{\vz}{T}, \auxmat) \explain{\pw}{\longrightarrow} (\serv{Z}_{1}, \dotsc, \serv{Z}_{T}, \serv{A}).
\end{align}In order to do so, we will introduce a perturbed iteration which approximates \eqref{eq:VAMP-unpert} and additionally satisfies the non-degeneracy condition. We will then infer \eqref{eq:degeneracy-goal} by applying \thref{VAMP} to the perturbed iteration.

\paragraph{Perturbed VAMP.} For each $\epsilon \in (0,1)$, we define a perturbed VAMP iteration of the form:
\begin{align}\label{eq:VAMP-pert}
    \iter{\vz}{t}_\epsilon = \mM_t^\epsilon \cdot \nonlin_t^\epsilon(\iter{\vz}{1}_\epsilon, \iter{\vz}{2}_\epsilon, \dotsc, \iter{\vz}{t-1}_\epsilon; \auxmat, \mW) \quad \forall \; t\; \in \; [T].
\end{align}
In the above display:
\begin{enumerate}
    \item $\mM_t^\epsilon = \mS(\mPsi_t + \epsilon \mG_t)\mS = \mM_t + \epsilon \mS \mG_t \mS$ where $\mG_{1:T}$ are i.i.d. $\dim \times \dim$ matrices drawn from the Gaussian Orthogonal Ensemble (GOE). The following lemma verifies that $\mM_{1:T}^\epsilon$ are semi-random, as required by \thref{VAMP}.  
\end{enumerate}
\begin{lemma}\label{lem:semirandom-pert} With probability 1, $\mM_{1:T}^\epsilon$ form a semi-random ensemble with limiting covariance matrix $\Omega_\epsilon \explain{def}{=} \Omega + \epsilon^2 I_T$.
\end{lemma}
\begin{enumerate}
  \setcounter{enumi}{1}
    \item The algorithm in \eqref{eq:VAMP-pert} uses two kinds of auxiliary information $\auxmat$ and $\mW$. The auxiliary information $\auxmat$ is the same auxiliary information that is used in the original VAMP iterations \eqref{eq:VAMP-unpert}  that we seek to approximate. The auxiliary information $\mW$ is a $\dim \times T$ matrix with columns $\vw_1, \dotsc, \vw_T$ which are sampled i.i.d. from $\gauss{\vzero}{\mI_{\dim}}$. 
    \item The non-linearities $\nonlin_t^\epsilon$ are given by:
    \begin{align} \label{eq:perturbed-nonlinearities}
        \nonlin_t^{\epsilon}(z_1, \dotsc, z_{t-1}; \auxvec, w_{1}, \dotsc, w_T) = \nonlin_t(z_1, \dotsc, z_{t-1}; \auxvec) + \epsilon \cdot w_t - \sum_{s=1}^{t-1} (\alpha_{t}^\epsilon)_s \cdot z_s.
    \end{align}
    In the above display, the \emph{correction coefficients} $\alpha_t^\epsilon \in \R^{t-1}$ will be specified adaptively with the state evolution recursion so that the non-linearities $\nonlin_t^{\epsilon}$ are divergence-free in the sense of \assumpref{div-free}. 
\end{enumerate}

\paragraph{State Evolution and Correction Vectors for Perturbed VAMP.} Next, we specify the state evolution associated with the perturbed VAMP algorithm (recall \eqref{eq:SE-VAMP}) and coefficients in \eqref{eq:perturbed-nonlinearities}. We will denote the Gaussian state evolution random variables associated with \eqref{eq:VAMP-pert} as $\serv{Z}_1^\epsilon, \dotsc, \serv{Z}_T^\epsilon$, which will be distributed as $\gauss{0}{\Sigma_T^{\epsilon}}$ where $\Sigma_T^\epsilon$ is the Gaussian state evolution covariance. Likewise we will denote the non-Gaussian state evolution covariance associated with \eqref{eq:VAMP-pert} by $\Phi_T^\epsilon$. For each $t \in \{0, 1, \dotsc, T-1\}$, the correction vectors $\alpha_{t+1} \in \R^t$ (which complete the definition of $\nonlin_{t+1}^\epsilon$ in \eqref{eq:perturbed-nonlinearities}), the entries of $\Phi_T, \Sigma_T$ are defined recursively as follows:
\begin{subequations}\label{eq:SE-pert-VAMP}
\begin{align}
    \alpha_{t+1}^\epsilon & \explain{def}{=} ({\Sigma_t^\epsilon})^{-1} \cdot  \E[\nonlin_{t+1}(\serv{Z}_1^\epsilon, \dotsc, \serv{Z}_t^\epsilon; \serv{A}) \cdot \serv{Z}_{[t]}^\epsilon], \\
     (\Phi_{T}^\epsilon)_{s,t+1} & \explain{def}{=} \E[\nonlin_{s}^\epsilon(\serv{Z}_1^\epsilon, \dotsc, \serv{Z}_{s-1}^\epsilon; \serv{A}, \serv{W}) \nonlin_{t+1}^\epsilon(\serv{Z}_1^\epsilon, \dotsc, \serv{Z}_{t}^\epsilon; \serv{A},\serv{W})] \quad \forall \; s \; \leq \; t+1, \\
   (\Sigma_{T}^\epsilon)_{s,t+1} & \explain{def}{=}  (\Omega_\epsilon)_{s,t+1} \cdot (\Phi_{T}^\epsilon)_{s,t+1}   \quad \forall \; s \; \leq \; t+1.
\end{align}
\end{subequations}
In the above display:
\begin{enumerate}
    \item $\Sigma_t^\epsilon$ denotes the leading $t \times t$ principal sub-matrix of $\Sigma_T^{\epsilon}$ (formed by the first $t$ rows and columns). Similarly, we will also use $\Phi_t^\epsilon$ to denote the leading $t \times t$ principal sub-matrix of $\Phi_T^{\epsilon}$.
    \item $\serv{Z}_{[t]}^\epsilon \in \R^t$ is the vector $(\serv{Z}_1^\epsilon, \dotsc, \serv{Z}_t^\epsilon)$.
    \item $\serv{A}$ is the auxiliary information random variable and $\serv{W} \sim \gauss{0}{I_T}$. These random variables are independent of each other and of $(\serv{Z}_1, \dotsc, \serv{Z}_T)$. 
    \item $\Omega_\epsilon$ is the limiting covariance matrix of the semi-random ensemble $\mM_{1:T}^\epsilon$ defined in \lemref{semirandom-pert}
\end{enumerate}
Note that for the recursion \eqref{eq:SE-pert-VAMP} to be well-defined, $\Sigma_t^\epsilon$ should be invertible. The following lemma shows that this is indeed the case and collects some useful properties of the state evolution \eqref{eq:SE-pert-VAMP}. 

\begin{lemma} \label{lem:pert-VAMP-SE} For each $t \in [T]$, we have:
\begin{multicols}{2}
\begin{enumerate}
    \item $\lambda_{\min}(\Phi_t^\epsilon) > 0$.
    \item $\lambda_{\min}(\Sigma_t^\epsilon) > 0$.
    \item ${\alpha_t^\epsilon} \tran \Sigma_{t-1}^{\epsilon} \alpha_t^{\epsilon} \rightarrow 0$ as $\epsilon \rightarrow 0$.
    \item $\Phi_t^\epsilon \rightarrow \Phi_t$ and $\Sigma_t^\epsilon \rightarrow \Sigma_t$ as $\epsilon \rightarrow 0$.
\end{enumerate}
\end{multicols}
\begin{enumerate}
\setcounter{enumi}{4}
    \item For any $h:\R^{t + \auxdim} \mapsto \R$ which satisfies:
    \begin{align*}
        |h(z;a) - h(z^\prime;a)| & \leq L \cdot (1 + \|z\| + \|z^\prime\| + \|a\|^D)\cdot \|z-z^\prime\|
    \end{align*}
    for some constants $L < \infty, D \in \N$, we have $\E h(\serv{Z}_1^\epsilon, \dotsc, \serv{Z}_t^\epsilon; \serv{A}) \rightarrow \E h(\serv{Z}_1, \dotsc, \serv{Z}_t; \serv{A})$ as $\epsilon \rightarrow 0$. 
\end{enumerate}
\end{lemma}
Finally, in order to complete the proof of \lemref{degeneracy}, we will also require the following perturbation bound on the distance between the perturbed VAMP \eqref{eq:VAMP-pert} and unperturbed VAMP \eqref{eq:VAMP-unpert} iterates. 

\begin{lemma}\label{lem:pert-VAMP-pert-bound} Assuming that \thref{VAMP} holds under the additional \sassumpref{non-degenerate}, then, for each $t \in [T]$, we have:
\begin{align*}
    \text{1. } \limsup_{\dim \rightarrow \infty} \frac{\E \|\iter{\vz}{t}\|^2}{\dim} < \infty, \quad \text{2. } \lim_{\epsilon \rightarrow 0} \limsup_{\dim \rightarrow \infty} \frac{\E \|\iter{\vz}{t}_{\epsilon}\|^2}{\dim} < \infty, \quad  \text{3. }\lim_{\epsilon \rightarrow 0} \limsup_{\dim \rightarrow \infty} \frac{\E \|\iter{\vz}{t}-\iter{\vz}{t}_\epsilon\|^2}{\dim} = 0.
\end{align*}
\end{lemma}

We defer the proof of the intermediate results introduced so far (\lemref{semirandom-pert}, \lemref{pert-VAMP-SE}, and \lemref{pert-VAMP-pert-bound}) to the end of this section, and provide a proof of \lemref{degeneracy}.

\begin{proof}[Proof of \lemref{degeneracy}] In order to show that $(\iter{\vz}{1}, \iter{\vz}{2}, \dotsc, \iter{\vz}{T}, \auxmat) \explain{\pw}{\longrightarrow} (\serv{Z}_{1}, \dotsc, \serv{Z}_{T}, \serv{A})$, we need to show that for any test function $h:\R^{T+\auxdim} \mapsto \R$ that satisfies the regularity hypothesis required by the definition of $\text{\pw}$ convergence (\defref{PW2}) we have, 
\begin{align*}
    H_\dim \explain{def}{=} \frac{1}{\dim} \sum_{\ell = 1}^\dim h(\iter{z}{1}_\ell, \dotsc, \iter{z}{T}_\ell; \auxvec_{\ell}) \explain{P}{\rightarrow} \E[ h(\serv{Z}_1, \dotsc, \serv{Z}_T; \serv{A})].
\end{align*}
Note that for any $\epsilon> 0$, the perturbed VAMP iterates $\iter{\vz}{t}_\epsilon$ satisfy all the requirements of \thref{VAMP} along with \sassumpref{non-degenerate}. Indeed \lemref{semirandom-pert} (item 2) guarantees that $\lambda_{\min}(\Omega_\epsilon) \geq \epsilon^2 > 0$ and \lemref{pert-VAMP-SE} guarantees that $\lambda_{\min}(\Phi_T) > 0$ and $\lambda_{\min}(\Sigma_T) > 0$. Hence, 
\begin{align}\label{eq:pert-vamp-se-conclusion}
    H_N^\epsilon \explain{def}{=} \frac{1}{\dim} \sum_{\ell = 1}^\dim h((\iter{z}{1}_\epsilon)_\ell, \dotsc, (\iter{z}{T}_\epsilon)_\ell; \auxvec_{\ell}) \explain{P}{\rightarrow} \E[ h(\serv{Z}_1^\epsilon, \dotsc, \serv{Z}_T^\epsilon; \serv{A})].
\end{align}
Furthermore, by \lemref{pert-VAMP-SE} (item 5):
\begin{align}\label{eq:se-conv}
    \lim_{\epsilon \rightarrow 0} \E[ h\big(\serv{Z}_1^\epsilon, \dotsc, \serv{Z}_T^\epsilon; \serv{A}\big)] & =  \E[ h(\serv{Z}_1, \dotsc, \serv{Z}_T; \serv{A})].
\end{align}
We also have the following bound on $\E| H_\dim^\epsilon - H_\dim|$:
\begin{align}
    &\lim_{\epsilon \rightarrow 0} \limsup_{\dim \rightarrow \infty} \E| H_\dim^\epsilon - H_\dim|    \leq \lim_{\epsilon \rightarrow 0} \limsup_{\dim \rightarrow \infty} \frac{1}{\dim} \sum_{\ell=1}^\dim \E \bigg|h\big((\iter{z}{1}_\epsilon)_\ell, \dotsc, (\iter{z}{T}_\epsilon)_\ell; \auxvec_{\ell}\big) -  h\big(\iter{z}{1}_\ell, \dotsc, \iter{z}{T}_\ell; \auxvec_{\ell}\big)\bigg| \nonumber\\
    & \qquad \qquad \qquad \qquad\explain{(a)}{\leq} \lim_{\epsilon \rightarrow 0} \limsup_{\dim \rightarrow \infty}  \E \left[ \frac{L}{N} \sum_{\ell = 1}^\dim (1 + \|(\iter{z}{1:T}_\epsilon)_{\ell}\| + \|\iter{z}{1:T}_{\ell}\| + \|\auxvec_{\ell}\|^\degree) \cdot \|(\iter{z}{1:T}_\epsilon)_{\ell} -  \iter{z}{1:T}_{\ell}\|  \right]\nonumber \\
    & \qquad \qquad\qquad \qquad\explain{(b)}{\leq} \lim_{\epsilon \rightarrow 0} \limsup_{\dim \rightarrow \infty}  2L \left\{ 1 + \sum_{t=1}^T \frac{\E\|\iter{\vz}{t}\|^2 + \E\|\iter{\vz}{t}_\epsilon\|^2}{\dim}  + \E \|\serv{A}\|^{2\degree} \right\}^{\frac{1}{2}} \left\{ \sum_{t=1}^T \frac{\E\|\iter{\vz}{t}_\epsilon-\iter{\vz}{t}\|^2}{\dim} \right\}^{\frac{1}{2}} \nonumber\\
    & \qquad \qquad\qquad \qquad\explain{(c)}{=} 0. \label{eq:pert-VAMP-test-pert-bound}
\end{align}
In the above display, inequality (a) follows from the regularity hypothesis on the test function (\defref{PW2}), step (b) uses Cauchy-Schwarz Inequality, and step (c) follows from \lemref{pert-VAMP-pert-bound}.  With \eqref{eq:pert-vamp-se-conclusion}, \eqref{eq:se-conv}, and \eqref{eq:pert-VAMP-test-pert-bound}, we can now verify that $H_\dim \explain{P}{\rightarrow} E[ h(\serv{Z}_1, \dotsc, \serv{Z}_T; \serv{A})]$. Indeed for any $\delta > 0$, we have:
\begin{align*}
    &\limsup_{\dim \rightarrow \infty} \P(|H_\dim -E[ h(\serv{Z}_1, \dotsc, \serv{Z}_T; \serv{A})]| > 3\delta )  \leq \limsup_{\epsilon \rightarrow 0} \limsup_{\dim \rightarrow \infty} \bigg\{\P(|H_\dim -H_\dim^\epsilon| > \delta )  \bigg. \\& \hspace{2cm}+ \bigg.   \P(|H_\dim^\epsilon-  \E[ h\big(\serv{Z}_1^\epsilon, \dotsc, \serv{Z}_T^\epsilon; \serv{A}\big)]| > \delta )  + \P(|\E[ h\big(\serv{Z}_1^\epsilon, \dotsc, \serv{Z}_T^\epsilon; \serv{A}\big)]-E[ h(\serv{Z}_1, \dotsc, \serv{Z}_T; \serv{A})]| > \delta ) \bigg\} \\
    & \hspace{2cm}= 0,
\end{align*}
where the last step follows from \eqref{eq:pert-vamp-se-conclusion}, \eqref{eq:se-conv}, and \eqref{eq:pert-VAMP-test-pert-bound}. This concludes the proof of \lemref{degeneracy}. 
\end{proof}

\subsubsection{Proof of \lemref{semirandom-pert}}
\begin{proof}[Proof of \lemref{semirandom-pert}]
We have for any $\epsilon,\eta > 0$ with probability $1$,
\begin{align*}
    \|\hat{\mM}_t\|_{\op} & \leq \|\mPsi_t\|_{\op} + \epsilon\|\mG_t\|_{\op} \explain{(a)}{\lesssim} 1, \\
    \|\hat{\mM}_t\|_{\infty} & \leq \|\mPsi_t\|_{\infty} + \epsilon \|\mG_t\|_{\infty} \explain{(b)}{\lesssim} \dim^{-1/2 + \eta}.
\end{align*}
In the above display the estimate (a) follows from standard bounds on the operator norm of Gaussian matrices (see e.g., \citep[Corollary 4.4.8]{vershynin2018high}) and (b) follows from standard bounds on the maximum of Gaussian random variables. Furthermore, observe that
\begin{align*}
    \hat{\mM}_s \hat{\mM}_t \tran & = \mM_s \mM_t \tran + \epsilon \mG_s \mM_t\tran + \epsilon \mM_s \mG_t \tran + \epsilon^2 \mG_s \mG_t \tran.  
\end{align*}
Observe that for any $\eta > 0$
\begin{align*}
    \|{\mM}_s {\mM}_t \tran - \dim^{-1} \cdot \Tr({\mM}_s {\mM}_t \tran)  \mI_{\dim}\|_{\infty} \explain{(c)}{\lesssim} \dim^{-\frac{1}{2}+\eta}, \quad \| \mG_s \mM_t\tran\|_{\infty} \explain{(d)}{\lesssim} \dim^{-\frac{1}{2}+\eta}, \quad \|\mG_s \mG_t\tran - \delta_{st} \mI_{\dim}\|_{\infty} \explain{(e)}{\lesssim} \dim^{-\frac{1}{2}+\eta}.
\end{align*}
In the above display, (c) follows because $\mM_{1:T}$ is a semi-random ensemble with limiting covariance $\Omega$, (d) follows from standard bounds on maximum of Gaussian random variables (see e.g., \citep[Exercise 2.5.10]{vershynin2018high}). In order to obtain (e), notice that each entry of $\mG_s \mG_t \tran$ is a sum of independent sub-exponential random variables. Hence, the Bernstein Inequality (see e.g., \citep[Theorem 2.8.1]{vershynin2018high}) along with a union bound over the $\dim^2$ entries of $\mG_s \mG_t \tran$ yields (e). Hence, $\hat{\mM}_{1:T}$ is semi-random with limiting covariance matrix $\Omega_\epsilon =\Omega + \epsilon^2 I_{T}$, as claimed. 
\end{proof}
\subsubsection{Proof of \lemref{pert-VAMP-SE}}
\begin{proof}[Proof of \lemref{pert-VAMP-SE}]
We show the claims by induction on $t \in [T]$. 
\paragraph{Base Case $t=1$.} Recalling \eqref{eq:SE-VAMP-unpert} and \eqref{eq:SE-pert-VAMP} we have,
\begin{align*}
    \Phi_1 & = \E[f_1(\serv{A})^2] ,\; \Sigma_1 = \Omega_{11} \cdot \E[f_1(\serv{A})^2], \quad \serv{Z}_1 \sim \gauss{0}{\Sigma_1} \quad \\
    f_1^\epsilon(\auxvec) &= f_1(\auxvec) + \epsilon w_1, \quad
    \Phi^\epsilon_{1} = \E[f_1(\serv{A})^2] + \epsilon^2, \quad
    \Sigma^\epsilon_{1}  = (\Omega_{11} + \epsilon^2) \cdot (\E[f_1(\serv{A})^2] + \epsilon^2), \quad \serv{Z}_1^\epsilon \sim \gauss{0}{\Sigma^\epsilon_1}.
\end{align*}
From these expressions it is immediate that $\lambda_{\min}(\Phi_{1}^\epsilon) > 0, \lambda_{\min}(\Sigma_1^\epsilon) > 0$ and that $\Phi_1^\epsilon \rightarrow \Phi_1$, $\Sigma_1^\epsilon \rightarrow \Sigma_1$ as $\epsilon \rightarrow 0$. This verifies claims (1), (2), and (4) in the statement of the lemma for $t=1$. There is nothing to prove regarding claim (3) when $t=1$, since no correction vector $\alpha_1$ is defined for $t=1$. To verify claim (5) observe that:
\begin{align*}
    |\E h(\serv{Z}_1^{\epsilon}; \serv{A}) - \E h({Z}_1; \serv{A})| & \leq \E | h(\sqrt{\Sigma_1^\epsilon} \serv{G}; \serv{A})  - h(\sqrt{\Sigma_1} \serv{G}; \serv{A})|, \quad \serv{G} \sim \gauss{0}{1} \\
    & \leq L \cdot \E[(1 + (\sqrt{\Sigma_1^\epsilon} + \sqrt{\Sigma_1}) \cdot |\serv{G}| + \|\serv{A}\|^\degree) \cdot |\sqrt{\Sigma_1^\epsilon} - \sqrt{\Sigma_1}| \cdot |\serv{G}|] \\
    & \rightarrow 0 \quad \text{as} \quad \epsilon \rightarrow 0.
\end{align*}
Hence, we have verified the claims of the lemma when $t=1$.
\paragraph{Induction Hypothesis.} Suppose all the claims of the lemma hold for all $s \leq t$. 
\paragraph{Induction Step.} We verify each of the claims of the lemma for $t+1$:
\begin{enumerate}
    \item We show that $\lambda_{\min}(\Phi_{t+1}^{\epsilon}) > 0$ by contradiction. Indeed if $\lambda_{\min}(\Phi_{t+1}^{\epsilon}) = 0$, since $\lambda_{\min}(\Phi_t^\epsilon) > 0$ it must be that:
    \begin{align*}
        \nonlin_{t+1}^\epsilon( \serv{Z}_1^{\epsilon}, \dotsc, \serv{Z}_t^{\epsilon}; \serv{A}, \serv{W}) & = \sum_{s=1}^t \beta_s \cdot \nonlin_s^\epsilon( \serv{Z}_1^{\epsilon}, \dotsc, \serv{Z}_s^{\epsilon}; \serv{A}, \serv{W}) \quad \text{almost surely},
    \end{align*}
    for some coefficients $\beta_{1:t}$. Recalling the definition the perturbed non-linearities $\nonlin_{1:T}^\epsilon$ from \eqref{eq:perturbed-nonlinearities}, this means that:
    \begin{align*}
        \epsilon\serv{W}_{t+1} & = \sum_{s=1}^t \beta_s \cdot \nonlin_s( \serv{Z}_1^{\epsilon}, \dotsc, \serv{Z}_s^{\epsilon}; \serv{A}, \serv{W}) - \nonlin_{t+1}( \serv{Z}_1^{\epsilon}, \dotsc, \serv{Z}_t^{\epsilon}; \serv{A})  + \sum_{s=1}^t (\alpha_t^\epsilon)_s \cdot  \serv{Z}_s^\epsilon  \quad \text{almost surely}.
    \end{align*}
    Note that the definition of the perturbed non-linearities \eqref{eq:perturbed-nonlinearities} guarantees that the RHS of the above equation is independent of $\serv{W}_{t+1} \sim \gauss{0}{1}$. This leads to a contradiction.
    \item Note that $\Sigma_{t+1}^\epsilon$ is the entry-wise product of $\Phi_{t+1}^\epsilon$ and the principal $(t+1) \times (t+1)$ submatrix of $\Omega_\epsilon = \Omega + \epsilon^2 I_{T}$. Since $\lambda_{\min}(\Phi_{t+1}^\epsilon) > 0$ and the smallest eigenvalue of any principal sub-matrix of  $\Omega_\epsilon$ is $ \geq \epsilon^2$, it follows from a result of \citet[Theorem 3, claim (ii)]{bapat1985majorization} that $\lambda_{\min}(\Sigma_{t+1}^\epsilon) \geq \lambda_{\min}(\Phi_{t+1}^\epsilon) \cdot  \epsilon^2 > 0$. 
    \item As $\epsilon \rightarrow 0$, we have:
    \begin{align*}
        {\alpha_{t+1}^\epsilon} \tran \Sigma_{t}^\epsilon \alpha_{t+1}^{\epsilon} &\explain{\eqref{eq:SE-pert-VAMP}}{=} \| {(\Sigma_{t}^\epsilon)}^{-\frac{1}{2}} \E[\nonlin_{t+1}(\serv{Z}_1^\epsilon, \dotsc, \serv{Z}_t^\epsilon; \serv{A}) \cdot \serv{Z}_{[t]}^\epsilon] \|^2 \\
        & \explain{(a)}{=} \|  \E[\nonlin_{t+1}({(\Sigma_{t}^\epsilon)}^{\frac{1}{2}}\serv{G}; \serv{A}) \cdot \serv{G}] \|^2, \quad \serv{G} \sim \gauss{0}{I_t} \\
        & \explain{(b)}{\rightarrow} \|  \E[\nonlin_{t+1}({(\Sigma_{t})}^{\frac{1}{2}}\serv{G}; \serv{A}) \cdot \serv{G}] \|^2, \quad \serv{G} \sim \gauss{0}{I_t}. 
    \end{align*}
    In the above display, in the step (a) we used the fact that since $(\serv{Z}_1^\epsilon, \dotsc, \serv{Z}_t^\epsilon) \sim \gauss{0}{\Sigma_t^\epsilon}$, we have $(\serv{Z}_1^\epsilon, \dotsc, \serv{Z}_t^\epsilon) \explain{d}{=} {(\Sigma_{t})}^{\frac{1}{2}}\serv{G} $ where $\serv{G} \sim \gauss{0}{I_t}$ (independent of $\serv{A}$) and step (b) follows from the induction hypothesis that $\Sigma_{t}^\epsilon \rightarrow \Sigma_t$. Define $\widetilde{\serv{G}} = \Sigma_t^{\frac{1}{2}} \serv{G} \explain{d}{=} (Z_1, \dotsc, Z_t)$. Observe that:
    \begin{align*}
        \lim_{\epsilon \rightarrow 0}  {\alpha_{t+1}^\epsilon} \tran \Sigma_{t}^\epsilon \alpha_{t+1}^{\epsilon}  = \|  \E[\nonlin_{t+1}(\widetilde{\serv{G}}; \serv{A}) \cdot \serv{G}] \|^2 \explain{(c)}{=}  \|  \E[\nonlin_{t+1}(\widetilde{\serv{G}}; \serv{A}) \cdot \E[\serv{G}| \widetilde{\serv{G}}] ] \|^2 &\explain{(d)}{=} \E[\nonlin_{t+1}(\widetilde{\serv{G}}; \serv{A}) \cdot Q\widetilde{\serv{G}} ] \|^2 \\&= \|Q \cdot  \E[\nonlin_{t+1}(\widetilde{\serv{G}}; \serv{A}) \cdot \widetilde{\serv{G}} ] \|^2 \explain{(e)}{=} 0.
    \end{align*}
    In the above display (c) follows from the Tower property and the fact that $\serv{G}, \widetilde{\serv{G}}$ are independent of $\serv{A}$, step (d) follows from the fact that since $(\serv{G}, \widetilde{\serv{G}})$ are jointly Gaussian $\E[\serv{G}| \widetilde{\serv{G}}] = Q\widetilde{\serv{G}}$ for some matrix $Q$ determined by the joint covariance matrix of $(\serv{G}, \widetilde{\serv{G}})$ (the precise formula will not be needed). Step (e) follows from the observation that  $\widetilde{\serv{G}} \explain{d}{=} (Z_1, \dotsc, Z_t)$ and the fact that the non-linearity $\nonlin_{t+1}$ is divergence free (\assumpref{div-free}). This verifies claim (3) of the lemma for $t+1$. 
    \item The sub-matrix formed by the first $t$ rows and columns of $\Phi_{t+1}^\epsilon$ is precisely $\Phi_t^\epsilon$ which converges to $\Phi_t$ as $\epsilon \rightarrow 0$ by the induction hypothesis. Hence, we only need to show that $(\Phi_{t+1}^\epsilon)_{s,t+1} \rightarrow (\Phi_{t+1})_{s,t+1}$ for each $s \leq t+1$. Recalling \eqref{eq:SE-pert-VAMP}, by Cauchy-Schwarz Inequality we have,
    \begin{align*}
        &\left|(\Phi_{t+1}^\epsilon)_{s,t+1} - \E[\nonlin_s(\serv{Z}_1^\epsilon, \dotsc, \serv{Z}_{s-1}^\epsilon ; \serv{A}) \cdot \nonlin_{t+1}(\serv{Z}_1^\epsilon, \dotsc, \serv{Z}_{t}^\epsilon ; \serv{A})] \right| \\& \leq  \sqrt{\epsilon^2\E[\serv{W}_s^2] + \E[\serv{\Delta}_s^2]} \sqrt{\E[\nonlin_{t+1}^\epsilon(\serv{Z}_1^\epsilon, \dotsc, \serv{Z}_{t}^\epsilon ; \serv{A}, \serv{W})^2]  } + \sqrt{\epsilon^2\E[\serv{W}_{t+1}^2] + \E[\serv{\Delta}_{t+1}^2]} \sqrt{\E[\nonlin_{s}^\epsilon(\serv{Z}_1^\epsilon, \dotsc, \serv{Z}_{s-1}^\epsilon ; \serv{A}, \serv{W})^2] }
    \end{align*}
    where we defined:
    \begin{align*}
        \serv{\Delta}_s \explain{def}{=} \sum_{\tau=1}^{s-1} (\alpha_s^\epsilon)_{\tau} \serv{Z}_\tau^\epsilon, \quad \serv{\Delta}_s \explain{def}{=} \sum_{\tau=1}^{t} (\alpha_{t+1}^\epsilon)_{\tau} \serv{Z}_\tau^\epsilon.
    \end{align*}
    Notice that by the induction hypothesis and the proof of claim (3) in the induction step, $\E[\serv{\Delta}_s^2]  = {\alpha_s^\epsilon} \tran \Sigma_{s-1}^\epsilon \alpha_s^\epsilon \rightarrow 0$ for each $s \leq t+1$. Furthermore by claim (5) of the induction hypothesis,
    \begin{align*}
        \lim_{\epsilon \rightarrow 0} \E[\nonlin_s(\serv{Z}_1^\epsilon, \dotsc, \serv{Z}_{s-1}^\epsilon ; \serv{A}) \cdot \nonlin_{t+1}(\serv{Z}_1^\epsilon, \dotsc, \serv{Z}_{t}^\epsilon ; \serv{A})] &= \E[\nonlin_s(\serv{Z}_1, \dotsc, \serv{Z}_{s-1} ; \serv{A}) \cdot \nonlin_{t+1}(\serv{Z}_1, \dotsc, \serv{Z}_{t} ; \serv{A})] \\&\explain{\eqref{eq:SE-VAMP-unpert}}{=} (\Phi_{t+1})_{s,t+1}.
    \end{align*}
    This shows that $\Phi_{t+1}^\epsilon \rightarrow \Phi_{t+1}$. Since $\Sigma_{t+1}^\epsilon$ is the entry-wise product of $\Phi_{t+1}^\epsilon$ and $\Omega_\epsilon$ and $\Omega_\epsilon \rightarrow \Omega$, we immediately obtain $\Sigma_{t+1}^\epsilon \rightarrow \Sigma_{t+1}$. This completes the induction step for claim (4) of the lemma. 
    \item Since $(\serv{Z}_1, \dotsc, \serv{Z}_{t+1}) \sim \gauss{0}{\Sigma_{t+1}}$ and $(\serv{Z}_1^\epsilon, \dotsc, \serv{Z}_{t+1}^\epsilon) \sim \gauss{0}{\Sigma_{t+1}^\epsilon}$, we can write for $\serv{G} \sim \gauss{0}{I_{t+1}}$:
    \begin{align*}
        &|\E[h(\serv{Z}_1, \dotsc, \serv{Z}_{t+1}; \serv{A})] - \E[h(\serv{Z}_1^\epsilon, \dotsc, \serv{Z}_{t+1}^\epsilon; \serv{A})]|  \leq \E | h(\Sigma_{t+1}^{\frac{1}{2}} \serv{G}; \serv{A}) - h({\Sigma_{t+1}^\epsilon}^{\frac{1}{2}} \serv{G}; \serv{A}) | \\
        &\hspace{4cm} \leq L \cdot \E[(1 + \|\serv{A}\|^\degree + \|\Sigma_{t+1}^{\frac{1}{2}}\serv{G}\| + \|{\Sigma_{t+1}^\epsilon}^{\frac{1}{2}} \serv{G}\|) \cdot \|{\Sigma_{t+1}^\epsilon}^{\frac{1}{2}} -\Sigma_{t+1}^{\frac{1}{2}}\| \cdot \|\serv{G}\| ] \\
        & \hspace{4cm} \rightarrow 0,
    \end{align*}
    where the last equation follows from $\Sigma_{t+1}^\epsilon \rightarrow \Sigma_{t+1}$ shown previously. 
\end{enumerate}
This concludes the claim of the lemma. 
\end{proof}
\subsubsection{Proof of \lemref{pert-VAMP-pert-bound}}
\begin{proof}[Proof of \lemref{pert-VAMP-pert-bound}] We consider each claim made in the statement of the lemma.
\begin{enumerate}
    \item We show the first claim by induction. Observe that:
    \begin{align*}
       \limsup_{\dim \rightarrow \infty} \frac{\E \|\iter{\vz}{1}\|^2}{\dim} \leq \limsup_{\dim \rightarrow \infty} \|\mPsi_1\|_{\op} \frac{\E\| \nonlin_1(\auxmat)\|^2}{\dim}  = \E[\nonlin_1(\serv{A})^2] \cdot \limsup_{\dim \rightarrow \infty} \|\mPsi_1\|_{\op}  <\infty. 
    \end{align*}
    Now assume that the claim for all $s \leq t$ for some $t \in [T]$. In the induction step, we verify the claim for $t+1$. Indeed,
    \begin{align*}
        \limsup_{\dim \rightarrow \infty} \frac{\E \|\iter{\vz}{t+1}\|^2}{\dim} &\leq \limsup_{\dim \rightarrow \infty} \|\mPsi_{t+1}\|_{\op} \frac{\E\| \nonlin_{t+1}(\iter{\vz}{1}, \dotsc, \iter{\vz}{t}; \auxmat)\|^2}{\dim}  \\&\explain{(a)}{\leq} \limsup_{\dim \rightarrow \infty}  \|\mPsi_{t+1}\|_{\op} \cdot \left( 2\E \nonlin^2_{t+1}(0, 0, \dotsc 0; \serv{A}) + \frac{2L^2}{\dim} \sum_{s=1}^ t {\E\|\iter{\vz}{s}\|^2}\right) \\
        & \explain{(b)}{\lesssim} 1.
    \end{align*}
    In the above display, step (a) follows from the assumption that the non-linearities $\nonlin_{1:T}$ are Lipschitz with constant $L$ made in the statement of \thref{VAMP} and (b) follows from the induction hypothesis. This proves the first claim of the lemma. 
    \item Recall that we assume that \thref{VAMP} holds under the additional \sassumpref{non-degenerate}. The perturbed VAMP iterates $\iter{\vz}{t}_\epsilon$ satisfy all the requirements of \thref{VAMP} along with \sassumpref{non-degenerate}. Indeed \lemref{semirandom-pert} (item 2) guarantees that $\lambda_{\min}(\Omega_\epsilon) \geq \epsilon^2 > 0$ and \lemref{pert-VAMP-SE} guarantees that $\lambda_{\min}(\Phi_T) > 0$ and $\lambda_{\min}(\Sigma_T) > 0$. Hence,
    \begin{align*}
        \lim_{\dim\rightarrow \infty } \frac{\E\|\iter{\vz_\epsilon}{t}\|^2}{\dim} & = \E |\serv{Z}_t^\epsilon|^2  \explain{$\epsilon \rightarrow 0$}{\longrightarrow} \E[\serv{Z}_t]^2 < \infty.
    \end{align*}
    In the above display, the claim regarding the $\epsilon \rightarrow 0$ limit follows from \lemref{pert-VAMP-SE} (item 5). This proves the second claim of the lemma.
    \item We prove the third claim by induction. Consider the base case $t=1$. Recall from \eqref{eq:VAMP-unpert}, \eqref{eq:VAMP-pert}, and \eqref{eq:perturbed-nonlinearities} that:
    \begin{align*}
        \iter{\vz}{1} & = \mM_1 \nonlin_1(\auxmat), \quad \iter{\vz}{1}_\epsilon = (\mM_1 + \epsilon \mS\mG_1\mS) \cdot (\nonlin_1(\auxmat) + \epsilon \vw_1).
    \end{align*}
    Hence,
    \begin{align*}
        \limsup_{\epsilon \rightarrow 0}\limsup_{\dim \rightarrow \infty} \frac{\E \|\iter{\vz}{1}_\epsilon  -\iter{\vz}{1}\|^2}{\dim} & \leq 2\limsup_{\epsilon \rightarrow 0} \epsilon^2 \cdot \left(\limsup_{\dim \rightarrow \infty} \frac{\E \|\nonlin_1(\auxmat)\|^2 +\dim \epsilon^2 }{\dim} \cdot \|\mG_1\|_{\op}^2  + \|\mM_1\|_{\op}^2 \right) \\&= \limsup_{\epsilon \rightarrow 0} \epsilon^2 \cdot (\E[\nonlin_1^2(\serv{A})]+\epsilon^2) \cdot \limsup_{\dim \rightarrow \infty} \|\mG_1\|_{\op} = 0,
    \end{align*}
    where the last equality follows form the fact that with probability $1$, $\|\mG_1\|_{\op} \lesssim 1$ (see e.g., \citep[Corollary~4.4.8]{vershynin2018high}). As the induction hypothesis, we assume that the claim holds at all iterations $s \leq t$ for some $t \in [T]$. In order to verify that the claim also holds at iteration $t+1$, we recall from \eqref{eq:VAMP-unpert}, \eqref{eq:VAMP-pert}, and \eqref{eq:perturbed-nonlinearities} that:
    \begin{align*}
        \iter{\vz}{t+1} & = \mM_{t+1} \nonlin_{t+1}(\iter{\vz}{1}, \dotsc, \iter{\vz}{t}; \auxmat), \\ \iter{\vz}{t+1}_\epsilon &= (\mM_{t+1} + \epsilon \mS\mG_{t+1}\mS) \cdot \Bigg(\nonlin_{t+1}(\iter{\vz}{1}_\epsilon, \dotsc, \iter{\vz}{t}_\epsilon; \auxmat) + \epsilon \vw_{t+1} - \underbrace{\sum_{s=1}^t (\alpha_{t+1}^\epsilon)_{s} \cdot  \iter{\vz}{s}_\epsilon}_{\explain{def}{=} \iter{\vDelta}{t}_\epsilon} \Bigg).
    \end{align*}
    Hence,
    \begin{align*}
        \iter{\vz}{t+1}_\epsilon -  \iter{\vz}{t+1}  & = (\diamondsuit) + (\spadesuit) + (\clubsuit),
    \end{align*}
    where:
    \begin{align*}
        (\diamondsuit) &\explain{def}{=} \mM_{t+1} \cdot (\nonlin_{t+1}(\iter{\vz}{1}_\epsilon, \dotsc, \iter{\vz}{t}_\epsilon; \auxmat) - \nonlin_{t+1}(\iter{\vz}{1}, \dotsc, \iter{\vz}{t}; \auxmat)), \\
        (\spadesuit) &\explain{def}{=} (\mM_{t+1} + \epsilon \mS\mG_{t+1}\mS) \cdot (\epsilon \vw_{t+1} - \iter{\vDelta}{t}_\epsilon), \\
        (\clubsuit) & \explain{def}{=} \epsilon \mS\mG_{t+1}\mS \nonlin_{t+1}(\iter{\vz}{1}_\epsilon, \dotsc, \iter{\vz}{t}_\epsilon; \auxmat).
    \end{align*}
    We analyze each of these terms. Recalling that the non-linearities are assumed to be uniformly Lipschitz with constant $L$ (cf. \thref{VAMP}):
    \begin{align*}
        \frac{\E\|(\diamondsuit)\|^2}{\dim} & \leq \|\mPsi_{t+1}\|_{\op}^2 \cdot L^2 \cdot \sum_{s=1}^t \frac{\E\|\iter{\vz}{s}_\epsilon-\iter{\vz}{s}\|^2}{\dim},
    \end{align*}
    Hence by the induction hypothesis, $\lim_{\epsilon \rightarrow 0} \limsup_{\dim \rightarrow \infty} \dim^{-1} \|(\diamondsuit)\|^2 = 0$. Next, we consider the term $(\spadesuit)$:
    \begin{align*}
        \frac{\E\|(\spadesuit)\|^2}{\dim} & \leq \E\|\mPsi_{t+1} + \epsilon\mG_{t+1}\|_{\op}^2 \cdot \left( \epsilon^2 + \frac{\E \|\iter{\vDelta}{t}_\epsilon\|^2}{\dim} \right).
    \end{align*}
    Since we assume that \thref{VAMP} holds under \sassumpref{non-degenerate}, $ \lim_{\dim \rightarrow \infty} \dim^{-1} \E \|\iter{\vDelta}{t}_\epsilon\|^2 = {\alpha_{t+1}^\epsilon} \tran \Sigma_t^\epsilon \alpha_{t+1}^\epsilon$. By \lemref{pert-VAMP-SE} (item (3)) ${\alpha_{t+1}^\epsilon} \tran \Sigma_t^\epsilon \alpha_{t+1}^\epsilon \rightarrow 0$ as $\epsilon \rightarrow 0$. Hence, we obtain: $$\lim_{\epsilon \rightarrow 0} \limsup_{\dim \rightarrow \infty} \dim^{-1} \|(\spadesuit)\|^2 = 0.$$ Finally, we analyze the term $(\clubsuit)$. Since the non-linearities are assumed to be uniformly Lipschitz with constant $L$:
    \begin{align*}
         \frac{\E\|(\clubsuit)\|^2}{\dim} & \leq 2 \cdot L^2 \cdot \epsilon^2 \cdot \E[ \|\mG_{t+1}\|^2] \cdot \left( \frac{\E\| \nonlin_{t+1}(\vzero, \dotsc, \vzero ; \auxmat)\|^2}{\dim} + L^2 \sum_{s=1}^t \frac{\E \| \iter{\vz}{s}_\epsilon\|^2}{\dim} \right)  \\
         & = 2 \cdot L^2 \cdot \epsilon^2 \cdot \E[ \|\mG_{t+1}\|^2] \cdot \left( \E[\nonlin_{t+1}^2(0, \dotsc, 0 ; \serv{A})^2] + L^2 \sum_{s=1}^t \frac{\E \| \iter{\vz}{s}_\epsilon\|^2}{\dim} \right).
    \end{align*}
    Since we have already shown that $\dim^{-1} \E\| \iter{\vz}{t}_\epsilon\|^2 \lesssim 1$ in item (2) of this lemma, we obtain $$\lim_{\epsilon \rightarrow 0} \limsup_{\dim \rightarrow \infty} \dim^{-1} \|(\clubsuit)\|^2 = 0.$$ Hence, we have shown that:
    \begin{align*}
        \limsup_{\epsilon \rightarrow 0}\limsup_{\dim \rightarrow \infty} \frac{\E \|\iter{\vz}{t+1}_\epsilon  -\iter{\vz}{t+1}\|^2}{\dim} & = 0,
    \end{align*}
    as desired.
\end{enumerate}
This proves the claim of \lemref{pert-VAMP-pert-bound}.
\end{proof}

\subsection{Balancing Semi-Random Matrices}\label{appendix:balancing}
This section is devoted to the proof of \lemref{balancing}. The proof relies on the following intermediate lemma. 

\begin{lemma}\label{lem:approx-by-balanced} Let $\mM_{1:T} = \mS \Psi_{1:T} \mS$ be a semi-random ensemble (\defref{semirandom}) with limiting covariance matrix $\Omega$ that satisfies $\lambda_{\min}(\Omega) > 0$.  Then there exists a  semi-random ensemble $\hat{\mM}_{1:T} = \mS \hat{\mPsi}_{1:T} \mS $ with limiting covariance matrix $\Omega$, which is balanced (that is, satisfies \sassumpref{balanced}) such that:
\begin{align*}
    \max_{t \in [T]} \|\hat{\mM}_i - {\mM}_i\|_{\op} =  \max_{t \in [T]} \|\hat{\mPsi}_i - {\mPsi}_i\|_{\op} \ll 1.
\end{align*}
\end{lemma}

We defer the proof of the above claim to the end of this section, and present the proof of \lemref{balancing}.

\begin{proof}[Proof of \lemref{balancing}] In order to prove \lemref{balancing}, we will assume that \thref{VAMP} holds in the situation when \sassumpref{non-degenerate} and \sassumpref{balanced} hold (in addition to the assumptions listed in the statement of \thref{VAMP}). We will show that this implies that \thref{VAMP} holds in the situation when \sassumpref{non-degenerate} holds (but not necessarily \sassumpref{balanced}). By \lemref{degeneracy}, this is sufficient to show that \thref{VAMP} holds without any additional assumptions. 

To this end, we consider $T$ iterations of a VAMP algorithm, which satisfies all the assumptions of \thref{VAMP} and \sassumpref{non-degenerate}, but need not satisfy  \sassumpref{balanced}:
\begin{align}\label{eq:VAMP-unbalanced}
    \iter{\vz}{t} = \mM_t \cdot \nonlin_t(\iter{\vz}{1}, \iter{\vz}{2}, \dotsc, \iter{\vz}{t-1}; \auxmat) \quad \forall \; t\; \in \; [T].
\end{align}
In the above display $\mM_{1:T}$ is a semi-random ensemble with limiting covariance matrix $\Omega$ which satisfies $\lambda_{\min}(\Omega)>0$. Recall from \defref{semirandom}, this means that $\mM_i = \mS \mPsi_i \mS$ where $\mS$ is a uniformly random sign diagonal matrix. 
Let: 
\begin{align}\label{eq:SE-VAMP-unbalanced}\serv{Z}_{1}, \dotsc, \serv{Z}_{T} \sim \gauss{0}{\Sigma_T}, \quad \Phi_T, \quad \Sigma_T
\end{align}  
denote the state evolution random variables and covariance matrices associated with the VAMP algorithm in \eqref{eq:VAMP-unbalanced}. Our goal is to show that $(\iter{\vz}{1}, \iter{\vz}{2}, \dotsc, \iter{\vz}{T}, \auxmat) \explain{\pw}{\longrightarrow} (\serv{Z}_{1}, \dotsc, \serv{Z}_{T}, \serv{A})$. In order to do so, we need to show that for any test function $h:\R^{T+\auxdim} \mapsto \R$ that satisfies the regularity hypothesis required by the definition of $\text{\pw}$ convergence (\defref{PW2}) we have, 
\begin{align}\label{eq:balancing-goal}
    H_\dim \explain{def}{=} \frac{1}{\dim} \sum_{\ell = 1}^\dim h(\iter{z}{1}_\ell, \dotsc, \iter{z}{T}_\ell; \auxvec_{\ell}) \explain{P}{\rightarrow} \E[ h(\serv{Z}_1, \dotsc, \serv{Z}_T; \serv{A})].
\end{align}
Using \lemref{approx-by-balanced}, we can obtain a  semi-random ensemble $\hat{\mM}_{1:T} = \mS \hat{\mPsi}_{1:T} \mS $ with limiting covariance matrix $\Omega$, which is balanced (that is, satisfies \sassumpref{balanced}) such that:
\begin{align}\label{eq:balanced-approx-guarantee}
    \max_{t \in [T]} \|\hat{\mM}_i - {\mM}_i\|_{\op} =  \max_{t \in [T]} \|\hat{\mPsi}_i - {\mPsi}_i\|_{\op} \ll 1.
\end{align}
Consider the VAMP iterations driven by the balanced semi-random ensemble:
\begin{align}\label{eq:VAMP-balanced}
    \iter{\hat{\vz}}{t} = \hat{\mM}_t \cdot \nonlin_t(\iter{\hat{\vz}}{1}, \iter{\hat{\vz}}{2}, \dotsc, \iter{\hat{\vz}}{t-1}; \auxmat) \quad \forall \; t\; \in \; [T].
\end{align}
Since the balanced semi-random ensemble $\hat{\mM}_{1:T}$ has the same limiting covariance matrix $\Omega$ as $\mM_{1:T}$, iteration \eqref{eq:VAMP-balanced} and \eqref{eq:VAMP-unbalanced} have the same state evolution. Furthermore, since iteration \eqref{eq:VAMP-balanced} satisfies \sassumpref{non-degenerate} and \sassumpref{balanced} in addition to the assumptions specified in \thref{VAMP}, we know that:
\begin{align*}
     \hat{H}_\dim \explain{def}{=} \frac{1}{\dim} \sum_{\ell = 1}^\dim h(\iter{\hat{z}}{1}_\ell, \dotsc, \iter{\hat{z}}{T}_\ell; \auxvec_{\ell}) \explain{P}{\rightarrow} \E[ h(\serv{Z}_1, \dotsc, \serv{Z}_T; \serv{A})].
\end{align*}
Hence \eqref{eq:balancing-goal} follows if we show that:
\begin{align*}
    \E | \hat{H}_\dim - H_\dim| \rightarrow 0.
\end{align*}
Indeed,
\begin{align*}
     \limsup_{\dim \rightarrow \infty}\E | \hat{H}_\dim - H_\dim| & \leq \limsup_{\dim \rightarrow \infty}\E \left[ \frac{1}{\dim} \sum_{\ell=1}^\dim \big| h(\iter{\hat{z}}{1}_\ell, \dotsc, \iter{\hat{z}}{T}_\ell; \auxvec_{\ell}) - h(\iter{z}{1}_\ell, \dotsc, \iter{z}{T}_\ell; \auxvec_{\ell}) \big| \right] \\
     & \explain{(a)}{\leq} \limsup_{\dim \rightarrow \infty} L  \left\{ \E\|\serv{A}\|^{2\degree} + \ \sum_{t=1}^T \frac{\E\|\iter{\vz}{t}\|^2}{\dim} + \frac{\E\|\iter{\hat{\vz}}{t}\|^2}{\dim}   \right\}^{1/2}  \left\{ \sum_{t=1}^T \frac{\E \|\iter{\hat{\vz}}{t} - \iter{\vz}{t}\|^2}{\dim}\right\}^{1/2} \\
     & \explain{(b)}{\leq} \limsup_{\dim \rightarrow \infty} L  \left\{ \E\|\serv{A}\|^{2\degree} + \ \sum_{t=1}^T \frac{2\E\|\iter{\hat{\vz}}{t}-\iter{\vz}{t}\|^2}{\dim} + \frac{3\E\|\iter{\hat{\vz}}{t}\|^2}{\dim}   \right\}^{1/2}  \left\{ \sum_{t=1}^T \frac{\E \|\iter{\hat{\vz}}{t} - \iter{\vz}{t}\|^2}{\dim}\right\}^{1/2} \\
     &\explain{(c)}{=} \limsup_{\dim \rightarrow \infty}  L  \left\{ \E\|\serv{A}\|^{2\degree} + 3 \Tr(\Sigma_T) +  \sum_{t=1}^T \frac{2\E\|\iter{\hat{\vz}}{t}-\iter{\vz}{t}\|^2}{\dim}    \right\}^{1/2}  \left\{ \sum_{t=1}^T \frac{\E \|\iter{\hat{\vz}}{t} - \iter{\vz}{t}\|^2}{\dim}\right\}^{1/2}.
\end{align*}
In the above display (a) follows from Cauchy-Schwarz Inequality and the continuity hypothesis on $h$ in \defref{PW2}. Step (b) follows from the triangle inequality and step (c) follows from \thref{VAMP} applied to the iteration \eqref{eq:VAMP-balanced} (which satisfies \sassumpref{non-degenerate} and \sassumpref{balanced}). Hence, the claim of the lemma follows if we can show that:
\begin{align}\label{eq:balancing-final-goal}
\lim_{\dim \rightarrow \infty} \frac{\E \|\iter{\hat{\vz}}{t} - \iter{\vz}{t}\|^2}{\dim} = 0 \quad \forall \; t \; \in \; [T].
\end{align}
This can be shown by induction. Indeed, for the base case $t=0$ we have:
\begin{align*}
    \lim_{\dim \rightarrow \infty} \frac{\E \|\iter{\hat{\vz}}{1} - \iter{\vz}{1}\|^2}{\dim} & = \lim_{\dim \rightarrow \infty} \frac{\E \|(\hat{\mM}_1 - \mM_1) \cdot \nonlin_1(\auxmat)\|^2}{\dim} \leq \lim_{\dim \rightarrow \infty} \|\hat{\mPsi}_1 - \mPsi_1\|_{\op} \cdot \E \nonlin_1^2(\serv{A}) \explain{\eqref{eq:balanced-approx-guarantee}}{=} 0. 
\end{align*}
Assuming that \eqref{eq:balancing-final-goal} holds for all iterations $s \leq t$ for some $t \in [T-1]$ as the induction hypothesis, we have:
\begin{align*}
     &\lim_{\dim \rightarrow \infty} \frac{\E \|\iter{\hat{\vz}}{t+1} - \iter{\vz}{t+1}\|^2}{\dim}  \\&\hspace{2.5cm}= \lim_{\dim \rightarrow \infty} \frac{\E \|(\hat{\mM}_{t+1} - \mM_{t+1}) \cdot \nonlin_{t+1}(\iter{\hat{\vz}}{1:t};\auxmat) + \mM_{t+1} \cdot (\nonlin_{t+1}(\iter{\hat{\vz}}{1:t};\auxmat) - \nonlin_{t+1}(\iter{{\vz}}{1:t};\auxmat))\|^2}{\dim} \\&\hspace{2.5cm} \explain{}{\leq} 2\lim_{\dim \rightarrow \infty} \|\hat{\mPsi}_{t+1} - \mPsi_{t+1}\|_{\op}^2 \cdot \frac{\E[\|\nonlin_{t+1}(\iter{\hat{\vz}}{1:t};\auxmat)\|^2]}{\dim} \\&\hspace{7cm}+ 2\lim_{\dim \rightarrow \infty} \|\mPsi_{t+1}\|_{\op}^2 \cdot \frac{\E[\|\nonlin_{t+1}(\iter{\hat{\vz}}{1:t};\auxmat) - \nonlin_{t+1}(\iter{{\vz}}{1:t};\auxmat))\|^2]}{\dim} \\ &\hspace{2.5cm}\explain{(d)}{\leq} 2\lim_{\dim \rightarrow \infty} \|\hat{\mPsi}_{t+1} - \mPsi_{t+1}\|_{\op}^2  \left(\frac{2\E[\|\nonlin_{t+1}(\vzero, \dotsc, \vzero;\auxmat)\|^2]}{\dim} +  \sum_{s=1}^t \frac{2L^2\E\| \iter{\hat{\vz}}{s}\|^2}{\dim} \right) \\&\hspace{10cm}+ 2L^2\lim_{\dim \rightarrow \infty} \|\mPsi_{t+1}\|_{\op}^2  \sum_{s=1}^t\frac{\E \|\iter{\hat{\vz}}{s} - \iter{\vz}{s}\|^2}{\dim} \\ &\hspace{2.5cm} \explain{(e)}{=} 2\lim_{\dim \rightarrow \infty} \|\hat{\mPsi}_{t+1} - \mPsi_{t+1}\|_{\op}^2  \left(\frac{2\E[\|\nonlin_{t+1}(\vzero, \dotsc, \vzero;\auxmat)\|^2]}{\dim} +  \sum_{s=1}^t \frac{2L^2\E\| \iter{\hat{\vz}}{s}\|^2}{\dim} \right)  \\ &\hspace{2.5cm} \explain{(f)}{=} 4(\E \nonlin_{t+1}^2(0, \dotsc, 0; \serv{A}) + \Tr(\Sigma_T)) \cdot \lim_{\dim \rightarrow \infty} \|\hat{\mPsi}_{t+1} - \mPsi_{t+1}\|_{\op}^2 \\
     &\hspace{2.5cm} \explain{\eqref{eq:balanced-approx-guarantee}}{=} 0.\end{align*}
In the above display (d) follows from the assumption that $\nonlin_{1:T}$ are uniformly Lipschitz (cf. \thref{VAMP}), (e) follows from the induction hypothesis, (f) follows from \thref{VAMP} applied to the iteration \eqref{eq:VAMP-balanced} (which satisfies \sassumpref{non-degenerate} and \sassumpref{balanced}). This concludes the proof of \lemref{balancing}. 
\end{proof}

\subsubsection{Proof of \lemref{approx-by-balanced}}
\begin{proof}[Proof of \lemref{approx-by-balanced}]
Recalling the definition of semi-random ensemble (\defref{semirandom}) we know that $\mM_i = \mS \mPsi_i \mS$ where $\mS$ is a uniformly random diagonal sign matrix and $\mPsi_{1:T}$ are deterministic matrices. Let $\mPsi_{i1}, \mPsi_{i2}, \dotsc, \mPsi_{i\dim}$ denote the $\dim$ rows of $\mPsi_i$. Let $\hat{\Omega}$ be the empirical covariance matrix of the semi-random ensemble $\mM_{1:T}$ with entries $\hat{\Omega}_{ij} \explain{def}{=} \Tr(\mPsi_i \mPsi_j \tran)/\dim$ and $\Omega$ be the limiting covariance matrix of $\mM_{1:T}$.  For each $\ell \in [\dim]$, we define  a $T\times \dim$ matrix $\mU_{\ell}$ with rows:
\begin{align*}
    \mU_{\ell}\tran = [\mPsi_{1\ell} , \mPsi_{2\ell}, \dotsc, \mPsi_{T\ell}].
\end{align*} Define $\Omega_{\ell} \explain{def}{=} \mU_{\ell} \mU_{\ell} \tran$.  Note that the definition of a semi-random ensemble guarantees that:
\begin{align} \label{eq:omega-ell-pert-bound}
    \max_{\ell \in [\dim]} \|\Omega_{\ell} - \hat{\Omega}\|_{\op} &\lesssim \dim^{-1/2+\epsilon} \; \forall \; \epsilon > 0, \quad
     \|\hat{\Omega} - \Omega\|_{\op}  \ll 1.
\end{align}
In particular, this means that there is a  $\dim_0 \in \N$ such that:
\begin{align}\label{eq:lambda-min}
    \min_{\ell \in [\dim]} \lambda_{\min}(\Omega_{\ell}) \geq \lambda_{\min}(\Omega)/2 > 0 \quad \forall \; \dim \; \geq \; \dim_0.
\end{align}
Hence, for large enough $\dim$, each of the matrices $\Omega_{\ell}$ is invertible simultaneously. In order to construct the desired balanced semi-random ensemble, for each $\ell \in [\dim]$, we will mix the rows $\mPsi_{1\ell} , \mPsi_{2\ell}, \dotsc, \mPsi_{T\ell}$ to form $\hat{\mPsi}_{1\ell} , \hat{\mPsi}_{2\ell}, \dotsc, \hat{\mPsi}_{T\ell}$ so that the Gram matrix of the new rows is exactly $\Omega$ (instead of $\Omega_{\ell} \approx \Omega$). Specifically, we define the matrix $\hat{\mU}_{\ell}$ with rows given by:
\begin{align} \label{eq:hat-U-def}
    \hat{\mU}_{\ell} \tran = [\hat{\mPsi}_{1\ell} , \hat{\mPsi}_{2\ell}, \dotsc, \hat{\mPsi}_{T\ell}] \explain{def}{=} \mU_{\ell} \tran \Omega_{\ell}^{-\frac{1}{2}} \Omega^{\frac{1}{2}}.
\end{align}
We construct the approximating balanced semi-random ensemble $\hat{\mM}_{1:T}$ as follows:
\begin{align*}
    \hat{\mM}_i \explain{def}{=} \mS\hat{\mPsi}_i \mS, \quad \hat{\mPsi}_i\tran \explain{def}{=} [\hat{\mPsi}_{i1}, \hat{\mPsi}_{i2}, \dotsc, \hat{\mPsi}_{i\dim}]. 
\end{align*}
We first show that $\max_{i \in [T]} \|\hat{\mPsi}_i - \mPsi_i\|_{\op} \lesssim \dim^{-1/2 + \epsilon}$ for any $\epsilon > 0$. Define the block matrix:
\begin{align*}
    \mXi \explain{def}{=} \begin{bmatrix}  \hat{\mU}_{1} \tran - {\mU}_{1} \tran & \vzero  & \hdots & \vzero \\ \vzero &  \hat{\mU}_{2} \tran - {\mU}_{2} \tran & \hdots & \vzero \\  \vdots & \vdots &   & \vdots \\  \vzero & \vzero & \hdots &  \hat{\mU}_{\dim} \tran - {\mU}_{\dim} \tran \end{bmatrix}.
\end{align*}
Observe that for any $i \in [T]$:
\begin{align*}
     \|\hat{\mPsi}_i - \mPsi_i\|_{\op} & = \max_{\substack{\vx \in \R^{\dim} \\ \|\vx\| = 1}} \| (\hat{\mPsi}_i - \mPsi_i) \cdot \vx  \| =  \max_{\substack{\vx \in \R^{\dim} \\ \|\vx\| = 1}} \|\mXi \mathscr{X}_i(\vx)\|,
\end{align*}
where:
\begin{align*}
    \mathscr{X}_i(\vx) \explain{def}{=} \begin{bmatrix} x_1 \ve_i  \\ x_2 \ve_i \\ \vdots \\ x_\dim \ve_i \end{bmatrix}.
\end{align*}
In the above display, $\ve_{1:T}$ denote the standard basis in $\R^T$. Hence,
\begin{align*}
    \max_{i \in [T]} \|\hat{\mPsi}_i - \mPsi_i\|_{\op} & = \max_{i \in [T]} \max_{\substack{\vx \in \R^{\dim} \\ \|\vx\| = 1}} \|\mXi \mathscr{X}_i(\vx)\|  \leq \max_{i \in [T]}\max_{\substack{\vx \in \R^{\dim} \\ \|\vx\| = 1}} \|\mXi\|_{\op}\| \mathscr{X}_i(\vx)\|  = \|\mXi\|_{\op}.
\end{align*}
We can upper bound $\|\mXi\|_{\op}$ as follows:
\begin{align*}
    \|\mXi\|_{\op}^2  = \max_{\ell \in \dim} \|\hat{\mU}_{\ell} \tran - \hat{\mU}_{\ell} \tran\|_{\op}^2 &\explain{\eqref{eq:hat-U-def}}{=} \max_{\ell \in \dim} \|\mU_{\ell}\tran \cdot (\Omega_{\ell}^{-\frac{1}{2}} \Omega^{\frac{1}{2}}-I_{T}) \|_{\op}^2 \\&= \max_{\ell \in \dim} \|(\Omega_{\ell}^{-\frac{1}{2}} \Omega^{\frac{1}{2}}-I_{T}) \tran \mU_{\ell} \mU_{\ell} \tran (\Omega_{\ell}^{-\frac{1}{2}} \Omega^{\frac{1}{2}}-I_{T})\|_{\op} \\
    & = \max_{\ell \in \dim} \|(\Omega_{\ell}^{-\frac{1}{2}} \Omega^{\frac{1}{2}}-I_{T})\tran \Omega_{\ell} (\Omega_{\ell}^{-\frac{1}{2}} \Omega^{\frac{1}{2}}-I_{T})\|_{\op} =  \max_{\ell \in \dim} \| \Omega^{\frac{1}{2}}-\Omega_{\ell}^{\frac{1}{2}}\|_{\op}^2.
\end{align*}
Using standard perturbation bounds on the matrix square roots (e.g. \citep[Lemma 2.2]{schmitt1992perturbation}):
\begin{align*}
    \max_{i \in [T]} \|\hat{\mM}_i - {\mM}_i\|_{\op} =  \max_{i \in [T]} \|\hat{\mPsi}_i - {\mPsi}_i\|_{\op} & \leq  \max_{\ell \in \dim} \| \Omega^{\frac{1}{2}}-\Omega_{\ell}^{\frac{1}{2}}\|_{\op} \leq \frac{\| \Omega-\Omega_{\ell}\|_{\op}}{\sqrt{\lambda_{\min}(\Omega_\ell)} + \sqrt{\lambda_{\min}({\Omega})}} \ll 1,
\end{align*}
where the final bound follows from  \eqref{eq:omega-ell-pert-bound} and \eqref{eq:lambda-min}. This proves the claimed bound on $\max_{i \in [T]} \|\hat{\mM}_i - {\mM}_i\|_{\op} =  \max_{i \in [T]} \|\hat{\mPsi}_i - {\mPsi}_i\|_{\op}$. In order to complete the proof of the lemma, we verify that $\hat{\mM}_{1:T}$ is a balanced semi-random ensemble. Recalling \eqref{eq:hat-U-def} we observe that that for any $i \in [T]$ and any $\ell \in [\dim]$, row $\ell$ of the matrix $\hat{\mPsi}_i$ is given by:
\begin{align*}
    \hat{\mPsi}_{i\ell} & = \sum_{t=1}^T  (\Omega_{\ell}^{-\frac{1}{2}} \Omega^{\frac{1}{2}})_{tj} \cdot \mPsi_{t \ell}.
\end{align*}
Hence,
\begin{align} \label{eq:balancing-semirandom-1}
    \max_{i \in [T]}\|  \hat{\mPsi}_{i}\|_{\infty} & \leq T \cdot \bigg(\max_{\ell \in [\dim]} \|\Omega_{\ell}^{-\frac{1}{2}} \Omega^{\frac{1}{2}}\|_{\op} \bigg) \cdot   \max_{i \in [T]}\|  {\mPsi}_{i}\|_{\infty} \explain{(a)}{\lesssim} \dim^{-1/2 + \epsilon}, \\
      \max_{i,j \in [T]} \max_{\substack{\ell, \ell^\prime \in [\dim] \\ \ell \neq \ell^\prime }} |(\hat{\mPsi}_i \hat{\mPsi}_j \tran)_{\ell, \ell^\prime}| & \leq T^2 \cdot \bigg(\max_{\ell \in [\dim]} \|\Omega_{\ell}^{-\frac{1}{2}} \Omega^{\frac{1}{2}}\|_{\op}^2 \bigg) \cdot \max_{i,j \in [T]} \max_{\substack{\ell, \ell^\prime \in [\dim] \\ \ell \neq \ell^\prime }} |({\mPsi}_i {\mPsi}_j \tran)_{\ell, \ell^\prime}| \explain{(a)}{\lesssim} \dim^{-1/2 + \epsilon}. \label{eq:balancing-semirandom-2}
\end{align}
In the above display, the steps marked (a) follow from \eqref{eq:lambda-min} and the fact that $\mM_{1:T} = \mS \mPsi_{1:T} \mS$ form a semi-random ensemble. Furthermore, for any $i,j \in [T]$ and any $\ell \in [\dim]$:
\begin{align} \label{eq:balancing-semirandom-3}
    (\hat{\mM}_i \hat{\mM}_j \tran)_{\ell,\ell} = \ip{\hat{\mPsi}_{i\ell}}{\hat{\mPsi}_{j\ell}} = (\hat{\mU}_{\ell} \hat{\mU}_{\ell} \tran)_{ij} \explain{\eqref{eq:hat-U-def}}{=} (\Omega^{\frac{1}{2}} \Omega_{\ell}^{-\frac{1}{2}} \mU_{\ell} \mU_{\ell} \tran \Omega_{\ell}^{-\frac{1}{2}} \Omega^{\frac{1}{2}} )_{ij} = (\Omega^{\frac{1}{2}} \Omega_{\ell}^{-\frac{1}{2}} \Omega_{\ell} \Omega_{\ell}^{-\frac{1}{2}} \Omega^{\frac{1}{2}} )_{ij} = (\Omega)_{ij}.
\end{align}
\eqref{eq:balancing-semirandom-1}, \eqref{eq:balancing-semirandom-2}, and \eqref{eq:balancing-semirandom-3} together show that the matrices $\hat{\mM}_{1:T} = \mS \hat{\mPsi}_{1:T} \mS$ form a balanced (\sassumpref{balanced}) semi-random ensemble (\defref{semirandom}) with limiting covariance $\Omega$. This concludes the proof of this lemma. 
\end{proof}

\subsection{Orthogonalization of VAMP Iterates} \label{appendix:orthogonalization}
\begin{proof}[Proof of \lemref{orthogonalization}]
 In order to prove \lemref{orthogonalization}, we will assume that \thref{VAMP} holds in the situation when the semi-random ensemble is balanced (\sassumpref{balanced}) and the VAMP iteration is orthogonalized (\sassumpref{orthogonality}). We will show that this implies that \thref{VAMP} holds in the situation when the VAMP iteration is non-degenerate (\sassumpref{non-degenerate}) and the semi-random ensemble is balanced (\sassumpref{balanced}). By \lemref{balancing}, this is sufficient to show that \thref{VAMP} holds without any additional assumptions. To this end, we consider $T$ iterations of a VAMP algorithm which satisfies all the assumptions of \thref{VAMP}, along with the non-degeneracy assumption (\sassumpref{non-degenerate}) and the balanced semi-random ensemble assumption (\sassumpref{balanced}):
\begin{align}\label{eq:VAMP-non-ortho}
    \iter{\vz}{t} = \mM_t \cdot \nonlin_t(\iter{\vz}{1}, \iter{\vz}{2}, \dotsc, \iter{\vz}{t-1}; \auxmat) \quad \forall \; t\; \in \; [T].
\end{align}
In the above display $\mM_{1:T}$ is a balanced semi-random ensemble with limiting covariance matrix $\Omega$. Recall from \defref{semirandom}, this means that $\mM_i = \mS \mPsi_i \mS$ where $\mS$ is a uniformly random sign diagonal matrix. 
Let: 
\begin{align}\label{eq:SE-VAMP-non-ortho}\serv{Z}_{1}, \dotsc, \serv{Z}_{T} \sim \gauss{0}{\Sigma_T}, \quad \Phi_T, \quad \Sigma_T
\end{align}  
denote the state evolution random variables and covariance matrices associated with the VAMP algorithm in \eqref{eq:VAMP-non-ortho} (as defined in \eqref{eq:SE-VAMP}). Our goal is to show that $(\iter{\vz}{1}, \iter{\vz}{2}, \dotsc, \iter{\vz}{T}, \auxmat) \explain{\pw}{\longrightarrow} (\serv{Z}_{1}, \dotsc, \serv{Z}_{T}, \serv{A})$. Note that the non-degeneracy assumption (\sassumpref{non-degenerate}) guaratees that:
\begin{align}
    \lambda_{\min}(\Omega) > 0, \quad \lambda_{\min}(\Phi_T) > 0, \quad \lambda_{\min}(\Sigma_T) > 0.
\end{align}

\paragraph{Step 1: Orthogonalization of Matrix Ensemble.} By viewing the matrices $\mPsi_{1:T}$ as vectors, we can orthogonalize them via Gram-Schmidt orthogonalization to obtain a matrix ensemble $\hat{\mPsi}_{1:T}$ and an invertible, lower triangular matrix $P \in \R^{T \times T}$ such that: 
\begin{subequations}\label{eq:matrix-ortho}
\begin{align}
    \ip{\hat{\mPsi}_i}{\hat{\mPsi}_j} & = \begin{cases} 1 &: i = j \\
    0&: i \neq j\end{cases},  \\
    P P\tran & = \Omega, \\
    \mPsi_i & = \sum_{j \leq i} P_{ij} \cdot \hat{\mPsi}_j,  \\
    \hat{\mPsi}_i & =  \sum_{j \leq i} (P^{-1})_{ij} \cdot {\mPsi}_j.
\end{align}
\end{subequations}
Define $\hat{\mM}_{1:T} = \mS \hat{\mPsi}_{1:T} \mS$. Observe that since $\hat{\mM}_{1:T}$ is a linear combination of a balanced semi-random ensemble $\mM_{1:T}$, $\hat{\mM}_{1:T}$ themselves form a balanced semi-random ensemble with limiting covariance matrix $I_T$. 

\paragraph{Step 2: Orthogonalization of $\nonlin_{1:T}$.}  For any two functions $h_1, h_2 : \R^{T + \auxdim} \mapsto \R$ define the inner product:
\begin{align*}
    \ip{h_1}{h_2} \explain{def}{=} \E[ h_1(\serv{Z}_1, \dotsc, \serv{Z}_T; \serv{A}) \cdot h_2(\serv{Z}_1, \dotsc, \serv{Z}_T; \serv{A})],
\end{align*}
where $\serv{Z}_{1:T} \sim \gauss{0}{\Sigma_T}$ are the Gaussian state evolution random variables for the VAMP algorithm \eqref{eq:VAMP-non-ortho}. Using Gram-Schmidt Orthogonalization with respect to the above inner product, we can find functions $\hat{\nonlin}_{1:T}$ with $\hat{\nonlin}_t : \R^{t-1+\auxdim} \mapsto \R$ and an invertible lower triangular matrix $Q$ such that:
\begin{subequations}\label{eq:nonlin-ortho}
\begin{align}
    \ip{\nonlin_s}{\nonlin_t} & \explain{def}{=} \E[ \nonlin_s(\serv{Z}_1, \dotsc, \serv{Z}_{s-1}; \serv{A}) \cdot \nonlin_t(\serv{Z}_1, \dotsc, \serv{Z}_t; \serv{A})] = \begin{cases} 1 &: s = t \\
    0&: s \neq t\end{cases},  \\
    Q Q\tran & = \Phi_{T}, \\
    \nonlin_t(z_1, \dotsc, z_{t-1}; \auxvec) & = \sum_{s \leq t} Q_{ts} \cdot \hat{\nonlin}_s(z_1, \dotsc, z_{s-1}; \auxvec),  \\
   \hat{\nonlin}_t(z_1, \dotsc, z_{t-1}; \auxvec) & = \sum_{s \leq t} (Q^{-1})_{ts} \cdot {\nonlin}_s(z_1, \dotsc, z_{s-1}; \auxvec)
\end{align}
\end{subequations}

\paragraph{Step 3: Orthogonalized Iterates.} We introduce the orthogonalized iterates:
\begin{align} \label{eq:ortho-iterates-def}
    \iter{\hat{\vz}}{t,i} \explain{def}{=} \hat{\mM}_i \hat{\nonlin}_t(\iter{\vz}{1}, \dotsc, \iter{\vz}{t-1}; \auxmat) \quad \forall \; t \; \in [T], \; i \; \in \; [T].
\end{align}
Observe that the VAMP iterates \eqref{eq:VAMP-non-ortho} $\iter{\vz}{t}$ can be expressed in terms of the orthogonalized iterates. Indeed for any $t \in [T]$:
\begin{align*}
    \iter{\vz}{t} &= \mM_t \nonlin_t(\iter{\vz}{1}, \dotsc, \iter{\vz}{t-1}; \auxmat) \\
    & \explain{\eqref{eq:matrix-ortho}}{=} \sum_{j \leq t} P_{tj} \cdot  \hat{\mM}_j \nonlin_t(\iter{\vz}{1}, \dotsc, \iter{\vz}{t-1}; \auxmat) \\
    & \explain{\eqref{eq:nonlin-ortho}}{=} \sum_{j \leq t} \sum_{s \leq t} P_{tj} \cdot Q_{ts} \cdot  \hat{\mM}_j \hat{\nonlin}_s(\iter{\vz}{1}, \dotsc, \iter{\vz}{s-1}; \auxmat) \\
    & = \sum_{j \leq t} \sum_{s \leq t} P_{tj} \cdot Q_{ts} \cdot  \iter{\hat{\vz}}{s,j}.
\end{align*}
For notational convenience, given a collection scalars $v_{sj} : j \in [T], s \in [T]$ for each $t \in [T]$,  define the linear map:
\begin{align} \label{eq:linmap-def}
    \linmap_t(v_{1,1:T}, \dotsc, v_{t,1:T}) \explain{def}{=} \sum_{j=1}^T \sum_{s=1}^t P_{tj} \cdot Q_{ts} \cdot v_{sj} \explain{(a)}{=}  \sum_{j=1}^t \sum_{s=1}^t P_{tj} \cdot Q_{ts} \cdot v_{sj},
\end{align}
where equality (a) follows from the fact that $P$ is a lower triangular matrix. Hence,
\begin{align} \label{eq:nonortho-to-ortho}
    \iter{\vz}{t} & = \linmap_t(\iter{\hat{\vz}}{1,1:T}, \dotsc, \iter{\hat{\vz}}{t,1:T}),
\end{align}
where the linear map $\linmap_t$ acts entrywise on the vectors $\iter{\hat{\vz}}{1:t,1:t}$. 
\paragraph{Step 4: Dynamics of Orthogonalized Iterates.} Next, we show that the orthogonalized iterates can be obtained by running $T^2$ iterations of a VAMP algorithm that satisfies the orthogonality assumption (\sassumpref{orthogonality}) in addition to the other assumptions stated in \thref{VAMP}. In order to see this, observe that any $\tau \in [T^2]$ has a unique representation of the form $\tau = T(t-1) + i$ where $i,t \in [T]$. Define the iterates: $\iter{\vv}{\tau} \explain{def}{=} \iter{\hat{\vz}}{t,i}$. Observe that:
\begin{align}
    \iter{\vv}{\tau}  \explain{def}{=} \iter{\hat{\vz}}{t,i}  &\explain{\eqref{eq:ortho-iterates-def}}{=} \hat{\mM}_i \hat{\nonlin}_t(\iter{\vz}{1}, \iter{\vz}{2},  \dotsc, \iter{\vz}{t-1}; \auxmat) \nonumber \\
    & \explain{\eqref{eq:nonortho-to-ortho}}{=}  \hat{\mM}_i \hat{\nonlin}_t(\linmap_1(\iter{\hat{\vz}}{1,1}), \linmap_2(\iter{\hat{\vz}}{1,1:T}, \iter{\hat{\vz}}{2,1:T}), \dotsc, \linmap_{t-1}(\iter{\hat{\vz}}{1,1:T},\dotsc, \iter{\hat{\vz}}{t-1,1:T}); \auxmat) \nonumber \\
    & =  \hat{\mM}_i \hat{\nonlin}_t(\linmap_1(\iter{\vv}{1}), \linmap_1(\iter{\vv}{1}, \dotsc, \iter{\vv}{2T}) \dotsc, \linmap_{t-1}(\iter{\vv}{1}, \dotsc, \iter{\vv}{T(t-1)}); \auxmat) \nonumber.
\end{align}
\begin{subequations}\label{eq:final-ortho-VAMP}
For each $\tau \in [T^2]$ with representation $\tau = T(t-1) + i$ define:
\begin{align}
    \mQ_{\tau} &\explain{def}{=} \hat{\mM}_i, \\ g_{\tau}(v_1, \dotsc, v_{\tau-1}; \auxvec) &\explain{def}{=} \hat{\nonlin}_t(\linmap_1(v_1,\dotsc, v_T), \linmap_2(v_1, \dotsc, v_{2T}), \linmap_3(v_1, \dotsc, v_{3T}), \dotsc, \linmap_{t-1}(v_1, \dotsc, v_{T(t-1)}); \auxvec).
\end{align}
Hence,
\begin{align}
    \iter{\vv}{\tau}  & = \mQ_\tau g_{\tau}(\iter{\vv}{1}, \dotsc, \iter{\vv}{\tau-1}; \auxmat). 
\end{align}
\end{subequations}
Next, we verify the VAMP algorithm in \eqref{eq:final-ortho-VAMP} satisfies the orthogonality assumption (\sassumpref{orthogonality}), balanced semi-random ensemble assumption (\sassumpref{balanced}) in addition to all the assumptions required by \thref{VAMP}. Observe that:
\begin{enumerate}
    \item Observe that $\mQ_{1:T^2}$ is a balanced semi-random ensemble with limiting covariance matrix $\Upsilon$, where for any $\tau, \tau^\prime \in [T^2]$ with representations $\tau = T(t-1) + i$ and $\tau^\prime = T(t^\prime - 1) + i^\prime$, $\Upsilon_{\tau\tau^\prime}$ is given by the formula:
    \begin{align} \label{eq:Q-ortho}
        \Upsilon_{\tau \tau^\prime} & = \begin{cases} 1 &: i = i^\prime \\ 0 &: i \neq i^\prime \end{cases}.
    \end{align}
    \item Let $\serv{V}_1, \dotsc, \serv{V}_{T^2}$ be i.i.d. $\gauss{0}{1}$ random variables. For each $t \in [T]$ define the random variables:
    \begin{align} \label{eq:lin-map-gaussian}
        \serv{W}_t & = \linmap_t(\serv{V}_1, \dotsc, \serv{V}_{Tt}).
    \end{align}Observe that for any $\tau, \tau^\prime \in [T^2]$ with representations $\tau = T(t-1) + i$ and $\tau^\prime = T(t^\prime - 1) + i^\prime$:
    \begin{align*}
        \E[g_\tau(\serv{V}_1, \dotsc, \serv{V}_{\tau-1}; \serv{A}) \cdot g_{\tau^\prime}(\serv{V}_1, \dotsc, \serv{V}_{\tau^\prime-1}; \serv{A})] & \explain{\eqref{eq:final-ortho-VAMP}}{=}  \E[ \hat{\nonlin}_t(\serv{W}_1, \dotsc, \serv{W}_{t-1}; \serv{A}) \cdot \hat{\nonlin}_{t^\prime}(\serv{W}_1, \dotsc, \serv{W}_{t^\prime-1}; \serv{A})].
    \end{align*}
    In order to further simplify the above formula, we note that $\serv{Z}_{1:T}$ are mean zero Gaussian random variables. Furthermore, for any $t,t^\prime \in [T]$, we can compute their covariance:
    \begin{align*}
        \E[\serv{W}_t \serv{W}_{t^\prime}] & \explain{\eqref{eq:lin-map-gaussian}}{=} \E[ \linmap_t(\serv{V}_1, \dotsc, \serv{V}_{Tt}) \cdot  \linmap_{t^\prime}(\serv{V}_1, \dotsc, \serv{V}_{Tt^\prime})] \\
        & \explain{\eqref{eq:linmap-def}}{=}  \sum_{j,j^\prime=1}^T \sum_{s=1}^t \sum_{s^\prime=1}^{t^\prime} P_{tj}  Q_{ts} \cdot P_{t^\prime j^\prime}  Q_{t^\prime s^\prime} \cdot \E[\serv{V}_{T(s-1) + j} \serv{V}_{T(s^\prime-1) + j^\prime}] \\
        & \explain{(b)}{=} \sum_{j=1}^T  \sum_{s=1}^{t \wedge t^\prime} P_{tj}  Q_{ts} \cdot P_{t^\prime j}  Q_{t^\prime s} \\
        & \explain{(c)}{=} (P P \tran)_{tt^\prime} (Q Q \tran)_{tt^\prime} \explain{(d)}{=} (\Omega)_{tt^\prime} (\Phi_T)_{ss^\prime} \explain{(e)}{=} (\Sigma_T)_{tt^\prime}.
    \end{align*}
    In the above display step (b) follows from the fact that $\serv{V}_{1:T^2}$ are i.i.d. $\gauss{0}{1}$ random variables and step (c) follows from the fact that $P,Q$ are lower triangular matrices, step (d) uses \eqref{eq:matrix-ortho} and \eqref{eq:nonlin-ortho}  and step (e) follows from the formula for the Gaussian state evolution covariance given in \eqref{eq:SE-VAMP}. In particular,
    \begin{align}\label{eq:W-deq-Z}
        (\serv{W}_1, \dotsc, \serv{W}_T) \explain{d}{=} (\serv{Z}_1, \dotsc, \serv{Z}_T),
    \end{align}
    where $(\serv{Z}_1, \dotsc, \serv{Z}_T)$ are the Gaussian state evolution random variables associated with the original (non-orthogonal) VAMP algorithm \eqref{eq:VAMP-non-ortho}. Hence, for any $\tau, \tau^\prime \in [T^2]$ with representations $\tau = T(t-1) + i$ and $\tau^\prime = T(t^\prime - 1) + i^\prime$:
    \begin{align}
         \E[g_\tau(\serv{V}_1, \dotsc, \serv{V}_{\tau-1}; \serv{A})  g_{\tau^\prime}(\serv{V}_1, \dotsc, \serv{V}_{\tau^\prime-1}; \serv{A})] &= \E[ \hat{\nonlin}_t(\serv{Z}_1, \dotsc, \serv{Z}_{t-1}; \serv{A}) \hat{\nonlin}_{t^\prime}(\serv{Z}_1, \dotsc, \serv{Z}_{t^\prime-1}; \serv{A})] \nonumber\\&\explain{\eqref{eq:nonlin-ortho}}{=} \begin{cases} 1 &: t^\prime = t \\
    0&: t^\prime \neq t\end{cases} \label{eq:g-ortho}
    \end{align}
    Hence, each of the conditions required by the orthogonality assumption (\sassumpref{orthogonality}) are met by \eqref{eq:Q-ortho} and \eqref{eq:g-ortho}.
    \item Finally, we verify the non-linearities $g_{1:T^2}$ are divergence-free in the sense of \assumpref{div-free}. Since the iteration \eqref{eq:final-ortho-VAMP} satisfies the orthogonality conditions (\sassumpref{orthogonality}), the state evolution random variables associated with \eqref{eq:final-ortho-VAMP} are i.i.d. $\gauss{0}{1}$, and can be taken as $\serv{V}_1, \dotsc, \serv{V}_{T^2}$. To verify the divergence-free condition, we compute for any $\tau^\prime < \tau$ with representations $\tau = T(t-1) + i$ and $\tau^\prime = T(t^\prime - 1) + i^\prime$
    \begin{align*}
        \E[g_\tau(\serv{V}_1, \dotsc, \serv{V}_{\tau-1}; \serv{A})\cdot  V_{\tau^\prime}] &\explain{\eqref{eq:lin-map-gaussian}}{=} \E[  \hat{\nonlin}_t(\serv{W}_1, \dotsc, \serv{W}_{t-1} ; \serv{A}) \cdot V_{\tau^\prime}] \\ &\explain{(f)}{=} \E[  \hat{\nonlin}_t(\serv{W}_1, \dotsc, \serv{W}_{t-1} ; \serv{A}) \cdot \E[V_{\tau^\prime} | \serv{W}_1, \dotsc, \serv{W}_{t-1}]] \\
        & \explain{(g)}{=} \sum_{s=1}^{t-1} \iter{\alpha}{\tau^\prime,t-1}_s \cdot  \E[  \hat{\nonlin}_t(\serv{W}_1, \dotsc, \serv{W}_{t-1} ; \serv{A}) \cdot \serv{W}_s] \\ 
        & \explain{\eqref{eq:W-deq-Z}}{=} \sum_{s=1}^{t-1} \iter{\alpha}{\tau^\prime,t-1}_s \cdot  \E[  \hat{\nonlin}_t(\serv{Z}_1, \dotsc, \serv{Z}_{t-1} ; \serv{A}) \cdot \serv{Z}_s] \\
        & \explain{\eqref{eq:nonlin-ortho}}{=} \sum_{s=1}^{t-1} \sum_{s^\prime = 1}^t \iter{\alpha}{\tau^\prime,t-1}_s \cdot (Q^{-1})_{ts^\prime} \cdot \E[\nonlin_{s^\prime}(\serv{Z}_1, \dotsc, \serv{Z}_{s^\prime-1} ; \serv{A}) \cdot \serv{Z}_s] \\
        & \explain{(h)}{=} \sum_{s=1}^{t-1} \sum_{s^\prime = 1}^t \sum_{r=1}^{s^\prime - 1} \iter{\alpha}{\tau^\prime,t-1}_s \cdot (Q^{-1})_{ts^\prime} \cdot \iter{\beta}{s,s^\prime-1}_r \cdot  \E[\nonlin_{s^\prime}(\serv{Z}_1, \dotsc, \serv{Z}_{s^\prime-1} ; \serv{A}) \cdot \serv{Z}_r] \\
        & \explain{(i)}{=} 0.
    \end{align*}
    In the above display, (f) follows from the tower property and the fact that $\serv{V}_{1:T^2}, \serv{W}_{1:T}$ are independent of $\serv{A}$, (g) follows from the fact that since $(\serv{W}_{1:t-1}, \serv{V}_{\tau^\prime})$ is a Gaussian vector $\E[V_{\tau^\prime} | \serv{W}_1, \dotsc, \serv{W}_{t-1}]]$ is a linear combination of $\serv{W}_{1:t-1}$ (the precise formula for the coefficients $\iter{\alpha}{\tau^\prime,t-1}_{1:t-1}$ is not needed). An identical argument was used in step (h). Finally equality (i) follows because the non-linearities $\nonlin_{1:T}$ are assumed to be divergence-free with respect to $\serv{Z}_{1:T}$.  
\end{enumerate}
Hence, we have verified that the VAMP algorithm in \eqref{eq:final-ortho-VAMP} satisfies the orthogonality assumption (\sassumpref{orthogonality}), balanced semi-random ensemble assumption (\sassumpref{balanced}) in addition to all the assumptions required by \thref{VAMP}. Consequently,
\begin{align*}
    (\iter{\vv}{1}, \dotsc, \iter{\vv}{T^2}, \auxmat) \explain{\text{\pw}}{\longrightarrow} (\serv{V}_1, \dotsc, \serv{V}_{T^2}, \serv{A}). 
\end{align*}
Using \eqref{eq:nonortho-to-ortho}:
\begin{align*}
 (\iter{\vz}{1}, \dotsc, \iter{\vz}{T}, \auxmat) &\explain{\text{\pw}}{\longrightarrow} ( \linmap_1(\serv{V}_1, \dotsc, \serv{V}_{T}) \dotsc, \linmap_T(\serv{V}_1, \dotsc, \serv{V}_{T^2}), \serv{A}) \\
 & \explain{\eqref{eq:lin-map-gaussian},\eqref{eq:W-deq-Z}}{=}  ( \serv{Z}_1, \dotsc,\serv{Z}_T, \serv{A}).
\end{align*}
This proves the claim of \lemref{orthogonalization}. 
\end{proof}

\section{Miscellaneous Results}
In this section, we collect the proofs of some miscellaneous results used in the paper. 
\subsection{Proof of \lemref{iid}}\label{appendix:iid}
This section provides a proof of \lemref{iid}, which identifies the universality class corresponding to left linear transformations of i.i.d. matrices. 
\begin{proof}[Proof of \lemref{iid}] Observe that since the entries of $\mZ$ are distributionally symmetric, $\mX_{\mathtt{tiid}} \explain{d}{=} \mT \mZ \mS$ where $\mS$ is a uniformly random sign diagonal matrix. Hence, it suffices to verify that $\mJ \explain{def}{=} \mT \mZ $ satisfies the requirements of \defref{univ-class}. In order to do so, we rely on three random matrix theory results on i.i.d. matrices:
\begin{enumerate}
    \item \citet{yin1986limiting} has shown that for any $k \in \N$:
    \begin{align*}
        \Tr[(\mJ \tran \mJ)^k]/\dim \explain{a.s.}{\rightarrow} \int \lambda^k  \; \pi \boxtimes \mu_{\mathrm{MP}}^\alpha(\diff \lambda),
    \end{align*}
    as required by \defref{univ-class}. 
    \item \citet{bai2008limit} have shown that $\lambda_{\max}(\mZ \tran \mZ) \explain{a.s.}{\rightarrow} R_{}(\alpha)$ where  $R(\alpha) \explain{def}{=}  \sqrt{\alpha} + 1/{\sqrt{\alpha}} + 2$. Consequently, the event:
    \begin{align}\label{eq:ll-good-event-A}
        \mathcal{A}_{\dim} \explain{def}{=} \left\{ \|\mJ\|_{\op} < \sqrt{2 K R(\alpha)} \right\},
    \end{align}
    satisfies:
    \begin{align}\label{eq:iid-rmt-io-events-A}
          \P( {\mathcal{A}}^c_\dim \text{ occurs infinitely often}) & = 0.
    \end{align}
    In \eqref{eq:iid-rmt-io-events-A}, $K$ is any upper bound on $\|\mT\|_{\op}^2$ (independent of $\dim$). In particular, $\|\mJ\|_{\op} \lesssim 1$ with probability $1$, as required by \defref{univ-class}. 
    \item Finally, \citet[Theorem 3.16 and Remark 3.17]{knowles2017anisotropic} (see also \citep[Remark 2.6 and Remark 2.7]{alex2014isotropic}) have obtained a local law for the resolvent $ (\mJ \tran \mJ -z \mI_{\dim} )^{-1}$ which identifies a function $m_{\dim} : \C  \mapsto \C$ such that for any $\epsilon > 0$ (independent of $\dim$), the event:
    \begin{align}\label{eq:ll-good-event-B}
     \mathcal{B}_{\dim}(\epsilon) \explain{def}{=} \Bigg\{ \sup_{\substack{z \in \C \\ |z| = 2KR_{}(\alpha)}} \left\| (\mJ \tran \mJ -z \mI_{\dim} )^{-1} - m_{\dim}(z) \cdot \mI_{\dim} \right\|_{\infty} \leq \dim^{-1/2+\epsilon} \Bigg\}.
\end{align}
satisfies $\P({\mathcal{B}}_{\dim}(\epsilon)^c) \leq C(\epsilon) \cdot \dim^{-2}$ for some constant $C(\epsilon)$ that is determined by $\epsilon$. Hence:
\begin{align} \label{eq:iid-rmt-io-events-B}
  \P( {\mathcal{B}}_\dim(\epsilon)^c\text{ occurs infinitely often})  = 0 \quad \forall \; \epsilon > 0.
\end{align}
The exact formula for $m_{\dim}(z)$ will not be important for our argument. 
\end{enumerate}
We claim that on the event $\mathcal{A}_\dim \cap \mathcal{B}_\dim({\epsilon})$ we have:
\begin{align}\label{eq:goal-iid-is-semirandom}
    \bigg\| (\mJ \tran \mJ)^k -  \frac{\Tr[ (\mJ \tran \mJ)^k]}{\dim} \cdot \mI_{\dim} \bigg\|_{\infty} \leq 2^{k+2} \cdot R_{}(\alpha)^{k+1} \cdot K^{k+1}\cdot \dim^{-1/2 + \epsilon} \quad \forall \; k \in \N.
\end{align}
To prove this claim, let $\{(\lambda_i, \vu_i): i \in [\dim]\}$ denote the eigenvalue-eigenvector pairs for $\mJ \tran \mJ$. By Cauchy's integral formula, on the event $\mathcal{A}_{\dim}$, we can write:
\begin{align*}
    (\mJ \tran \mJ)^k & = \sum_{j=1}^\dim \vu_j \vu_j \tran \cdot  \lambda_j^k = \sum_{j=1}^\dim \vu_j \vu_j \tran \cdot \frac{1}{2\pi \mathrm{i}} \oint_{|z|=2KR_{}(\alpha)}  \frac{z^k}{z-\lambda_j} \diff z = - \frac{1}{2\pi \mathrm{i}}\oint_{|z|=2KR_{}(\alpha)} (\mJ \tran \mJ -z \mI_{\dim} )^{-1}  \cdot z^k \diff z,
\end{align*}
where $\mathrm{i} = \sqrt{-1}$. Hence, on $\mathcal{A}_\dim \cap \mathcal{B}_\dim({\epsilon})$, we have the upper bound:
\begin{align*}
      \bigg\| (\mJ \tran \mJ)^k -  \frac{\Tr[ (\mJ \tran \mJ)^k]}{\dim} \cdot \mI_{\dim} \bigg\|_{\infty}  & = \bigg\|  \frac{1}{2\pi \mathrm{i}}\oint_{|z|=2KR_{}(\alpha)} z^k \cdot \bigg( (\mJ \tran \mJ -z \mI_{\dim} )^{-1} - \frac{\Tr[(\mJ \tran \mJ -z \mI_{\dim} )^{-1}]}{\dim} \cdot \mI_{\dim} \bigg)  \diff z \bigg\|_{\infty} \\
     & \leq   \frac{1}{2\pi}\oint_{|z|=2KR_{}(\alpha)} |z|^k \cdot\bigg\|  (\mJ \tran \mJ -z \mI_{\dim} )^{-1} - \frac{\Tr[(\mJ \tran \mJ -z \mI_{\dim} )^{-1}]}{\dim} \cdot \mI_{\dim} \bigg\|_{\infty}  \diff z.
\end{align*}
Observe that on the event $\mathcal{B}_\dim(\epsilon)$:
\begin{align*}
    &\bigg\|  (\mJ \tran \mJ -z \mI_{\dim} )^{-1} - \frac{\Tr[(\mJ \tran \mJ -z \mI_{\dim} )^{-1}]}{\dim} \cdot \mI_{\dim} \bigg\|_{\infty} \\& \hspace{4cm} \leq \bigg\|  (\mJ \tran \mJ -z \mI_{\dim} )^{-1} - m_{\dim}(z) \cdot \mI_{\dim} \bigg\|_{\infty} + \left| m_{\dim}(z)  - \frac{\Tr[(\mJ \tran \mJ -z \mI_{\dim} )^{-1}]}{\dim}
    \right| \\
    & \hspace{4cm} \leq 2\bigg\|  (\mJ \tran \mJ -z \mI_{\dim} )^{-1} - m_{\dim}(z) \cdot \mI_{\dim} \bigg\|_{\infty} \leq 2 \dim^{-1/2 + \epsilon}.
\end{align*}
Hence,
\begin{align*}
     \bigg\| (\mJ \tran \mJ)^k -  \frac{\Tr[ (\mJ \tran \mJ)^k]}{\dim} \cdot \mI_{\dim} \bigg\|_{\infty} & \leq 2^{k+2} \cdot R_{}(\alpha)^{k+1} \cdot K^{k+1}\cdot \dim^{-1/2 + \epsilon}. 
\end{align*}
which proves the claim \eqref{eq:goal-iid-is-semirandom}. Combining \eqref{eq:goal-iid-is-semirandom} with \eqref{eq:iid-rmt-io-events-A} and \eqref{eq:iid-rmt-io-events-B} we obtain:
\begin{align*}
    \P\left( \bigg\| (\mJ \tran \mJ)^k -  \frac{\Tr[ (\mJ \tran \mJ)^k]}{\dim} \cdot \mI_{\dim} \bigg\|_{\infty} \lesssim  \dim^{-1/2 + \epsilon} \quad \forall \; k \in \N \right) & = 1 \quad \forall \; \epsilon > 0.
\end{align*}
Taking a union bound over $\epsilon \in \mathbb{Q}$ (the set of rationals) we obtain:
\begin{align*}
    \P\left( \bigg\| (\mJ \tran \mJ)^k -  \frac{\Tr[ (\mJ \tran \mJ)^k]}{\dim} \cdot \mI_{\dim} \bigg\|_{\infty} \lesssim  \dim^{-1/2 + \epsilon} \quad \forall \; k \in \N, \; \epsilon > 0 \right) & = 1.
\end{align*}
Hence, $\mJ$ satisfies the requirements of \defref{univ-class} with probability $1$. This concludes the proof of this lemma. 
\end{proof}

\subsection{Concentration Inequality for Random Permutations}\label{appendix:permutation}
This section provides the statement of the concentration inequality of \citet{bercu2015concentration} for random permutations, which was used in the proof of \lemref{sign-perm-inv-ensemble}. 
\begin{fact}[{\citet[Theorem 4.3]{bercu2015concentration}}]\label{fact:permutation} Let $\mM$ be a random matrix with eigen-decomposition $\mM = \mO \mP \mLambda \mP \tran \mO \tran$ where:
\begin{enumerate}
    \item $\mLambda = \diag(\lambda_1, \dotsc, \lambda_{\dim})$ is a deterministic diagonal matrix.
    \item $\mO \in \mathbb{O}(\dim)$ is a deterministic orthogonal matrix. 
    \item $\mP$ is a uniformly random $\dim \times \dim$ permutation matrix. 
\end{enumerate}
Then, there is a universal constant $K$ such that:
\begin{align*}
    \P \left( \|\mM - \E[\mM]\|_{\infty} >  K \cdot \|\mO\|_{\infty}^2  \cdot \|\mLambda\|_{\op} \cdot (\sqrt{\dim \ln(\dim)}  +  \ln(\dim))\right) & \leq 4/\dim^2.
\end{align*}
\end{fact}
\begin{proof} {\citet[Theorem 4.3]{bercu2015concentration}} have shown that given an array $\mA = (A_{i,j})_{i,j \in [\dim]}$, the permutation statistic:
\begin{align} \label{eq:permutation-statistic}
    T(\mA) \explain{def}{=} \sum_{\ell=1}^ \dim A_{\ell,\tau(\ell)},
\end{align}
constructed using a uniformly random permutation $\tau: [\dim] \mapsto [\dim]$ satisfies the concentration estimate:
\begin{align} \label{eq:bercu-concentration}
    \P\left( |T(\mA) - \E[T(\mA)] | > C \|\mA\|_{\infty} \cdot (\sqrt{\dim t} + t) \right) & \leq 4 e^{-t} \quad \forall \; t \geq 0,
\end{align}
for some explicit, universal constant $C$. We observe that if $\tau$ is the random permutation corresponding to the permutation matrix $\mP$, the entries of $\mM$ can be expressed as permutation statistics of the form \eqref{eq:permutation-statistic}. Indeed for any $i,j \in [\dim]$
\begin{align*}
    \mM_{ij} & = \sum_{\ell = 1}^\dim O_{i\ell} O_{j\ell} \lambda_{\tau(\ell)}  = T(\iter{\mA}{ij})
\end{align*}
where entries of $\iter{\mA}{ij}$ are given by:
\begin{align*}
    \iter{A}{ij}_{\ell,k} & \explain{def}{=} O_{i\ell} O_{j\ell} \lambda_k \quad \forall \; \ell,k \in [\dim].
\end{align*}
Hence applying the concentration inequality \eqref{eq:bercu-concentration} with $t = 4 \ln(\dim)$ we obtain:
\begin{align*}
     \P\left( |M_{ij} - \E[M_{ij}] | > C \cdot \|\mO\|_{\infty}^2  \cdot \|\mLambda\|_{\op} \cdot (2\sqrt{\dim \ln(\dim)} + 4\ln(\dim)) \right) & \leq 4/\dim^4 \quad \forall i,j \; \in \; [\dim].
\end{align*}
Now, taking a union bound over $i,j \in [\dim]$ immediately gives us the claimed concentration bound. 
\end{proof}

\subsection{Polynomial Approximation}\label{appendix:approximation}
This appendix is devoted to the proof of \lemref{low-degree-approx} from \sref{poly-approx}. 
\begin{proof}[Proof of \lemref{low-degree-approx}] Note that the assumption:
\begin{align*}
    \E \nonlin_i( \serv{Z}_1, \dotsc, \serv{Z}_k; \serv{A})\nonlin_j( \serv{Z}_1, \dotsc, \serv{Z}_k; \serv{A}) & \in \{0,1\},
\end{align*}
guarantees that any two function $\nonlin_i, \nonlin_j$ are either identical or orthogonal. Consequently, we need to only construct approximations for collection of orthogonal functions among $\nonlin_{1:k}$. Hence without loss of generality, we may assume that:
\begin{align}\label{eq:poly-approx-wlog}
     \E \nonlin_i( \serv{Z}_1, \dotsc, \serv{Z}_k; \serv{A})\nonlin_j( \serv{Z}_1, \dotsc, \serv{Z}_k; \serv{A}) & = 0  \; \forall \; i  \neq j.
\end{align}
We begin by describing the construction of the approximating functions. For any $\auxvec \in \R^{\auxdim}$, consider the Hermite decompositions of the functions $\nonlin_{1:k}( \cdot ; \auxvec)$ and $h( \cdot; \auxvec)$:
\begin{subequations}\label{eq:poly-approx-hermite-decomp}
\begin{align} 
    h(z; \auxvec) & = \sum_{r \in \W^k} \alpha_r(\auxvec) \cdot \hermite{r}(z), \quad \alpha_r(\auxvec) \explain{def}{=} \E\left[h(\serv{Z}_1, \dotsc, \serv{Z}_k; \auxvec) \cdot \prod_{i=1}^k \hermite{r_i}(\serv{Z}_i)\right],  \\
    \nonlin_{i}(z ; \auxvec) & = \sum_{r \in \W^k} c_{i,r}(\auxvec) \cdot \hermite{r}(z), \quad c_{i,r}(\auxvec) \explain{def}{=} \E\left[\nonlin_i(\serv{Z}_1, \dotsc, \serv{Z}_k; \auxvec) \cdot \prod_{i=1}^k \hermite{r_i}(\serv{Z}_i)\right].
\end{align}
In the above display, $\{\hermite{r}: r \in \W^\order\}$ denote the collection of $\order$-variate orthonormal Hermite polynomials. 
\end{subequations}
For each $D \in \N$, we define the low-degree approximations:
\begin{subequations}\label{eq:poly-approx-lowdegree-def}
\begin{align}
    h^{\leq D}(z; \auxvec) & \explain{def}{=} \sum_{\substack{r \in \W^k \\ \|r\|_1 \leq D}} \alpha_r(\auxvec) \cdot \hermite{r}(z), \\
    \nonlin_{i}^{\leq D}(z ; \auxvec) & \explain{def}{=} \sum_{\substack{r \in \W^k \\ \|r\|_1 \leq D}} c_{i,r}(\auxvec) \cdot \hermite{r}(z).
\end{align}
\end{subequations}
The regularity assumptions (specifically, polynomial growth) imposed on $h, \nonlin_{1:k}$ guarantee that $\E[\nonlin_i^2( \serv{Z}_1, \dotsc, \serv{Z}_k ; \serv{A})]$ and $\E[h^2(\serv{Z}_1, \dotsc, \serv{Z}_k; \serv{A})]$ are finite. Hence,
\begin{align} 
    \lim_{D \rightarrow \infty} \E[(\nonlin_i^{\leq D}( \serv{Z}_1, \dotsc, \serv{Z}_k ; \serv{A}) - \nonlin_i( \serv{Z}_1, \dotsc, \serv{Z}_k ; \serv{A}))^2] &= 0, \label{eq:poly-approx-nonlin-conv}\\
     \lim_{D \rightarrow \infty} \E[(h^{\leq D}( \serv{Z}_1, \dotsc, \serv{Z}_k ; \serv{A}) - h( \serv{Z}_1, \dotsc, \serv{Z}_k ; \serv{A}))^2] &= 0 \label{eq:poly-approx-test-conv}
\end{align}
As a consequence, the matrix $Q^{\leq D} \in \R^{k \times k}$ with entries defined as:
\begin{align}\label{eq:Q-low-degree}
    Q_{ij}^{\leq D} \explain{def}{=} \E[\nonlin_i^{\leq D}( \serv{Z}_1, \dotsc, \serv{Z}_k ; \serv{A}) \cdot \nonlin_j^{\leq D}( \serv{Z}_1, \dotsc, \serv{Z}_k ; \serv{A})]
\end{align}
also satisfies:
\begin{align}\label{eq:poly-approx-Q-conv}
    Q^{\leq D} \rightarrow I_k \quad \text{as } D \rightarrow \infty.
\end{align}
For each $i \in [k]$, define the following sequence of orthogonalized functions (indexed by $D$) as:
\begin{align}\label{eq:poly-approx-whitening}
    \hat{\nonlin}_i^{\leq D}(z_1, \dotsc, z_k ; \auxvec) \explain{def}{=} \sum_{j=1}^k \left[\left(Q^{\leq D} \right)^{-\frac{1}{2}}\right]_{ij} \cdot {\nonlin}_j^{\leq D}(z_1, \dotsc, z_k ; \auxvec).
\end{align}
Note that \eqref{eq:poly-approx-Q-conv} implies that $Q^{\leq D}$ is invertible for large enough $\degree$ and hence, $\left(Q^{\leq D} \right)^{-\frac{1}{2}}$ in \eqref{eq:poly-approx-whitening} is well-defined for large $\degree$. As a consequence of \eqref{eq:poly-approx-nonlin-conv} and \eqref{eq:poly-approx-Q-conv}:
\begin{align}\label{eq:poly-approx-hat-nonlin-conv}
     \lim_{D \rightarrow \infty} \E[(\hat{\nonlin}_i^{\leq D}( \serv{Z}_1, \dotsc, \serv{Z}_k ; \serv{A}) - \nonlin_i( \serv{Z}_1, \dotsc, \serv{Z}_k ; \serv{A}))^2] &= 0
\end{align}
In light of \eqref{eq:poly-approx-test-conv} and \eqref{eq:poly-approx-hat-nonlin-conv}, for any $\epsilon \in (0,1)$ we can find $D_\epsilon \in \N$ such that:
\begin{align*}
    \max_{i \in [k]}\E[(\hat{\nonlin}_i^{\leq D_\epsilon}( \serv{Z}_1, \dotsc, \serv{Z}_k ; \serv{A}) - \nonlin_i( \serv{Z}_1, \dotsc, \serv{Z}_k ; \serv{A}))^2] & \leq \epsilon^2, \\
    \E[(h^{\leq D_\epsilon}( \serv{Z}_1, \dotsc, \serv{Z}_k ; \serv{A}) - h( \serv{Z}_1, \dotsc, \serv{Z}_k ; \serv{A}))^2] & \leq \epsilon^2. 
\end{align*}
This gives us the desired approximating functions $\hat{\nonlin}_{1:k}^\epsilon \explain{def}{=} \hat{\nonlin}_{1:k}^{\leq D_\epsilon}$ and $\hat{h}^\epsilon \explain{def}{=} h^{D_\epsilon}$. This proves the first two claims of the lemma. We now consider each of the remaining claims.
\begin{enumerate}
\setcounter{enumi}{2}
    \item Observe that:
    \begin{align*}
        \E[\serv{Z}_i \hat{\nonlin}^\epsilon_j( \serv{Z}_1, \dotsc, \serv{Z}_k; \serv{A})]  &\explain{\eqref{eq:poly-approx-whitening}}{=} \sum_{\ell = 1}^k \left[\left(Q^{\leq D_\epsilon} \right)^{-\frac{1}{2}}\right]_{j\ell} \cdot  \E[\serv{Z}_i \cdot {\nonlin}_{\ell}^{\leq D_{\epsilon}}( \serv{Z}_1, \dotsc, \serv{Z}_k; \serv{A})] \\&\explain{\eqref{eq:poly-approx-lowdegree-def}}{=} \sum_{\ell = 1}^k \left[\left(Q^{\leq D_\epsilon} \right)^{-\frac{1}{2}}\right]_{j\ell} \E[\serv{Z}_i \cdot {\nonlin}_{\ell}( \serv{Z}_1, \dotsc, \serv{Z}_k; \serv{A})] \\
        & = 0,
    \end{align*}
    where the last step follows from the assumption $\E[\serv{Z}_i \cdot {\nonlin}_{\ell}( \serv{Z}_1, \dotsc, \serv{Z}_k; \serv{A})] = 0$ made in the statement of the lemma. Using \eqref{eq:Q-low-degree} and \eqref{eq:poly-approx-whitening}, we can also compute:
    \begin{align*}
         \E[\hat{\nonlin}^\epsilon_i( \serv{Z}_1, \dotsc, \serv{Z}_k; \serv{A}) \cdot  \hat{\nonlin}^\epsilon_j( \serv{Z}_1, \dotsc, \serv{Z}_k; \serv{A})] 
         & = (I_k)_{ij} \explain{\eqref{eq:poly-approx-wlog}}{=}  \E[\nonlin_i( \serv{Z}_1, \dotsc, \serv{Z}_k; \serv{A})\nonlin_j( \serv{Z}_1, \dotsc, \serv{Z}_k; \serv{A})].
    \end{align*}
    This verifies the third claim made in the lemma.
    \item Recall the definition of the coefficients $\alpha_r(\auxvec)$ and $c_{i,r}(\auxvec)$ from \eqref{eq:poly-approx-hermite-decomp}. Since the functions $\nonlin_{1:k}, h$ are assumed to be continuous and polynomially bounded, by the Dominated Convergence Theorem, $\alpha_r(\cdot)$ and $c_{i,r}(\cdot)$ are continuous functions on $\R^{\auxdim}$. As a consequence, the approximations $\hat{h}^\epsilon$ and $\hat{\nonlin}_{1:k}^\epsilon$ defined in \eqref{eq:poly-approx-lowdegree-def} are also continuous. Furthermore observe that the coefficient  $\alpha_r(\auxvec)$ is polynomially bounded since:
    \begin{align*}
        |\alpha_r(\auxvec)|^2 & \explain{\eqref{eq:poly-approx-hermite-decomp}}{=} \left| \E\left[h(\serv{Z}_1, \dotsc, \serv{Z}_k; \auxvec) \cdot \prod_{i=1}^k \hermite{r_i}(\serv{Z}_i)\right] \right|^2 \explain{(a)}{\leq} \E[h^2(\serv{Z}_1, \dotsc, \serv{Z}_k; \auxvec) ] \explain{(b)}{\leq} 3L^2 \cdot (1 + \E\|\serv{Z}_{1:k}\|^{2\degree} + \|\auxvec\|^{2\degree}).
    \end{align*}
    In the above display, (a) follows from Cauchy-Schwarz Inequality and the orthonormality of the Hermite polynomials, (b) follows from the assumption that $h$ is polynomially bounded. The same bound applies to the coefficients $c_{i,r}(\auxvec)$. Hence, the approximations $\hat{h}^\epsilon$ and $\hat{\nonlin}_{1:k}^\epsilon$ defined in \eqref{eq:poly-approx-lowdegree-def} are polynomially bounded.
\end{enumerate}
This concludes the proof of \lemref{low-degree-approx}. 
\end{proof}

\end{document}